%% file: main.tex
\documentclass[10pt]{article}
\pdfoutput=1
\usepackage[margin=1in]{geometry}
\usepackage{amsthm,amsfonts,amssymb,amsmath,comment,slashed,mathtools, bm}
\usepackage[T1]{fontenc} 
\usepackage[utf8]{inputenc}

\usepackage{
	enumerate,		    
  	thmtools,           
	csquotes
}

\usepackage[affil-it]{authblk} 
\usepackage[hidelinks]{hyperref}    
\usepackage{thm-restate} 
\usepackage{xargs}          
\usepackage{mathrsfs} 
\usepackage{xcolor}  
\usepackage{here} 
\usepackage[capitalize]{cleveref} 
\usepackage[backend=biber,giveninits=true,doi=false,isbn=false,url=false,sorting=nyt]{biblatex} 
\renewbibmacro{in:}{} 
\bibliography{main.bib}
\allowdisplaybreaks

\newcommand\wt[1]{\widetilde{#1} }

\usepackage{subcaption}
\usepackage{standalone}
\usepackage{tikz}
\usepackage{tikz-cd}
\usetikzlibrary{
	calc,
	decorations.pathmorphing,
	hobby,
	matrix,arrows,
	positioning,
	shapes.geometric, snakes
}
\usepackage{pgfplots}
\pgfplotsset{compat=newest}

\theoremstyle{plain}
\newtheorem{dfn}{Definition} 
\newtheorem{prop}{Proposition}
\newtheorem{lem}{Lemma}
\newtheorem{cor}{Corollary}
\newtheorem*{conject}{Conjecture}
\newtheorem{rmk}{Remark}
\declaretheorem[name=Theorem]{theorem}

\makeatletter
\renewcommand{\paragraph}[1]{%
	\par 
	\addvspace{\medskipamount}
	\textit{#1\@addpunct{.}}\enspace\ignorespaces
}
\makeatother
\numberwithin{equation}{section}
\numberwithin{prop}{section}
\numberwithin{rmk}{section}
\numberwithin{lem}{section}
\numberwithin{dfn}{section}
\numberwithin{cor}{section}

\title{A construction of approximately self-similar naked singularities for the spherically symmetric Einstein-scalar field system}
\author[1]{Jaydeep Singh \thanks{jaydeeps@math.princeton.edu}}
\affil[1]{\small  Department of Mathematics, Princeton University, Washington~Road,~Princeton,~NJ~08544,~United~States~of~America \vskip.1pc \ }
\date{December 8, 2022}

\begin{document}
\maketitle
\input{abstract.tex}
\tableofcontents
\input{intro.tex}
\input{prelims.tex}
\input{sec3.tex}
\input{sec4.tex}
\input{sec5.tex}
\input{appA.tex}
\printbibliography[heading=bibintoc]
\end{document}

%% file: abstract.tex
\begin{abstract}
    In this work we investigate the stability and instability properties of a class of naked singularity spacetimes. The first rigorous study of naked singularities in the spherically symmetric Einstein-scalar field system is due to Christodoulou \cite{chris2}, who identified singularity formation in the class $(\overline{g}_k, \overline{\phi}_k)$ of \textit{$k$-self-similar} solutions, for any $k^2 \in (0,\frac{1}{3})$. Here we extend the construction to produce examples of interior and exterior regions of naked singularity spacetimes locally modeled on $k$-self-similar solutions, without requiring exact self-similarity.
    
    The main result is a global stability statement under \textit{fine-tuned} data perturbations, for a class of naked singularity spacetimes satisfying self-similar bounds. Given the well-known blueshift instability for suitably regular naked singularities in the Einstein-scalar field model, we require non-generic conditions on the data perturbations. In particular, the scalar field perturbation along the past lightcone of the singular point $\mathcal{O}$ vanishes to high order near $\mathcal{O}$. Technical difficulties arise from the singular behavior of the background solution, as well as regularity considerations at the axis and past lightcone of the singularity. The interior region is constructed via a backwards stability argument, thereby avoiding activating the blueshift instability. The extension to the exterior region is treated as a global existence problem to the future of $\mathcal{O}$, adapting techniques developed for vacuum spacetimes in \cite{igoryak2}.
\end{abstract}

%% file: intro.tex
\section{Introduction}
As examples of singular solutions arising from regular, asymptotically flat initial data, naked singularities are model cases for studying the formation of gravitational singularities. In this paper we restrict attention to the spherically symmetric setting, and consider solutions to the Einstein-scalar field system:
\begin{equation}
\begin{cases}
\textbf{Ric}_{\mu\nu}-\frac{1}{2}\textbf{R}\textbf{g}_{\mu\nu} = 2\textbf{T}_{\mu\nu}[\bm{\phi}], \\[.3em]
\Box_{\textbf{g}_{\mu\nu}}\bm{\phi} = 0.
\end{cases}
\label{ref:eqneqn}
\end{equation}
The main results of this paper establish the existence and stability under \textit{finely-tuned} initial data perturbations of a class of naked singularity solutions to (\ref{ref:eqneqn}) satisfying self-similar bounds.

The problem of singularity formation arises already in the simplest case of spherically-symmetric vacuum spacetimes, as the famous Schwarzschild solution illustrates. Here the singularity assumes the form of a spacelike $\{r=0\}$ boundary, across which curvature invariants blowup and the solution fails to be extendible as a regular spacetime. Despite this breakdown in regularity, the presence of the event horizon decouples the singularity from the causal futures of observers who remain in the black hole exterior region. In particular, \enquote{infinitely far away} observers along $\mathcal{I}^+$ exist for infinite proper time.

\begin{figure}[h]
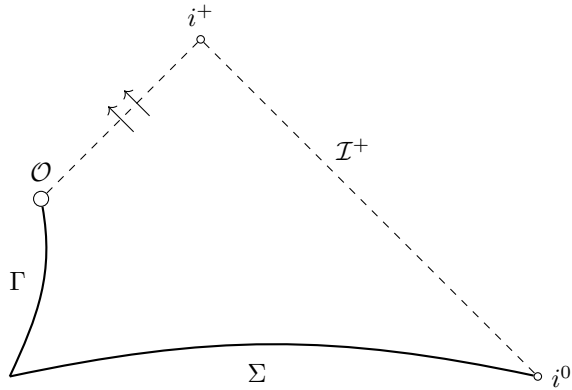

    \centering
  \includestandalone[]{Figures/fig_intronakedsing}
  \caption{Formation of a globally naked singularity from regular Cauchy data on $\Sigma$. Observers (represented by arrows) leave the maximal development associated to data, despite not encountering a breakdown in regularity.}
  \label{fig0}
\end{figure}
In contrast, the spacetimes considered in this paper contain an incomplete $\mathcal{I}^+$ (see Figure \ref{fig0}). Imagined observers along $\mathcal{I}^+$ thus reach the future lightcone of the singularity $\mathcal{O}$ in finite proper time, rendering the singularity \enquote{naked}. The region of spacetime uniquely determined by initial data abruptly comes to an end for these distant observers, even though they see no local breakdown in regularity.

The conceptual issues raised by these pathological spacetimes were taken up by Penrose in \cite{penroseWCC}, leading to the conjecture that naked singularity formation should be \textit{non-generic}. In practice, such a statement must be accompanied by a choice of matter model, a class of initial data and solutions, and a notion of genericity. Setting aside these issues, one form of this so-called weak cosmic censorship (WCC) conjecture, adapted from \cite{chris5}, is the following:
\begin{conject}[Weak Cosmic Censorship]
The maximal development of generic regular, asymptotically flat, initial data admits a complete $\mathcal{I}^+$.
\end{conject}
For the spherically symmetric Einstein-scalar field system, a positive resolution of (WCC) was given by Christodoulou in the series of works \cite{chris1}, \cite{chris3}. The foundational work \cite{chris1} identified limited-regularity solution classes in which well-posedness and continuation criteria hold. These classes provide enough flexibility to construct unstable perturbations, and \cite{chris3} established the existence of an unstable two-dimensional plane of directions in the space of initial data for suitably regular naked singularities\footnote{More precisely, \cite{chris3} shows a co-dimension $1$ instability result for absolutely continuous (AC) naked singularities. In the more general bounded variation (BV) setting, he identifies an unstable two-dimensional plane of perturbuations, subject to an infinite blueshift assumption along the past lightcone of the singularity.}. Note that the \enquote{genericity} assumption in the statement of (WCC) was shown to be necessary in the work \cite{chris2}, which provided examples of absolutely continuous naked singularity spacetimes. See also the work \cite{liuli} providing a more robust instability proof.

A key role in the argument is played by the dynamics along the backwards light cone of $\mathcal{O}$, given by $\{v=0\}$ in appropriate double null coordinates. Along this null surface the evolution equations for transversal double null unknowns (e.g. $\partial_v \phi$) satisfy ordinary differential equations (ODEs), with coefficients and source terms determined by the quantities intrinsic to the surface. An important discovery of \cite{chris3} was that such ODEs exhibit a blueshift instability. For generic perturbations of the value of $\partial_v \phi$ at a single point, say $(u,v) = (-1,0)$, the instability leads to rapid growth of the scalar field and trapped surface formation.

An additional feature of the proof in \cite{chris3} is that the unstable perturbations are supported on the exterior region of the naked singularity (i.e. $\{v \geq 0\}$). As a result, the argument says little about the stability properties of the interior. Some questions left unanswered by the analysis of \cite{chris3} include: 
\begin{enumerate}[(i)]
    \item Are there examples of naked singularity interiors that are \enquote{stable} in a suitable sense, either generically or up to perturbations of finite codimension?
    \item Can the support of the unstable perturbations identified by \cite{chris3} be extended into the region $\{v=0\}$?
    \item Under generic perturbations, is the existence of a central singularity stable? Or can appropriate perturbations de-singularize the solution?
\end{enumerate}
This work addresses a simpler question of \textit{constructing} additional examples of naked singularity interiors. Our main result extends the construction of $k$-self-similar naked singularities in \cite{chris2} outside the class of exact self-similarity. The proof moreover guarantees the existence of stable directions in the space of initial data for the forwards problem, but does not give any estimate for the dimension of a stable manifold.

We take a perturbative approach to constructing the interior regions. Namely, given a self-similar naked singularity as in \cite{chris2}, or more generally any admissble spacetime (see Section \ref{sec:assumptionsbackground} for the full list of requirements), our first result establishes the existence of \textit{fine-tuned} data perturbations along the backwards light cone of $\mathcal{O}$, for which the naked singularity interior is backwards stable. The class of allowed data is explicit, and is non-generic in the space of all initial data.  

A preliminary statement of the backwards stability result for naked singularity interiors is given below. For precise statements of all results, see Section \ref{sec:thmstatements}.

\begin{theorem}
\label{thm:introthm1}
Let $(\mathcal{Q}^{(in)}, \overline{g}, \overline{\phi})$ be a given admissible\footnote{Roughly, admissible spacetimes admit an appropriate double null gauge, lie in a fixed regularity class, and satisfy quantitative self-similar bounds. For details, see Definition \ref{dfn:admissibility}.} naked singularity interior. Choose data for the scalar field along the past light cone of $\mathcal{O}$ satisfying
\begin{equation}
    \label{eqn:introtemp3}
    \partial_u \phi(u,0) = \partial_u \overline{\phi}(u,0) + \epsilon f_0(u),
\end{equation}
where $f_0(u) = O(|u|^{\alpha-1})$ as $u \rightarrow 0$, and where $\alpha$ is a large constant depending on the background spacetime. See Figure \ref{fig:0.1}a).

There exists an $\epsilon_0$ depending on the background solution and $f_0$ such that for $\epsilon \leq \epsilon_0$, there exists a bounded variation solution to the spherically symmetric Einstein-scalar field system in $\mathcal{Q}^{(in)}$ achieving the data for $\partial_u \phi$ along $\{v=0\}$. The solution satisfies appropriate gauge and regularity conditions, as well as self-similar bounds. 

Finally, the solution is asymptotic to the background solution as $u \rightarrow 0$, with rates controlled by $\alpha$.
\end{theorem}

We emphasize that the property of being a naked singularity is global, and relies on the behavior of the solution in the asymptotically flat region of the exterior. It is thus necessary to verify that there exists a choice of outgoing data for the scalar field along $\{u=-1, \ v \geq 0\}$ for which the above interior regions extend to genuine naked singularities. In particular, the exterior construction is a forwards problem. Our second main result constructs the exterior region globally up to $\mathcal{I}^+$.

\begin{theorem}
\label{thm:introthm2}
Fix an admissible background spacetime $(\mathcal{Q}, \overline{g}, \overline{\phi})$. Choose data for the scalar field along $\{v=0\}$ as in (\ref{eqn:introtemp3}), and let $(\mathcal{Q}^{(in)}, g, \phi)$ denote the solution constructed by Theorem \ref{thm:introthm1}. 

Choose data for the scalar field along $\{u=-1, \ v \geq 0\}$ satisfying
\begin{equation}
    \partial_v \phi(-1,v) = \partial_v \overline{\phi}(-1,v) + \epsilon g_0(v) + O(\epsilon),
\end{equation}
where $g_0(v)$ is a free function vanishing sufficiently quickly as $v \rightarrow 0$ and $v \rightarrow \infty$. The additional $O(\epsilon)$ terms are explicitly determined by the interior solution $(\mathcal{Q}^{(in)}, g, \phi)$. See Figure \ref{fig:0.1}b).

There exists $\epsilon_0$ small depending on the background solution and $g_0$ such that for $\epsilon \leq \epsilon_0$, there exists a solution in the exterior region $\mathcal{Q}^{(ex)}$ attaining the data along $\{v=0\}\cup \{u=-1, \ v \geq 0\}$. The solution satisfies appropriate gauge and regularity conditions, self-similar bounds, and is asymptotically flat. The exterior and interior spacetimes glue across $\{v=0\}$ as a solution of bounded variation, and the resulting solution admits an incomplete $\mathcal{I}^+$.
\end{theorem}

\begin{figure}[h]
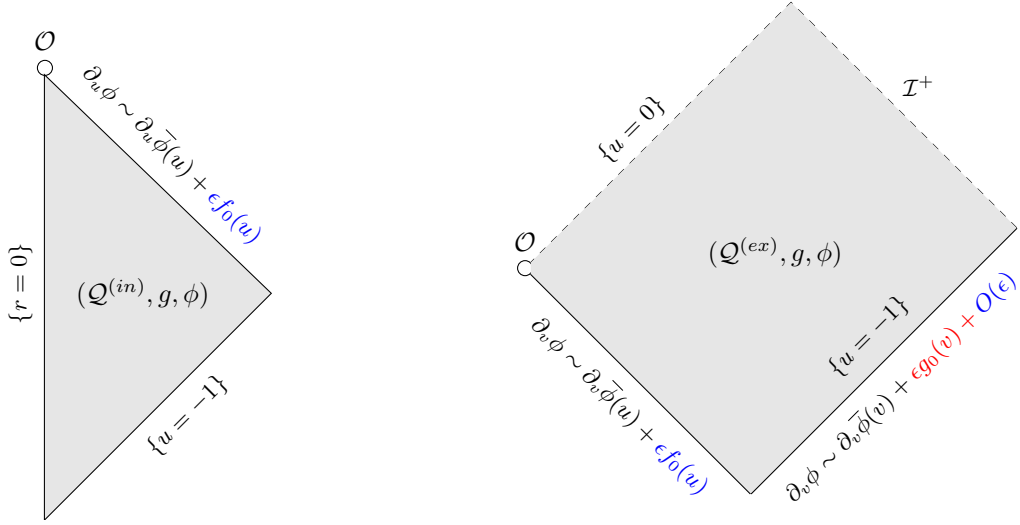

    \centering
    \begin{subfigure}[b]{.4\linewidth}
        \includestandalone[]{Figures/fig_introthm1}
    \end{subfigure}
    \begin{subfigure}[b]{.4\linewidth}
        \includestandalone[]{Figures/fig_introthm2}
    \end{subfigure}
    \caption{(a) Domain of existence for the interior solutions of Theorem \ref{thm:introthm1}. Ingoing data perturbation highlighted in blue. (b) Domain of existence for the interior solutions of Theorem \ref{thm:introthm2}. Outgoing data perturbation highlighted in red. }
    \label{fig:0.1}
\end{figure}

\subsection{Christodoulou's $k$-self-similar solutions}
\label{subsection:chrissolutions}
A detailed discussion of the solutions constructed in \cite{chris2}, including their global expression in double null gauge and sharp regularity properties, is given in Appendix \ref{appA}. In this introduction we review some motivating aspects.

\subsubsection*{Underlying $k$-self-similarity and expression in Bondi coordinates}
We begin by defining $k$-self-similarity. The spherically symmetric Einstein-scalar field system admits a two parameter symmetry group $\mathbb{R}_+ \times \mathbb{R}$, where $a \in \mathbb{R}_+, b \in \mathbb{R} $ act on solutions $(r,m,\phi)$ via
\begin{equation*}
    r \rightarrow ar, \ m \rightarrow am, \ \phi \rightarrow \phi + b.
\end{equation*}
One can also write the action on the quotient metric as 
\begin{equation*}
    g_{\mu \nu} \rightarrow a^2 g_{\mu \nu}.
\end{equation*}
Fixing a parameter $k \in \mathbb{R}$, $k$-self-similarity formalizes the notion of solutions invariant under the above scaling transformation. More precisely, $k$-self-similar solutions admit a homothetic vector field $S$, generating a one-parameter family of diffeomorphisms $f_a$ of $\mathcal{Q}$ under which the solution transforms as 
\begin{equation*}
   ( f_a^*g)_{\mu \nu} = a^2 g_{\mu\nu}, \ \ \ f_a^*r = r, \ \ \ f_a^* \phi = \phi - k\log a.
\end{equation*}
These equations in turn imply 
\begin{equation*}
    (\mathcal{L}_S g)_{\mu\nu} = 2g_{\mu\nu}, \ \ \ Sr = r, \ \ \ S\phi = -k.
\end{equation*}
The parameter $k$ is critical for the underlying mechanism of singularity formation. Roughly, nonzero $k$ corresponds to the logarithmic growth of the scalar field along the past lightcone of $\mathcal{O}$.

The construction of \cite{chris2} takes place in \textit{self-similar Bondi coordinates} $(u,r)$, where $u$ is an outgoing null coordinate and $r$ the area radius. With respect to this gauge the homothetic vector field can be written as
\begin{equation*}
    S = r\partial_r + u\partial_u.
\end{equation*}
The ansatz of $k$-self-similarity implies the following form for the metric $\overline{g}_k$ and scalar field $\overline{\phi}_k$:
\begin{align}
    \overline{g}_k = e^{2\beta(u,r)}du^2 - 2e^{\beta(u,r)+\gamma(u,r)}dudr, \label{intro:e1}
     \ \ \ \overline{\phi}_k = \chi(u,r) - k\log|u|,
\end{align}
where $\beta(u,r) = \mathring{\beta}(-\frac{r}{u}), \ \gamma(u,r) = \mathring{\gamma}(-\frac{r}{u}), \ \chi(u,r) = \mathring{\chi}(-\frac{r}{u})$ reduce to functions of a self-similar parameter $-\frac{r}{u}$.

The subclass of solutions with bounded scalar field ($k = 0$) were studied in the earlier work \cite{chris1} as prototypical examples of bounded variation solutions. These solutions are called scale-invariant, and can in fact be written explicitly in double null gauge; however, they do not model a breakdown in regularity relative to initial data, and are thus not examples of naked singularity formation.

The case with $k \neq 0$ does not lend itself to analytical expressions for the solution. Still, \cite{chris2} shows that the ansatz (\ref{intro:e1}) reduces the full Einstein-scalar field system to an autonomous 2x2 system of ordinary differential equations. A detailed phase plane analysis establishes the global properties of solutions, and for the subrange $k^2 \in (0,\frac{1}{3})$, there exist solutions that extend to the future light cone of $\mathcal{O}$. The resulting spacetimes are not asymptotically flat, as a consequence of self-similarity; however, an appropriate truncation in $\{r \gg 1\}$ remedies this issue, and leads to a naked singularity spacetime $\mathcal{Q}$. In coordinates $\mathcal{Q}$ takes the form
\begin{equation*}
    \mathcal{Q} = \{(u,r) \ | \ u \in [-1,0), \ r \in [0, \infty)\}.
\end{equation*}
We will only consider the region to causal future of a fixed outgoing null ray, e.g. $\{u=-1\}$, although extensions to past null infinity are possible as well.

\subsubsection*{Expression in renormalized double null coordinates}
For technical reasons, it is more convenient in this paper to work in a global double null coordinate system. We thus turn to recasting the existence result of \cite{chris2} in an appropriate double null gauge.

The analogous notion of self-similar double null coordinates, with respect to which $S = \hat{u}\partial_{\hat{u}} + \hat{v}\partial_{\hat{v}}$, is unfortunately not regular across the past self-similar horizon $\{v=0\}$ (see Section \ref{appa:regularcoords} for further discussion of this point). We therefore work in a \textit{renormalized double null gauge} $(u,v)$, with respect to which $\mathcal{Q}$ takes the form (see Figure \ref{fig1})
\begin{equation}
\label{eqn:introtemp10}
\mathcal{Q} = \{(u,v) \ | \ -1 \leq u < 0, \ -|u|^{1-k^2} \leq v < \infty\}.
\end{equation}
The natural self-similar parameter is $-\frac{v}{u},$ and the metric functions and scalar field assume the functional form
\begin{align}
\label{eqn:introtemp11}
    \overline{g}_k = - \mathring{\Omega}^2(-\frac{v}{u}) dudv, \ \ \quad  \overline{r}_k(u,v) =  |u|\mathring{r}(-\frac{v}{u}), \ \ \ \quad \overline{\phi}_k(u,v) = \mathring{\phi}(-\frac{v}{u}) - k\log |u|.
\end{align}

The following theorem collects the key geometric features of Christodoulou's solutions, and motivates the class of admissible spacetimes defined in Section \ref{sec:assumptionsbackground}. See Sections \ref{sec:solnclasses} - \ref{subsec:nakedsingdfn} for precise definitions of terms.
\begin{theorem}[\cite{chris2} and Appendix \ref{appA}]
\label{thm:christodoulou_solutions}
Fix $k \in \mathbb{R}$ with $k^2 \in (0,\frac{1}{3})$. There exists an asymptotically flat solution $S = (\overline{r}_k,\overline{m}_k,\overline{\phi}_k)$ to (\ref{SSESF:1:1})-(\ref{SSESF:1:5}) in the domain $\mathcal{Q}$ satisfying the following properties:
\begin{enumerate} 
    \itemsep0em
    \item There exists a global double null coordinate system with respect to which $\mathcal{Q}$ takes the form (\ref{eqn:introtemp10}). The metric and scalar field assume the form (\ref{eqn:introtemp11}).
    \item The solution $S$ has $C^1$ regularity\footnote{Recall that the singular point $\mathcal{O}$ is not included in spacetime. More precisely, statements of regularity are assumed to exclude an arbitrarily small neighborhood of $\mathcal{O}$} on $\mathcal{Q} \setminus (\{v=0\} \cup \{u=0\}).$ Moreover, $\overline{\phi}_k \in C^{1,\frac{k^2}{1-k^2}}(\mathcal{Q})$, and this regularity is sharp near $\{v=0\}$.
    \item The null derivatives of the scalar field up to second order satisfy modified self-similar bounds. In particular, 
    \begin{equation}
        \label{eq:introtemp11}
        \Big|\frac{1}{\overline{\nu}_k}\partial_u \overline{\phi}_k \Big| \lesssim |u|^{-1}, \ \ \ \Big|\frac{1}{\overline{\lambda}_k}\partial_v \overline{\phi}_k \Big| \lesssim |u|^{-1},
    \end{equation}
    where we have set $\overline{\nu}_k \doteq \partial_u \overline{r}_k, \ \overline{\lambda}_k \doteq \partial_v \overline{r}_k$. Moreover, $\partial_u \phi(u,0) = \frac{k}{u}.$ See (A2)-(A3), (B1)-(B2) in Section \ref{sec:assumptionsbackground} for additional self-similar bounds.
    \item In a self-similar neighborhood of the axis, the solution is smooth.
    \item The solution is exactly $k$-self-similar in $\{v \leq 1\}$. Moreover, the generator of self-similarity takes the form $u\partial_u + (1-k^2)v\partial_v$. The generator is timelike in $v < 0$, null along the past lightcone of $\mathcal{O}$, and spacelike in $v > 0$. 
    \item All outgoing null surfaces $\{u=c\}$ with $c < 0$ are asymptotically flat as $v \rightarrow \infty$.
    \item For $(u,v) \in \Gamma$ with $u < 0$, the Hawking mass satisfies $$\lim_{\delta \rightarrow 0} \overline{\mu}_k(u - \delta,v) = \lim_{\delta \rightarrow 0}\overline{\mu}_k(u, v+\delta) = 0.$$ Along $\{v=0\}$ and along $\{u=0\}$, $ \overline{\mu}_k(u,0)= \frac{k^2}{1+k^2}$. As a consequence, the solution does not extend\footnote{By definition of the BV class, along null lines intersecting $\Gamma$ at a regular point we must have $\mu \rightarrow 0$. The failure of BV extendibility also follows from point (3) above, as any BV extension must have $\partial_u \phi \in L^1(\{v=0\}).$ See \cite{chris1} or Section \ref{sec:solnclasses} for further discussion of solution classes.} as a BV solution in any neighborhood of $\mathcal{O}$.
    \item $\overline{r}_k(u,v)$ extends to a continuous function on $\{u=0, \ v >0\}$. Letting $\overline{r}_k(0,v)$ denote this limit, we have $\lim_{v\rightarrow \infty}\overline{r}_k(0,v) = \infty$.
    \item $\mathcal{Q}$ contains an incomplete null infinity in the sense of Definition \ref{dfn:incompleteI}.
\end{enumerate}
\end{theorem}

\subsubsection*{Instability under (rough) exterior perturbations}
We briefly turn to the instability of $k$-self-similar solutions under \textit{generic} perturbations of initial data. Without loss of generality we consider perturbations of the outgoing scalar field derivative $\partial_v \phi(-1,v)$, and work with characteristic initial data. 

As mentioned above, key to the instability are ODEs for various double null quantities along $\{v=0\}$. Self-similarity allows for explicit computations of the intrinsic geometry along $\{v=0\}$, leading to the following ODE for the scale invariant quantity $\frac{1}{\lambda}\partial_v \phi$:
\begin{equation}
    \partial_u \Big(\frac{1}{\lambda}\partial_v \phi \Big) - \frac{1+k^2}{|u|}\Big(\frac{1}{\lambda}\partial_v \phi \Big) = F(u).
\end{equation}
The source term $F(u) \sim |u|^{-2}$ is determined by the tangential derivative of the scalar field. The solutions of such singular transport equations are discussed in \cite{igoryak1}, a consequence of which is the following: for $k^2 > 0$, the generic solution to such an ODE satisfies non-self-similar bounds $\lambda^{-1}\partial_v \phi \sim |u|^{-1+k^2}$. Moreover, given any point along $\{v=0\}$, there is a \textit{unique} initial value of $\lambda^{-1}\partial_v \phi$ there that leads to the self-similar bound $|u|^{-1}$ consistent with (\ref{eq:introtemp11}) above.

A strategy for showing nonlinear instability is outlined in Figure \ref{fig0.5}. Choosing a perturbation for outgoing data $\partial_v \phi(-1,v)$ with support on $\{v \geq 0\}$, one can simultaneously arrange that 1) the geometry and scalar field along $\{v=0\}$ are unchanged, by domain of dependence arguments, and 2) the value of $\lambda^{-1}\partial_v \phi(-1,0)$ is shifted away from the unique value leading to self-similar bounds. The ODE analysis above implies the resulting solution cannot remain globally close to the background self-similar solution. We emphasize here that the perturbation thus described has a low regularity, with $\partial_v \phi (-1,v)$ experiencing a jump across $v=0$.

A priori this growth occurs only along $\{v=0\}$; however, coupled with \textit{stability} estimates for the solution in an open set containing $\{v=0\}$, one can show that the non-self-similar growth propagates to a full region near the singularity. Careful consideration of the behavior of the Hawking mass $m$ and use of the trapped surface formation result in \cite{chris1.5} implies a trapped surface must form to the future of $\{v=0\}$. As a consequence of the classification given in \cite{komm}, we conclude the spacetime has complete $\mathcal{I}^+$. Further details on the stability analysis near $\{v=0\}$ is given in \cite{liuli}.

\begin{figure}
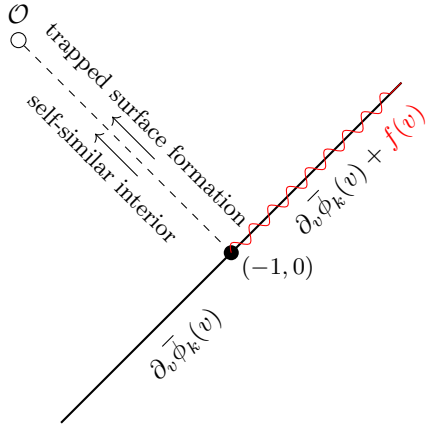

    \centering
  \includestandalone[]{Figures/fig_instability}
  \caption{Establishing non-linear instability of the $k$-self-similar solutions by applying a perturbation (red) to the scalar field supported on $\{v\geq 0\}$.}
  \label{fig0.5}
\end{figure}

\subsection{Further constructions of naked singularities}
A wide variety of naked singularity solutions have been either constructed, or numerically conjectured to exist, leading to a complicated phenomenology. Here we review several examples for different Einstein-matter systems, and comment on the relation to the self-similar examples considered in this paper.

Explicit examples are provided by over-extremal Kerr-Newman spacetimes or negative mass Schwarzschild. Work in the physics literature (see \cite{wald}) suggests that these examples may not be dynamically relevant, and that \enquote{super-charging} initially sub-extremal black holes should not be possible. In various spherically symmetric settings, this argument has been validated in \cite{komm}.

A rigorous construction of naked singularities to the Einstein equations coupled to inhomogeneous dust clouds was given by Christodoulou in \cite{chris4}. In fact, the set of initial data giving rise to such qualitative behavior is open, a violation of (WCC). Still, this behavior is believed to reflect a defective choice of matter model, and may not be representative of the situation in vacuum.

More closely connected to the construction in \cite{chris2} are \cite{anzhang}, \cite{guohadzicjang}. The work \cite{anzhang} lifts spherically symmetric, self-similar solutions of the Einstein-scalar field system with potential to generate vacuum naked singularities in $5+1$ dimensions. Similarly, \cite{guohadzicjang} applies a detailed ODE analysis in self-similarity to construct naked singularities in the $3+1$ dimensional Einstein-Euler system. We also mention the recent work \cite{serban}, studying lifts of $k$-self-similar solutions to vacuum solutions in $4+1$ dimensions. Although the aim is not to directly construct naked singularity solutions, a consequence of the analysis is the existence of suitably defined \enquote{locally} naked singularities.

The remarkable works \cite{igoryak2}-\cite{yakov1} construct the exterior and interior regions respectively of vacuum naked singularities in $3+1$ dimensions, outside of any exact symmetry. The authors identify an analog of $k$-self-similarity in the vacuum setting, and several qualitative features of Christodoulou's solutions carry over. However, the mechanism of singularity formation is unrelated to growth in a matter field (indeed, there is no scalar field in vacuum), and relies on twisting of the shift vector as the singular point is approached. In particular, the mechanism is highly non-spherically symmetric. Although the construction is beyond the scope of this introduction, we note that \cite{igoryak2} provides a very general framework for proving existence of naked singularity exteriors given data satisfying self-similar bounds. The constructions in this paper, in particular the proof of Theorem \ref{thm:introthm2}, draw heavily from the methods developed in \cite{igoryak2}.

We conclude this section by returning to the spherically symmetric Einstein-scalar field model, for which an intriguing picture of naked singularity formation has emerged in the literature surrounding critical collapse; for a review, see \cite{gundlach}. Investigations of $1$-parameter families of initial data on the threshold of dispersion and black hole formation, beginning with the numerical work of Choptuik \cite{choptuik1}, suggest the existence of universal \enquote{critical} solutions containing naked singularities. In contrast to the solutions of \cite{chris2}, which are continuously self-similar and have limited regularity, the solutions observed by Choptuik are discretely self-similar and analytic across the past self-similarity horizon, with a distinguished timescale on which the solution exhibits \enquote{echoes.} The precise relationship between Christodoulou's continuously self-similar solutions and the solutions associated to critical collapse is not yet clear.

Studying threshold solutions to general matter systems provides an alternative path to constructing naked singularities. Critical behavior and naked singularity formation have been numerically illustrated for higher dimensional vacuum spacetimes and a variety of matter models, including kinetic models and massive/charged scalar fields \cite{gundlach}. 

\subsection{Related works and future directions}
While the literature on singularity formation in general relativity is vast, the study of naked singularities and their evolutionary properties is comparatively less developed. Given an expected positive resolution of (WCC) for physically reasonable matter models (and vacuum), generic stability for globally naked singularities cannot be expected. Still, many questions remain related to the existence of stable subspaces, instability mechanisms, and the stability of qualitative features of such spacetimes (e.g. the existence of a central singularity), as well as the critical collapse picture described in the previous section. 

Similar problems of stable singularity formation arise in the context of black hole interiors and cosmological spacetimes. A related work is \cite{greg1}, which shows a backward stability result for the Schwarzschild spacelike singularity in vacuum. This singularity is also unstable under generic perturbations due to the presence of the Kerr family; however, the backwards construction employed there allows identification of a subspace of data along the final timeslice for which the backwards problem admits a solution. The choice of data, and reliance on a high rate of vanishing for the perturbation data near the singular point, is similar to the data used here in the proof of Theorem \ref{thm1:interior}. 

Work on backwards stability of the Schwarzschild singularity is complemented by the work \cite{greg2}, which studies the corresponding \textit{forward} problem. In a suitable symmetry class, the authors pose data on a timeslice to the past of the putative singular surface. A result on the forward problem for naked singularity formation is not yet available, even in spherical symmetry.

We finally remark that although the main motivation for this paper an investigation of the stability properties of Christodoulou's continuously self-similar spacetimes, it is natural to ask whether a wider class of solutions satisfies the assumptions of Section \ref{sec:assumptionsbackground}. The techniques used here do require a double null gauge and self-similar bounds consistent with Christodoulou's spacetimes, but importantly do not require continuous self-similarity. It is feasible that the results proved here could be extended to apply to Choptuik's critical collapse spacetime, although closer study of the latter is required.

\subsection{Guide to the paper}
We conclude this section with an outline to the paper. Section \ref{sec:prelims} provides a series of analytical preliminaries, including the formulations of the Einstein-scalar field system needed here, a review of solution classes, and useful integration lemmas. Section \ref{sec:solnclasses} reviews local existence results for the Einstein-scalar field system, and derives an existence result for $C^1$ solutions in a non-standard geometry. The latter result, Proposition \ref{prop:weirdlocalex}, is necessary for the proof of Theorem \ref{thm1:interior}.

Section \ref{sec:mainresults} covers the paper's main results. In Sections \ref{sec:assumptionsbackground}, \ref{sec:admissibledata}, we outline the class of admissible background spacetimes and initial data perturbations respectively. Section \ref{sec:thmstatements} states the main theorems, asserting the existence of perturbed naked singularity interiors and exteriors. A detailed outline of the proof is given in \ref{subsec:introproofoutline}.

Section \ref{sec:proofthm1} completes the proof of Theorem \ref{thm1:interior}, and Section \ref{sec:proofthm2} the proof of Theorem \ref{thm2:exterior}. 

Appendix \ref{appA} discusses Christodoulou's $k$-self-similar solutions in detail, translating the existence proof of \cite{chris2} to a double null gauge. Moreover, it is shown that Christodoulou's solutions satisfy all the assumptions required for Theorems \ref{thm1:interior} and \ref{thm2:exterior}.

\subsection{Acknowledgements} 
The author is grateful to his advisor, Igor Rodnianski, for suggesting this problem and providing guidance throughout the project. The author also acknowledges valuable conversations with Yakov Shlapentokh-Rothman and Mihalis Dafermos.

%% file: Figures/fig_intronakedsing.tex
\begin{tikzpicture}[scale=.6]
\coordinate (n1) at (0,0) {};
\node[circle, label= above:{$\mathcal{O}$}, draw, inner sep =0, minimum size = .2cm] (n2) at (80:4) {};
\draw[thick] (n1) to[out = 65, in = -80] (n2);

\coordinate (a1) at ($(n2) + (45:2.5) + (-45:.4)$) {};
\coordinate (a2) at ($(n2) + (45:2.5)+ (135:.4)$) {};
\draw[->] (a1) -- (a2);
\coordinate (b1) at ($(n2) + (45:3) + (-45:.4)$) {};
\coordinate (b2) at ($(n2) + (45:3)+ (135:.4)$) {};
\draw[->] (b1) -- (b2);

\node[left] at (75:2.2) {$\Gamma$};

\node[circle, label= above:{$i^+$}, draw, inner sep =0, minimum size = .1cm] (n3) at ($(n2) + (45:5)$) {};

\node[circle, label= right:{$i^0$}, draw, inner sep =0, minimum size = .1cm] (n4) at ($(n1) + (0:11.7)$) {};

\draw[dashed] (n2) -- (n3);

\draw[dashed] (n3) -- (n4);

\draw[thick] (n1) to[out = 12, in = 168] (n4);

\node[below] at (5:5.5) {$\Sigma$};

\node[above] at ($(n3) + (-45:4.5) + (90:.3)+(0:.2)$) { $\mathcal{I}^+$};

\end{tikzpicture}

%% file: Figures/fig_introthm1.tex
\begin{tikzpicture}[scale=1.0]
\node[circle, label= above:{$\mathcal{O}$}, draw, inner sep =0, minimum size = .2cm] (n1) at (0,0) {};
\coordinate (n2) at (-90:6) {};
\node[label = {[xshift = -2mm, rotate = 90]left: \small$\{r=0\}$}] at (-90:2.2) {};
\coordinate (n3) at (-45:4.242) {};
\draw[thick] (n1) -- (n2) -- (n3) -- cycle;

\coordinate (n4) at ($(-90:6) + (45:2.4) $ ) {};
\node[label = {[rotate=45]below: \small$\{u=-1\}$}] at (n4) {};

\coordinate (n5) at (-45:2.1) {};
\node[label = {[rotate =-45]above: \small$\partial_u \phi \sim \partial_u \overline{\phi}(u) + \color{blue}\epsilon f_0(u)    $ }] at (n5) {};

\fill[lightgray!40] (n1) -- (n2) -- (n3) -- cycle;

\node[] at ($(-90:3) + (0:1.3)$) {$(\mathcal{Q}^{(in)}, g, \phi) $};

\end{tikzpicture}

%% file: Figures/fig_introthm2.tex
\begin{tikzpicture}[scale=1.0]
\node[circle, label= above:{$\mathcal{O}$}, draw, inner sep =0, minimum size = .2cm] (n1) at (0,0) {};
\draw[thick] (n1) -- (n3);
\coordinate (n3) at (-45:4.242) {};
\coordinate (n5) at (-45:2.1) {};
\node[label = {[rotate =-45]below: \small$\partial_v \phi \sim \partial_v \overline{\phi}(u) + \color{blue}\epsilon f_0(u)    $ }] at (n5) {};

\coordinate (n6) at ($(n3) + (45:5) $) {};
\coordinate (n7) at ($(n6) + (135:4.242) $ ) {};
\draw[thick] (n3) -- (n6);
\draw[dashed] (n6) -- (n7) -- (n1);

\coordinate (n8) at ($(n3) + (45:2.5) $ ) {};
\node[label = {[rotate =45]below: \small$\partial_v \phi \sim \partial_v \overline{\phi}(v) + \color{red}\epsilon g_0(v) + \color{blue} O(\epsilon) $ }] at (n8) {};

\fill[lightgray!40] (n1) -- (n3) -- (n6) -- (n7) -- (n1);
\node[label = {[rotate =45]above: \small$\{u=-1\}$ }] at ($ (-45:4.242)+ (45:2.5)$ ) {};

\node[] at ($(-45:2.2) + (45:2.5)$) {$(\mathcal{Q}^{(ex)}, g, \phi) $};

\coordinate (n9) at ($(n6) + (135:2.121) $ ) {};
\node[label = {[xshift = 2mm]\small$\mathcal{I}^+$ }]at (n9) {};

\coordinate (n10) at (45:3) {};
\node[label = {[xshift = -1mm, yshift =2mm, rotate=45]left: \small$\{u=0\}$}] at (n10) {};

\end{tikzpicture}

%% file: Figures/fig_instability.tex
\begin{tikzpicture}[scale = .8]
\node[circle, label= above:{$\mathcal{O}$}, draw, inner sep =0, minimum size = .2cm] (n1) at (0,0) {};
\draw[dashed] (n1) -- (-45:5);
\draw[thick] ($(-45:5) + (-135:4)$ ) -- ($(-45:5) + (45:4)$ );
\node[circle,fill, label = {[xshift = -1mm, yshift=-2mm]right:\small $(-1,0)$},inner sep =0, minimum size = .2cm] (n2) at ($(-45:5)$) {};
\draw[snake=coil,segment aspect=0,red]    ($(-45:5)$ ) -- ($(-45:5) + (45:4)$ );

\node[label = {[rotate=45]below:$\partial_v \overline{\phi}_k(v)$ }] at ($(-45:5) + (-135:1.5)$ ) {};
\node[label = {[rotate=45]below:$\partial_v \overline{\phi}_k(v) + \color{red} f(v)$ }] at ($(-45:5) + (45:2.6)$ ) {};

\draw[->] ($(-45:3)+(-135:.2)$) -- ($(-45:2)+(-135:.2)$);
\node[rotate = -45] at ($(-45:2.5)+(-135:.5)$) {\small self-similar interior};

\draw[->] ($(-45:3)+(45:.2)$) -- ($(-45:2)+(45:.2)$);
\node[rotate = -45] at ($(-45:2.5)+(45:.5)$) {\small trapped surface formation};

\end{tikzpicture}

%% file: prelims.tex
\section{Preliminaries}
\label{sec:prelims}
\subsection{Outline of spherical symmetry and the solution manifold $\mathcal{Q}$}
\label{subsec:conseqsphersym}
The starting point for the study of the Einstein-scalar field system is a Lorentzian $3+1$ dimensional manifold $(\mathcal{M}, \textbf{g}_{\mu \nu}),$ with $\textbf{g}_{\mu \nu}$ a Lorentzian metric. The manifold is equipped with a real valued massless scalar field, $\phi: \mathcal{M} \rightarrow \mathbb{R}$. 

We consider the subclass of spherically symmetric spacetimes. These spacetimes admit an isometric action of $SO(3)$, such that the orbits of points $x \in \mathcal{M}$ are either 1) a single point, or 2) a sphere in $\mathcal{M}$ that is spacelike with respect to $\textbf{g}_{\mu \nu}$. Under these conditions one can consider the $1+1$ dimensional quotient spacetime $\mathcal{Q} \doteq \mathcal{M} / SO(3)$, with Lorentzian quotient metric $g_{\mu \nu}$. We additionally require that $\phi$ descend to a function on $\mathcal{Q}$, for which we use the same symbol.

The quotient spacetime is a manifold with boundary, with the boundary corresponding to the projection of fixed points of the $SO(3)$ action. Define the projection map $\pi: \mathcal{M} \rightarrow \mathcal{Q}$, and label this boundary $\Gamma$. This boundary, alternatively called the axis or center, will be assumed to be a connected, timelike curve in $\mathcal{Q}$.

A convenient characterization of the boundary $\Gamma$ is given by the introduction of the geometric radius function $r: \mathcal{Q} \rightarrow \mathbb{R}_{\geq 0},$ where
\begin{equation*}
    r(p) \doteq \sqrt{ \frac{A(\text{proj}^{-1}(p))}{4\pi}}.
\end{equation*}
Here $A(U)$ is the area of a set $U \subset \mathcal{M}$. The boundary is then characterized by the condition $\Gamma = \{p\ |\ r(p) = 0\}.$

The discussion thus far has been coordinate independent. We specialize to the case where the quotient spacetime $\mathcal{Q}$ is covered by a global double null coordinate system. In such a coordinate system the metrics on $\mathcal{M}$ and $\mathcal{Q}$ respectively take the form
\begin{align*}
    \textbf{g}_{\mu\nu} &= -\Omega^2(u,v) du dv + r(u,v)^2 d\sigma_{S^2}, \\
    g_{\mu\nu} &= -\Omega^2(u,v) du dv.
\end{align*}
Here $\Omega^2$ is the coordinate-dependent null lapse, and $d\sigma_{S^2}$ is the round metric on $S^2$. It follows that the metric and scalar field together determine three functions on $\mathcal{Q}$, namely $(r, \Omega, \phi)$. In the next section the Einstein-scalar field system is recast as a closed system for these unknowns.

Additionally define the Hawking mass 
\begin{equation}
    \label{eq:defnofm}
    m \doteq \frac{r}{2}(1 - g(\nabla r, \nabla r )) = \frac{r}{2}(1 + \frac{4\partial_u r \partial_v r}{\Omega^2}).
\end{equation}
The mass is a geometric quantity, and may be used in place of $\Omega^2$ as a primitive unknown quantity. We largely take such an approach in this paper.

A convenient proxy for studying the mass, which is often more useful when approaching the axis, is the mass ratio 
\begin{equation*}
    \mu \doteq \frac{2m}{r}.
\end{equation*}
Finally, define the coordinate derivatives
\begin{equation}
    \label{eq:defnofcoordders}
    \nu \doteq \partial_u r, \ \ \ \lambda \doteq \partial_v r.
\end{equation}

\begin{figure}
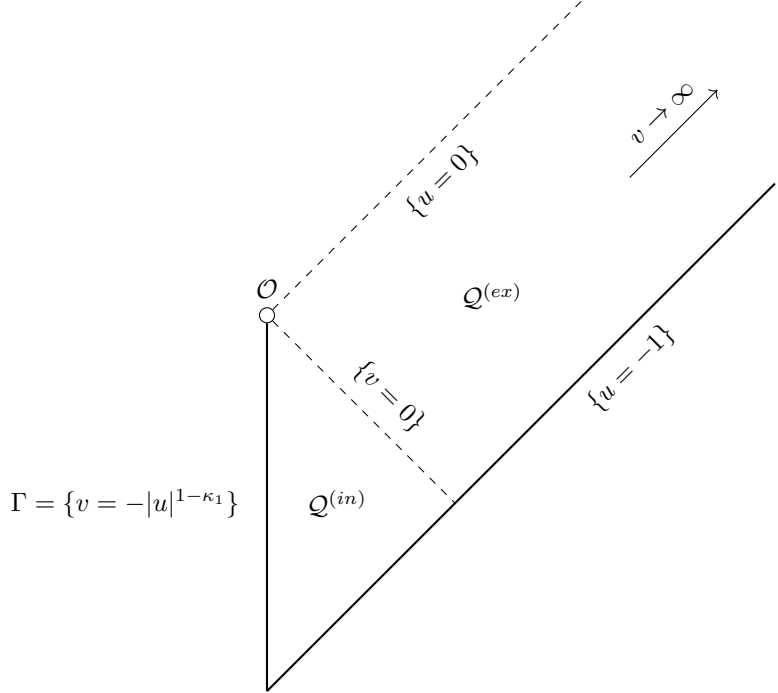

    \centering
  \includestandalone[]{Figures/fig_mainpenrose}
  \caption{The solution manifold, with divisions into the interior $\mathcal{Q}^{(in)}$ and exterior $\mathcal{Q}^{(ex)}$. Note the parameter $\kappa_1,$ which is a property of the admissible background solution. Christodoulou's solutions have $\kappa_1 = k^2$.}
  \label{fig1}
\end{figure}
\vspace{-1.5em}
\subsubsection{The solution manifold $\mathcal{Q}$}
The normalizations for the global double null coordinate system $(u,v)$ used in this paper are adapted to the $k$-self-similar models. For further discussion of this point, see Section \ref{sec:assumptionsbackground}. Associated to the spacetime is a parameter $\kappa_1 \in (0, \frac{1}{2})$ capturing the properties of the underlying self-similar model near the past lightcone of the singular point $
\mathcal{O}$. With respect to these coordinates $\mathcal{Q}$ takes the form in Figure \ref{fig1}, with coordinate description 
\begin{equation}
    \mathcal{Q} = \{(u,v) \ | \ -1 \leq u < 0, \ -|u|^{1-\kappa_1} \leq v < \infty\}.
\end{equation}
The spacetime contains a regular center $\Gamma$ given by 
\begin{equation}
    \Gamma = \{v = -|u|^{1-\kappa_1}, \ u < 0\}.
\end{equation}
The assumptions on the background solution will imply that $\Gamma$ is future incomplete, terminating at $\mathcal{O}$. It will follow that no suitably regular extensions exist in a neighborhood of $\mathcal{O}$, and in particular, $\mathcal{O}$ cannot be included as a point of the spacetime. We will often still represent $\mathcal{O}$ in coordinates as the point $(u,v) = (0,0)$. $\mathcal{Q}$ will be the maximal globally hyperbolic development of characteristic data along $\{u=-1\}$.

We now develop notation for distinguished subsets of $\mathcal{Q}$. The main decomposition of the spacetime introduces the \enquote{interior} and \enquote{exterior} region of the singular point $\mathcal{O}$, defined as 
\begin{equation*}
    \mathcal{Q}^{(in)} = \mathcal{Q} \cap \{v \leq 0 \} , \ \ \ \mathcal{Q}^{(ex)} = \mathcal{Q} \cap \{v \geq 0 \}.
\end{equation*}
For fixed $(u,v) \in \mathcal{Q}$ we define the characteristic subset
\begin{equation*}
    \mathcal{Q}_{u,v} \doteq \mathcal{Q} \cap \{u' \leq u, \ v' \leq v \}.
\end{equation*}
For an arbitrary subset $\mathcal{U} \subset \mathcal{Q}$ we define $\mathcal{U}_{u,v}$ similarly. 

Constant $u$ and $v$ null hypersurfaces are denoted by 
\begin{equation*}
    \Sigma_u \doteq \{(u',v') \in \mathcal{Q} \ | \ u' = u \}, \ \ \ \underline{\Sigma}_v \doteq \{(u',v') \in \mathcal{Q} \ | \ v' = v \}.
\end{equation*}
Denote the restriction of $\Sigma_u$ to $\mathcal{Q}^{(in)}$ by $\Sigma^{(in)}_u$, with similar definitions for $\Sigma^{(ex)}_u, \underline{\Sigma}^{(in)}_v, \underline{\Sigma}^{(ex)}_v$.

For working near the axis, define $u_\Gamma(v)$ to denote the unique $u$ value (if one exists) satisfying $(u_\Gamma(v), v) \in \Gamma$. Similarly define $v_\Gamma(u)$.

We finally introduce the self-similar coordinate $s_1 \doteq \frac{v}{|u|^{q_1}}$. With respect to this coordinate we can identify various important surfaces in $\mathcal{Q}$ as level sets of $s_1$:
\begin{equation*}
    \Gamma = \{s_1 = -1 \}, \ \ \ \underline{\Sigma}_0 = \{s_1 = 0\}.
\end{equation*}
Although $\Sigma_0$ is not a subset of spacetime, we will still suggestively write $\Sigma_0 = \{s_1 = \infty \}$. Neighborhoods with respect to the $s_1$ coordinate are called \textit{self-similar neighborhoods}, and many of our estimates will be adapted to such regions.

We give names to the following self-similar neighborhoods of $\underline{\Sigma}_0, \ \Gamma$, respectively, which will be required in the proof of Theorem \ref{thm1:interior}:
\begin{align}
    \mathcal{S}_{near} \doteq \mathcal{Q}^{(in)} \cap \{s_1 \leq \frac{1}{2} \} \label{dfn:snear}, \ \ \ \mathcal{S}_{far} \doteq \mathcal{Q}^{(in)} \cap \{s_1 \geq \frac{1}{2} \}. 
\end{align}

\subsection{Formulations of the Einstein-scalar field system and consequences}

The condition of spherical symmetry reduces the full Einstein-scalar field system on $(\mathcal{M},\textbf{g}_{\mu\nu})$ to a simpler $1+1$ system on $(\mathcal{Q},g_{\mu\nu})$. Several equivalent formulations exist, emphasizing different primitive functions of the metric and scalar field.

Written as a system for the unknowns $(r, \Omega, \phi)$, the relevant equations become
\begin{align}
    r \partial_u \partial_v r &= -\partial_u r \partial_v r - \frac{1}{4}\Omega^2, \label{SSESF:1:1} \\[.3em]
    \label{SSESF:1:2}
     r^2 \partial_u \partial_v \log \Omega &= \partial_u r \partial_v r + \frac{1}{4}\Omega^2 - r^2 \partial_u \phi \partial_v \phi, \\[.3em]
     \label{SSESF:1:3}
    r\partial_u \partial_v \phi &= -\partial_u r \partial_v \phi - \partial_v r \partial_u \phi, \\[.3em]
    \label{SSESF:1:4}
    2 \Omega^{-1} \partial_u r \partial_u \Omega &= \partial_u^2 r + r (\partial_u \phi)^2, \\[.3em]
     \label{SSESF:1:5}
     2 \Omega^{-1} \partial_v r \partial_v \Omega &= \partial_v^2 r + r (\partial_v \phi)^2.
\end{align}
Schematically, (\ref{SSESF:1:1})-(\ref{SSESF:1:3}) consist of wave equations for $r, \Omega, \phi,$ and (\ref{SSESF:1:4})-(\ref{SSESF:1:5}) transport equations along constant $u$ and $v$ hypersurfaces. The latter importantly only concern the geometry and scalar field quantities intrinsic to the hypersurfaces, and may thus be thought of as constraint equations.

It will often be convenient to estimate the geometric quantities $(r,m,\phi)$ and their coordinate derivatives. Making use of the definitions (\ref{eq:defnofm})-(\ref{eq:defnofcoordders}), we recast the system in the following form:
\begin{align}
    \label{SSESF:2:1}
    \partial_u \lambda &= \frac{\mu \lambda \nu}{(1-\mu)r}, \\[.3em]
    \label{SSESF:2:2}
    \partial_v \nu &= \frac{\mu \lambda \nu}{(1-\mu)r},\\[.3em]
    \label{SSESF:2:3}
    r\partial_u \partial_v \phi &= -\nu \partial_v \phi - \lambda \partial_u \phi,\\[.3em]
    \label{SSESF:2:4}
    2\nu\partial_u m &= r^2(1-\mu)(\partial_u \phi)^2,\\[.3em]
    \label{SSESF:2:5}
    2\lambda\partial_v m &= r^2(1-\mu)(\partial_v \phi)^2.
\end{align}
An alternative form of the wave equation (\ref{SSESF:2:3}), more suitable for study near the axis and in asymptotically flat regions, is given by considering the equation for $r\phi$:
\begin{equation}
    \label{SSESF:2:6}
    \partial_u \partial_v (r\phi) = \frac{\mu \lambda \nu}{(1-\mu)r^2}(r\phi) = (\partial_v \nu) \phi.
\end{equation} 
We will also make use of equations for the auxiliary quantities $\mu, \ \frac{\lambda}{1-\mu}, \ $and $\frac{\nu}{1-\mu}$ due to their favorable structure. Straightforward calculations with (\ref{SSESF:2:1})-(\ref{SSESF:2:5}) yield
\begin{align}
    \label{SSESF:2:7}
    \partial_u \mu &= -\frac{\nu}{r}\mu + \frac{r}{\nu}(1-\mu)(\partial_u \phi)^2, \ \ \ \
    \partial_v \mu = -\frac{\lambda}{r}\mu + \frac{r}{\lambda}(1-\mu)(\partial_v \phi)^2,
\end{align}
as well as the Raychaudhuri equations 
\begin{align}
    \label{SSESF:2:9} 
    \partial_u \log \Big(\frac{\lambda}{1-\mu} \Big) &= \frac{r}{\nu}(\partial_u \phi)^2, \ \ \ \ 
    \partial_v \log \Big(\frac{\nu}{1-\mu} \Big) = \frac{r}{\lambda}(\partial_v \phi)^2.
\end{align}
The system is complemented by the boundary conditions along $\Gamma$, assuming a regular center:
\begin{equation}
     \label{SSESF:2:10.5}
     (r\phi)|_\Gamma = r|_\Gamma = m|_\Gamma = 0.
\end{equation}

We conclude this section by remarking how the above system specializes in the regime of this paper, namely the setting of naked singularity spacetimes. 

We will work in regions containing no trapped and no anti-trapped surfaces, for which the following sign conditions hold pointwise:
\begin{equation*}
    \nu < 0, \ \ \ \lambda > 0, \ \ \  0 < (1-\mu) < 1.
\end{equation*}
Inspecting (\ref{SSESF:2:1}), (\ref{SSESF:2:2}), (\ref{SSESF:2:4}), (\ref{SSESF:2:5}) shows that under these sign conditions various error terms have definite signs. We will be able to make use of this monotonicity information in deriving a priori estimates.

It is also important to note that despite the factors of $r^{-1}$ appearing in the system, we will work with a regular center to the past of $\mathcal{O}$, i.e. double null unknowns will have finite limits as $r \rightarrow 0$. Propagating such finiteness requires the solution to satisfy regularity conditions at the center, and requires one to propagate bounds on $m, \mu$ consistent with vanishing as $r \rightarrow 0$. The solution classes with which we work are regular enough to ameliorate difficulties at the axis.

It remains to define solution classes that incorporate the above system, appropriate boundary/initial conditions, and gauge conditions.

\subsection{Solution classes and local existence results}
\label{sec:solnclasses}

The sharp regularity of the solutions constructed here, just as for Christodoulou's solutions, varies across the spacetime (cf. Theorem \ref{thm:christodoulou_solutions}). In addition to the severe loss of regularity as $u \rightarrow 0$, the solutions also display limited regularity across the horizons $\{v=0\} \cup \{u=0\}$. For this reason we will work with two classes of solutions introduced in \cite{chris1}, namely the low regularity bounded variation (BV) class and the $C^1$ class. After introducing essential features of the solution classes, we recall well-posedness results for the characteristic initial value problem. We finally discuss well-posedness in a certain non-standard geometry, which will be of use in the proof of Theorem \ref{thm1:interior}.

In the following, let $AC(U), \ BV(U)$ denote the space of absolutely continuous and bounded variation functions on a subset $U$, respectively. We denote the total variation norm by $T.V$. Recall the parameter $\kappa_1$ reflecting the choice of double null gauge, and define $q_1 = 1-\kappa_1, \ p_1 = (1-\kappa_1)^{-1}$.

Fix a point $(u,v) \in \mathcal{Q}$. We first define the notion of a BV solution on the characteristic domain $\mathcal{Q}_{u,v}.$

\begin{dfn}
\label{dfn:BV}
A \textbf{BV solution} to the system (\ref{SSESF:1:1})-(\ref{SSESF:1:5}) in a domain $\mathcal{Q}_{u,v}$ consists of a triple $(r,m,\phi)$ satisfying the equations distributionally, subject to the following requirements:
\begin{enumerate}
    \item $\sup_{\mathcal{Q}_{u,v}} (-\nu) < \infty$, \ \text{and} \ \ $\sup_{\mathcal{Q}_{u,v}} \lambda^{-1} < \infty$, 
    \item $\lambda \in BV(\Sigma_{u} \cap \mathcal{Q}_{u,v})$ uniformly in $u$, and $\nu \in BV(\underline{\Sigma}_{v} \cap \mathcal{Q}_{u,v})$ uniformly in $v$.
    \item $\phi \in AC(\Sigma_{u} \cap  \mathcal{Q}_{u,v})$, and has $T.V.$ norm bounded uniformly in $u$. Similarly, $\phi \in AC(\underline{\Sigma}_{v} \cap  \mathcal{Q}_{u,v})$, and has $T.V.$ norm bounded uniformly in $v$.
    \item  $\partial_v(r\phi) \in BV(\Sigma_{u} \cap \mathcal{Q}_{u,v} )$ uniformly in $u$, and $\partial_u(r\phi) \in BV(\underline{\Sigma}_{v} \cap \mathcal{Q}_{u,v})$ uniformly in $v$.
     \item For any point $(u,v) \in \Gamma$, we have the regularity conditions
    \begin{equation*}
        \lim_{\epsilon \rightarrow 0} \sup_{0 < \delta \leq \epsilon} T.V._{\{u-\delta\} \times (v_\Gamma(u-\delta),v)}[\phi] = 0, \ \ \  \lim_{\epsilon \rightarrow 0} \sup_{0 < \delta \leq \epsilon} T.V._{(u_\Gamma(v-\epsilon), u_\Gamma(v-\delta)) \times \{v - \delta \} }[\phi] = 0,
    \end{equation*}
    \begin{equation*}
        \lim_{\epsilon \rightarrow 0} \sup_{0 < \delta \leq \epsilon} T.V._{(u,u_\Gamma(v + \delta)) \times \{v+\delta \} }[\phi] = 0, \ \ \ \lim_{\epsilon \rightarrow 0} \sup_{0 < \delta \leq \epsilon} T.V._{\{u+\delta\} \times (v_\Gamma(u+\delta), v_\Gamma(u+\epsilon)) }[\phi]=0.
    \end{equation*} 
    \item The boundary conditions (\ref{SSESF:2:10.5}) hold along $\Gamma \cap \mathcal{Q}_{u,v}$.
    \item  For any point $(u,v)\in \Gamma$, 
    \begin{align*}
        &\lim_{\delta \rightarrow 0} \ (\lambda + p_1|u|^{\kappa_1} \nu)(u,v+\delta) = 0, \\[.3em]
        &\lim_{\delta \rightarrow 0} \ (\partial_v(r\phi) +p_1|u|^{\kappa_1} \partial_u(r\phi))(u,v+\delta) = 0.
    \end{align*}
\end{enumerate}
\end{dfn}
For the solutions considered in this paper, the key requirement of BV regularity is that $\partial_u(r\phi), \nu$ be uniformly integrable along ingoing null hypersurfaces, and similarly for $\partial_v(r\phi), \lambda$. It is \textit{not} required that these quantities have bounded derivatives. The solutions considered in this paper will have this property for the outgoing derivatives $\partial_v^2(r\phi), \partial_v \lambda$, which will generically only be integrable functions across $\{v=0\}$. For this reason BV is a natural regularity class in this region.

Away from the horizons, the higher order derivatives of $r, \phi$ enjoy better bounds. This added regularity is captured in the notion of a $C^1$ solution below:

\begin{dfn}
\label{dfn:C1}
A $\mathbf{C^1}$ \textbf{solution} to the system (\ref{SSESF:1:1})-(\ref{SSESF:1:5}) in a domain $\mathcal{Q}_{u,v}$ consists of a triple $(r,m,\phi)$ satisfying the equations pointwise, subject to the following requirements:
\begin{enumerate}
    \item $ \sup_{\mathcal{Q}_{u,v}} (-\nu) < \infty$, \ \text{and} \ \ $\sup_{\mathcal{Q}_{u,v}} \lambda^{-1} < \infty$, 
    \item $\nu, \ \lambda, \ \partial_u(r\phi), \ \partial_v(r\phi) \in C^1(\mathcal{Q}_{u,v})$.
    \item The boundary conditions (\ref{SSESF:2:10.5}) hold along $\Gamma \cap \mathcal{Q}_{u,v}$.
    \item For any point $(u,v)\in \Gamma$, 
    \begin{align*}
        &\lim_{\delta \rightarrow 0} (\lambda + p_1|u|^{\kappa_1} \nu)(u,v+\delta) = \lim_{\delta \rightarrow 0} (\lambda + p_1|u|^{\kappa_1} \nu)(u-\delta,v)=0, \\[.3em]
        &\lim_{\delta \rightarrow 0} (\partial_v(r\phi) + p_1|u|^{\kappa_1} \partial_u(r\phi))(u,v+\delta) = \lim_{\delta \rightarrow 0} (\partial_v(r\phi) + p_1|u|^{\kappa_1} \partial_u(r\phi))(u-\delta,v)=0.
    \end{align*}
\end{enumerate}
\end{dfn}
Given a BV ($C^1$) solution on $\mathcal{Q}_{u,v}$ for all $(u,v) \in \mathcal{Q}$, we will speak of a BV ($C^1$) solution on $\mathcal{Q}$. We will often speak of BV ($C^1$) solutions in self-similar neighborhoods in $\mathcal{Q}$, by which we mean a corresponding solution in the neighborhood intersected with all $\mathcal{Q}_{u,v}.$

\subsubsection{Standard local existence results}
We now discuss local existence theory for (\ref{SSESF:1:1})-(\ref{SSESF:1:5}), or equivalently, for the Hawking mass formulation (\ref{SSESF:2:1})-(\ref{SSESF:2:5}).  The spherically symmetric Einstein-scalar field system admits well-posed initial value problems in both of the above classes, with propagation of higher regularity. 

\vspace{5pt}
\noindent
A standard local existence problem is that of a single outgoing null surface $\Sigma_{u_0}$, for some $u_0 \in [-1,0)$. Here $\Sigma_{u_0}$ is assumed to intersect the axis at a point $(u_0, v_\Gamma(u_0)).$ The free data consists of a gauge choice for $\lambda(u_0,v),$\footnote{In the following we assume gauges along initial data are always selected to satisfy $\nu < 0, \ \ \lambda >0$. It follows that we do not allow trapped surfaces already in data.} and a choice of scalar field $\partial_v(r\phi)(u_0,v)$, for $v \in [v_\Gamma(u_0),\infty)$. Up to a coordinate change one may in fact assume $\lambda(u_0,v)  =\frac{1}{2}$. Along with the condition $r(u_0, v_\Gamma(u_0)) = 0$, we derive $r(u_0,v)$ along the outgoing surface. Given $\partial_v(r\phi),$ integrating in $v$ gives $\phi$ and $\partial_v\phi$ along the initial data surface. Therefore, it is equivalent to give $\phi$ or $\partial_v(r\phi)$ initially. Finally, one can integrate (\ref{SSESF:2:5}) and use the boundary condition $m(u_0, v_\Gamma(u_0)) = 0$ to compute $m(u_0,v)$. 

In this case, cf. \cite{chris1}, given data for $\partial_v (r\phi)$ in $BV(\Sigma_{u_0}) $, there exists a domain of the form $\{u_0 \leq u \leq u_0 +\delta\}$ and a unique BV solution achieving the data. For data that is additionally $C^1$, the solution is in the $C^1$ class (indeed, for the natural notion of $C^k$ solution, propagation of higher regularity holds as well). A key challenge in this setting is developing a robust scheme of higher order estimates for establishing regularity near the axis, and avoiding singularities due to $r^{-1}$ factors in the scalar field system.

\vspace{5pt}
\noindent
The second local existence result we need is the fully characteristic initial value problem, in which data is posed on intersecting null hypersurfaces $\Sigma_{u_0} \cup \underline{\Sigma}_{v_0},$ intersecting at a point $(u_0,v_0)$. The appropriate data is a choice of $r(u_0,v_0) > 0, \ m(u_0,v_0), \ \phi(u_0,v_0)$ at the intersection point, gauge conditions for $\lambda(u_0,v), \ \nu(u,v_0)$, and choices of scalar field derivatives $\partial_v(r\phi)(u_0,v), \ \partial_u(r\phi)(u,v_0)$ along the data surfaces. From the boundary conditions at the intersection point and the specified data, one can fully compute the value of $r, m, \phi$, and tangential derivatives thereof, along the initial data surfaces. 

With appropriate gauge choice, and given BV initial data for the scalar field derivatives, there exists a neighborhood $\{u_0 \leq u \leq u_0 + \delta, \ v_0 \leq v \leq v_0 + \delta \}$ and a unique BV solution achieving the data. A similar propagation of regularity statement for $C^1$ and $C^k$ solutions holds. Moreover, provided $\inf_{\Sigma_{u_0}\cup\underline{\Sigma}_{v_0}} r \geq c > 0$, one can solve in a full neighborhood $\{u_0 \leq u \leq u_0+\delta\} \cup \{v_0 \leq v \leq v_0 +\delta\}$ of data. 

To establish existence in the \textit{full} neighborhood of a null hypersurface, say the ingoing $\underline{\Sigma}_{v_0}$, one can no longer integrate $u$ transport equations and apply smallness of the region in the $u$ direction. A resolution of this problem in the vacuum setting is due to \cite{luk1}, relying on a reductive structure in the equations. One need only utilize smallness in one direction, applying Grönwall in the long direction to estimate all remaining quantities. The argument applies equally well to the spherically symmetric Einstein-scalar field system, provided one stays a finite distance away from the axis. In the next section we discuss an extension to the setting where $r$ is not bounded below.

\subsubsection{Local existence in a mixed timelike-characteristic region}

\begin{figure}[h]
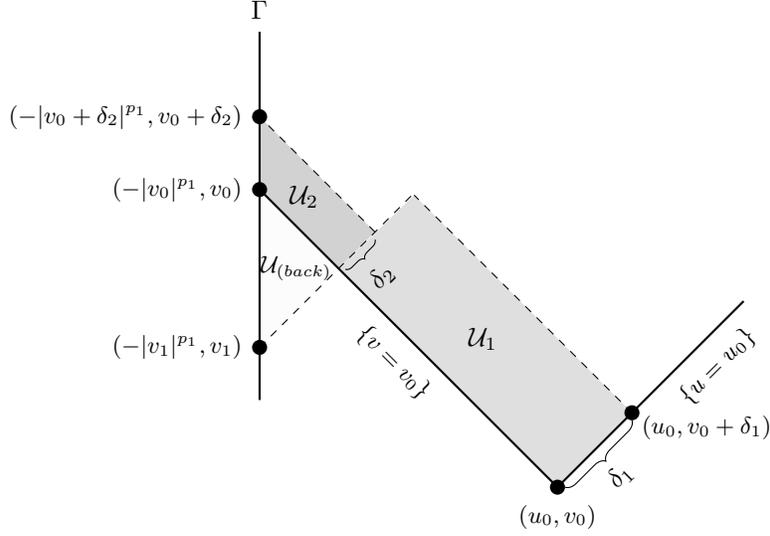

    \centering
  \includestandalone[]{Figures/fig_newlocalexistence}
  \caption{Diagram for the local existence result of Proposition \ref{prop:weirdlocalex}}
  \label{figidk}
\end{figure}

In the sequel we need a slight extension of the previous local existence results to a setting with both a bifurcate null hypersurface and an axis. See Figure \ref{figidk}. Assume a point $(u_0,v_0) \in \mathcal{Q}$ is fixed, along with a pair of ingoing and outgoing null surfaces $\underline{\Sigma}_{v_0}, \ \Sigma_{u_0}$. Pose the data  $\lambda(u_0,v), \ \partial_v (r\phi)(u_0,v) \in C^{1}(\Sigma_{u_0}),$ and $\nu(u,v_0),  \partial_u(r\phi)(u,v_0)\in C^{1}(\underline{\Sigma}_{v_0})$. It follows from this choice of data, the scalar field system, and boundary conditions at the axis that the values of $r,m,\phi$ are everywhere determined on $\underline{\Sigma}_{v_0} \cup \Sigma_{u_0}$.

Unlike the problem considered above, we assume $\underline{\Sigma}_{v_0}$ intersects the axis at a single point, given by $(-|v_0|^{p_1},v_0)$ with respect to our double null coordinates. Local existence in a full neighborhood of the ingoing cone $\underline{\Sigma}_{v_0}$ must contend with two issues, already suggested in the previous section: 1) integrating $u$ equations does not come with any smallness, and hence estimates for quantities only satisfying $u$ equations (e.g. $\partial_v \phi$)  require additional structure; 2) the standard scheme for estimating scalar field derivatives near the axis require estimating \textit{higher} derivatives of $r\phi$. Commuting the system in order to establish these bounds complicates the structure of the equations, rendering estimates for quantities like $\partial_v^2(r\phi)$ difficult.

Christodoulou overcame these challenges in \cite{chris1}, which introduces a flexible extension principle for establishing regularity near the axis. We first recall this result, rephrased to apply in our non-standard gauge. Assume as above that $\Gamma = \{-|u|^{q_1} = v, \ u < 0\}.$ For fixed $(\wt{u},\wt{v}) \in \mathcal{Q}$ define
\begin{equation*}
    \mathcal{D}(\wt{u},\wt{v}) \doteq \{(u,v)\ | \ \wt{u} \leq u \leq -|v|^{p_1}, \ -|\wt{u}|^{q_1} \leq v \leq \wt{v}\},
\end{equation*}
This set may be interpreted as the domain of dependence of $\Sigma_u \cap \{v \leq \wt{v}\}.$ 

\begin{prop}[\cite{chris2}]
Fix an outgoing null surface $\Sigma_u$ with past endpoint along $\Gamma$, and data $\lambda(v), \partial_v(r\phi) \in C^1(\Sigma_{u})$. Assume there exists a $v_* > -|u|^{q_1}$ and a $C^1$ solution on $\mathcal{D}(u,v)$ for all $v < v_*$. There exists a constant $\epsilon > 0$  such that the following extension principle holds:

Suppose that as $(u,v_*)$ is approached, $\mu$ satisfies the smallness assumption
\begin{equation}
    \label{eq:extensionprinciple}
    \lim_{u \rightarrow -|v_*|^{p_1}} \sup_{\mathcal{D}(u,v_*)} \mu(u,v) < \epsilon.
\end{equation}
Then there exists a $v_{**} > v_*$ and a $C^1$ extension of the solution on the region $\mathcal{D}(u,v_{**})$.
\end{prop}

The local existence result for data in the form of Figure \ref{figidk} will now follow easily from the extension principle, alongside the aforementioned standard local existence problems.
\begin{prop}
\label{prop:weirdlocalex}
Assume data along $\underline{\Sigma}_{v_0} \cup \Sigma_{u_0}$ is specified as above, and is $C^1$. Then there exists a $\delta > 0$ depending on the $C^1$ norms of the initial data, and a unique local $C^1$ solution to the Einstein-scalar field system on the domain 
\begin{equation}
    \{v_0 \leq v \leq v_0 + \delta , \ u_0 \leq u \leq -|v|^{p_1}\}.
\end{equation}
\end{prop}

\begin{proof}
To apply the extension principle, we require a local solution in a region given by $\mathcal{D}(u,v)$ for some point $(u,v)$. We construct such a solution in the region $\mathcal{U}_{(back)}$ by considering the backwards problem with data posed along $\underline{\Sigma}_{u_0}$. One may think of $\underline{\Sigma}_{u_0}$ as an \textit{outgoing} cone for the backwards problem with past vertex along $\Gamma$. The above local existence theory provides a $v_1 < v_0$ and a $C^1$ solution on the region $\mathcal{U}_{(back)},$ equivalently given as $\mathcal{D}(-|v_1|^{p_1},v_0).$

We next consider the solution in a region $\mathcal{U}_1$, at a fixed distance away from the axis. Restricting to the problem of characteristic data along $(\underline{\Sigma}_{u_0} \cap \{u \leq -|v_1|^{p_1}\}) \cup \Sigma_{u_0}$, we apply the local existence for characteristic initial data with $r$ bounded away from $0$. It follows that we have a local $C^1$ solution in the domain $$\mathcal{U}_1 = \{u_0 \leq u \leq -|v_1|^{p_1}, \ v_0 \leq v \leq v_0+\delta_1\},$$
for some $\delta_1 > 0$. In particular, the solution in $\mathcal{U}_1$ generates $C^1$ outgoing data along $\{u=-|v_1|^{p_1}\}$ to the future of $\{v = v_0\}$. We now wish to apply the extension principle and generate an additional piece of the solution in $$\mathcal{U}_2 = \{-|v_1|^{p_1} \leq u \leq -|v|^{p_1}, \ v_0 \leq v \leq v_0 + \delta_2 \},$$
for some $\delta_2 > 0$. The smallness condition on $\mu$ in $\mathcal{U}_{(back)}$ is simple to check, as the local existence theorem guarantees that the solution is $C^1$ there uniformly up to the cone point $(-|v_0|^{p_1},v_0).$ The bound $\mu \lesssim r$ then follows by standard arguments.

The extension principle thus suffices to generate a $C^1$ solution in $\mathcal{U}_2$. Taking the union $(\mathcal{U}_1 \cap \{v \leq v_0+\delta_2 \}) \cup \mathcal{U}_2$ gives the desired local solution.
\end{proof}

\subsection{Asymptotic flatness and naked singularities}
\label{subsec:nakedsingdfn}
Naked singularities capture the notion of singularities \enquote{visible} to infinitely far away observers. Formalizing this notion requires the concept of asymptotically flat spacetimes, for which an ideal boundary at infinite distance (future null infinity) is well-defined.

Definitions of asymptotic flatness are not standard in the literature. We will only work on $\mathcal{Q}$, for which a sufficient definition is the following:
\begin{dfn}
A BV solution to the scalar field system (\ref{SSESF:1:1})-(\ref{SSESF:1:5}), $(\mathcal{Q},g_{\mu\nu},\phi)$,  is \textbf{asymptotically flat} if for all $u \in [-1,0)$ the following conditions hold:
\begin{enumerate}
    \item $\lim_{v\rightarrow \infty} r(u,v) = \infty$,
    \item $\lim_{v \rightarrow \infty} m(u,v)$ exists and is finite.
    \item $\lim_{v \rightarrow \infty} (r\phi)(u,v)$ exists and is finite.
\end{enumerate}
\end{dfn}
For such spacetimes we formally think of the set $\{(u,\infty)| \ u \in [-1,0)\}$ as an ideal boundary of all future directed null geodesics, denoted $\mathcal{I}^+$. In a conformal compactification of the spacetime, it makes sense to think of future null infinity $\mathcal{I}^+$ as a bona fide boundary of the spacetime embedded in $\mathbb{R}^{1+1}$. It then follows by \cite{daf1} that $\mathcal{I}^+$ is a connected, ingoing null ray. 

One can ask whether this ray is future geodesically complete, i.e. whether geodesic generators of $\mathcal{I}^+$ exist for infinite proper time. Failure to satisfy this condition is referred to as an incomplete null infinity, and is the method by which we detect the presence of naked singularities in the spacetime. The following formalization of this notion, without requiring a conformal compactification, is drawn from \cite{igoryak1}.
\begin{dfn}
\label{dfn:incompleteI}
Let $\mathcal{H}$ be an outgoing null ray in the asymptotically flat spacetime $(\mathcal{Q},g_{\mu\nu},\phi)$. The spacetime is said to contain an \textbf{incomplete null infinity} if there exists a sequence $p_i \in \mathcal{H}, \ v(p_i) \rightarrow \infty$, such that the ingoing null geodesics emanating from $p_i$ have uniformly bounded affine length. 

A spacetime with incomplete null infinity arising as the maximal globally hyperbolic development of regular Cauchy data is said to possess a \textbf{naked singularity}.
\end{dfn}

Following \cite{daf1}, it is easy to check in our setting that such affinely parameterized null geodesics have tangent vector 
\begin{equation*}
    X \doteq \frac{\Omega^2(u,v)}{\Omega^2(u_0,v)}\frac{\partial}{\partial u},
\end{equation*}
where $u_0$ is the value of the $u$ coordinate along $\mathcal{H}$. To establish incompleteness of $\mathcal{I}^+$, it suffices to check that there exists a constant $A < \infty$ such that 
\begin{equation}
    \limsup_{i \rightarrow \infty} \int\limits_{(u_0,v(p_i))}^{(0,v(p_i))}(Xu)(u',v(p_i))du' \leq A.
\end{equation}

\subsection{Einstein-scalar field system for differences}
The main results of the paper are of an essentially perturbative nature, probing the behavior of the scalar field system in a neighborhood of an exact solution denoted $(\overline{g},\overline{\phi})$. It will therefore be useful to derive equations for the differences of two solutions to the system (\ref{SSESF:2:1})-(\ref{SSESF:2:5}). Let $(r,m,\phi)$ and $(\overline{r},\overline{m},\overline{\phi})$ denote two solutions to (\ref{SSESF:1:1})-(\ref{SSESF:1:5}) in a domain $\mathcal{U}$. Define the differences
\begin{align}
  \label{renormdef1}
     r_p &\doteq r - \overline{r}, \ \  \nu_p \doteq \nu - \overline{\nu}, \ \ \lambda_p \doteq \lambda - \overline{\lambda}, \\[.3em]
     \label{renormdef2}  
    m_p &\doteq m - \overline{m},\ \  \mu_p \doteq \mu - \overline{\mu}, \ \ \phi_p \doteq \phi - \overline{\phi}.
\end{align}
More generally, if $A,B,C,D$ are double null unknowns $\in \{r,\nu,\lambda,m, \mu,1-\mu,\phi\}$ (or derivatives of these), define 
\begin{equation*}
   \Big( \frac{AB}{CD}\Big)_p \doteq \frac{AB}{CD} - \frac{\overline{A}\overline{B}}{\overline{C}\overline{D}}.
\end{equation*}
We will often find it convenient to work with the mass ratio $\mu$ rather than the mass $m$ itself. On the level of differences, these quantities are related by 
\begin{equation}
    \label{eq:mpvsmup}
    \mu_p = \frac{2m_p}{r} - \frac{2\overline{m}r_p}{r\overline{r}}.
\end{equation}
Inserting these definitions in (\ref{SSESF:2:1})-(\ref{SSESF:2:7}) gives equations for the differences:
\begin{align}
    \label{PSSESF:1}
    \partial_u \lambda_p &= \frac{\mu \nu}{(1-\mu) r}\lambda_p + \Big(\frac{\mu \nu}{(1-\mu)r}\Big)_p \overline{\lambda}, \\
    \label{PSSESF:2}
    \partial_v \nu_p &= \frac{\mu \lambda}{(1-\mu) r}\nu_p + \Big(\frac{\mu \lambda }{(1-\mu)r}\Big)_p \overline{\nu}, \\
    \label{PSSESF:3}
    \partial_u \mu_p &= -\Big(\frac{\nu}{r} + \frac{r}{\nu}(\partial_u \phi)^2  \Big)\mu_p  -\Big(\frac{\nu}{r} + \frac{r}{\nu}(\partial_u \phi)^2  \Big)_p \overline{\mu} + \Big(\frac{r}{\nu}(\partial_u \phi)^2 \Big)_p, \\
    \label{PSSESF:4}
     \partial_v \mu_p &= -\Big(\frac{\lambda}{r} + \frac{r}{\lambda}(\partial_v \phi)^2  \Big)\mu_p  -\Big(\frac{\lambda}{r} + \frac{r}{\lambda}(\partial_v \phi)^2  \Big)_p \overline{\mu} + \Big(\frac{r}{\lambda}(\partial_v \phi)^2 \Big)_p, \\
     \label{PSSESF:5.5}
    \partial_u \partial_v \phi_p &= -\frac{\nu}{r}\partial_v \phi_p - \Big(\frac{\nu}{r}\Big)_p\partial_v \overline{\phi} - \frac{\lambda}{r}\partial_u \phi_p - \Big(\frac{\lambda}{r} \Big)_p \partial_u \overline{\phi}.
\end{align}
We note here other forms of the wave equation that will be useful in the interior construction: 
\begin{align}
   \label{PSSESF:6}
    \partial_u\partial_v (r\phi_p) &= \frac{\mu\lambda\nu}{(1-\mu)r^2}(r\phi_p) - (r_p \partial_u \partial_v \overline{\phi} + \nu_p \partial_v \overline{\phi} + \lambda_p \partial_u \overline{\phi} ), \\
    \label{PSSESF:5}
    \partial_u (r\partial_v \phi_p) &= -r_p\partial_u \partial_v \overline{\phi} - \nu_p \partial_v \overline{\phi} - \lambda_p \partial_u \overline{\phi}.
\end{align}
To avoid expressions becoming unwieldy, we give special names to many of the terms appearing in the above difference system:
\begin{align}
    \mathcal{G}_1 &\doteq \frac{\mu \lambda}{(1-\mu)r}, \ \ \  \mathcal{G}_2 \doteq \frac{\mu \nu}{(1-\mu)r}, \\
    \mathcal{G}_3 &\doteq \frac{\lambda}{r}+\frac{r}{\lambda}(\partial_v \phi)^2, \ \ \  \mathcal{I}_3 \doteq \frac{r}{\lambda}(\partial_v \phi)^2,\\
    \mathcal{G}_4 &\doteq \frac{\nu}{r}+\frac{r}{\nu}(\partial_u \phi)^2 , \ \ \ \mathcal{I}_4 \doteq \frac{r}{\nu}(\partial_u \phi)^2,\\
    \mathcal{G}_5 &\doteq \frac{\mu\lambda\nu}{(1-\mu)r^2}, \ \ \  \mathcal{I}_5 \doteq r_p\partial_u\partial_v \overline{\phi} + \nu_p \partial_v \overline{\phi} + \lambda_p \partial_u \overline{\phi},\\
    \mathcal{T}_1 &\doteq \frac{\lambda}{r}, \ \ \  \mathcal{T}_2 \doteq \frac{\nu}{r}.
\end{align}
Define $\overline{(\mathcal{G}_1)}, (\mathcal{G}_1)_p$ in the natural manner, and similarly for the other named terms.

We will need to commute various equations above by $\partial_u$ or $\partial_v$ in order to estimate the solution at the $C^1$ level. The relevant commuted equations are 
\begin{align}
    \label{PSSESF:7}
     \partial_u (\partial_v\lambda_p) &= \mathcal{G}_2 \partial_v \lambda_p + \partial_v \mathcal{G}_2 \lambda_p + (\mathcal{G}_2)_p \partial_v \overline{\lambda} + \partial_v (\mathcal{G}_2)_p \overline{\lambda}, \\[.3em] 
      \label{PSSESF:8}
    \partial_v (\partial_u \nu_p) &= \mathcal{G}_1 \partial_u \nu_p + \partial_u \mathcal{G}_1 \nu_p + (\mathcal{G}_1)_p \partial_u \overline{\nu} + \partial_u(\mathcal{G}_1)_p\overline{\nu}, \\[.3em] 
    \partial_u (\partial_v^2 \phi_p) &= - \mathcal{T}_2 \partial_v^2 \phi_p -\mathcal{T}_1 (\partial_v \partial_u \phi)_p - \partial_v \mathcal{T}_1 \partial_u \phi_p - \partial_v (\mathcal{T}_1)_p \partial_u \overline{\phi} - (\mathcal{T}_1)_p \partial_v\partial_u \overline{\phi} \nonumber\\[.3em] 
    &- \partial_v \mathcal{T}_2 \partial_v\phi_p - \partial_v (\mathcal{T}_2)_p \partial_v\overline{\phi}-(\mathcal{T}_2)_p \partial_v^2\overline{\phi}, \label{PSSESF:9} \\[.3em] 
    \partial_v (\partial_u^2 \phi_p) &= -\mathcal{T}_1 \partial_u^2 \phi_p - \partial_u \mathcal{T}_1 \partial_u \phi_p - \partial_u(\mathcal{T}_1)_p\partial_u\overline{\phi} - (\mathcal{T}_1)_p \partial_u^2\overline{\phi} - \partial_u \mathcal{T}_2 \partial_v \phi_p \nonumber \\[.3em] 
    &- \mathcal{T}_2 (\partial_u\partial_v \phi_p) - \partial_u(\mathcal{T}_2)_p \partial_v \overline{\phi} - (\mathcal{T}_2)_p \partial_u\partial_v \overline{\phi}, \label{PSSESF:10} \\[.3em] 
    \label{PSSESF:11}
    \partial_u(\partial_v^2(r\phi_p)) &= \mathcal{G}_5 \partial_v(r\phi_p) + \partial_v \mathcal{G}_5 (r\phi_p) - \partial_v \mathcal{I}_5, \\[.3em] 
    \label{PSSESF:12}
    \partial_v(\partial_u^2(r\phi_p)) &= \mathcal{G}_5 \partial_u(r\phi_p) + \partial_u \mathcal{G}_5 (r\phi_p) - \partial_u \mathcal{I}_5.
\end{align}
\subsection{Integration lemmas}

\subsubsection{Basic transport estimate}
The following lemma estimates transport equations containing zeroth order terms with a favorable sign. Although straightforward, this lemma will be used repeatedly.
\begin{lem}
\label{lem:integrationmain}
Fix $v_0 < v_1$. Let $f(v):[v_0,v_1] \rightarrow \mathbb{R}$ be continuous on $[v_0,v_1]$ and $C^1$ on $(v_0,v_1]$. Assume $f$ satisfies the equation 
$$\partial_v f + c(v) f = E(v),$$
in $(v_0,v_1],$ where $c(v), E(v)$ are given functions with $c(v) \geq 0$. Assume $E(v) \in L^\infty([v_0,v_1])$, and $c(v) \in L^\infty([v',v_1])$for all $v' > v_0$. Then 
\begin{equation}
\label{eq:basictransportestimate}
\sup_{v \in [v_0,v_1]} |f(v)| \leq |f(v_0)|+\int_{v_0}^{v_1} |E|(v')dv'.
\end{equation}
An analogous result holds in the $u$ direction, for functions $g(u):[u_0,u_1] \rightarrow \mathbb{R}$ satisfying the equation
$$\partial_u g + c(u) g = E(u),$$
with $u_0 < u_1 $ and $c(u) \geq 0$.
\end{lem}
\begin{proof}
Fix $v \in (v_0,v_1]$, as well as a $\delta \ll 1$ sufficiently small satisfying $v_0 + \delta < v$. Since $c(v) \in L^\infty([v_0+\delta,v))$, we may apply the method of integrating factors to give
\begin{equation*}
    f(v) = e^{-\int_{v_0+\delta}^v c(v')dv'}f(v_0+\delta) +  e^{-\int_{v_0+\delta}^v c(v')dv'}\int_{v_0+\delta}^v e^{\int_{v_0+\delta}^{v'} c(v'')dv''} E(v')dv'.
\end{equation*}
By the non-negativity of $c(v),$ we estimate
\begin{align*}
    |f(v)| &\leq |f(v_0 + \delta)| + \int_{v_0+\delta}^v |E(v')|dv' \\
    &\leq |f(v_0 + \delta)| + \int_{v_0}^{v_1} |E(v')|dv'.
\end{align*}
Taking $\delta \rightarrow 0$ and taking supremum over $v \in [v_0,v_1]$ yields (\ref{eq:basictransportestimate}). The estimate for $g(u)$ follows in a similar manner.
\end{proof}

\begin{rmk}
The point of the above lemma is that zeroth order terms with favorable signs may be dropped from transport estimates, even when the coefficient is singular at one endpoint. We will often integrate in the \textit{decreasing} $u$ or $v$ direction, for which the relevant sign condition is $c(u), c(v) \leq 0$.
\end{rmk}

\subsubsection{Integration lemmas near the axis}
\label{subsubsec:axisintlemmas}
The spacetimes considered in this paper contain a regular center $\Gamma$ to the past of $\mathcal{O}$. With the goal of deriving pointwise bounds on double null unknowns, the axis presents challenges due to the various factors of $r^{-1}$ in the scalar field system. In this section we discuss one approach to deriving uniform estimates near the axis.

Consider the problem of showing boundedness of $\frac{m}{r^3},$ which is a desirable statement of regularity near $\Gamma$ for $C^1$ solutions. Assume bounds on $r\phi$ and its derivatives are given, say by estimating the wave equation (\ref{SSESF:2:6}). Inspecting (\ref{SSESF:2:4}) or (\ref{SSESF:2:5}) shows that an estimate on $m$ consistent with $O(r^3)$ behavior requires pointwise bounds on $\partial_u \phi$ or $\partial_v \phi$.

In regions with $r > 0$, the relation
\begin{equation*}
    \partial_u \phi = \frac{1}{r}(\partial_u(r\phi) - \nu \phi)
\end{equation*}
is sufficient to estimate $\partial_u \phi$ given a bound at the same order on $\partial_u(r\phi)$. However, it is clear that such a procedure cannot work up to the axis. 

An alternative approach is discussed in \cite{chris1} and \cite{lukohyang1}, with the latter presenting a general framework for near-axis analysis using the language of \textit{averaging operators}. Observe that the BV (or $C^1$) boundary conditions require $r\phi|_\Gamma =0$ pointwise, and hence one can write
\begin{equation}
    \label{eq:phiasanaverage}
    \phi(u,v) = -\frac{1}{r(u,v)}\int\limits_{(u,v)}^{(u_\Gamma(v),v)} \partial_u(r\phi)(u',v)du'.
\end{equation}
$\phi$ can thus be interpreted as an average in $u$ of $\partial_u(r\phi)$ (a corresponding formula holds in the $v$ direction). Differentiating averages of this form gives the required bounds on $\partial_u \phi$, at the cost of requiring bounds on $\partial^2_u(r\phi)$. 

More generally, let $(r,m,\phi)$ be a $C^1$ solution to the system (\ref{SSESF:1:1})-(\ref{SSESF:1:5}) in some domain $\mathcal{U} \subset \mathbb{R}^2$. Fix $(u,v) \in \mathcal{U}$, and assume that both the future directed constant $v$ line, and the past directed constant $u$ line, intersect the axis $\Gamma$ and are contained in $\mathcal{U}$. Let $f \in C^1(\mathcal{U})$ be a given function, and consider the quantities
\begin{equation}
    I_u[f](u,v) \doteq r(u,v)^{-1}\int\limits_{(u,v)}^{(u_\Gamma(v),v)} f(u',v)du', \ \ \
    I_v[f](u,v) \doteq r(u,v)^{-1}\int\limits_{(u,v_\Gamma(u))}^{(u,v)} f(u,v')dv'.
\end{equation}
The following lemma illustrates how to bound derivatives of these averages, assuming suitable control on the geometry through $\nu, \lambda$.
\begin{lem}
\label{lem:axisaverage}
\begin{equation} 
    \label{eq:averageestu}
   | \partial_u I_u[f]|(u,v)  \lesssim  \frac{\sup_{u' \in [u,u_\Gamma(v)]}|\nu|(u',v) }{\inf_{u' \in [u,u_\Gamma(v)]}|\nu|(u',v)} \sup_{u' \in [u,u_\Gamma(v)]}|\partial_u(\nu^{-1}f)|(u',v),
\end{equation}
\begin{equation} 
     \label{eq:averageestv}
   | \partial_v I_v[f]|(u,v)  \lesssim \frac{\sup_{v' \in [v_\Gamma(u),v]}| \lambda|(u,v') }{\inf_{v' \in [v_\Gamma(u),v]}|\lambda|(u,v')}\sup_{v' \in [v_\Gamma(u),v]}|\partial_v(\lambda^{-1}f)|(u,v').
\end{equation}
\end{lem}
\begin{proof}
The proofs of (\ref{eq:averageestu}) and (\ref{eq:averageestv}) are nearly identical, so we just discuss (\ref{eq:averageestu}). Explicitly differentiating $I_u[f](u,v)$ and integrating by parts gives 
\begin{align}
    \partial_u I_u[f](u,v) &= -r(u,v)^{-2}\nu(u,v) \int\limits_{(u,v)}^{(u_\Gamma(v),v)} f(u',v)du' - r(u,v)^{-1}f(u,v) \nonumber\\
    &= -r(u,v)^{-2}\nu(u,v) \int\limits_{(u,v)}^{(u_\Gamma(v),v)}\partial_u r(u',v) \Big(\frac{f}{\nu}\Big)(u',v)du'- r(u,v)^{-1}f(u,v) \nonumber\\
    &= r(u,v)^{-2}\nu(u,v)\int\limits_{(u,v)}^{(u_\Gamma(v),v)} r(u',v) \partial_u \Big(\frac{f}{\nu}\Big)(u',v)du'. \label{eq:averagingtemp}
\end{align}
Pulling out the supremum of $\partial_u(\nu^{-1}f)$ and factor of $r$, the integral can then be explicitly evaluated to give $|u-u_\Gamma(v)|$. Given a positive lower bound on $-\nu$, this difference can be bounded by a factor of $r(u,v)$. All the singular factors of $r$ cancel, and we are left with the estimate (\ref{eq:averageestu}).
\end{proof}

The above framework immediately gives bounds on derivatives of the scalar field. 
\begin{cor}
\label{cor:integralestimateinterior}
Introduce the shorthand notation $\sup_{u'} \doteq \sup_{u' \in [u,u_\Gamma(v)]}$, $\sup_{v'} \doteq \sup_{v' \in [v_\Gamma(u),v]}$, and similarly for $\inf_{u'}, \inf_{v'}$. With the running assumptions on $(r,m,\phi)$ and $\mathcal{U}$, we have
\begin{equation}
    \label{eq:lukohestimateu}
    |\partial_u \phi|(u,v) \lesssim \frac{\sup_{u'}|\nu|}{\inf_{u'}|\nu|^3}\sup_{u'}|\partial_u \nu| \sup_{u'}|\partial_u(r\phi)| + \frac{\sup_{u'}|\nu|}{\inf_{u'}|\nu|^2}\sup_{u'}|\partial_u^2(r\phi)|,
\end{equation}
\begin{equation}
    \label{eq:lukohestimatev}
    |\partial_v \phi|(u,v) \lesssim \frac{\sup_{v'}|\lambda|}{\inf_{v'}|\lambda|^3}\sup_{v'}|\partial_v \lambda| \sup_{v'}|\partial_v(r\phi)| + \frac{\sup_{v'}|\lambda|}{\inf_{v'}|\lambda|^2}\sup_{v'}|\partial_v^2(r\phi)|.
\end{equation}
We have dropped the arguments $(u',v), (u,v')$ of the terms appearing in the above estimates.
\end{cor}
\begin{proof}
Apply Lemma \ref{lem:axisaverage} with $f = \partial_u(r\phi), \partial_v(r\phi)$ and use the formula (\ref{eq:phiasanaverage}) along with its variant in the $v$ direction.
\end{proof}

\begin{rmk}
An analogous result to Corollary \ref{cor:integralestimateinterior} will be needed when the future directed constant $v$ line does not intersect the axis, but rather intersects data along some null line $\{u = u_\delta\}$. The proof of the previous lemma and corollary still goes through, because 
\begin{enumerate}
    \item The data term $r(u_\delta,v)\phi(u_\delta,v)$ drops out during the integration by parts step.
    \item The $u$ difference produced by evaluating the integral is now $|u - u_\delta|$, which can be bounded by $|u - u|_\Gamma|$, and therefore one still gains the extra power of $r$. 
\end{enumerate}
\end{rmk}
The argument extends to higher derivatives on $I_u, I_v$. We will need estimates at one derivative higher, captured in the following lemma:
\begin{lem}
\label{lem:axisaverage2}
\begin{align}
    | \partial_u^2 I_u[f]|(u,v)  &\lesssim \frac{\sup_{u'}|\partial_u \nu| }{\inf_{u'}|\nu|}\sup_{u'}|\partial_u(\nu^{-1}f)| 
    + \frac{\sup_{u'}|\nu|^2 }{\inf_{u'}|\nu|}\sup_{u'}|\partial_u(\nu^{-1}\partial_u(\nu^{-1}f))|, \label{eq:averageestu2}
\end{align}
\begin{align}
    | \partial_u^2 I_v[f]|(u,v)  &\lesssim \frac{\sup_{v'}|\partial_v \lambda| }{\inf_{v'}|\lambda|}\sup_{v'}|\partial_v(\lambda^{-1}f)|
    + \frac{\sup_{v'}|\lambda|^2 }{\inf_{v'}|\lambda|}\sup_{v'}|\partial_v(\lambda^{-1}\partial_v(\lambda^{-1}f))|. \label{eq:averageestv2}
\end{align}
\end{lem}
\begin{proof}
Differentiating (\ref{eq:averagingtemp}) gives 
\begin{align*}
    \partial_u^2 I_u[f](u,v) &= \partial_u \nu(u,v) r(u,v)^{-2} \int\limits_{(u,v)}^{(u_\Gamma(v),v)} r(u',v) \partial_u \Big(\frac{f}{\nu}\Big)(u',v)du' \\
    &+ \nu(u,v) \partial_u\Big(r(u,v)^{-2} \int\limits_{(u,v)}^{(u_\Gamma(v),v)} r(u',v) \partial_u \Big(\frac{f}{\nu}\Big)(u',v)du'\Big).
\end{align*}
The first term can be estimated by $|\partial_u \nu| |\partial_u(\nu^{-1}f)|.$ The second term can be treated precisely as in the previous lemma, i.e. integration by parts and placing an extra derivative on $\nu^{-1}\partial_u(\nu^{-1}f)$. The result is the stated bound (\ref{eq:averageestu2}). 

The same argument with $v$ instead of $u$ gives (\ref{eq:averageestv2}).
\end{proof}

\begin{rmk}
\label{rmk:higherorderaveraging}
The proof of Lemma \ref{lem:axisaverage} applies to a slightly more general class averaging operators of the form 
\begin{equation}
    \label{eq:higherorderaverageu}
    I_{s,u}[f](u,v) \doteq r(u,v)^{-s}\int\limits_{(u,v)}^{(u_\Gamma(v),v)}r(u',v)^{s-1}f(u',v)du'
\end{equation}
and
\begin{equation}
    \label{eq:higherorderaveragev}
    I_{s,v}[f](u,v) \doteq r(u,v)^{-s}\int\limits_{(u,v_\Gamma(u))}^{(u,v)}r(u',v)^{s-1}f(u,v')du.
\end{equation}
We will have use for this extension when considering estimates for quantities like $\frac{m}{r^3}$, which by (\ref{SSESF:2:4})-(\ref{SSESF:2:5}) can be written in the form (\ref{eq:higherorderaverageu})-(\ref{eq:higherorderaveragev}) with $s = 3$. Moreover, the explicit formula (\ref{eq:averagingtemp}) for derivatives continues to hold for higher $s$, with the $r$ weights naturally adjusted as a function of $s$.
\end{rmk}

\subsection{Notation and constants}
Throughout the paper, we will follow common convention for constants. Large constants will generally be denoted $C$, and can change from line to line. Similarly for small constants, which we will denote by $c$. When used without additional context, such constants can be taken to depend only on the background $(\overline{g},\overline{\phi})$ solution, as well as various absolute constants.

\vspace{5pt}
Denote by $f \lesssim g$ the statement that there exists a constant $C$ for which $f \leq C g$. In the context of a bootstrap argument, the constant $C$ is assumed to be independent of bootstrap assumptions. Similarly define $f \gtrsim g$. Finally, denote by $f \sim g$ the statements $f \lesssim g$ and $f \gtrsim g$.

\vspace{5pt}
The notation $\overline{C},\overline{c}$ is used for constants depending on the background solution when we wish to emphasize the role of the background.

\vspace{5pt}
When referring to constants that depend on a bootstrap parameter $A$, we write $C(A)$ for an increasing function of $A$ that is free to change throughout an argument. The exact dependence on $A$ will not be of interest from the perspective of improving bootstrap assumptions.

\vspace{5pt}
We say functions $f(u) \in C^k([u_1,u_2])$ if $f(u) \in C^k((u_1,u_2)),$ with all $k$ derivatives continuous at the endpoints. Similarly for functions $g(v)$.

\vspace{5pt}
Various expressions involving $\kappa_i$ will appear throughout the paper. To ease readability, we label some of the common expressions: 
$$q_i \doteq 1-\kappa_i, \ \ p_i \doteq \frac{1}{1-\kappa_i}. $$
Observe the algebraic relations $p_i q_i = 1$, and  $p_i \kappa_i = p_i-1$.

\vspace{5pt}
Smallness parameters like $\delta, \sigma, \beta$ are reused across different arguments. The parameter $\epsilon$, however, refers uniquely to the smallness parameter in the initial data perturbations. The following parameters have single uses throughout the paper, and reflect properties of the solution being constructed: $\kappa_i, q_i, p_i, a, b, c, d, \alpha, \epsilon$.

The letter $\mathcal{I}$ in general refers to initial data norms, $A$ to (large) bootstrap constants, and $\mathfrak{A}, \mathfrak{B}, \ldots$ to spacetime norms of the solution.

%% file: Figures/fig_mainpenrose.tex
\begin{tikzpicture}
\node[circle, label= above:{$\mathcal{O}$}, draw, inner sep =0, minimum size = .2cm] (n1) at (0,0) {};
\node[label = {[align=center]left: $\Gamma =\{v=-|u|^{1-\kappa_1}\}$ } ] (n2) at (-90:2.5) {};
\node[label = {[rotate = -45]above:$\{v=0\}$} ] (n3) at (-45:2) {};

\draw[thick] (n1) -- (-90:5);
\draw[dashed] (n1) -- (-45:3.55);
\draw[thick] (-90:5) -- ($(-90:5) + (45:9.55)$);

\draw[dashed] (n1) -- (45:6); 

\node[] (n4) at (.93,-2.5) {$\mathcal{Q}^{(in)}$};
\node[] (n5) at (5:3) {$\mathcal{Q}^{(ex)}$};

\node[label = {[rotate = 45]below:$\{u=-1\}$}] (n6) at ($(-90:5) + (45:6.55)$) {};
\node[label = {[rotate = 45]below:$\{u=0\}$}] (n7) at (45:3) {};

\node[] (n8) at (20:5) {};
\node[] (n9) at ($(20:5) + (45:2)$) {};
\node[label = {[rotate = 45]above:$v \rightarrow \infty$} ] (n10) at ($(20:5) + (45:1)$) {};

\draw[->] (n8) -- (n9);

\end{tikzpicture}

%% file: Figures/fig_newlocalexistence.tex
\begin{tikzpicture}[scale=.7]
\coordinate[] (n1) at (0,0) {};
\coordinate (n2) at (-45:8) {};
\coordinate (n3) at ($(n2) + (45:5) $) {};

\coordinate (n4) at ($(n2) + (45:2) $) {};
\coordinate (n5) at ($(n4) + (135:11)$) {};
\coordinate (n6) at (-90:3);
\coordinate (n7) at ($(n6) + (45:4.1) $) ;
\coordinate (n11) at ($(n6) + (45:3.1) $) ;
\coordinate (n8) at ($(n6) + (45:2.1) $ ) ;

\coordinate (n9) at ($(n4) + (135:5.85) $ );

\coordinate (n12) at ($(n11)+(135:3.1) $ );

\fill[lightgray!70] (n1) -- (n8) -- (n11) -- (n12);
\fill[lightgray!50] (n8) -- (n2) -- (n4) -- (n9);
\fill[lightgray!05] (n1) -- (n6) -- (n8);

\draw[thick] (n1) -- (90:3.0) {};
\draw[thick] (n1) -- (n2) -- (n3);

\draw[dashed] (n4) -- (n9);

\draw[dashed] (n11) -- (n12);

\node[above, label = {[yshift=-2mm]$\Gamma$}] at (90:3.0) {};
\node[label = {[rotate = -45]below: \small$\{v=v_0 \} $ }] at (-45:4) {};
\node[label = {[rotate = 45]below: \small$\{u=u_0 \} $ }] at ($(n2) + (45:3.9) $ ) {};

\node[circle, fill = black, inner sep = 0, minimum size = .2cm, label = below:{\small$(u_0,v_0)$}] at (n2) {};
\node[circle, fill = black, inner sep = 0, minimum size = .2cm, label = left:{\small$(-|v_0|^{p_1},v_0)$}] at (n1) {};
\node[circle, fill = black, inner sep = 0, minimum size = .2cm, label = left:{\small$(-|v_1|^{p_1},v_1)$}] at (n6) {};
\node[circle, fill = black, inner sep = 0, minimum size = .2cm, label = left:{\small$(-|v_0+\delta_2|^{p_1},v_0+\delta_2)$}] at (n12) {};
\node[circle, fill = black, inner sep = 0, minimum size = .2cm, label = {[xshift = -1mm, yshift = -2mm]right:\small$(u_0,v_0+\delta_1)$}] at (n4) {};

\draw[thick] (n1) -- (-90:4);

\draw[dashed] (n6) -- (n7);

\node[] at ($(-90:1.5) + (0:.7) $) {\small$\mathcal{U}_{(back)}$};
\node[] at ($(-45:5) + (45:1) $) {$\mathcal{U}_{1}$};
\node[] at ($(-45:.7) + (45:.5) $) {$\mathcal{U}_{2}$};

\draw [decorate, decoration = {brace,raise=2pt,mirror}] ($(n8)+(45:.1)$) -- ($(n11)+(-135:.1)$)
node[pos = .5, below = 6pt, sloped,below]{$\delta_2$};
\draw [decorate, decoration = {brace,amplitude = 4pt, raise=2pt,mirror}] ($(n2)+(45:.1)$) -- ($(n4)+(-135:.1)$)
node[pos = .5, below = 6pt, sloped,below]{$\delta_1$};

\end{tikzpicture}

%% file: sec3.tex
\section{Main Results}
\label{sec:mainresults}

\subsection{Admissible background spacetimes}
\label{sec:assumptionsbackground}
Theorems \ref{thm1:interior}, \ref{thm2:exterior} below will apply to a class of naked singularity solutions $(\overline{g}, \overline{\phi})$ defined on a fixed quotient manifold $\mathcal{Q}$. The metrics are modeled off of the $k$-self-similar Christodoulou spacetimes, in that they 1) satisfy similar gauge conditions and self-similar bounds, and 2) are asymptotically flat. The precise conditions are given below. One should think of assumptions (A1)-(A4) as the minimal requirements for constructing the naked singularity interior, and assumptions (B1)-(B3) as additional requirements for the exterior region.

It follows from Appendix \ref{appA} that the $k$-self-similar naked singularities satisfy all the assumptions below, for appropriate choices of parameters. For further discussion see Remark \ref{rmk:Christodoulourmk}.

\vspace{5pt}
\noindent
\underline{\textbf{(A1): Solution manifold and regularity}} \newline

\noindent
There exists a $\kappa_1 \in (0, \frac{1}{2})$ and a global double null coordinate system $(u,v)$ such that $\mathcal{Q}$ takes the coordinate form 
\begin{equation}
    \label{eq:coordformspacetime}
    \mathcal{Q} = \{(u,v) \ | \ -1 \leq u < 0, \ -|u|^{q_1} \leq v < \infty\}.
\end{equation}
Moreover, $\mathcal{Q}$ has a regular center $\Gamma$ parameterized by 
\begin{equation*}
    \Gamma = \{(u,v) \ | \ v = -|u|^{q_1}, \  -1 \leq u < 0 \},
\end{equation*}
with generator $T = p_1 |u|^{\kappa_1}\partial_u + \partial_v$. $\Gamma$ is a future incomplete timelike curve terminating in a singular point $\mathcal{O}$. The latter is denoted in coordinates as $(u,v) = (0,0)$, although we stress that $\mathcal{O}$ is not contained in $\mathcal{Q}$.

\vspace{4pt}
\noindent
With respect to this double null gauge, the spacetime $(\mathcal{Q},\overline{g},\overline{\phi})$ is a BV solution to the spherically symmetric Einstein-scalar field system away from $\mathcal{O}$ in the sense of Definition \ref{dfn:BV}, and a $C^1$ solution on $\mathcal{Q} \setminus (\{v=0\}\cup \{u=0\})$ away from $\mathcal{O}$, in the sense of Definition \ref{dfn:C1}.

\vspace{2pt}
\noindent
Finally, $\mathcal{Q}$ is assumed to be globally free of trapped surfaces, and satisfy the bound 
\begin{equation}
    \sup_{\mathcal{Q}}|\log(1-\mu)| \lesssim 1.
\end{equation}

\begin{rmk}
The choice of $\{u=-1\}$ for the past boundary of $\mathcal{Q}$ is arbitrary, and may be replaced by $\{u=u_{min}\}$ for any $u_{min}<0.$
\end{rmk}
\begin{rmk}
The requirement that $\kappa_1 \in (0,\frac{1}{2})$ is in fact only needed for the proof of Theorem \ref{thm2:exterior} concerning the exterior region. For the proof of Theorem \ref{thm1:interior} it is sufficient to assume $\kappa_1 \in (0, 1)$.
\end{rmk}

\vspace{2pt}
\noindent
In the assumptions (A2)-(B1) below, unless otherwise specified, all double null unknowns are computed with respect to the above $(u,v)$ coordinate system. Before proceeding with the self-similar bounds, recall from Section \ref{subsec:conseqsphersym} the definitions of the self-similar coordinate $s_1$, as well as the self-similar neighborhoods $\mathcal{S}_{near}, \mathcal{S}_{far}$. 

\vspace{7pt}
\noindent
\underline{\textbf{(A2): Gauge conditions and elementary bounds in $\mathcal{Q}^{(in)}$}} \newline

\noindent
Along $\{v=0\}$, fix the gauge condition
\begin{equation}
    \overline{r}(u,0) = |u|.
\end{equation}
It follows that $\overline{\nu}(u,0) = -1$. In the near-horizon region $\mathcal{S}_{near}$ we assume 
\begin{equation}
    \overline{r} \sim |u|.
\end{equation}
Additional low order bounds on the geometry and scalar field will be required in $\mathcal{Q}^{(in)}$. For the background geometry we assume
\begin{equation}
    (-\overline{\nu}) \sim 1, \ \  \overline{\lambda} \sim |u|^{\kappa_1}, \ \ |\overline{r}| \lesssim |u|, \ \  |\overline{m}| \lesssim r\label{eq:A2rates1},
\end{equation}
and for the scalar field we require bounds consistent with self-similarity:
\begin{equation}
    |\overline{\partial_u \phi}| \lesssim |u|^{-1}, \ \  |\overline{\partial_v \phi}| \lesssim |u|^{-1 + \kappa_1}.
\end{equation}

\vspace{5pt}
\noindent
\underline{\textbf{(A3): Higher order bounds in $\mathcal{Q}^{(in)}$}} \newline

\noindent
In $\mathcal{Q}^{(in)}$ we have the following higher order bounds:
\begin{align}
    \label{eq:A2rates2}
    |\partial_u \overline{\nu}| &\lesssim |u|^{-1} , \ \ \ |\partial_v \overline{\lambda}| \lesssim |u|^{-1+2\kappa_1}|s_1|^{-a}, \\[.5em]
    \label{eq:A2rates3}
    |\partial_u^2 \overline{\phi}| &\lesssim |u|^{-2}, \ \ \ |\partial_v^2 \overline{\phi}| \lesssim |u|^{-2+2\kappa_1}|s_1|^{-b}, \\[.5em]
    |\partial_u^2 \overline{\nu}| &\lesssim |u|^{-2}, \ \ \  |\partial_u^3 \overline{\phi}| \lesssim |u|^{-3}.
\end{align}
Here, $a,b \in [0,1)$ are parameters satisfying $a \leq b$. Finally, in the region $\mathcal{S}_{far}$ we assume
\begin{equation}
    |\partial_v^2 \overline{\lambda}| \lesssim |u|^{-2+3\kappa_1}, \ \ \ |\partial_v^3 \overline{\phi}| \lesssim |u|^{-3+3\kappa_1}.
\end{equation}

\begin{rmk}
We briefly comment on the regularity of the second order unknowns transversal to $\{v=0\}$. For $a,b = 0$ it follows that the solution is regular across $\{v=0\}$ away from the singular point, and may be viewed as a $C^1$ solution there. In general, the constraints of self-similarity lead to blowup of sufficiently high derivatives of transverse quantities as $v \rightarrow 0$. The bounds above with general $a,b$ therefore assert that in a suitable coordinate system, first derivatives of the metric and scalar field remain bounded, whereas second derivatives may blowup. The blowup is sufficiently slow that the rates are integrable in $v$, and the solution thus remains in the bounded variation class. In fact, the solutions are Hölder continuous with exponent depending on $a,b$. 

Observe that the third order bounds on transversal quantities do not specify blowup rates in $v$. In fact, we will only need these bounds in a self-similar neighborhood $\mathcal{S}_{far}$ of the axis, where $v$ and $u$ are equivalent (up to powers depending on $\kappa_1$). For this reason we ignore considerations of higher regularity of the solution near $\{v=0\}$.
\end{rmk}

\begin{rmk}
Observe that quantities tangential to $\{v=0\}$, namely $\overline{r}, \overline{m}, \partial_u \overline{\phi}$ and $u$ derivatives thereof, satisfy self-similar bounds with blowup/decay rates given by integer powers of $|u|$. To each $u$ derivative of a double null unknown is associated a gain of a single power of $|u|^{-1}$. 

In contrast, transversal quantities to $\{v=0\}$ see their rates modified by $\kappa_1$. Each $v$ derivative is associated heuristically with a gain of $|u|^{-1+\kappa_1}$. In the case of unknowns with singular behavior in $v$, the heuristic is made precise as follows: given an estimate $|\overline{\Psi}| \lesssim |u|^{-m}|v|^{-n}$, define the total weight to be $m + (1-\kappa_1)n$. Then each $v$ derivative of a $\overline{\Psi}$ is accompanied by a loss in $|u|$ and $|v|$, such that the total change in weight is $1-\kappa_1$.
\end{rmk}

\vspace{5pt}
\noindent
\underline{\textbf{(A4): Boundary conditions along $\Gamma$}}\newline

\noindent
The solution satisfies $C^1$ boundary conditions along $\Gamma$. More precisely, for $(u_\Gamma,v_\Gamma) \in \Gamma$ we have
\begin{align}
    &(\partial_v + p_1|u|^{\kappa_1}\partial_u)\overline{r}(u_\Gamma,v_\Gamma) = 0, \\[.3em] 
    &(\partial_v + p_1|u|^{\kappa_1}\partial_u)(\overline{r\phi})(u_\Gamma,v_\Gamma) = 0, 
\end{align}
\begin{equation}
    \overline{r}(u_\Gamma,v_\Gamma) = \overline{m}(u_\Gamma,v_\Gamma) = \overline{r\phi}(u_\Gamma,v_\Gamma)= 0.
\end{equation}
The derivatives are understood as limits along sequences of interior points approaching $\Gamma$. Note that the boundary conditions, in conjunction with $C^1$ bounds on the scalar field, in fact gives the stronger vanishing condition on the mass $\overline{m} = O(\overline{r}^3)$ to the past of $\mathcal{O}$.

Although the center fails to be regular precisely as $u \rightarrow 0$, we require certain asymptotic boundary conditions near the singular point:
\begin{equation}
   \lim_{(u,v) \rightarrow (0,0)} \overline{r}(u,v) = \lim_{(u,v) \rightarrow (0,0)} \overline{m}(u,v) = \lim_{(u,v) \rightarrow (0,0)} \overline{r \phi}(u,v) = 0.
\end{equation}
These conditions are to be understood along any sequence of points in $\mathcal{Q}$ approaching the singular point.

\vspace{5pt}
\noindent
\underline{\textbf{(B1): Bounds in $\mathcal{Q}^{(ex)}$ near $\{v=0\}$}} \newline 

\noindent
In a self-similar neighborhood of $\{v=0\}$ given by $\{0 \leq s_1 \leq 1\}$, the bounds (\ref{eq:A2rates1})-(\ref{eq:A2rates3}) continue to hold. \newline

\noindent
Moreover, the parameters $a,b$ controlling blowup of second-order unknowns as $v \rightarrow 0$ are required to satisfy the constraints\footnote{Note that $\kappa_1 < \frac{1}{2}$ is required for these conditions to be non-vacuous.}
\begin{equation}
\label{eqn:a,brequirements}
        \begin{cases}
            a \leq b < 1 - \kappa_1 p_1, & \text{or} \\[.5em]
            a < b = 1 - \kappa_1 p_1.
        \end{cases}
    \end{equation}

\begin{rmk}
For simplicity we assume the bounds on the blowup rates for $\partial_v \overline{\lambda}, \ \partial_v^2\overline{\phi}$ as $v \rightarrow 0$ are the same in the interior and exterior regions. This assumption is for notational convenience, and is not necessary.
\end{rmk}

\vspace{5pt}
\noindent
\underline{\textbf{(B2): Bounds in $\mathcal{Q}^{(ex)}$ near $\{u=0\}$ }} \newline 

\noindent
In the region $\{s_1 >1\}$ containing a self-similar neighborhood of $\{u=0\}$, we separately apply assumptions to the region near the singularity $\{s_1 > 1, v \leq 1 \}$ and to the asymptotically flat region $\{s_1 > 1, v \geq 1 \}$. Assumptions on the latter region are given in (B3) below.

In the near-singularity region, the solution is modeled on a self-similar solution with parameter $\kappa_2 \in (0,\frac{1}{2})$. It is therefore helpful to state bounds in a coordinate system adapted to both the new self-similar model and the $\{u=0\}$ hypersurface. Define
\vspace{-.5em}
\begin{equation}
    \label{eq:asympflat_coords}
    U = -|u|^{1-\kappa_2}, \ \ \ V = v^{\frac{1}{1-\kappa_1}},
\end{equation}
with associated metric functions ${^{(U)}}\overline{\nu} \doteq \partial_U \overline{r} , \  {^{(V)}}\overline{\lambda} \doteq \partial_V \overline{r}$, and ${^{(U,V)}}\overline{\Omega}$. Let $s_2 = \frac{V^{1-\kappa_2}}{|U|}$ be the natural self-similar coordinate. With respect to the $(U,V)$ coordinates, in the near-singularity region $\{s_2 \geq 1, V \leq 1 \}$ the solution satisfies bounds analogous to (A2)-(A3):
\begin{align}
    \overline{r} \sim V, \ \ &\ \ \   (-{^{(U)}}\overline{\nu}) \sim V^{\kappa_2}, \ \ \   {^{(V)}}\overline{\lambda} \sim 1, \ \ \ |\overline{m}| \lesssim V,
\end{align}
\vspace{-2.2em}
\begin{alignat}{2}
    |\partial_U \overline{\phi}| &\lesssim V^{-1+\kappa_2}, \quad \ \ \ \quad \quad |\partial_V \overline{\phi}| &&\lesssim V^{-1}, \\[.5em]
    |{\partial_U} {^{(U)}}\overline{\nu}| &\lesssim V^{-1+2\kappa_2}s_2^{c}, \ \ \ |{\partial_V}{^{(V)}}\overline{\lambda}| &&\lesssim V^{-1}, \\[.5em]
    |\partial_U^2 \overline{\phi}| &\lesssim V^{-2+2\kappa_2}s_2^{d}, \quad \ \ \quad |\partial_V^2 \overline{\phi}| &&\lesssim V^{-2}.
\end{alignat}
Here $c,d \in [0,1)$ are given, and we assume $c \leq d$. These latter two parameters are the analog of the $a,b$ introduced in assumption (A3), and both sets can be seen to control the level of Hölder regularity of the solution across the past and future lightcones of $\mathcal{O}$.

\vspace{5pt}
\noindent
\underline{\textbf{(B3): Bounds in $\mathcal{Q}^{(ex)}$ in the asymptotically flat region}} \newline 

\noindent
Assume the gauge is normalized along $\{U = -1\}$ so that ${^{(V)}}\overline{\lambda}(-1,V) = \frac{1}{2}$ for all sufficiently large $V$. Moreover assume outgoing null rays $\{U = c, \ V \geq 1\}$ for $c\in [-1,0)$ are asymptotically flat, and the solution satisfies bounds
\begin{align}
    \overline{r} &\sim V, \ \ \ (-{^{(U)}}\overline{\nu}) \sim 1, \ \ \ {^{(V)}}\overline{\lambda} \sim 1, \ \ \ \overline{m} \sim 1, 
\end{align}
\vspace{-2em}
\begin{alignat}{3}
    |\partial_U(\overline{r\phi})| &\lesssim 1,  \ \ \ &&|\partial_V (\overline{r\phi})| &&\lesssim V^{-2}, \\[.6em]
    |{\partial_U}{^{(U)}}\overline{\nu}| &\lesssim |U|^{-c}, \ \ \  &&|{\partial_V}{^{(V)}}\overline{\lambda}| &&\lesssim V^{-2},\\[.6em]
    |\partial_U^2 (\overline{r\phi})| &\lesssim |U|^{-d},\ \ \  &&|\partial_V^2(\overline{r\phi})| &&\lesssim V^{-3}.
\end{alignat}

\vspace{2em}
\noindent
We are now in a position to define the class of admissible spacetimes for which our main results apply.
\begin{dfn}
\label{dfn:admissibility}
Fix a parameter $\kappa_1 > 0$, and define the quotient manifold $\mathcal{Q}$ by (\ref{eq:coordformspacetime}). An \textbf{admissible naked singularity spacetime} is a solution to the spherically symmetric Einstein-scalar field system on $\mathcal{Q}$ satisfying assumptions (A1)-(A4) and (B1)-(B3) for a suitable set $(\kappa_2, a,b,c,d)$. An \textbf{admissible naked singularity interior} is a spacetime defined on $\mathcal{Q}^{(in)}$ satisfying the restricted set (A1)-(A4) of assumptions for a set of parameters $(a,b)$.
\end{dfn}

\begin{rmk}
A direct consequence of the estimates satisfied by the solution is that any admissible naked singularity spacetime admits an incomplete $\mathcal{I}^+$. The existence of $\mathcal{I}^+$ follows from the asymptotically flat bounds in (B3), which moreover imply that ingoing null geodesics reach $\{u=0\}$ in finite affine time for all sufficiently large $v$. For details on this computation, see Section \ref{subsubsec:asympflattrunc} in Appendix \ref{appA}.

Inextendibility of the solution in a neighborhood of $\mathcal{O}$ in turn follows from the positive lower bound on $\overline{\mu}(v,0)$. The structure of BV solutions is such that $\overline{\mu} \rightarrow 0$ holds along null lines approaching any regular central point.

Combining these observations, we conclude that $\mathcal{Q}$ is the maximal globally hyperbolic development for regular Cauchy data along $\{u = -1\}$, and null infinity is incomplete. The solutions considered here therefore represent globally naked singularity spacetimes.
\end{rmk}

\begin{rmk}
A consequence of (B2)-(B3) is that $\overline{r}, {^{(U)}}\overline{\nu}$, ${^{(V)}}\overline{\lambda},$ $\overline{\mu},$ $\partial_U \overline{\phi},$ $\partial_V \overline{\phi}$ achieve limits along ingoing null hypersurfaces as $U \rightarrow 0$. Therefore, these spacetimes extend continuously to the future lightcone of $\mathcal{O}$, and it follows that $\overline{r}(0,V) \rightarrow \infty$ as $V \rightarrow \infty$. 
\end{rmk}

\begin{rmk}
\label{rmk:Christodoulourmk}
After asymptotically flat truncation and a global scaling of the double null gauge, Christodoulou's $k$-self-similar solutions are examples of admissible naked singularity spacetimes (and their interior regions are consequently admissible interiors). It follows from the results of Appendix \ref{appA} that for these solutions, we have 
\begin{align*}
    (\kappa_1, \kappa_2) = (k^2, k^2), \ \ \ (a,b) = (0, 1-\frac{k^2}{1-k^2}), \ \ \ (c,d) = (0, 1-\frac{k^2}{1-k^2}).
\end{align*}
\end{rmk}

\subsection{Admissible data}
\label{sec:admissibledata}

The solutions constructed here are parameterized by their scalar field data along $\underline{\Sigma}_0 \cup \Sigma_{-1}^{(ex)}$. As perturbations of the background $(\overline{g},\overline{\phi})$ solutions, the data will be required to be $\epsilon$-close to the induced data from the background solution. However, generic $\epsilon$-small perturbations cannot be expected to produce globally naked singularities. The approach taken here to constructing data consistent with naked singularity formation is to choose data that vanishes as $u \rightarrow 0$ and $v \rightarrow 0$ at suitable rates.

\vspace{5pt}
\noindent
\textbf{Ingoing data along $\underline{\Sigma}_0$:} We first discuss the data posed along the null surface ingoing to the singularity, $\underline{\Sigma}_0$. The data consists of a choice of the scalar field perturbation (or the tangential derivative thereof), a gauge choice for $r$, and boundary conditions for $r, \ m$ as the singularity is approached. It will follow from Theorem \ref{thm1:interior} that this data is sufficient for constructing the interior region of the naked singularity.

In the following, $\alpha \in \mathbb{N}$ will be a large integer depending only on the admissible background solution, and $\epsilon > 0$ a small number measuring the size of the data perturbation.

The gauge is chosen such that $r_p(u,0) = 0$, i.e. the function $r$ agrees with its background value along $\underline{\Sigma}_0$. The free data then consists of a function $f_0(u) \in C^2(\underline{\Sigma}_0),$ representing a choice of $\partial_u(\overline{r}\phi_p)(u,0)$. In addition to the assumed regularity of $f_0$, we require that $f_0$ vanish \textit{sufficiently quickly} as $u \rightarrow 0$. More precisely, assume the bounds
\begin{equation}
    |\partial_u^i f_0|(u) \lesssim |u|^{\alpha-i},
\end{equation}
for $i=0,1,2$. The data along $\underline{\Sigma}_0$ is thus given by 
\begin{align}
     r(u,0) &= \overline{r}(u,0),\label{intdata1}\\[.5em]
    \partial_u (r\phi)(u,0) &= \partial_u (\overline{r \phi})(u,0) + \epsilon f_0(u),
\end{align}
with boundary condition at $(u,v) = (0,0)$ given by 
\begin{equation}
   \lim_{u\rightarrow 0} m_p(u,0) =  \lim_{u\rightarrow 0} \overline{r}\phi_p(u,0) = 0. \label{intdata2}
\end{equation}

The perturbation data is specified for the scalar field derivative $\partial_u(\overline{r}\phi_p)$, which after integrating and using the boundary conditions (\ref{intdata2}) yields the values of $\phi_p(u,0), \ \partial_u \phi_p(u,0)$ as well. A similar procedure using (\ref{SSESF:2:4}) and the boundary condition (\ref{intdata2}) gives the value of $m(u,0),$ and thus $m_p(u,0)$. 

\vspace{8pt}
\noindent
\textbf{Outgoing data along $\Sigma^{(ex)}_{-1}$:}
Constructing the exterior region requires the additional choice of data along the outgoing null surface $\Sigma^{(ex)}_{-1}$. Care must be taken to ensure the outgoing data does not activate the blueshift instability along $\underline{\Sigma}_0$, and that the interior and exterior solutions glue with maximal regularity. 

It is instructive to first consider the point $(u,v) = (-1,0)$. It will follow from the interior construction that the quantities transversal to $\underline{\Sigma}_0$, namely $\lambda$ and $\partial_v (r\phi)$, will have well-defined limits as $v \rightarrow 0^{-}$. Label these limits $\lambda^{(in)}(-1,0), \ $ $\partial_v (r\phi)^{(in)}(-1,0)$. To maximize the regularity of the spacetime arising from gluing the interior and exterior, it is natural to demand that the limits of these transversal quantities along $\Sigma^{(ex)}_{-1}$ as $v \rightarrow 0^+$ agree with the interior limits. Moreover, we require that the imposed values $r(-1,0) \ $, $m(-1,0)$ agree with their limits along $\underline{\Sigma}_0$.

With these regularity requirements in mind, we proceed to define the outgoing data. Define a smooth cutoff function $\chi_1(v)$ on $[0,\infty)$ satisfying $\chi_1(v) = 1$ for $v \in [0,1], \ $ $\chi_1(v) = 0$ for $v \geq 2, \ $ and $\chi_1 \leq 1$. Moreover assume all derivatives are bounded, and $\chi$ is mean zero. This cutoff is adapted to a neighborhood of $(u,v) = (-1,0)$, and will be needed in the following.

Assuming that an appropriate interior solution is given, define the constant
\begin{equation}
    \Lambda \doteq \lambda^{(in)}(-1,0) - \overline{\lambda}(-1,0),
\end{equation}
and specify the outgoing data for $\lambda(-1,v)$ as
\begin{equation}
    \label{extdata1}
    \lambda(-1,v) = \overline{\lambda}(-1,v) + \chi_1(v) \Lambda.
\end{equation}
Along with the condition $r(-1,0) = \overline{r}(-1,0)$, the choice of $\lambda(-1,v)$ fully determines $r(-1,v)$. Note the term proportional to $\Lambda$ is only active near $(-1,0)$, and serves to match the interior and exterior value of $\lambda$ there. Moreover, as the cutoff function $\chi_1(v)$ is mean zero, we have that $r_p(-1,v) = 0$ for all $v \geq 1$.

It remains to specify the outgoing scalar field. The free data consists of a function $g_0(v)\in C^1(\Sigma^{(ex)}_{-1})$ satisfying appropriate asymptotically flat bounds. In the $(U,V)$ gauge (cf. assumption (B2), (\ref{eq:asympflat_coords}) in Section \ref{sec:assumptionsbackground}), assume the bounds
\begin{equation}
    |g_0(V)| \lesssim V^{-2}, \ \ \ |\partial_V g_0(V)| \lesssim V^{-3}.
\end{equation}
Moreover, define the constant
\begin{equation}
    \Phi \doteq  \partial_v (r\phi)^{(in)}(-1,0) - \partial_v(\overline{r\phi})(-1,0),
\end{equation}
which depends on the solution in the interior region, as well as the background solution values. The constant $\Phi$ serves a similar role as $\Lambda$ above, ensuring the matching of outgoing data in the interior and exterior region. 

We may now express the outgoing scalar field data as
\begin{equation}
     \partial_v (r\phi)(-1,v) =  \partial_v (\overline{r\phi})(-1,v)+ \chi_1(v) \Phi +\epsilon \Bigl \{ \chi_1(v) v + \Big(1-\chi_1(v)\Big) \Bigr \}  g_0(v),\label{extdata2}
\end{equation}
with $\chi_1$ denoting the same cutoff function as above. The data prescribed for $ \phi(-1,v)$ is rather complicated, so we comment here on its behavior in various regimes. Near $(-1,0)$, the contribution of $g_0(v)$ is negligible, and $ \partial_v (r\phi)$ approaches the value $\partial_v (r\phi)^{(in)}$. One should interpret the rigidity in the choice of $ \partial_v (r\phi)(-1,0)$ as a reflection of the instability of the naked singularity exterior. It is only the higher derivatives of $r\phi$ that may be specified essentially arbitrarily.

In contrast, for large $v$ the only requirement for the data is asymptotic flatness. Thus for $v \gg 1$ the data reduces to that of the background, plus the arbitrary asymptotically flat contribution of $g_0(v)$. The coefficient of $g_0(v)$ in the data above serves to smoothly interpolate between the regimes of vanishing as $v \rightarrow 0$, and decay as $v \rightarrow \infty$.

\vspace{8pt}
\noindent
\textbf{Initial data norms:}
We conclude this section by defining norms encompassing the decay assumptions on initial data. Define
\begin{equation}
    \label{eq:intdatanormE1}
    \mathcal{E}_{1,\alpha}  \doteq \sum_{i=0}^{2} \sup_{u \in [-1,0]}||u|^{-(\alpha-i)}\partial_u^i f_0(u)| ,
\end{equation}
and 
\begin{equation}
\label{eq:intdatanormE2}
    \mathcal{E}_{2,\alpha} \doteq \mathcal{E}_{1,\alpha} + \sum_{i=0}^1 \sup_{V \in [0,\infty)}|V^{2+i} \, \partial_V^i g_0(V)|.
\end{equation}

\subsection{Statement of theorems}
\label{sec:thmstatements}
The first theorem asserts that admissible naked singularity interiors are backwards stable with respect to a special class of data perturbations along $\{v=0\}$. The resulting spacetime, defined on $\mathcal{Q}^{(in)}$, asymptotically converges to the background $(\overline{g},\overline{\phi})$ as $u \rightarrow 0$.

\begin{theorem}
\label{thm1:interior}
Fix a choice of an admissible naked singularity interior $(\mathcal{Q}^{(in)}, \overline{g}, \overline{\phi})$. There exists $\alpha \gg 1$ large enough depending on the background solution such that the following backwards stability result holds:

Let $f_0(u) \in C^2([-1,0])$ be given, and define data for $r$, $\phi$  along $\{v=0\}$ as in (\ref{intdata1})-(\ref{intdata2}). There exists $\epsilon$ sufficiently small depending on the background solution and the norm $\mathcal{E}_{1,\alpha}$, and a BV solution\footnote{More precisely, the solution is BV on domains of the form $\mathcal{Q}^{(in)} \setminus \{ u \leq u_\delta\}$ for all $u_\delta < 0$. Similar considerations apply for all statements of regularity on domains containing a neighborhood of $\mathcal{O}$.} to the spherically symmetric Einstein-scalar field system in $\mathcal{Q}^{(in)}$ achieving the data on $\{v = 0\}$. 

The following properties of the solution hold: 
\begin{itemize}
    \item (Regularity) The interior solution is $C^1(\mathcal{Q}^{(in)} \setminus \{v=0\})$ away from $\mathcal{O}$.
    \item (Convergence to background) The solution asymptotically converges to the background as $u \rightarrow 0$, with rates
    \begin{align}
      \label{eq:thm1bounds1}
        |r - \overline{r}| &\lesssim \epsilon \overline{r}|u|^{\alpha}, \quad \ \ \quad \quad  \quad |\nu - \overline{\nu}| \lesssim \epsilon|u|^{\alpha}, \quad \quad \quad \quad |\lambda - \overline{\lambda}| \lesssim \epsilon|u|^{\alpha+\kappa_1}, \\[.5em]
        \label{eq:thm1bounds2}
        |\mu - \overline{\mu}| &\lesssim \epsilon \overline{r}^2|u|^{\alpha-2}, \ \ \ |\partial_u \phi - \partial_u \overline{\phi}| \lesssim \epsilon|u|^{\alpha-1}, \ \ \ |\partial_v \phi - \partial_v \overline{\phi}| \lesssim \epsilon|u|^{\alpha-1+\kappa_1}.
    \end{align}
\end{itemize}
\end{theorem}

\vspace{5pt}
\noindent
Our second result considers the forward problem in the exterior region of the naked singularity. These exterior solutions are compatible across $\{v=0\}$ with the interior regions constructed in the above theorem, and after gluing we arrive the desired globally naked singularity spacetimes.

\begin{theorem}
\label{thm2:exterior}
Fix a choice of admissible naked singularity spacetime $(\mathcal{Q}, \overline{g},\overline{\phi})$ and choice of data (\ref{intdata1})-(\ref{intdata2}) along $\{v=0\}$. Assume an interior solution defined on $\mathcal{Q}^{(in)}$ has been constructed as in Theorem \ref{thm1:interior}. Supply outgoing data along $\Sigma^{(ex)}_{-1}$ given by (\ref{extdata1}), (\ref{extdata2}).

Then there exists $\epsilon$ sufficiently small depending on the background solution and the norm $\mathcal{E}_{2,\alpha}$, and a BV solution to the spherically symmetric Einstein-scalar field system on $\mathcal{Q}^{(ex)}$ achieving the data on $\underline{\Sigma}_0 \cup \Sigma^{(ex)}_{-1}$. 

The following properties of the solution hold:
\begin{itemize}
    \item (Regularity) The exterior solution is $C^1(\mathcal{Q}^{(ex)} \setminus (\{v=0 \} \cup \{u=0\} ))$ away from $\mathcal{O}$.
    \item (Convergence to background near $\{v=0\}$) There exists $0 < \delta <1 $ small such that in $\mathcal{Q}^{(ex)} \cap \{\frac{v}{|u|^{q_1}} \lesssim 1\}$, the solution asymptotically converges to the background spacetime as $u \rightarrow 0$ with the rates
    \begin{align}
        |r - \overline{r}| &\lesssim \epsilon \overline{r}|u|^{\delta}, \quad \ \  \quad |\nu - \overline{\nu}| \lesssim \epsilon|u|^{\delta}, \quad \quad \ \ \quad \quad  |\lambda - \overline{\lambda}| \lesssim \epsilon|u|^{\kappa_1+\delta}, \\[.5em]
        |\mu - \overline{\mu}| &\lesssim \epsilon |u|^{\delta}, \ \ \  |\partial_u \phi - \partial_u \overline{\phi}| \lesssim \epsilon |u|^{-1+\delta}, \ \ \ |\partial_v \phi - \partial_v \overline{\phi}| \lesssim \epsilon|u|^{-1+\kappa_1+\delta}.
    \end{align}
    \item (Convergence to background near $\{u=0\}$) There exists $0 < \delta, s <1 $ small such that in $\mathcal{Q}^{(ex)} \cap \{\frac{V^{q_2}}{|U|} \gtrsim 1\}$, the solution asymptotically converges to the background spacetime as $V \rightarrow 0$ with the rates
    \begin{align}
        |r - \overline{r}| &\lesssim \epsilon^{1-s} \overline{r} V^{\delta}, \ \ \ |^{(U)}\nu - ^{(U)}\overline{\nu}| \lesssim \epsilon^{1-s}V^{\kappa_2+\delta}, \quad \quad \ \  |^{(V)}\lambda - ^{(V)}\overline{\lambda}| \lesssim \epsilon^{1-s}V^{\delta}, \\[.5em]
        |\mu - \overline{\mu}| &\lesssim \epsilon^{1-s} V^{\delta}, \quad \ \  |\partial_U \phi - \partial_U \overline{\phi}| \lesssim \epsilon^{1-s} V^{-1+\kappa_2+\delta}, \ \ \ |\partial_V \phi - \partial_V \overline{\phi}| \lesssim \epsilon^{1-s} V^{-1+\delta}. \label{eq:lastbound}
    \end{align}
    \item (Existence of $\{u=0\}$ cone) The unknowns $r, \ ^{(U)}\nu, \ ^{(V)}\lambda, \mu, \partial_U \phi, \partial_V \phi,$ extend continuously to $\{u=0\}$ for all $v \neq 0$. Moreover, $\lim_{v \rightarrow \infty}r(0,v) = \infty$.
    \item (Incomplete null infinity) The spacetime $\mathcal{Q}^{(ex)}$ is asymptotically flat, and contains an incomplete $\mathcal{I}^+$ in the sense of Definition \ref{dfn:incompleteI}.
\end{itemize}
\end{theorem}

We conclude this section with remarks on the proof statement. 

\begin{rmk}
A consequence of the proof of Theorem \ref{thm1:interior} is that the interior solutions may not be unique. Besides the specified data along $\{v=0\}$, the proof requires additional (arbitrary) choices of data throughout the construction. 
\end{rmk}

\begin{rmk}
All the solutions considered in this paper (including the original $k$-self-similar solutions of \cite{chris2}) lie in the absolutely continuous class of solutions for which the instability proof of \cite{chris3} holds. It follows that the solutions constructed by Theorems \ref{thm1:interior}, \ref{thm2:exterior} are unstable to trapped surface formation with respect to a two dimensional family of perturbations supported in the exterior region.
\end{rmk}

\begin{rmk}
The theorems assert increased regularity of the $\Psi_p$ as $(u,v) \rightarrow (0,0)$ with respect to the background $\overline{\Psi}$. One can also ask about the optimal regularity across $\{v=0\}$. Recall $\partial_v \overline{\lambda}, \ \partial_v^2 \overline{\phi}$ are allowed to be singular across $\{v=0\}$ at integrable rates. In fact, one can show using the methods here that $\partial_v^2 \phi_p $ is less singular than the corresponding background solution, and provided $\partial_v \overline{\lambda}$ is bounded across $\{v=0\}$, so is $\partial_v^2 \phi_p$. In particular, all perturbations of $k$-self-similar solutions constructed in this paper have $\partial_v \lambda_p, \ \partial_v^2 \phi_p$ bounded across $\{v=0\}$.
\end{rmk}

\begin{rmk}
The value of $\alpha$ is not optimized in the proof of Theorem \ref{thm1:interior}, but the requirements are explicit in terms of norms of the background solution. For Christodoulou's $k$-self-similar solutions with $\kappa_1 = k^2$, it is instructive to note that pointwise norms of the background blowup as $k \rightarrow 0$ due to the largeness of terms in the background solution, e.g. $\partial_v \overline{\phi} \sim \frac{1}{k}.$ One could restate the bounds (\ref{eq:thm1bounds1})-(\ref{eq:lastbound}) with appropriate norms of the background solution to arrive at a result uniform as $k \rightarrow 0$, although we do not do this here.
\end{rmk}

\begin{rmk}
The rates of convergence to the background in both theorems are stated only for a subset of all double null unknowns. Similar statements exist for the higher order quantities, although subtleties exist at the axis and future horizon $\{u=0\}$. More precisely,
\begin{enumerate}[(i)]
    \item In $\mathcal{Q}^{(in)}$, corresponding estimates for $\partial_u \nu_p, \partial_v \lambda_p, \partial_u^2 \phi_p, \partial_v^2 \phi_p, \partial_u^3(r\phi_p), \partial_v^3(r\phi_p)$ are given in the course of the proof. However we are not able to estimate $\phi_p \in C^3$ near the axis, despite assuming such a bound on $\overline{\phi}$. The loss in regularity is related to the techniques used to estimate non-degenerate derivatives near $\Gamma$, and potentially could be avoided using different techniques.
    \item In $\mathcal{Q}^{(ex)}$, as $U \rightarrow 0$ we lose control on the singular behavior of transversal quantities ${\partial_U}{^{(U)}}\nu_p, \partial_U^2 \phi_p$ relative to the background solution. Although we assume only that ${\partial_U} {^{(U)}}\overline{\nu} \lesssim |U|^{-c}, \ \partial_U^2 \overline{\phi} \lesssim |U|^{-d}$ for integrable rates $c,d$, the rates propagated on the differences are of the form $|U|^{-1+\sigma},$ for a small constant $\sigma$. It follows that the solution remains in BV up to $\{U=0\},$ but higher order transversal quantities are not guaranteed to stay $\epsilon$-close to the background solution.
\end{enumerate}
\end{rmk}

\begin{rmk}
An important simplification of the interior construction is the backwards nature of the problem, namely data is supplied along the \enquote{final} surface $\{v=0\}$ containing the singularity. A backwards formulation allows one to avoid discussing singularity formation, and at key steps in the argument it is shown that integrating \textit{away} from the singular point is what allows the strong control on the perturbations $\Psi_p$ as $u \rightarrow 0$.

The exterior problem has the added complications of being a forwards construction, as well as including both similarity horizons $\{v=0\}$ and $\{u=0\}$. To decompose the exterior into subregions, each with controllable dynamics, we rely heavily on the corresponding decomposition of \cite{igoryak2}. The forwards problem relies more delicately on the size of constants associated to the background solution, and thus an unfortunate feature of the proof is that the strong $|u|^{\alpha}$ decay proven for the differences in the interior cannot be propagated into any full self-similar neighborhood of $\{v=0\}$. It is possible to propagate only a small $|u|^{\delta}$ improvement for metric and scalar field perturbations.
\end{rmk}

\subsection{Proof outline}
\label{subsec:introproofoutline}

\subsubsection{Theorem \ref{thm1:interior}: Interior construction}

\textit{Approximate interiors: }Due to the low regularity of ingoing scalar field data along $\underline{\Sigma}_0$, standard local existence results for domains of the form $\mathcal{Q}^{(in)}$ do not apply. Indeed, the ultimate goal is to construct a solution in a neighborhood of the cone point $\mathcal{O}$, but we will have to approach the solution in a different manner. 

Instead we consider a series of \textit{approximate interior solutions}, adapted to truncated domains of the form $\mathcal{Q}^{(in),u_\delta} \doteq \mathcal{Q}^{(in)} \cap \{u \leq u_\delta\}$ (see Figure \ref{fig:approxint}). Here, $u_\delta < 0$ is a small parameter that will eventually be taken to $0$. We proceed to solve the problem with trivial outgoing data for the perturbations $\Psi_p$ along $\{v \leq 0, \ u=u_\delta \}$, and cutoff ingoing data along $\{v=0, \ u\leq u_\delta\}$. The resulting problem takes place a finite distance from $\mathcal{O}$, and local existence theory (see Proposition \ref{prop:weirdlocalex}) gives a local solution in a full neighborhood of $\{v \leq 0, \ u=u_\delta \}$. The key analytical step of the proof is extending this local solution to the domain $\mathcal{Q}^{(in),u_\delta}$, which we turn to next.

\begin{figure}[h]
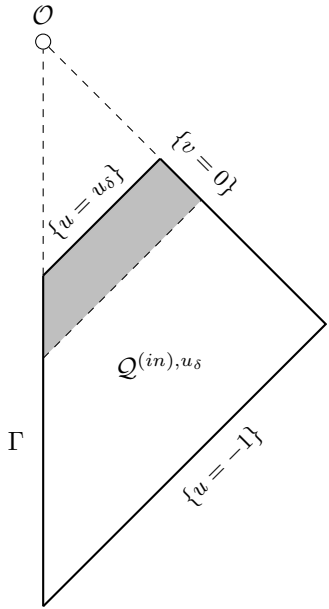

    \centering
  \includestandalone[]{Figures/fig_approxint}
  \caption{Setup for the approximate interior construction. Region guaranteed by local existence theorem shaded.}
  \label{fig:approxint}
\end{figure}

\noindent
\textit{Main transport estimates: }Although $\mathcal{Q}^{(in),u_\delta}$ occupies a compact domain in $(u,v)$ coordinates, proving existence of a solution uniformly in $u_\delta$ requires solving a global existence problem. The admissible backgrounds are in general large in pointwise norms (and, due to self-similar bounds, may be quite singular) for $|u|$ small, and thus we must use the structure of the background solution in an essential way. 

The main difficulty is apparent already in the linearized problem. Let $\Psi \in \{ \nu, \lambda, \phi, \mu, \partial_u \phi, \partial_v \phi\}$ denote a double null unknown\footnote{This list is only for illustrative purposes. In practice a larger set of quantities are required, including derivatives of $r\phi$ and higher derivatives of the solution.}. We wish to show estimates on the $\Psi_p$ consistent with $\epsilon$-smallness.

Associated to each $\Psi$ is a self-similar bound, which roughly captures the blowup/decay rate of the corresponding unknown in the $k$-self-similar model solutions. For example, $w_{\partial_u \phi} = |u|$ and $w_{\partial_v \phi} = |u|^{1-\kappa_1}$. We stress that the discussion here is informal, and in the bootstrap some weights are modified by factors of $\log |u|, r,$ and $v$ as required.

The linearized system for the weighted $w_\Psi \Psi_p$ becomes a coupled system of transport equations, of the schematic form
\begin{align}
    \partial_u (w_\Psi \Psi_p) + \frac{\overline{c}+ O(\epsilon)}{|u|} (w_\Psi \Psi_p) = O\Big(\frac{\epsilon}{|u|}\Big), \ \ \ \partial_v (w_\Psi \Psi_p) + \frac{\overline{d} + O(\epsilon)}{|u|^{1-\kappa_1}} (w_\Psi \Psi_p) = O\Big(\frac{\epsilon}{|u|}\Big). \label{eq:1.4pfoutline}
\end{align}
The constants $\overline{c}, \overline{d}$ vary from equation to equation, and reflect properties of the background solution. 

The $u$ equations are problematic for two reasons: 1) the zeroth order terms may have unfavorable signs, reflecting a tendency for quantities to grow, and 2) the error terms just fail to be integrable. For $v$ equations the integrability of singular $u$ coefficients is no longer a problem. Still, without knowledge of the signs of the $\overline{c}, \overline{d}$, deriving pointwise estimates from the above equations alone may not be possible.

The strategy is to exploit both 1) the rapidly decaying initial data for the $\Psi_p$, and 2) the backwards nature of the problem. The latter implies that integration along characteristics proceeds \textit{away} from the singular point. Taken together, these factors allow us to propagate strong $|u|$ decay on all the $\Psi_p$ uniformly as $u \rightarrow 0$. 

We thus attempt to propagate bounds of the form $|\Psi_p| \lesssim w_\Psi |u|^{\alpha}$. Conjugating the $u$ transport equation in (\ref{eq:1.4pfoutline}) by $|u|^{-\alpha}$ yields an new low order term, giving 
\begin{equation*}
    \partial_u (|u|^{-\alpha}w_\Psi \Psi_p) - \frac{\alpha+\bar{c}+ O(\epsilon)}{|u|}(|u|^{-\alpha}w_\Psi \Psi_p) = O\Big(\frac{\epsilon}{|u|}\Big).
\end{equation*}
For $\alpha$ large the zeroth order coefficient becomes strictly negative. Contracting with $|u|^{-\alpha}w_\Psi \Psi_p$ and integrating in $u$ from $\{u = u_\delta\}$, we conclude
\begin{align*}
    ||u|^{-\alpha}w_\Psi \Psi_p|^2(u,v) +& (\alpha - \bar{c} + O(\epsilon))\int\limits_{(u,v)}^{(u_\delta,v)}\frac{1}{|u'|}(|u'|^{-\alpha}w_\Psi \Psi_p)^2(u',v)du' \\
    &\lesssim O(\epsilon) \int\limits_{(u,v)}^{(u_\delta,v)} \frac{1}{|u'|} ||u'|^{-\alpha}w_\Psi \Psi_p|(u',v)du' \\
    &\lesssim \frac{\alpha}{2}\int\limits_{(u,v)}^{(u_\delta,v)} \frac{1}{|u'|}(|u'|^{-\alpha}w_\Psi \Psi_p)^2(u',v)du' + \alpha^{-1}O(\epsilon^2)\int\limits_{(u,v)}^{(u_\delta,v)}\frac{1}{|u'|}du'.
\end{align*}
We have used that the $\Psi_p$ vanish along $\{u=u_\delta\}$. Absorbing the first integral on the right hand side leaves an error term of size $\alpha^{-1}O(\epsilon^2)$, which up to a logarithmic divergence, gives an improvement for $\alpha$ large as a function of the implied constants. Importantly all implied constants can be chosen to depend only on the background solution. It is worth noting that the ability to cure the logarithmic divergence is again a function of integrating away from the singular point.

This pattern is carried through for unknowns satisfying $u$ transport equations. A similar pattern works for those only satisfying $v$ equations, after replacing the weight $|u|^{-\alpha}$ with a weight $(|u|^{q_1} + |v|)^{-\alpha p_1}.$ The geometry of the interior region $\mathcal{Q}^{(in)}$ implies that this weight is comparable to $|u|^{-\alpha}$, and also produces good lower order terms after conjugating.

We conclude this section by remarking on various technical features of the proof:\\

\vspace{4pt}
\noindent
\textit{Boundary conditions and regularity at the axis: }Although we only discussed boundary terms appearing along $\{u=u_\delta\}$ above, integrating the $u$ transport equations often produces boundary terms along $\Gamma$. The values along $\Gamma$ are not known a priori; however, boundary conditions relate these terms to the values of \textit{different} unknowns, which themselves satisfy $v$ equations. It follows that as long as $v$ unknowns are estimated first, their improved bootstrap assumptions are sufficient to control boundary terms on $\Gamma$.

An additional complication due to $\Gamma$ concerns the presence of various singular factors in the system proportional to $r^{-1}$. We have ignored these issues above, only identifying terms with singular $|u|$ weights. As is already apparent in proofs of BV well-posedness in \cite{chris1}, key to avoiding singularities at the axis is control on the second order unknowns $\partial_u \nu, \ \partial_v \lambda, \ \partial_u^2(r\phi), \ $ and $\partial_v^2(r\phi)$ near $\Gamma$. Combined with the averaging estimates of \cite{lukohyang1}, we will have sufficient regularity to bound factors of the form $\frac{m}{r^3}$, which in turn leads to estimates on all unknowns near $\Gamma$.

\vspace{5pt}
\noindent
\textit{Convergence of approximate solutions: }The final step in the proof is extracting a limit of the truncated solutions as $u_\delta \rightarrow 0$. By extending the solutions on $\mathcal{Q}^{(in),u_\delta}$ appropriately on $u \geq u_\delta$, we arrive at a sequence of solutions defined on the common coordinate domain $\mathcal{Q}^{(in)}$. 

In studying the differences of solutions along the sequence, we arrive at a problem similar to that of estimating the $\Psi_p$ above, but with non-trivial data along outgoing hypersurfaces $\{u=u_\delta\}$. For this reason, in the construction of approximate interiors we allow for non-trivial outgoing data from the start, and derive all estimates subject only to finiteness of an appropriate initial data norm.

An additional complication arises due to a loss of regularity occurring in the construction of the approximate interior solutions. Given that $\overline{\phi} \in C^2(\mathcal{Q}^{(in)})$ for the admissible background solution, we will only be able to propagate $r\phi \in C^2(\mathcal{Q}^{(in)})$ on the approximate interior. In particular, near the axis such an estimate only gives $\phi \in C^1(\mathcal{Q}^{(in)})$. In order to apply the regularity results on approximate interiors to the \textit{difference} of approximate interiors, we revisit the proof of scalar field bounds and propagate higher regularity. This analysis takes place in Section \ref{subsubsec:thirdorderestimates}. It is here that the additional requirement $\overline{\phi} \in C^3(\mathcal{Q}^{(in)})$ on the background solution (modulo blowup as $v \rightarrow 0$) is needed, and we are then able to conclude  $\phi \in C^2(\mathcal{Q}^{(in)})$ uniformly along the sequence of approximate solutions. With this added regularity, we are able to take the $u_\delta \rightarrow 0$ limit in Section \ref{sec:limitingsolution} and complete the proof of Theorem \ref{thm1:interior}.

\subsubsection{Theorem \ref{thm2:exterior}: Exterior construction}

The starting point for the proof of Theorem \ref{thm2:exterior} is an application of characteristic local well-posedness to guarantee the existence of a solution in a neighborhood of $(u,v) = (-1,0)$. The region of existence can be extended to a full neighborhood of the set $\{v = 0, \ u \leq -u_\delta\} \cup \{u=-1, \ v \geq 0\}$ for all $u_\delta < 0$. However, the presence of the singular point $\mathcal{O}$ rules out the possibility of using local existence to take $u_\delta \rightarrow 0$. As for the interior, a hypothetical extension to a neighborhood of $\mathcal{O}$ requires proving uniform bounds. Fortunately, one can avoid working with approximate solutions as was required above, and it is sufficient to estimate the solution on $\mathcal{Q}^{(ex)}.$

We note here that unlike in the interior region, where the estimates for first and second order double null unknowns are mutually coupled (in particular, near $\Gamma$), in the exterior it is strictly easier to estimate the higher order quantities $\partial_u \nu, |v|^{a}\partial_v \lambda, \partial_u^2 \phi, |v|^{b}\partial_v^2 \phi$. The latter will thus not form a central part of the bootstrap arguments.

An additional difference in the analysis of the exterior concerns the strength of pointwise bounds we are able to propagate on the $\Psi_p$. Whereas in the interior region, the backwards nature of the integration allows us to conjugate with large ($\approx |u|^{-\alpha}$) weights, in the exterior we cannot estimate $u$ transport equations this way. As a result, the solution is constructed instead in a piecemeal fashion, using a decomposition of $\mathcal{Q}^{(ex)}$ into regions and choosing weights locally. Ultimately we are only able to propagate a control of the form $|\Psi_p| \lesssim w_\Psi ||u|+v^{p_1}|^{\delta},$ where $\delta \ll 1$ is a small parameter.

A decomposition of $\mathcal{Q}^{(ex)}$ conducive to analysis of solutions with self-similar bounds was introduced in \cite{igoryak2}. In a form modified for our spherically symmetric setting, we write $\mathcal{Q} = \mathcal{R}_{\text{I}} \cup \mathcal{R}_{\text{II}} \cup \mathcal{R}_{\text{III}} \cup \mathcal{R}_{\text{IV}},$ where the regions are depicted below in Figure \ref{fig:exttotal}. We briefly comment on the role of each region:

\begin{figure}[h]
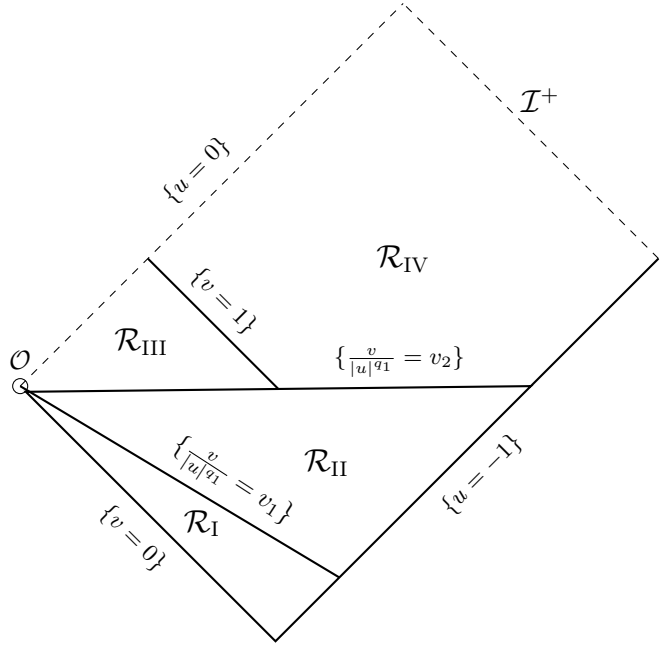

    \centering
  \includestandalone[]{Figures/fig_wholeext}
  \caption{Decomposition of the exterior into subregions.}
  \label{fig:exttotal}
\end{figure}

\vspace{10pt}
\noindent
\underline{$\mathcal{R}_{\text{I}}$: Self-similar neighborhood of $\underline{\Sigma}_0$} \newline

\vspace{-.5em}
In this region we extend the local solution to a full self-similar neighborhood of $\underline{\Sigma}_0$, with a parameter $v_1$ controlling the size of the region. The parameter $v_1$ can be chosen small as a function of the size of the background solution (measured in appropriate weighted norms), implying that existence in $\mathcal{R}_{\text{I}}$ is essentially a local existence result.

The transport system for the $\Psi_p$ takes a form similar to (\ref{eq:1.4pfoutline}). For unknowns $\nu, \mu, \partial_u \phi$ which satisfy $v$ equations, the integration is straightforward: all singularities appearing are expressible in the $u$ coordinate and are no worse than $|u|^{-q_1}$, and thus integration in $v$ leads to factors $\frac{v}{|u|^{q_1}} \leq v_1 \ll 1$.

Recall the approach to estimating $u$ equations in the interior consisted of conjugating by large weights $|u|^{-\alpha}$ enforcing decay, contributing \textit{good} lower order terms. In the exterior data for $u$ equations is given along $\Sigma^{(ex)}_{-1}$, and such a procedure cannot work. Weights of the form $|u|^{\alpha}$ would solve this problem, but are not consistent with the $\Psi_p$ remaining $\epsilon$-small up to $\mathcal{O}$.

The main problematic quantity is $\partial_v \phi,$ satisfying an equation of the form
\begin{equation}
    \label{eq:proofoutlineregion1ext1}
    \partial_u (|u|^{1-\kappa_1}\partial_v \phi_p) - \frac{\kappa_1+O(v_1)}{|u|}(|u|^{1-\kappa_1}\partial_v \phi_p) = O\Big(\frac{\epsilon}{|u|}\Big),
\end{equation}
where the right hand side depends on the bootstrap assumptions for the other double null unknowns. Directly integrating this equation in $u$ does not close.

Following \cite{igoryak1} and \cite{igoryak2}, we instead consider the system for a renormalized set of unknowns 
\begin{equation*}
    \wt{\Psi}_p(u,v) \doteq \Psi(u,v) - \overline{\Psi}(u,v) - \lim_{v\rightarrow 0^{-}} \Psi^{(in)}(u,v),
\end{equation*}
where $\Psi^{(in)}(u,v)$ denotes the value of the corresponding unknown in the interior region. We refer to these quantities arising from the interior construction as \enquote{correctors,} the function of which is to ensure the renormalized unknowns vanish along $\{v=0\}$. We are then able to propagate weighted bounds of the form
\begin{equation*}
    | w_\Psi \wt{\Psi}_p| \lesssim \Big(\frac{v}{|u|^{q_1}} \Big)^l,
\end{equation*}
for powers $l$ depending on the unknown. The effect of the modified weights is to introduce positive powers of $|u|,$ which in turn lead to favorable lower order terms when conjugated through equations of the type (\ref{eq:proofoutlineregion1ext1}). It is important that all unknowns, including those satisfying $v$ equations, are estimated in a manner consistent with vanishing along $\{v=0\}$; however, establishing such estimates in the case of quantities satisfying $v$ equations is much simpler.

Subject to conditions on $\kappa_1,$ as well as the regularity of various transversal quantities to $\{v=0\}$ (cf. assumption (B1) in Section \ref{sec:assumptionsbackground}), this scheme closes. Note that along the future boundary of $\mathcal{R}_{\text{I}}$ the weight $\frac{|u|^{q_1}}{v}$ is a positive constant, and thus these additional weights do not contribute to the bounds in later regions.

Moreover, there is some flexibility in choosing $l$, allowing us to to propagate an extra decaying $|u|^{\delta}$ weight on the $\wt{\Psi}_p$. This added decaying $|u|$ weight is the source of enhanced regularity in the exterior region, and should be compared with the added $|u|^{\alpha}$ weights in the interior.

\vspace{1.5em}
\noindent
\underline{$\mathcal{R}_{\text{II}}$: Transition region} 
\newline

\vspace{-.5em}
The analysis in $\mathcal{R}_{\text{II}}$ takes place away from both $\{v=0\}$ and $\{u=0\}$, and serves to transition between the solution in these two regions. It is convenient for the remainder of the proof to work in a modified double null gauge, $(U,V)$, which we recall is given by 
\begin{equation*}
    (U,V) = (-|u|^{1-\kappa_2}, \ v^{p_1}).
\end{equation*}
This gauge is motivated by a corresponding coordinate transformation in Christodoulou's $k$-self-similar solutions (with $\kappa_2 = \kappa_2 = k^2$), with respect to which various derivatives of $r, \phi$ transversal to $\{u=0\}$ remain uniformly bounded. The self-similar bounds in the $(U,V)$ gauge satisfied in a neighborhood of $\{u=0\}$ are then formally identical to those satisfied in the $(u,v)$ gauge near $\{v=0\}$, with the roles of the ingoing and outgoing null coordinate reversed. It follows from the definition above that surfaces of constant self-similar parameter $\frac{v}{|u|^{q_1}}$ correspond to those of constant $\frac{V^{q_2}}{|U|}$ in the new gauge, and that we have $\{u=0\} = \{U=0\}, \ \{v=0\} = \{V =0\}$. For the purposes of this introduction the coordinate systems may be thought of interchangeably.

In $\mathcal{R}_{\text{II}}$ we have the coordinate condition $|U| \sim V^{q_2}$, where the implied constants depend on a parameter $V_2$ defining the future boundary $\{\frac{V^{q_2}}{|U|} = V_2\}$ of the region. The coordinate condition renders the transport equations directly integrable, as we can ``trade" $U$ and $V$ weights to render all singular terms harmless. However, we aim to prove existence independently of $V_2$, which will be chosen large in the later regions. Therefore there is no smallness in the estimates coming directly from the size of the region.

The key to improving bootstrap bounds in this region is an exponential weight $w^* \doteq \exp(-D\frac{V^{q_2}}{|U|})$. Conjugating through the equations by $w^*$ leads to favorable zeroth order terms of size $\approx D$. For $D$ large, the weights thus serve a similar role as $|u|^{\alpha}$ weights in the interior region. For any $V_2 < \infty$, $D$ can be chosen appropriately to improve the bounds. 

\vspace{1.5em}
\noindent
\underline{$\mathcal{R}_{\text{III}}$: Near-singularity region} 
\newline

\vspace{-.5em}
In keeping with the division of the spacetime into an approximately self-similar region for $V \lesssim 1$, and an asymptotically flat region for $V \gtrsim 1$, we first consider the near-singularity region $\mathcal{R}_{\text{III}}$. Importantly, singular weights appear only with respect to the $V$ weight. It will follow that $U$ transport equations are directly integrable, and integration leads to factors of $V_2^{-1}$. With the flexibility to choose $V_2$ large still remaining from the analysis of the previous region, we now use $V_2^{-1}$ as a smallness parameter to improve the bootstrap bounds.

The $V$ transport equations for $\nu, \ \partial_U \phi$ now see the singular weight. However, these equations are integrated \textit{away} from the singularity, and therefore the additional regularity of $|u|^{\delta}$ (which in $\mathcal{R}_{\text{III}}$, becomes an additional $V^{\delta}$ weight) goes towards rendering error terms integrable. A renormalization scheme similar to $\mathcal{R}_{\text{I}}$ is therefore not needed here. Importantly, since no natural smallness appears in this direction, we exploit a reductive structure in the equations allowing the scheme to only rely on quantities that have been previously estimated. 

It is helpful at points in this region to normalize the double null gauge yet again, in order to fix the gauge \enquote{along} $\{U=0\}$. Fixing the gauge gives explicit control on various zeroth order coefficients. See Section \ref{subsec:regionIII} for details.

\vspace{1.5em}
\noindent
\underline{$\mathcal{R}_{\text{IV}}$: Asymptotically flat region} 
\newline

\vspace{-.5em}
The remainder of the construction takes place a fixed $V$ distance away from the singular point, and therefore the analysis simplifies considerably. Establishing global bounds in the unbounded domain $\mathcal{R}_{\text{IV}}$ with asymptotically flat data is a well-studied problem, see \cite{pricelaw}, \cite{lukoh1}. A helpful observation, clarified in \cite{pricelaw}, is that the wave equation for $r\phi$ has error terms that are strongly decaying in $r$, and are conducive to estimating the scalar field in regions with $r \gg 1$ (as is the case here). We therefore work primarily with $r\phi$ and its derivatives in this region.

As is the case for $\mathcal{R}_{\text{III}},$ the $U$ equations gain smallness factors in the form of $V_2^{-1}$. In particular, these unknowns naturally inherit the fast decay of initial data along $\{U=-1\}$. In the $V$ direction we again utilize the structure of the equations to estimate outgoing quantities in a specific order, and thus at each point only rely on quantities that have previously been estimated.

%% file: Figures/fig_approxint.tex
\begin{tikzpicture}[scale=1.1]
\coordinate (n1) at (0,0) {};
\node[circle, label= above:{$\mathcal{O}$}, draw, inner sep =0, minimum size = .2cm] (n2) at (90:2.83) {};
\fill[lightgray] (0,0) -- ++(-90:1) -- ++(45:2.708) -- ++ (135:.707);
\draw[dashed] (-90:1) -- +(45:2.708);

\draw[dashed] (n1) -- (n2);
\draw[dashed] (n2) --(45:2);

\draw[thick] (n1) -- node[sloped,above] {\small $\{u = u_\delta\}$}+ (45:2) -- + (45:2);
\draw[thick] (45:2) -- +(-45: 2.83) -- (-90:4);
\node[label = {[rotate=45]below:\small$\{u=-1\}$}] at ($(-90:4) + (45:2.814)$) {};
\node[label = left:$\Gamma$] at (-90:2) {};

\draw[thick] (n1) -- +(-90:4); 

\node[label = {[rotate = -45]above:\small$\{v=0\}$}] at ($(90:2.83) + (-45:2.5)$) {};

\node[] at ($(-90:2.5) + (45:2)$ ) {$\mathcal{Q}^{(in),u_\delta}$};

\end{tikzpicture}

%% file: Figures/fig_wholeext.tex
\begin{tikzpicture}[scale=1.2]
\node[circle, label= above:{$\mathcal{O}$}, draw, inner sep =0, minimum size = .2cm] (n1) at (0,0) {};

\draw[thick] (n1) -- ++(-45:4) -- +(45:4) -- cycle;
\draw[thick] ($(-45:4)+(45:4)$) -- +(45:2);
\draw[dashed] ($(-45:4)+(45:6)$) -- +(135:4);
\node[above] at ($(-45:4)+(45:6) + (135:2)+(45:.15)$) {\large$\mathcal{I}^+$};

\draw[dashed] (0,0) -- + (45:6);
\draw[thick] (45:2) -- + (-45:2.05);
\node[label = {[rotate = -45]above:\small$\{v=1\} $  }] (n2) at ($(45:2) + (-45:1) +(-135:.05) $) {};


\node[] (n4) at ($(45:4) + (-45:2) $) {\large$\mathcal{R}_{\text{IV}}$};

\node[label = {[rotate=45]below:\small$\{u=-1\}$}] (n6) at ($(-45:4) + (45:3)$) {};
\node[label = {[rotate=-45]below:\small$\{v=0\}$}]  at (-45:2) {};
\node[label = {[rotate=45]above:\small$\{u=0\}$}] (n7) at (45:3) {};
\node[label = {above:\small$\{\frac{v}{|u|^{q_1}}=v_2\}$}] (n8) at ($(0:4.2)+(-90:.1)$) {};

\node[] (n9) at ($(45:6) + (-45:2) $ ) {};
\node[] (n10) at ($(45:6) + (-45:2) + (45:1)$) {};

\draw[thick] (0,0) -- ($(-45:4) + (45:1)$);

\node[label = {[rotate=-33]above:$\{\small\frac{v}{|u|^{q_1}}=v_1\}$}]  at (-33:2.6) {};

\node[] at ($(-45:2.5) + (45:.35) $) {\large$\mathcal{R}_{\text{I}}$};
\node[] at ($(-45:3) + (45:1.8) $) {\large$\mathcal{R}_{\text{II}}$};
\node[] at ($(45:1.3) + (-45:.6) $) {\large$\mathcal{R}_{\text{III}}$};

\end{tikzpicture}

%% file: sec4.tex
\section{Interior solution: Proof of Theorem \ref{thm1:interior} }
\label{sec:proofthm1}

\subsection{Approximate interiors}
\label{subsec:approxint}
Assume a $u_\delta \in (-1,0)$ is fixed, and recall the definition of the truncated domain 
\begin{equation}
\label{def:intgeom}
\mathcal{Q}^{(in),u_\delta} \doteq \mathcal{Q}^{(in)} \cap \{u \leq u_\delta \}.
\end{equation}
For \textbf{this section alone}, we apply the convention that the surfaces $\Gamma, \Sigma_u,  \underline{\Sigma}_v$ are restricted to $\mathcal{Q}^{(in)} \cap \{u \leq u_\delta\}$. This convention will simplify notation.

Let $\mathcal{L}$ denote the set 
\begin{equation*}
    \mathcal{L} \doteq \Gamma \cup \Sigma_{u_\delta},
\end{equation*}
containing the future endpoints of ingoing null lines. $\mathcal{L}$ will be the natural boundary from which transport equations in the $u$ direction are integrated. For $v$ transport equations, any boundary terms will lie on the initial data surface $\underline{\Sigma}_0$. Given any $(u,v) \in \mathcal{Q}^{(in),u_\delta}$, let $u_\mathcal{L}(v)$ be the future endpoint of $\underline{\Sigma}_v$ on $\mathcal{L}$. Similarly one can define $v_\mathcal{L}(u).$

Additional subregions of $\mathcal{Q}^{(in),u_\delta}$ that will be important in the following are
\begin{equation*}
    \mathcal{B}^{(1),u_\delta} \doteq \mathcal{Q}^{(in),u_\delta} \cap \{v \geq -|u_\delta|^{q_1} \}, \ \ \ \ 
     \mathcal{B}^{(2),u_\delta} \doteq \mathcal{Q}^{(in),u_\delta} \cap \{v \leq -|u_\delta|^{q_1}  \},
\end{equation*}
The pair $\mathcal{B}^{(1),u_\delta}$ and $\mathcal{B}^{(2),u_\delta}$ divides $\mathcal{Q}^{(in),u_\delta}$ along the ingoing null line $v = -|u_\delta|^{q_1},$ separating points for which the boundary terms of $u$ transport equations lie on $ \Sigma_{u_\delta}$ and $\Gamma$ respectively. 

\subsubsection{Initial data and local existence}
We now prescribe data on the initial data surface $\underline{\Sigma}_0 \cup \Sigma_{u_\delta}$. The gauge is normalized along $\{v=0\}$ so that $r_p(u,0) = 0$, and we therefore are left to specify the gauge along the outgoing piece $\Sigma_{u_\delta}$, as well as the scalar field. Recall that we allow non-trivial data along $\Sigma_{u_\delta}$, which we require to be consistent with the regularity of the background spacetime.

The free data is parameterized by 
\begin{equation*}
    f_\delta(u) \in C^1(\underline{\Sigma}_0), \ \ \ k_\delta(v) \in C^2(\Sigma_{u_\delta} \cap \{v < 0\}), \ \ \ h_\delta(v) \in C^1(\Sigma_{u_\delta} \cap \{v < 0\}).
\end{equation*}
The functions $k_\delta, h_\delta$ will correspond to data for $r_p, \partial_v(r\phi_p)$ along $\Sigma_{u_\delta}$. We allow for singular behavior as $v \rightarrow 0$ consistent with that of the background solution. We thus require that
\begin{equation*}
    |v|^{a}\partial_v^2 k_\delta(v), \ |v|^b \partial_v h_\delta(v) \in L^\infty(\Sigma_{u_\delta}).
\end{equation*}

Our goal in this section is to construct solutions of the scalar field system in $\mathcal{Q}^{(in),u_\delta}$, uniformly in $u_{\delta}$, achieving the data
\begin{align}
        &r(u,0) = \overline{r}(u,0), \label{intdata2 } \\[.3em]
        &\partial_u(r\phi)(u,0) = \partial_u(r\overline{\phi})(u,0) + \epsilon f_\delta(u) , \\[.3em]
        &r(u_\delta, v) = \overline{r}(u_\delta,v) + \epsilon k_\delta(v), \\[.3em]
        &\partial_v(r\phi)(u_{\delta},v) = \partial_v(r\overline{\phi})(u_\delta,v) + \epsilon  h_\delta(v),\label{intdata3 }
\end{align}
and boundary conditions
\begin{align}
    \label{eq:axisbc1}
    &r|_\Gamma = (r\phi)|_\Gamma = m|_\Gamma = 0, \\[.3em]
    \label{eq:axisbc2}
    &(\partial_v + p_1|u|^{\kappa_1}\partial_u)r|_\Gamma = (\partial_v + p_1|u|^{\kappa_1}\partial_u)(r\phi)|_\Gamma  = 0.
\end{align}

Much of the analysis will proceed on the level of differences $\Psi_p$. The initial data can equivalently be expressed as 
\begin{align}
    r_p(u,0) &= 0 , \ \ \quad \quad \quad \partial_u(r\phi_p)(u,0) =  \epsilon f_\delta(u), \\[.3em]
    r_p(u_{\delta},v) &= \epsilon k_\delta (v) ,\ \ \  \partial_v(r\phi_p)(u_\delta,v) = \epsilon h_\delta(v) , 
\end{align}
with corresponding boundary conditions.

\vspace{10pt}
The starting point of the construction is a local solution for the mixed characteristic-timelike problem with data on $\underline{\Sigma}_0 \cup \Sigma_{u_\delta}$ and boundary conditions on $\Gamma$. Local existence is guaranteed\footnote{Note that although the statement of Proposition \ref{prop:weirdlocalex} considers only $C^1$ solutions, one may proceed by first using standard local well-posedness away from $r=0$ to solve in a region $\mathcal{Q}^{(in),u_\delta} \cap \{v \geq -\sigma\}$ for a small $\sigma >0$. The resulting induced data on $\{v=-\sigma\}$ will be $C^1$, and Proposition \ref{prop:weirdlocalex} may then be applied to continue the solution to the axis.} by Proposition \ref{prop:weirdlocalex}, applied here with reversed time orientation (Section \ref{sec:solnclasses} considered the \textit{forwards} problem). 

The following result is an easy consequence of the local existence of solutions (with $C^1$ bounds away from $\{v=0\}$), and estimates for the difference system (\ref{PSSESF:1})-(\ref{PSSESF:5}).

\begin{prop} 
\label{prop:localexint}
There exists a $u_* \in [-1,u_\delta)$ and a BV solution to the system (\ref{SSESF:2:1})-(\ref{SSESF:2:5}) in the domain $\mathcal{Q}^{(in),u_\delta} \cap \{u_* \leq u \leq u_\delta \}$ achieving the data (\ref{intdata2 })-(\ref{intdata3 }) and boundary conditions (\ref{eq:axisbc1})-(\ref{eq:axisbc2}). 

Moreover, there exists a constant $A$ depending on the initial data, but independent of $\epsilon$, for which the following estimates hold pointwise in $\mathcal{Q}^{(in),u_\delta} \cap \{u_* \leq u \leq u_\delta \}$:
\begin{enumerate}
    \item (Control on the geometry) \ $ |\nu_p|, \ \ |\lambda_p|, \ \ |\partial_u \nu_p|, \ \ ||v|^a\partial_v \lambda_p| \leq A\epsilon, $ 
    \item (Control on the Hawking mass) \ $|\mu_p|, \ \ |\frac{\mu_p}{r^2}|  \leq A\epsilon,$
    \item (First order control on the scalar field)  \ $ |\phi_p|,\ \  |\partial_u \phi_p|, \ \ |\partial_v \phi_p|, \ \ |\partial_u(r\phi_p)|, \ \ |\partial_v (r\phi_p)| \leq A\epsilon,$
    \item (Second order control on the scalar field)  \ $|\partial_u^2(r\phi_p)|, \ \ ||v|^b \partial_v^2(r\phi_p) | \leq A\epsilon.$
\end{enumerate}
\end{prop}

We emphasize that the interval of existence depends in an essential way on pointwise bounds for the background solution, which in the context of singular spacetimes are not uniform as $u_\delta \rightarrow 0$. For this reason one cannot take the $u_\delta \rightarrow 0$ limits on the solutions constructed by the above proposition to deduce the proof of Theorem \ref{thm1:interior}.
\vspace{-.5em}
\subsubsection{Norms}
\label{subsection:interiornorms}
For the construction of the approximate interior, we require only a subset of the bounds specified in assumptions (A1)-(A4) of Section \ref{sec:assumptionsbackground}. The relevant bounds are captured by the $\overline{\mathfrak{N}}_1$ below, which satisfies $\overline{\mathfrak{N}}_1 < \infty$.
\begin{align}
    \label{eq:totalbackgroudnorm}
    \overline{\mathfrak{N}}_1 \doteq &\sup_{\mathcal{Q}^{(in),u_\delta}}|\log |u|^{-\kappa_1}\overline{\lambda}| + \sup_{\mathcal{Q}^{(in),u_\delta}}|\log (-\overline{\nu})| + \sup_{\mathcal{Q}^{(in),u_\delta}}|\log (1-\overline{\mu})|+\sup_{\mathcal{Q}^{(in),u_\delta}}||u|^2\frac{\overline{m}}{\overline{r}^3}| \nonumber \\[.5em]
    &+ \sup_{\mathcal{Q}^{(in),u_\delta}}||u| \partial_u \overline{\nu}| + \sup_{\mathcal{Q}^{(in),u_\delta}}||u|^{1-2\kappa_1}s_1^{a} \partial_v \overline{\lambda}| + \sup_{\mathcal{Q}^{(in),u_\delta}}||u| \partial_u \overline{\phi}| + \sup_{\mathcal{Q}^{(in),u_\delta}}||u|^{1-\kappa_1} \partial_v \overline{\phi}|\nonumber \\[.5em]
    &+\sup_{\mathcal{Q}^{(in),u_\delta}}||u|^2 \partial_u^2 \overline{\phi}| + \sup_{\mathcal{Q}^{(in),u_\delta}}||u|^{2-2\kappa_1}s_1^{b} \partial_v^2 \overline{\phi}|.
\end{align}
In particular, higher order bounds are not needed at this stage of the construction, and will appear in Section \ref{subsubsec:thirdorderestimates}.

The main result of this section will be a semi-global existence result in $\mathcal{Q}^{(in),u_\delta}$ for data having finite initial data norm $\mathcal{I}_\alpha^{(in)}$:
\begin{align}
     \mathcal{I}_\alpha^{(in)} \doteq \sum_{i=0}^{1}&\sup_{\Sigma_{u_\delta}}|u|^{-\alpha+i}|\partial_u^i f_\delta(u)| +\sup_{\Sigma_{u_\delta}} |u_\delta|^{-\alpha}|r^{-1} k_\delta(v)| \nonumber\\[.5em]
     &+  \sup_{\Sigma_{u_\delta}} |u_\delta|^{-\alpha-\kappa_1}| \partial_v k_\delta(v)|
    + \sup_{\Sigma_{u_\delta}} |u_\delta|^{-\alpha+1-2\kappa_1-q_1a}|v|^a|\partial_v^2 k_\delta(v)| \nonumber \\[.5em]
    &+ \sup_{\Sigma_{u_\delta}} |u_\delta|^{-\alpha-\kappa_1}|h_\delta(v)|  + \sup_{\Sigma_{u_\delta}} |u_\delta|^{-\alpha+1-2\kappa_1-q_1b}|v|^{b}|\partial_v h_\delta(v)|. \label{eq:intinitialdatanorm}
\end{align}
In the case of trivial outgoing data along $\Sigma_{u_\delta},$ the norm reduces to $\||u|^{-\alpha}f_\delta\|_{C^0(\underline{\Sigma}_0)} + \||u|^{-\alpha+1}\partial_u f_\delta\|_{C^0(\underline{\Sigma}_0)}$. Note that the norms on the outgoing piece include an explicit scaling in the parameter $u_\delta$. This added technicality plays no role in the construction for fixed cutoff $u_\delta$, but will be helpful when considering the convergence of approximate interiors. To simplify notation in this section we will simply denote the initial data norm simply by $\mathcal{I}$.

The solution will be controlled by various pointwise norms of the $\Psi_p$. In addition to an $\approx \epsilon$ smallness, we propagate a strong decay in $|u|$. Roughly, the bounds correspond to rates $|u|^{\alpha}$ faster than the self-similar bounds on the $\overline{\Psi}$. To close a scheme with these norms, we also track vanishing along the axis for the mass ratio $\mu$, as well as potential blow up as $v \rightarrow 0$ for the second order quantities $\partial_v \lambda_p,\  \partial_v^2 (r\phi_p)$. 

Define the following pointwise norms:
\begin{align}
    \mathfrak{N}_{r_p} &\doteq \sup_{\mathcal{Q}^{(in),u_\delta}}\Big|\frac{1}{\overline{r}|u|^{\alpha}}r_p \Big|, \ \ \ \quad \mathfrak{N}_{\nu_p} \doteq \sum_{i=0}^1 \sup_{\mathcal{Q}^{(in),u_\delta}}\Big| \frac{1}{|u|^{\alpha-i}}\partial_u^i \nu_p \Big|, \\[.5em]
     \mathfrak{N}_{\lambda_p} &\doteq \sup_{\mathcal{Q}^{(in),u_\delta}}\Big|\frac{1}{|u|^{\alpha+\kappa_1}} \lambda_p\Big| + \sup_{\mathcal{Q}^{(in),u_\delta}}\Big|\frac{1}{|u|^{\alpha-1+2\kappa_1}}\Big(\frac{|v|}{|u|^{q_1}} \Big)^a\partial_v \lambda_p \Big| ,\\[.5em]
     \mathfrak{N}_{\mu_p} &\doteq \sup_{\mathcal{Q}^{(in),u_\delta}}\Big|\frac{1}{\overline{r}|u|^{\alpha-1}}\mu_p \Big|, \ \ \  \quad \mathfrak{N}_{\phi_p} \doteq \sup_{\mathcal{Q}^{(in),u_\delta}}\Big|\frac{1}{|u|^{\alpha}}\phi_p \Big|, \\[.5em] 
     \mathfrak{N}_{\partial_u(r\phi_p)} &\doteq \sum_{i=1}^2 \sup_{\mathcal{Q}^{(in),u_\delta}}\Big|\frac{1}{|u|^{\alpha-i}}\partial_u^i(r\phi_p) \Big|, \\[.5em]
     \mathfrak{N}_{\partial_v(r\phi_p)} &\doteq  \sup_{\mathcal{Q}^{(in),u_\delta}}\Big|\frac{1}{|u|^{\alpha+\kappa_1}}\partial_v(r\phi_p) \Big|+\sup_{\mathcal{Q}^{(in),u_\delta}}\Big|\frac{1}{|u|^{\alpha-1+2\kappa_1}}\Big(\frac{|v|}{|u|^{q_1}} \Big)^b \partial_v^2(r\phi_p) \Big|,
\end{align}
as well as the total spacetime norm
\begin{align}
\label{eq:inttotalspacetimenorm}
\mathfrak{N}^{(in)}_{tot} &\doteq \mathfrak{N}_{r_p}+ \mathfrak{N}_{\nu_p}+\mathfrak{N}_{\lambda_p}+ \mathfrak{N}_{\mu_p}
+\mathfrak{N}_{\phi_p}+\mathfrak{N}_{\partial_u (r\phi_p)}+\mathfrak{N}_{\partial_v (r\phi_p)}.
\end{align}

\subsubsection{Main result and bootstrap assumptions}
We now state the main result on the existence of our approximate interior solutions.
\begin{prop}
\label{approxint:mainprop}
There exists an $\alpha \gg 1$ large, such that for any $f_\delta, h_\delta, k_\delta$ with $\mathcal{I} < \infty$, there exists $\epsilon$ small enough depending only on $\overline{\mathfrak{N}}_1$ and $\mathcal{I}$ such that the solution constructed in Proposition \ref{prop:localexint} extends to $\mathcal{Q}^{(in),u_\delta}$. Moreover, the norm $\mathfrak{N}^{(in)}_{tot}$ is bounded in terms of $\mathcal{I}$.
\end{prop}

A continuity argument applied in conjunction with local existence theory shows that Proposition \ref{approxint:mainprop} will follow from estimates for the double null unknowns (and a fixed number of derivatives) uniformly in $\mathcal{Q}^{(in),u_\delta}$. The following proposition reduces the desired result to one about improving bootstrap assumptions on $\mathfrak{N}^{(in)}_{tot}$.

\begin{prop}
\label{approxint:mainboot}
Let $(r,m,\phi)$ be a solution constructed by Proposition \ref{prop:localexint} in a region $\mathcal{Q}^{(in),u_\delta} \cap \{u \geq u_* \}$ for some $u_* \in (-1,u_\delta)$. Assume that for some large constant $A$, there is the bound
\begin{equation}
    \label{int:bootassump}
    \mathfrak{N}^{(in)}_{tot} \leq 2A\epsilon,
\end{equation} 
holding on the region of existence. Then for $\alpha$ large enough depending on the background solution through $\overline{\mathfrak{N}}_1$, and $\epsilon$ sufficiently small depending on $\mathcal{I}$ and $\overline{\mathfrak{N}}_1$, the bound can be improved to 
\begin{equation}
    \mathfrak{N}^{(in)}_{tot} \leq A\epsilon.
\end{equation} 
\end{prop}

\subsubsection{Bounds on initial data}

\begin{lem}
    \label{lem:intinitialdata}
For $\epsilon$ sufficiently small, we have
\begin{equation}
    \label{e:id1}
    |\partial_v \phi_p|(u_\delta,v) \lesssim \epsilon \mathcal{I} |u_\delta|^{\alpha-1+\kappa_1} ,
\end{equation}
\begin{equation}
    \label{e:id2}
    |\frac{\mu_p}{r^2}|(u_\delta,v) \lesssim \epsilon \mathcal{I} |u_\delta|^{\alpha-2}.
\end{equation}
\end{lem}
\begin{proof}
To estimate $\partial_v \phi_p(u_\delta,v)$ in a neighborhood of $\{v=0\}$, it suffices to use the identity
\begin{equation}
    \label{eq:inttemp-1}
    \partial_v \phi_p(u_\delta,v) = \frac{1}{r(u_\delta,v)}(\partial_v (r\phi_p)(u_\delta,v) - \lambda(u_\delta,v) \phi_p(u_\delta,v) ).
\end{equation}
The bound $|r_p(u_\delta,v)| \lesssim \epsilon \mathcal{I} |u_\delta|^{\alpha}$ implies that for $\epsilon$ sufficiently small, $r = \overline{r}+r_p \sim \overline{r} \sim |u_\delta|$, the final equivalence holding near $\{v=0\}$.

Similarly $\lambda(u_\delta,v) \sim \overline{\lambda}(u_\delta,v)$. To estimate $\phi_p$, apply boundary conditions along the axis and the fundamental theorem of calculus in the $v$ direction to see 
\begin{align*}
    |(r\phi_p)(u_\delta,v)|  &\lesssim \int\limits_{(u_\delta,v_\Gamma(u_\delta))}^{(u_\delta,v)}|\partial_v(r\phi_p)|(u_\delta,v')dv' \\
    &\lesssim \epsilon \mathcal{I} |u_\delta|^{\alpha+\kappa_1}|v - v_\Gamma(u_\delta)|  \\
    &\lesssim \epsilon \mathcal{I} |u_\delta|^{\alpha} r(u_\delta,v).
\end{align*}
Dividing through by $r(u_\delta,v)$ gives the bound $|\phi_p| \lesssim \epsilon \mathcal{I} |u_\delta|^{\alpha}$. Inserting these estimates into (\ref{eq:inttemp-1}) gives
\begin{equation*}
    |\partial_v \phi_p(u_\delta,v)| \lesssim \epsilon \mathcal{I} |u_\delta|^{\alpha-1+\kappa_1}.
\end{equation*}
To show a similar bound near the axis we apply (\ref{eq:lukohestimatev}). Choose $\epsilon$ small such that $|\partial_v \lambda| \lesssim |u|^{-1+2\kappa_1}s_1^{-a}.$ In a self-similar neighborhood of $\Gamma$ we have $s_1 \sim 1$, and thus $|\partial_v \lambda| \lesssim |u|^{-1+2\kappa_1}$. Inserting bounds for $\lambda, \partial_v \lambda, \partial_v(r\phi_p), \partial_v^2(r\phi_p)$ into (\ref{eq:lukohestimatev}) gives
\begin{equation*}
    |\partial_v\phi_p(u_\delta,v)| \lesssim \epsilon \mathcal{I} |u_\delta|^{\alpha-1+\kappa_1},
\end{equation*}
as desired.

\vspace{5pt}
The strategy for deriving (\ref{e:id2}) is to integrate (\ref{PSSESF:4}) from $(u,v) = (u_\delta,v_\Gamma(u_\delta))$. The coefficient of the zeroth order term
\begin{equation*}
-\mathcal{G}_3 \mu_p \doteq -\Big(\frac{\lambda}{r} + \frac{r}{\lambda}(\partial_v \phi)^2) \Big)\mu_p 
\end{equation*}
appearing on the right hand side of (\ref{PSSESF:4}) has a favorable sign (cf. Lemma \ref{lem:integrationmain}). For $\epsilon$ small enough, estimate 
$$|\big(\mathcal{G}_3\big)_p \overline{\mu}|(u_\delta,v) \lesssim \epsilon \mathcal{I} |u_\delta|^{\alpha-2+\kappa_1}\overline{r}(u_\delta,v),$$
$$|\big(\mathcal{I}_3\big)_p|(u_\delta,v) \lesssim  \epsilon \mathcal{I} |u_\delta|^{\alpha-2+\kappa_1}\overline{r}(u_\delta,v).$$
Integrating  (\ref{PSSESF:4}) thus gives
\begin{align*}
    |\mu_p|(u_\delta,v) &\lesssim |\mu_p|(u_\delta, v_\Gamma(u_\delta)) + \epsilon \mathcal{I} |u_\delta|^{\alpha-2+\kappa_1}\overline{r}(u_\delta,v) |v - v_\Gamma(u_\delta)|  \\
    &\lesssim |\mu_p|(u_\delta, v_\Gamma(u_\delta)) +  \epsilon \mathcal{I} |u_\delta|^{\alpha-2}\overline{r}^2(u_\delta,v).
\end{align*}
It remains to show $|\mu_p|(u_\delta, v_\Gamma(u_\delta)) = 0$. This is a direct consequence of local existence theory and $C^1$ bounds near the axis.
\end{proof}

\subsubsection{Consequences of the bootstrap assumptions}
The first lemma establishes that for $\epsilon$ small, the solution $(\overline{r} + r_p, \ \overline{\mu}+\mu_p, \ \overline{\phi} + \phi_p)$ is quantitatively close to the background solution in the region of existence. In particular, the bounds propagated on the $\Psi_p$ are \textit{at least as} regular as those satisfied by the $\overline{\Psi}.$
\begin{lem}
\label{int:lem1}
Assume (\ref{int:bootassump}). For $\epsilon$ sufficiently small, we have 
\begin{align}
    r \sim \overline{r},\  \ \ r\lesssim |u|, \ \ \ (-\nu) \sim 1, \ \ \ \lambda \sim |u|^{\kappa_1}, \ \ \ (1-\mu) \sim 1.
\end{align}
Moreover, the following higher order bounds hold:
\begin{equation}
    |\partial_u \nu| \lesssim |u|^{-1}, \ \ \ |\partial_v \lambda| \lesssim |u|^{-1+2\kappa_1}s_1^{-a}.
\end{equation}
\end{lem}
\begin{proof}
The bounds follow by writing $\Psi = \overline{\Psi} + \Psi_p$, inserting bounds for $\overline{\Psi}$ from $\overline{\mathfrak{N}}_1$ and bootstrap bounds for $\Psi_p$, and finally choosing $\epsilon$ sufficiently small.
\end{proof}

\begin{lem}
\label{lem:intdphibounds}
For $\epsilon$ sufficiently small, we have
\begin{align}
    \label{eq:int10}
    |\partial_u \phi_p| &\lesssim A\epsilon |u|^{\alpha-1},\\[.5em]
    \label{eq:int10.5}
    |\partial_v \phi_p| &\lesssim A\epsilon |u|^{\alpha-1+\kappa_1}.
\end{align}
Therefore
\begin{equation}
   |\partial_u \phi| \lesssim |u|^{-1}, \ \ \ |\partial_v \phi| \lesssim |u|^{-1+\kappa_1}.
\end{equation}
\end{lem}
\begin{proof}
(\ref{eq:int10}) follows from (\ref{eq:lukohestimateu}) after inserting upper and lower bounds on $-\nu$, as well as bootstrap bounds on $\partial_u \nu$, $\partial_u(r\phi_p)$, and $\partial_u^2(r\phi_p)$.

To estimate $\partial_v \phi_p$ we apply (\ref{PSSESF:5}). Estimating the right hand side gives $|\partial_u(r \partial_v \phi_p)| \lesssim A \epsilon |u|^{\alpha-1+\kappa_1}$. For $(u,v) \in \mathcal{B}^{(1),u_\delta}$, integrating in $u$ gives
\begin{align*}
    |r \partial_v \phi_p|(u,v) &\lesssim |r \partial_v \phi_p|(u_\delta,v) + A \epsilon |u|^{\alpha+\kappa_1}|u - u_\delta| \\[.3em]
    &\lesssim \epsilon r(u_\delta,v)  |u_\delta|^{\alpha-1+\kappa_1}\mathcal{I} +A \epsilon r(u,v) |u|^{\alpha-1+\kappa_1}.
\end{align*}
Dividing through by $r(u,v)$ and observing $r(u_\delta,v) \leq r(u,v)$ gives the result, after potentially choosing $A$ larger as a function of $\mathcal{I}$.

For $(u,v) \in \mathcal{B}^{(2),u_\delta}$ the integration is similar, and the boundary term vanishes along the axis. 
\end{proof}

\begin{lem}
\label{lem:intmubounds}
For $\epsilon$ sufficiently small depending on $A$, 
\begin{equation}
    \label{eq:axisfullreg}
    |\frac{\mu_p}{r^2}| \lesssim A\epsilon |u|^{\alpha-2}.
\end{equation}
Therefore 
\begin{equation}
    \mu \lesssim r^2|u|^{-2}.
\end{equation}
\end{lem}
\begin{proof}
    Lemma \ref{lem:intinitialdata}, along with boundary conditions on $\Gamma$, bounds $\mu_p(u,v)$ for any $(u,v) \in \mathcal{L}$. We proceed by integrating (\ref{PSSESF:3}) in the direction of decreasing $u$, using that the bootstrap assumptions imply the zeroth order term has a good sign. Apply bootstrap assumptions and (\ref{eq:int10}) with $\epsilon$ sufficiently small to estimate
    \begin{equation}
        \label{est:temp1}
        |\big(\mathcal{G}_4\big)_p \overline{\mu}|,\ \  |\big(\mathcal{I}_4\big)_p | \lesssim A\epsilon |u|^{\alpha-2} \overline{r}.
    \end{equation}
    In the above estimate, there is potential danger due to the $r^{-1}$ factors appearing in (\ref{PSSESF:3}). Note however the presence of $\overline{\mu}$ factors in these terms, which allows one to bypass the singularity at the axis and write 
    $$\frac{\overline{\mu}}{r} \lesssim \frac{r}{|u|^2}. $$
    Manipulations of this form are used repeatedly in the later analysis.
    
    For $(u,v) \in \mathcal{B}^{(1),u_\delta}$, integrating (\ref{PSSESF:3}) in $u$ and applying (\ref{est:temp1}) gives 
    \begin{align*}
        |\mu_p|(u,v) & \lesssim |\mu_p|(u_\delta,v) + A\epsilon |u|^{\alpha-2} \overline{r}(|u| - |u_\delta|) \\
        &\lesssim |\mu_p|(u_\delta,v) + A\epsilon |u|^{\alpha-2} \overline{r}^2.
    \end{align*}
    Dividing by $\overline{r}(u,v)^2$ and applying (\ref{e:id2}) gives the result.
    
    For $(u,v) \in \mathcal{B}^{(2),u_\delta}$ the integration is the same, and we simply note that the boundary term vanishes.
\end{proof}

\subsubsection{Closing the main bootstrap}
We now proceed to the main set of estimates. First we discuss the exceptional cases of $r_p, \phi_p$, for which bounds follow simply from the fundamental theorem of calculus.
\begin{lem}
For $\epsilon$ sufficiently small we have the estimates
\begin{align}
     \mathfrak{N}_{r_p} &\lesssim \mathfrak{N}_{\lambda_p}, \\[.5em]
     \mathfrak{N}_{\phi_p} &\lesssim \mathfrak{N}_{\partial_v(r\phi_p)}.
\end{align}
\end{lem}
\begin{proof}
For all $(u,v) \in \mathcal{Q}^{(in),u_\delta}$, the past directed constant $u$ line intersects $\Gamma$ at some $(u,v_\Gamma(u)).$ Integrating $\partial_v r_p = \lambda_p$ from the axis, where $r_p(u,v_\Gamma(u)) = 0$, we conclude
    \begin{align*}
        |r_p|(u,v) &\leq \int_{(u,v_\Gamma(u))}^{(u,v)}|\lambda_p|(u,v')dv' \\
        &\leq \mathfrak{N}_{\lambda_p} |u|^{\alpha+\kappa_1} \int_{(u,v_\Gamma(u))}^{(u,v)} dv' \\
        &= \mathfrak{N}_{\lambda_p} |u|^{\alpha+\kappa_1}(v - v_\Gamma(u)) \\
        &\lesssim \mathfrak{N}_{\lambda_p} |u|^{\alpha}\overline{r}(u,v).
    \end{align*}
Dividing by $|u|^{\alpha}\overline{r}$ gives the desired result.
    
An analogous argument for $\phi_p$ proceeds by integrating $\partial_v(r\phi_p)$ from the axis, where boundary conditions require $r\phi_p(u,v_\Gamma(u)) =0$. Using the positive lower bound on $\lambda$ to exchange a factor of $v - v_\Gamma(u)$ for $r(u,v)|u|^{-\kappa_1}$, and dividing by $r$, we conclude a bound on $\phi_p$. 
\end{proof}

We next estimate those unknowns satisfying $v$ equations. In following with the outline sketched in Section \ref{subsec:introproofoutline}, the key is to conjugate by powers of $w \doteq (|u|^{q_1} + |v|)^{p_1}$ to gain good low order terms. Note $\mu$ satisfies both a $u$ and $v$ equation; in this section it is convenient to estimate $\mu$ via its $u$ equation.

\begin{lem}
For $\epsilon$ sufficiently small,
\begin{equation}
    \mathfrak{N}_{\nu_p},\  \mathfrak{N}_{\partial_u (r\phi_p)} \leq \frac{1}{10}A\epsilon.
\end{equation}
\end{lem}
\begin{proof}
Start by estimating the right hand side of (\ref{PSSESF:2}). Inserting the bootstrap assumptions, and choosing $\epsilon$ small to absorb terms of $O(\epsilon^2)$ gives
\begin{equation*}
    |\mathcal{G}_1 \nu_p|, \ |(\mathcal{G}_1)_p \overline{\nu}| \lesssim  A\epsilon |u|^{\alpha-1+\kappa_1} \lesssim A\epsilon(|u|^{q_1}+|v|)^{-1+ \alpha p_1}.
\end{equation*}
It follows that $\partial_v \nu_p = O(A \epsilon (|u|^{q_1}+|v|)^{-1+ \alpha p_1})$. Conjugate through by the weight $w^{-\alpha},$ giving
\begin{equation*}
    \partial_v (w^{-\alpha} \nu_p) - \frac{p_1\alpha}{|u|^{q_1}+|v|} (w^{-\alpha} \nu_p) = O(A\epsilon(|u|^{q_1}+|v|)^{-1}).
\end{equation*}
Contracting by $w^{-\alpha}\nu_p$ and integrating backwards in $v$ from data gives
\begin{align*}
    |w^{-\alpha} \nu_p|^2(u,v) &+ p_1\alpha \int\limits_{(u,v)}^{(u,0)}\frac{1}{|u|^{q_1}+|v'|}(w^{-\alpha}\nu_p)^2(u,v')dv' \lesssim O(A\epsilon)\int\limits_{(u,v)}^{(u,0)}\frac{1}{|u|^{q_1}+|v'|}w^{-\alpha}|\nu_p|dv' \nonumber\\
    &\lesssim \frac{p_1\alpha}{2} \int\limits_{(u,v)}^{(u,0)}\frac{1}{|u|^{q_1}+|v'|}(w^{-\alpha}\nu_p)^2(u,v')dv' + \alpha^{-1} O(A^2\epsilon^2) \int\limits_{(u,v)}^{(u,0)}\frac{1}{|u|^{q_1}+|v'|}dv' \nonumber\\
    &\lesssim \frac{p_1\alpha}{2} \int\limits_{(u,v)}^{(u,0)}\frac{1}{|u|^{q_1}+|v'|}(w^{-\alpha}\nu_p)^2(u,v')dv' + \alpha^{-1} O(A^2\epsilon^2).
\end{align*}
In the above, we have used that $\nu_p(u,0) = 0$ to drop the initial data term. After absorbing the remaining integral into the left hand side, we arrive at an estimate for $w^{-\alpha}\nu_p$. Since $w \sim |u|$, the bound on $w^{-\alpha} \nu_p$ translates to a bound on $|u|^{-\alpha} \nu_p$. We finally choose $\alpha$ sufficiently large to improve the bootstrap assumption.

To estimate $\partial_u \nu_p$, consult the $u$-commuted equation (\ref{PSSESF:8}). The main new terms to estimate are of the following form:
\begin{enumerate}
    \item Terms $\frac{\mu}{r^2}$ and $\frac{\mu_p}{r^2}$, which have been estimated in Lemma \ref{lem:intmubounds}.
    \item Derivatives $\partial_u \mu, \partial_u \lambda, \partial_u \mu_p, \partial_u \lambda_p$ which may be estimated by the respective transport equations satisfied by these quantities. Note the quantity $\partial_u \phi_p$ appearing in the equation for $\partial_u \mu_p$ has already been estimated in Lemma \ref{lem:intdphibounds}.
\end{enumerate}
The result is an estimate 
\begin{equation*}
    |\partial_v \partial_u\nu_p| \lesssim A\epsilon w^{\alpha-2+\kappa_1} \
\end{equation*}
Conjugating by $w^{-(\alpha-1)}$ and proceeding as for $\nu_p$ gives the desired improvement. Note $\partial_u \nu_p(u,0) = 0$, and hence the data term vanishes.

We next estimate $\partial_u (r\phi_p)$ using the wave equation (\ref{PSSESF:6}). Inserting the bootstrap assumptions and bounds on the derivatives of the background scalar field yields
\begin{equation*}
    |\partial_v \partial_u (r\phi_p)| \lesssim A\epsilon w^{\alpha-1+\kappa_1}.
\end{equation*}
Conjugating by $w^{-\alpha}$, integrating in $v$, and proceeding as for $\nu_p$ gives
\begin{equation*}
    |w^{-\alpha} \partial_u(r\phi_p)|^2(u,v) \lesssim |w^{-\alpha }\partial_u(r\phi_p)|^2(u,0) + \alpha^{-1} O(A^2\epsilon^2),
\end{equation*}
where the data term may be estimated by $\epsilon^2 \mathcal{I}^2$. It follows that for $\alpha$ large we again improve the bootstrap assumption.

Finally we consider $\partial_u^2(r\phi_p)$. Estimating the terms appearing in (\ref{PSSESF:12}) yields
\begin{equation*}
    |\partial_v \partial_u^2(r\phi_p)| \lesssim A\epsilon w^{\alpha-2+\kappa_1}.
\end{equation*}
It follows that we can conjugate by $w^{-(\alpha-1)}$ and integrate in $v$, giving the desired estimate.
\end{proof}

It remains to estimate the quantities satisfying $u$ equations. As mentioned in the proof outline, integration in $u$ must account for boundary terms appearing along the axis.

\begin{lem}
    \label{lem:inttheonlylemmareferredto}
For $\epsilon$ sufficiently small,
\begin{equation}
    \mathfrak{N}_{\lambda_p},\  \mathfrak{N}_{\mu_p}, \ \mathfrak{N}_{\partial_v (r\phi_p)} \leq \frac{1}{10}A\epsilon.
\end{equation}
\end{lem}
\begin{proof}
Start by estimating $(\ref{PSSESF:1})$. Inserting the bootstrap assumptions, and choosing $\epsilon$ small to absorb terms of $O(\epsilon^2)$ gives
\begin{equation*}
    |\mathcal{G}_2 \lambda_p|, \ |(\mathcal{G}_2)_p \overline{\lambda}| \lesssim A \epsilon |u|^{\alpha-1 + \kappa_1}.
\end{equation*}
It follows that $\partial_u \lambda_p = O(A \epsilon |u|^{\alpha-1+\kappa_1})$. Choose a $0 < \sigma  \ll 1$ small, and conjugate through by $|u|^{-\alpha-\kappa_1+\sigma},$ giving 
\begin{equation*}
    \partial_u(|u|^{-\alpha-\kappa_1 +\sigma} \lambda_p) - \frac{\alpha+\kappa_1-\sigma}{|u|}(|u|^{-\alpha-\kappa_1 +\sigma} \lambda_p) = O(A\epsilon |u|^{-1+\sigma}).
\end{equation*}
Contracting with $|u|^{-\alpha-\kappa_1 +\sigma} \lambda_p$ and integrating in $u$ yields two cases. For $(u,v) \in \mathcal{B}^{(1),u_\delta}$ the future directed constant $v$ line through $(u,v)$ intersects the boundary of the interior region at $\{u = u_\delta\}$, where data for $\lambda_p$ is explicitly specified. For $(u,v) \in \mathcal{B}^{(2),u_\delta}$, the boundary term is along the axis, and here we only have the boundary condition $\lambda_p = -p_1 |u_\Gamma(v)|^{\kappa_1}\nu_p$. 

For $(u,v) \in \mathcal{B}^{(1),u_\delta}$ we arrive at the estimate
\begin{align*}
    ||u|^{-\alpha-\kappa_1 +\sigma}\lambda_p|^2(u,v) &+ (\alpha+\kappa_1-\sigma) \int\limits_{(u,v)}^{(u_\delta,v)}\frac{1}{|u'|}(|u'|^{-\alpha-\kappa_1 +\sigma}\lambda_p)^2(u',v)du' \nonumber \\
    &\lesssim ||u_\delta|^{-\alpha-\kappa_1 +\sigma}\lambda_p|^2(u_\delta,v)+ O(A\epsilon)\int\limits_{(u,v)}^{(u_\delta,v)}\frac{1}{|u'|^{1-\sigma}}|u'|^{-\alpha-\kappa_1 +\sigma}|\lambda_p|(u',v)du' \nonumber\\
   &\lesssim \epsilon^2|u_\delta|^{2\sigma}\mathcal{I}^2 + \frac{\alpha}{2}\int\limits_{(u,v)}^{(u_\delta,v)}\frac{1}{|u'|}(|u'|^{-\alpha-\kappa_1 +\sigma}\lambda_p)^2(u',v) + \alpha^{-1}O(A^2\epsilon^2)(|u|^{2\sigma}-|u_\delta|^{2\sigma}).
\end{align*}
Absorbing the integral expression into the left hand side, we can then divide by $|u|^{2\sigma}$ to get the desired improvement. Note $\frac{|u_\delta|^{2\sigma}}{|u|^{2\sigma}} \leq 1$, and hence the data term remains regular.

For $(u,v) \in \mathcal{B}^{(2),u_\delta}$ a similar procedure yields
\begin{align*}
    ||u|^{-\alpha-\kappa_1 +\sigma}\lambda_p|^2(u,v) \lesssim ||u_\Gamma(v)|^{-\alpha-\kappa_1 +\sigma}\lambda_p|^2(u_\Gamma(v),v) + \alpha^{-1}O(A^2\epsilon^2)(|u|^{2\sigma}-|u_\Gamma(v)|^{2\sigma}).
\end{align*}
Boundary conditions imply
\begin{equation*}
    ||u_\Gamma(v)|^{-\alpha-\kappa_1 +\sigma}\lambda_p|^2(u_\Gamma(v),v) = p_1^2 ||u_\Gamma(v)|^{-\alpha+\sigma}\nu_p|^2(u_\Gamma(v),v),
\end{equation*}
and we have already improved the estimate for the $\nu_p$. Inserting the improved bound for $\nu_p$ allows one to complete the estimate for $\lambda_p$.

Now turn to $\partial_v \lambda_p$. The starting point is the commuted equation (\ref{PSSESF:7}). Inserting bootstrap estimates and choosing $\epsilon$ small implies
\begin{equation*}
    |\partial_u \partial_v \lambda_p| \lesssim A\epsilon |u|^{\alpha-2+2\kappa_1}\Big(\frac{|u|^{q_1}}{|v|} \Big)^a.
\end{equation*}
The primary difference for the estimates on $\partial_v \lambda_p$ and $\partial_v^2(r\phi_p)$ is the added singular $v$ weights as $v \rightarrow 0$, arising from the behavior of $\partial_v \overline{\lambda}$.

Define $w_1 = |u|^{-(\alpha-1+2\kappa_1)}\Big(\frac{|v|}{|u|^{q_1}}\Big)^a$ and conjugate (\ref{PSSESF:7}) by $w_1 |u|^{\sigma},$ giving 
\begin{equation*}
    \partial_u (w_1 |u|^{\sigma} \partial_v \lambda_p) - \frac{\alpha-1+2\kappa_1+q_1a-\sigma}{|u|}(w_1 |u|^{\sigma} \partial_v \lambda_p) = O(A\epsilon |u|^{-1+\sigma}).
\end{equation*}
Assume $\alpha$ is chosen large enough so that the zeroth order coefficient is negative. Fix $(u,v) \in \mathcal{B}^{(1),u_\delta}$. Contracting with $w_1 |u|^{\sigma}\partial_v \lambda_p$, integrating in $u$, and absorbing terms as above gives the estimate 
\begin{align*}
    |w_1 |u|^{\sigma}\partial_v \lambda_p|^2(u,v) \lesssim |w_1 |u_\delta|^{\sigma}\partial_v \lambda_p|^2(u_\delta,v) + \alpha^{-1}O(A^2\epsilon^2)|u|^{2\sigma}.
\end{align*}
The decaying $v$ weight built into $w_1$ ensures the data term is bounded by $\epsilon^2 |u_\delta|^{2\sigma} \mathcal{I}$. Therefore it is enough to divide by $|u|^{2\sigma}$ and choose $\alpha$ large in order to improve the bound on $\partial_v \lambda_p$.

For $(u,v) \in \mathcal{B}^{(2),u_\delta}$, integration gives the boundary term $|w_1|u_\Gamma(v)|^{\sigma}\partial_v\lambda_p|^2(u_\Gamma(v),v)$. Boundary conditions imply 
\begin{equation*}
    (\partial_u + p_1|u_\Gamma(v)|^{\kappa_1}\partial_v)^2 r_p(u_\Gamma(v),v) = 0,
\end{equation*}
which after expanding gives a relation for $\partial_v \lambda_p$ along the axis in terms of lower order quantities and $\partial_u \nu_p$. These quantities have all been estimated in the proof already, and it follows that 
\begin{equation*}
    |w_1|u_\Gamma(v)|^{\sigma}\partial_v\lambda_p|^2(u_\Gamma(v),v) \lesssim \epsilon^2 |u|^{2\sigma}.
\end{equation*}
The boundary term therefore does not pose a problem, and the argument proceeds as above.

We next consider $\partial_v(r\phi_p)$. Estimating (\ref{PSSESF:6}) gives
\begin{equation*}
    |\partial_u \partial_v (r\phi_p)| \lesssim A\epsilon |u|^{\alpha-1+\kappa_1}.
\end{equation*}
Conjugating through by $|u|^{-\alpha-\kappa_1+\sigma}$ we arrive at a schematic equation of the same form as that of $\lambda_p$. For $(u,v) \in \mathcal{B}^{(1),u_\delta}$, directly integrating in $u$ and using that the data term along $\{u = u_\delta\}$ is bounded by hypothesis, we conclude the estimate.

For $(u,v) \in \mathcal{B}^{(2),u_\delta}$, the boundary term lies on the axis. Boundary conditions imply $\partial_v(r\phi_p)(u_\Gamma(v),v) = -p_1 |u_\Gamma(v)|^{\kappa_1}\partial_u(r\phi_p)(u_\Gamma(v),v)$, with the latter having already been estimated. The desired estimate follows as above.

It remains to consider $\partial_v^2(r\phi_p)$, which includes a singular $v$ weight in addition to the decaying $|u|$ weights. Estimating the commuted wave equation (\ref{PSSESF:11}) gives
\begin{equation*}
    |\partial_u \partial_v^2(r\phi_p)| \lesssim A\epsilon |u|^{\alpha-2+2\kappa_1}\Big(\frac{|u|^{q_1}}{|v|} \Big)^b.
\end{equation*}
We have used that $a \leq b$ to ensure that the above bound reflects the worst possible behavior as $v \rightarrow 0$. Define $w_2 = |u|^{-\alpha+1-2\kappa_1}\Big(\frac{|v|}{|u|^{q_1}} \Big)^{b}$. Conjugate through by $w_2 |u|^{\sigma}$, giving an equation of the form
\begin{equation*}
    \partial_u (w_2 |u|^{\sigma}\partial_v^2(r\phi_p)) - \frac{\alpha-1+2\kappa_1+bq_1-\sigma}{|u|}(w_2 |u|^{\sigma}\partial_v^2(r\phi_p)) = O(A\epsilon|u|^{-1+\sigma}).
\end{equation*}
For $(u,v) \in \mathcal{B}^{(1),u_\delta}$ it is therefore sufficient to contract with $w_2|u|^{\sigma}\partial_v^2(r\phi_p)$ and integrate in $u$, using bounds on the data along $\{u = u_\delta\}$.

In analogy with $\partial_v \lambda_p$, the boundary terms along $\Gamma$ may be dealt with using the conditions 
\begin{equation*}
    (\partial_u + p_1|u_\Gamma(v)|^{\kappa_1}\partial_v)^2 (r\phi_p)(u_\Gamma(v),v) = 0.
\end{equation*}
By expanding it follows that $\partial_v^2(r\phi_p)(u_\Gamma(v),v)$ may be expressed in terms of $\partial_u^2(r\phi_p)$ and lower order quantities, all of which have already been estimated. It is therefore enough to integrate the conjugated equation in $u$ for $(u,v) \in \mathcal{B}^{(2),u_\delta}$, and choose $\alpha$ large to improve the bootstrap assumption.

\end{proof}

\subsubsection{Constructing the approximate interior solution}
It follows from Proposition \ref{approxint:mainprop} that for fixed $u_\delta$, given data for the perturbation along $\{u=u_\delta\} \cup \{v = 0, \ u \leq u_\delta\}$, there exists $\epsilon$ small and a regular perturbation of the $(\overline{g},\overline{\phi})$ background in the region $\mathcal{Q}^{(in),u_\delta}.$ We summarize some key features and limitations of the argument here:
\begin{enumerate}
    \item The necessary smallness in $\epsilon$ depends on the initial data only through $\mathcal{I}^{(in)}_\alpha$.
    \item The region $\mathcal{Q}^{(in),u_\delta}$ extends all the way to $\{u = -1\}$, independently of $u_\delta$.
    \item The solution is quantitatively regular at the axis. In fact, the solution inherits all the bounds appearing in the background norm $\overline{\mathfrak{N}}_1$, with the caveat that instead of $\phi \in C^2(\mathcal{Q}^{(in),u_\delta} \setminus \{v=0\})$, we show only that $r\phi \in C^2(\mathcal{Q}^{(in),u_\delta}\setminus \{v=0\})$.
\end{enumerate}
The idea of this section is to specialize this construction in order to define the approximants $(r^{u_\delta}, m^{u_\delta}, \phi^{u_\delta})$. These approximants will be used in the next section to take a limit as $u_\delta \rightarrow 0$.

Recall the desired data for the scalar field along $\{v = 0\}$ is 
\begin{equation*}
    \partial_u(r\phi)(u,0) = \partial_u(\overline{r\phi})(u,0) + \epsilon f_0(u) ,
\end{equation*}
with the gauge condition
\begin{equation*}
    r(u,0) = \overline{r}(u,0).
\end{equation*}
Assume $f_0(u) \in C^2([-1,0])$, and for any $u_\delta < 0$ small define the cutoff piece of data
\begin{equation}
    \label{dfn:cutoffinitialdata}
    f_\delta(u) \doteq \chi_{u_\delta}(u)f_0(u).
\end{equation}
Here $\chi_{u_\delta}(u)$ is a smooth cutoff function supported on $[-1,0]$ equal to $1$ on $[-1,u_\delta(1+\sigma)]$ for some small $0 <\sigma \ll 1$, and equal to zero on $(u_\delta(1+ \frac{1}{2}\sigma),0].$ In particular, we can arrange to have the estimates $|\partial^i \chi_{u_\delta}| \lesssim |u_\delta|^{-i}.$

For a given $u_\delta$ pose the data 
\begin{align}
   \label{eq:dataapproxint1}
    r_p(u,0) &= 0, \ \ \ \partial_u(r\phi_p)(u,0) = \epsilon f_\delta(u), \\[.3em]
    r_p(u_\delta,v)  &= 0, \ \ \ \partial_v(r\phi_p)(u_\delta,v) = 0.\label{eq:dataapproxint2}
\end{align}
The initial data norm corresponding to $f_0(u)$ is given by $\mathcal{E}_{1,\alpha}$ (cf. (\ref{eq:intdatanormE1})). Similarly, let $\mathcal{E}^{\delta}_{1,\alpha}$ denote the value of the same norm computed on $f_\delta$. One can estimate that 
\begin{equation*}
    \mathcal{E}^{\delta}_{1,\alpha} \lesssim \mathcal{E}_{1,\alpha},
\end{equation*}
holding independently of $u_\delta$. By Proposition \ref{approxint:mainprop}, for $\epsilon$ small depending on $\mathcal{E}_{1,\alpha}$ there exists a solution to the scalar field system in the region $\mathcal{Q}^{(in),u_\delta}$ achieving the data.  The differences $\Psi_p$ satisfy bounds
\begin{equation}
\label{eq:seqeuncebounds}
\mathfrak{N}^{(in)}_{(tot)} \lesssim  \mathcal{E}_{1,\alpha},
\end{equation}
where $\mathfrak{N}^{(in)}_{(tot)}$ is the value of the norm (\ref{eq:inttotalspacetimenorm}) applied to the $\Psi_p$. Label the resulting solution $$(r^{u_\delta}, m^{u_\delta}, \phi^{u_\delta}) \doteq (\overline{r} + r_p, \overline{m} + m_p, \overline{\phi} + \phi_p).$$
It remains to extend this solution on $\mathcal{Q}^{in} \setminus \mathcal{Q}^{(in),u_\delta}.$

By construction, along $\{u = u_\delta\} \cup \{v = 0, u \in [u_\delta(1+\frac{1}{2}\sigma),u_\delta] \}$ the data for the perturbation vanishes. A domain of dependence (and uniqueness in BV) argument implies that in the region $\mathcal{Q}^{(in),u_\delta} \cap \{u \geq  u_\delta(1+\frac{1}{2}\sigma)\}$, we have 
$$(r^{u_\delta}, m^{u_\delta}, \phi^{u_\delta}) = (\overline{r},\overline{m},\overline{\phi}). $$
Therefore it is enough to define the solution in $\mathcal{Q}^{in} \setminus \mathcal{Q}^{(in),u_\delta}$ to be identically equal to the value of the background in the same gauge. This extension is globally BV away from the singular point (in addition to the improved regularity near the axis, and pointwise bounds). 

From here on this extended solution on $\mathcal{Q}^{(in)}$ is denoted $(r^{u_\delta}, m^{u_\delta}, \phi^{u_\delta}).$ It follows that (\ref{eq:seqeuncebounds}) holds for the extended solution, where the pointwise bounds are measured over the domain $\mathcal{Q}^{(in)}$.

\subsubsection{Auxiliary higher order estimates}
\label{subsubsec:thirdorderestimates}
Before considering the $u_\delta \rightarrow 0$ limit, we prove some higher order estimates on the approximate interior solutions. The need for such estimates is motivated by the desire to treat $(r^{u_\delta},m^{u_\delta},\phi^{u_\delta})$ as a \textit{new} admissible interior solution, around which we can consider perturbations. 

Recall that the bounds on the background solution needed for the argument of Proposition \ref{approxint:mainprop} are exactly those contained in $\overline{\mathfrak{N}}_1$. Comparing $\overline{\mathfrak{N}}_1$ and $\mathfrak{N}^{(in)}_{tot},$ it is evident that we fail to control $\partial_u^2 \phi^{u_\delta}_p, \ \partial_v^2 \phi^{u_\delta}_p$, and therefore we have not shown enough regularity on $(r^{u_\delta}, \mu^{u_\delta}, \phi^{u_\delta})$ for it to satisfy the requirements of an admissible interior solution.

We turn now to estimating $\partial_u^2 \phi_p$ and $\partial_v^2 \phi_p$. Appealing to Lemma \ref{lem:axisaverage2} suggests that near the axis, such control will only come by considering the third order quantities 
$$\partial_u^3(r\phi_p), \ \ \partial_v^3(r\phi_p), \ \ \partial_u^2 \nu_p, \ \ \partial_v^2 \lambda_p.$$

Fix $u_\delta <0$, and consider a solution $(r^{u_\delta}, \mu^{u_\delta}, \phi^{u_\delta})$ in $\mathcal{Q}^{(in),u_\delta}$ with data described by (\ref{eq:dataapproxint1})-(\ref{eq:dataapproxint2}). For simplicity, we drop the superscripts from the solution variables. 

We first state the bounds on the background solution needed here. Let 
\begin{equation}
    \overline{\mathfrak{N}}_2 \doteq \sup_{\mathcal{Q}^{(in),u_\delta}}\big||u|^3 \partial_u^3 \overline{\phi} \big| +  \sup_{\mathcal{Q}^{(in),u_\delta}}\big||u|^2 \partial_u^2 \overline{\nu} \big| +  \sup_{\mathcal{S}_{far}}\big||u|^{3-3\kappa_1}\partial_v^3 \overline{\phi} \big| +  \sup_{\mathcal{S}_{far}}\big||u|^{2-3\kappa_1}\partial_v^2 \overline{\lambda} \big| ,
\end{equation}
and define the norms
\begin{align}
    \mathfrak{N}_{\partial_u^3(r\phi_p)} &\doteq \sup_{\mathcal{Q}^{(in),u_\delta}} \Big|\frac{1}{|u|^{\alpha-2}}\partial_u^3(r\phi_p) \Big|, \ \ \ \
    \mathfrak{N}_{\partial_u^2 \nu_p} \doteq \sup_{\mathcal{Q}^{(in),u_\delta}} \Big|\frac{1}{|u|^{\alpha-2}} \partial_u^2 \nu_p \Big|, \\[.3em]
     \mathfrak{N}_{\partial_v^3(r\phi_p)} &\doteq \sup_{\mathcal{S}_{far}} \Big| \frac{1}{|u|^{\alpha-2+3\kappa_1}}\partial_v^3(r\phi_p)\Big|, \ \ \ \ \mathfrak{N}_{\partial_v^2 \lambda_p} \doteq \sup_{\mathcal{S}_{far}} \Big| \frac{1}{|u|^{\alpha-2+3\kappa_1}}\partial_v^2 \lambda_p \Big|,
\end{align}
\begin{equation}
    \mathfrak{N}^{(in)}_{aux} \doteq \mathfrak{N}_{\partial_u^3(r\phi_p)} + \mathfrak{N}_{\partial_u^2 \nu_p} +\mathfrak{N}_{\partial_v^3(r\phi_p)}+ \mathfrak{N}_{\partial_v^2 \lambda_p}.
\end{equation}

Note that $\partial_v^3(r\phi_p)$ and $\partial_v^2 \lambda_p$ are only estimated in a neighborhood $\mathcal{S}_{far}$ of the axis. Near $\{v =0\}$ these quantities are potentially singular, and while in principle one can track the precise blow up rates, it is not necessary here. This is because the aim is to estimate $\partial_u^2 \phi_p$ and $\partial_v^2 \phi_p$, and such estimates follow in a more direct manner near $\{v=0\}$. It is only near the axis where one must treat the full system at third order.

It is not the case that all estimates can be closed in $\mathcal{S}_{far}$ alone, however. The quantities $\partial_u^3(r\phi_p)$ and $\partial_u^2 \nu_p$ satisfy $v$ equations, and must be integrated from data at $\{v =0\}$. These quantities are easier to analyze, though, as they stay bounded near $\{v = 0\}$ in the background solution.

We now proceed with the argument. Preservation of regularity and the local existence argument implies that in $\mathcal{Q}^{(in),u_\delta} \cap \{u_* \leq u \leq u_\delta\}$ for some $u_*$, the bound
\begin{equation}
    \label{eq:inthigherbootstrap}
    \mathfrak{N}^{(in)}_{aux} \leq 2A\epsilon
\end{equation}
holds for $A$ large depending on $\mathcal{E}_{1,\alpha}.$ Recall we are working with solutions that have vanishing perturbation data on the outgoing line $\{u = u_\delta \}$.

We first show that bounds on $\partial_u^2\phi_p, \ \partial_v^2 \phi_p$ near $\{v= 0\}$ follow from the analysis concluded in the previous section.
\begin{lem}
\label{lem:3rdorder1}
Assume $\mathfrak{N}^{(in)}_{tot} < \infty$ in $\mathcal{Q}^{(in),u_\delta}.$ Then in $\mathcal{S}_{near}$ we have
\begin{align}
    \sup_{\mathcal{S}_{near}}|\partial^2_u\phi_p|  &\lesssim A\epsilon |u|^{\alpha-2}, \\[.5em]
    \sup_{\mathcal{S}_{near}}|\partial^2_v\phi_p| &\lesssim A \epsilon |u|^{\alpha-2+2\kappa_1}s_1^{-b}.
\end{align}
\end{lem}
\begin{proof}
It is enough to write $\partial_u \phi_p, \ \partial_v \phi_p$ in terms of $\phi_p, \ \partial_u(r\phi_p), \ \partial_v(r\phi_p), \ \partial_u^2(r\phi_p), \ \partial_v^2(r\phi_p)$. The factors of $r$ do not contribute adversely to the estimate in $\mathcal{S}_{near}$, as $r \sim |u|$ there.
\end{proof}

Next, we argue that the bootstrap assumption (\ref{eq:inthigherbootstrap}) is sufficient to control $\partial_u^2 \phi_p$ and $\partial_v^2 \phi_p$ in $\mathcal{S}_{far}$ as well:
\begin{lem}
\label{lem:3rdorder2}
Assume (\ref{eq:inthigherbootstrap}) holds. Then in $\mathcal{Q}^{(in),u_\delta}$ we estimate 
\begin{align}
    \sup_{\mathcal{Q}^{(in),u_\delta}} |\partial_u^2 \phi_p| &\lesssim A\epsilon |u|^{\alpha-2}, \\[.3em]
    \label{eq:3rdorder2v}
    \sup_{\mathcal{Q}^{(in),u_\delta}} |\partial_v^2 \phi_p| &\lesssim A\epsilon |u|^{\alpha-2+2\kappa_1}s_1^{-b}.
\end{align}
\end{lem}
\begin{proof}
The boundedness in $\mathcal{S}_{near}$ follows from the previous lemma. In $\mathcal{S}_{far}$, it is enough to use the averaging estimates (\ref{eq:averageestu2})-(\ref{eq:averageestv2}). Note we only require the averaging estimates in $\mathcal{S}_{far}$, and hence it is enough to have the higher order bounds on $\partial_v^2 \lambda, \partial_v^3 (r\phi_p)$ there. Working in $\mathcal{S}_{far}$ also allows one to drop the dependence on $s_1$, which satisfies $|s_1| \sim 1$.
\end{proof}

\begin{lem}
There exists $\alpha$ large enough depending on $\overline{\mathfrak{N}}_1, \overline{\mathfrak{N}}_2$, and $\epsilon$ small enough depending on $\mathcal{I}_{final,\alpha}$ such that 
\begin{equation}
    \label{eq:thirdorderbootimproved1}
    \mathfrak{N}_{\partial_u^3(r\phi_p)}, \mathfrak{N}_{\partial_u^2 \nu_p} \leq \frac{1}{10}A\epsilon.
\end{equation}
\end{lem}
\begin{proof}
Commuting (\ref{PSSESF:2}), (\ref{PSSESF:6}) with $\partial_u^2$ gives 
\begin{align}
    \label{eq:thirdordernu}
    \partial_v \partial_u^2 \nu_p &= \mathcal{G}_1 \partial_u^2 \nu_p + 2\partial_u \mathcal{G}_1 \partial_u \nu_p + \partial_u^2 \mathcal{G}_1 \nu_p + \partial_u^2\big((\mathcal{G}_1)_p \overline{\nu} \big), \\[.3em]
    \label{thirdorderuderphi}
    \partial_v \partial_u^3 (r\phi_p) &= \partial_u^2 \mathcal{G}_5 \phi_p + 2\partial_u \mathcal{G}_5 \partial_u \phi_p + \mathcal{G}_5 \partial_u^2 \phi_p - \partial_u^2 \mathcal{I}_5.
\end{align}

The terms that have not already been estimated at lower order are $\partial_u^2 \mathcal{G}_1, \ \partial_u^2 (\mathcal{G}_1)_p, \ \partial_u^2 \mathcal{G}_5, \ \partial_u^2 \mathcal{I}_5.$ We take each term individually. Because $\partial_u^2\mathcal{G}_1 = \partial_u^2 \overline{\mathcal{G}}_1+ \partial_u^2 (\mathcal{G}_1)_p,$ we start by estimating $\partial_u^2 \overline{\mathcal{G}}_1.$ 

If at least one derivative falls on $\overline{\lambda}$ or $1-\overline{\mu}$, then the singularity in $\overline{r}$ is mild enough to be absorbed by the remaining factor of $\overline{\mu}$ in $\overline{\mathcal{G}}_1$. The terms that emerge already fall under the control on the background solution contained in $\overline{\mathfrak{N}}_1$, $\overline{\mathfrak{N}}_2$. Otherwise, both derivatives fall on the term $\frac{\overline{\mu}}{\overline{r}}$. 

Inserting (\ref{SSESF:2:7}) for derivatives falling on $\overline{\mu}$ shows that the term with worst potential $\overline{r}$ weights arises from $\partial_u ( \frac{\overline{\mu}}{\overline{r}^2})$. All other terms appearing after taking two derivatives can be estimated via already controlled quantities, including 
$\partial_u^2 \overline{\nu}, \partial_u^2 \overline{\phi}.$

It remains to estimate $\partial_u (\frac{\overline{\mu}}{\overline{r}^2})$. Differentiating leads to similar problems as discussed in Section \ref{subsubsec:axisintlemmas}. It would seem that derivatives falling on $\overline{r}$ lead to terms of the form $\frac{\overline{\mu}}{\overline{r}^3},$ which one can only bound by $\overline{r}^{-1}$ near the axis.

The estimate is saved by first using (\ref{SSESF:2:4}) to write
\begin{equation*}
    \frac{\overline{m}}{\overline{r}^3}(u,v) = \overline{r}(u,v)^{-3}\int\limits_{(u,v)}^{(u_\Gamma(v),v) }\overline{r}(u',v)^2 \Big(\frac{1}{2\overline{\nu}}(1-\overline{\mu})(\partial_u \overline{\phi})^2\Big)(u',v)du'.
\end{equation*}
The expression is an averaging operator with $s=2$ (see Remark \ref{rmk:higherorderaveraging}). Differentiating in $u$ and applying (\ref{eq:averagingtemp}) gives 
\begin{equation*}
    \big|\partial_u\big(\frac{\overline{m}}{\overline{r}^3} \big) \big| \lesssim \big|\partial_u \big(\frac{1}{2\overline{\nu}^2}(1-\overline{\mu})(\partial_u \overline{\phi})^2 \big)\big| \lesssim |u|^{-3}.
\end{equation*}
The conclusion of this argument is that $\partial_u^2 \overline{\mathcal{G}}_1$ is regular at the axis, with 
$$|\partial_u^2 \overline{\mathcal{G}}_1| \lesssim |u|^{-3+\kappa_1}. $$

The same argument can now be applied to $\partial_u^2 (\mathcal{G}_1 )_p,$ $\partial_u^2 \mathcal{G}_5,$ and $\partial_u^2 \mathcal{I}_5$. We only comment here on the major features of these terms.
\begin{enumerate}
 \item In estimating  $\partial_u^2 (\mathcal{G}_1 )_p,$ $\partial_u^2 (\mathcal{G}_5 )_p,$ and $\partial_u^2 \mathcal{I}_5$ we encounter derivatives of the form $\partial_u^2\big(\frac{r_p}{r} \big)$, $\partial_u^2\big(\frac{r_p}{\overline{r}} \big)$, $\partial_u^2\big(\frac{\overline{r}}{r} \big)$. These may each be estimated using the averaging operator formulas, along with estimates on up to two $u$ derivatives of $\overline{\nu},  \nu_p$.
 \item Analogously to the above estimate on $\partial_u\big(\frac{\overline{m}}{\overline{r}^3}\big)$, we must contend with the term $\partial_u\big(\frac{\mu_p}{r^2}\big)$. Here we use the relation (\ref{eq:mpvsmup}), the estimates on derivatives of $r_p$ discussed in the previous bullet point, and an averaging operator argument.
    \item Using the wave equation (\ref{SSESF:1:3}) to rewrite mixed derivatives $\partial_u \partial_v \overline{\phi}$ ensures that we are using information on at most three derivatives of the background solution. At least two of the derivatives are in the $u$ direction, and hence one does not see scalar field derivatives with singular limits as $v \rightarrow 0$. Note that terms of the form $\partial_u^3 \overline{\phi}$ do in fact appear in $\partial_u^2 \mathcal{I}_5$, and hence we require information on up to three derivatives of $\overline{\phi}$ in the $u$ direction.
\end{enumerate}

Applying this strategy yields the bounds
\begin{align*}
    |\partial_u^2 (\mathcal{G}_1 )_p| \lesssim A\epsilon |u|^{\alpha-3+\kappa_1}, \ \ \ |\partial_u^2 \mathcal{G}_5| \lesssim |u|^{-3+\kappa_1}, \ \ |\partial_u^2 \mathcal{I}_5| \lesssim A\epsilon |u|^{\alpha-3+\kappa_1}.
\end{align*}

As in the case of lower order estimates, the strategy is now to conjugate (\ref{eq:thirdordernu}) and (\ref{thirdorderuderphi}) by $w^{\alpha-2}$ and integrate backwards in $v$ from data. The scalar field contributes a data term along $\{v = 0\}$, controlled by $\mathcal{I}_{final,\alpha}$. The radius term $\partial_u^2 \nu_p$ vanishes along $\{v=0\}$ by hypothesis.

Choosing $\alpha$ large enough and $\epsilon$ small enough, the bootstrap assumption is improved.

\end{proof}

\begin{lem}
There exists $\alpha$ large enough depending on $\overline{\mathfrak{N}}_1, \overline{\mathfrak{N}}_2$, and $\epsilon$ small enough depending on $\mathcal{I}_{final,\alpha}$ such that 
\begin{equation}
 \label{eq:thirdorderbootimproved2}
   \mathfrak{N}_{\partial_v^3(r\phi_p)}, \ \mathfrak{N}_{\partial_v^2 \lambda_p} \leq \frac{1}{10}A\epsilon.
\end{equation}
\end{lem}
\begin{proof}
The starting point is the same as in the previous lemma. Commuting (\ref{PSSESF:1}) and (\ref{PSSESF:6}) with $\partial_v^2$ gives
\begin{align}
     \label{eq:thirdorderlambda}
    \partial_u \partial_v^2 \lambda_p &= \mathcal{G}_2 \partial_v^2 \lambda_p + 2\partial_v \mathcal{G}_2 \partial_v \lambda_p + \partial_v^2 \mathcal{G}_2 \lambda_p + \partial_v^2 \big( (\mathcal{G}_2)_p \overline{\lambda} \big), \\[.3em]
    \label{eq:thirdordervphi}
    \partial_u \partial_v^3 (r\phi_p) &= \partial_v^2 \mathcal{G}_5 \phi_p + 2\partial_v \mathcal{G}_5 \partial_v \phi_p + \mathcal{G}_5 \partial_v^2 \phi_p - \partial_v^2 \mathcal{I}_5.
\end{align}
The new terms to estimate are 
\begin{equation}
    \label{listofterms}
\partial_v^2 \mathcal{G}_2, \ \ \partial_v^2 (\mathcal{G}_2)_p, \ \ \partial_v^2 \mathcal{G}_5,\ \  \partial_v^2 \mathcal{I}_5.  
\end{equation}
Comparison with the previous lemma shows that key quantities to estimate are $\partial_v (\frac{\overline{\mu}}{\overline{r}^2}),$ and $\partial_v (\frac{\mu_p}{\overline{r}^2}).$ Applying the $\partial_v \mu$ equation and the averaging operator formalism yields
$$|\partial_v \big(\frac{\overline{\mu}}{\overline{r}^2}\big)| \lesssim |u|^{-3+2\kappa_1}, \ \ \ |\partial_v \big(\frac{\mu_p}{\overline{r}^2}\big)| \lesssim A\epsilon |u|^{\alpha-3+2\kappa_1}.$$
Inserting bounds on the background solution and bootstrap assumptions for the remainder of the terms appearing in (\ref{listofterms}), and taking $\epsilon$ sufficiently small, leads to the estimates 
\begin{equation*}
    |\partial_v^2 \mathcal{G}_2| \lesssim |u|^{-3+2\kappa_1}, \ \ \ |\partial_v^2 (\mathcal{G}_2)_p | \lesssim A\epsilon|u|^{\alpha-3+2\kappa_1}, 
\end{equation*}
\begin{equation*}
    |\partial_v^2 \mathcal{G}_5| \lesssim |u|^{-3+2\kappa_1}, \ \ \ | \partial_v^2 \mathcal{I}_5 | \lesssim A\epsilon|u|^{\alpha-3+2\kappa_1}.
\end{equation*}
Similar comments as given in the proof of the previous lemma are relevant here. In particular, one must use averaging formulas to estimate $\partial_v^2\big( \frac{r_p}{\overline{r}}\big)$, and apply estimates on mixed derivatives of the background of the form $\partial_v^2 \partial_u \overline{\phi}$.  

Note that by restricting attention to $\mathcal{S}_{far},$ one can use the estimate (\ref{eq:3rdorder2v}) to estimate $\partial_v^2 \phi_p$ without requiring singular $v$ weights.

It follows that the right hand sides of (\ref{eq:thirdorderlambda}), (\ref{eq:thirdordervphi}) satisfy estimates with homogeneity consistent with the bootstrap assumptions for $\partial_v^2 \lambda_p$ and $\partial_v^3(r\phi_p)$. The strategy is therefore to conjugate both equations with $|u|^{-(\alpha-2+3\kappa_1-\sigma)}$ for some small $0<\sigma \ll 1$, contract, and integrate backwards in $u$ from $\mathcal{L}$. 

For any $(u,v) \in \mathcal{S}_{far}$, the future directed constant $v$ line stays contained in  $\mathcal{S}_{far}$, and intersects either the axis or the line $\{u = u_\delta\}$. In the latter case, by hypothesis the data for $\partial_v^3(r\phi_p)$ and $\partial_v^2 \lambda_p$ is zero, and so no data terms are picked up. 

In the case when the future directed constant $v$ line intersects the axis, we use regularity there to write 
\begin{align*}
   &(\partial_v + p_1|u_\Gamma(v)|^{\kappa_1}\partial_u)^3(r\phi_p)(u_\Gamma(v),v) = 0, \\[.3em]
   &(\partial_v + p_1|u_\Gamma(v)|^{\kappa_1}\partial_u)^3 r_p(u_\Gamma(v),v) = 0.
\end{align*}
Expanding and solving for $\partial_v^3(r\phi_p), \partial_v^2 \lambda_p$, we are able to apply estimates on $\partial_u^3(r\phi_p), \partial_u^2 \nu_p$, and lower order quantities, to arrive at estimates for the boundary terms. 

For $\epsilon$ small enough depending on the bootstrap assumptions, and $\alpha$ large enough depending on the implicit constants in the above estimates (which depend only on the background), integrating in $u$ backwards from $\mathcal{L}$ improves the bootstrap assumptions.
\end{proof}

\subsection{Limiting solution} 
\label{sec:limitingsolution}
To construct a limit of $(r^{u_\delta}, m^{u_\delta}, \phi^{u_\delta})$ as $u_\delta \rightarrow 0$, it is enough to construct a function space in which the sequence is Cauchy. This space should include enough regularity to ensure pointwise convergence of the double null unknowns, and that the resulting limits are BV solutions to the scalar field system.

We will work with the sequence of differences $(r^{u_\delta}_p ,\mu^{u_\delta}_p,\phi^{u_\delta}_p)$. The total solution $(r^{u_\delta}, m^{u_\delta}, \phi^{u_\delta})$ can easily be recovered from $(r^{u_\delta}_p ,\mu^{u_\delta}_p,\phi^{u_\delta}_p)$ by adding back the background piece (and defining $m$ in terms of $\mu$), and so there is no loss in working at the level of perturbations.

Define the Banach space
\begin{equation}
    \mathcal{Y} \doteq \{(r,\mu,\phi) \in C^1(\mathcal{Q}^{(in)}) \times C^0(\mathcal{Q}^{(in)}) \times C^1(\mathcal{Q}^{(in)})\},
\end{equation}
with norm
\begin{equation}
    \|(r,\mu,\phi) \|_{\mathcal{Y}} \doteq \|r \|_{C^1(\mathcal{Q}^{(in)})} + \|\mu \|_{C^0(\mathcal{Q}^{(in)})} + \|\phi \|_{C^1(\mathcal{Q}^{(in)})}.
\end{equation}

For $-1 \leq u_2 < u_1 < 0$, let $(r^{(1)},\mu^{(1)}, \phi^{(1)})$, $(r^{(2)},\mu^{(2)},\phi^{(2)})$ be the approximate interiors constructed with trivial data along $\{u = u_1\},$ $\{u = u_2\}$ respectively. We proceed to estimate the differences $(r^{(1)}_p - r^{(2)}_p,\mu^{(1)}_p - \mu^{(2)}_p, \phi^{(1)}_p- \phi^{(2)}_p)$ in $\mathcal{Q}^{(in)}$, and show the $\mathcal{Y}$ norm of this difference vanishes as $|u_1|, 
 |u_2| \rightarrow 0$. Without loss of generality, we will always assume $u_2 < u_1$ holds.
 
We will consider the behavior of the solutions in the three regions (see Figure \ref{fig:approxintconv})
\begin{align*}
    \mathcal{X}_{\text{I}} &\doteq \mathcal{Q}^{(in)} \cap \{u \geq u_1 \}, \\
    \mathcal{X}_{\text{II}} &\doteq \mathcal{Q}^{(in)} \cap \{u_1 > u \geq u_2(1+\sigma) \}, \\
    \mathcal{X}_{\text{III}} &\doteq \mathcal{Q}^{(in)} \cap \{u_2(1+\sigma) > u \geq -1 \}.
\end{align*}
Here, $\sigma$ is the parameter associated to the cutoff scale for the initial data along $\{v=0\},$ cf. the definition of $f_\delta$ in (\ref{dfn:cutoffinitialdata}).

\begin{figure}
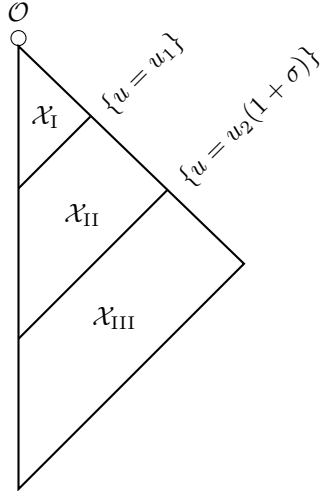

    \centering
  \includestandalone[]{Figures/fig_approxint_convergence}
  \caption{Various regions formed from taking differences of the solutions $(r^{(1)},\mu^{(1)},\phi^{(1)})$, $(r^{(2)},\mu^{(2)},\phi^{(2)})$.}
  \label{fig:approxintconv}
\end{figure}

\vspace{5pt}
\noindent
\textbf{Region I: }
In Region $\mathcal{X}_{\text{I}}$ both solutions $(r^{(i)}_p, \mu^{(i)}_p, \phi^{(i)}_p)$ vanish, so we get
$$\|(r^{(1)}_p,\mu^{(1)}_p,\phi^{(1)}_p) - (r^{(2)}_p,\mu^{(2)}_p,\phi^{(2)}_p) \|_{\mathcal{Y}} =0.$$

\vspace{5pt}
\noindent
\textbf{Region II: }
In $\mathcal{X}_{\text{II}}$ the solutions may no longer vanish; however, the size of the region in the $u$ direction is proportional to $u_2$, which we will use to show the contribution of the solution to the $\mathcal{Y}$ norm vanishes as $u_2 \rightarrow 0$. 

More precisely, Proposition \ref{approxint:mainprop} applied to the individual solutions $(r^{(i)}_p, \mu^{(i)}_p, \phi^{(i)}_p)$ gives the uniform estimates 
\begin{equation*}
    \|r_p^{(i)} \|_{C^1(\mathcal{X}_{\text{II}})} + \|\mu_p^{(i)} \|_{C^0(\mathcal{X}_{\text{II}})} + \|\phi_p^{(i)} \|_{C^1(\mathcal{X}_{\text{II}})} \lesssim \mathcal{E}_{1,\alpha}|u_2(1+\sigma)|^{\alpha-1},
\end{equation*}
where the constants are independent of $u_1$ provided $u_1 > u_2$. By the triangle inequality, similar estimates hold for the differences of the solution variables as well. Sending $u_2 \rightarrow 0$ gives the result.

\vspace{5pt}
\noindent
\textbf{Region III: }
In $\mathcal{X}_{\text{III}}$ we have to understand the difference of the solutions more carefully. The strategy will be to reuse the formalism of Section \ref{subsec:approxint}. Recall that the idea there was to construct the approximate interior solution by writing down the system for the differences between a putative perturbed solution and a \enquote{background} solution $(\overline{g},\overline{\phi})$. Provided a certain norm of the background, $\overline{\mathfrak{N}}_1$, was finite, the analysis of the difference quantities proceeded largely independently of the fine scale structure of the background.

In $\mathcal{X}_{\text{III}}$ one is considering a difference of two solutions, with the \enquote{background} solution now given by 
\begin{equation}
(\overline{r}', \overline{\mu}', \overline{\phi}') \doteq (r^{(1)}, \mu^{(1)}, \phi^{(1)}) = (\overline{r},\overline{\mu},\overline{\phi} ) + (r^{(1)}_p, \mu^{(1)}_p, \phi^{(1)}_p).
\end{equation}

The differences $(\widehat{r}_p, \widehat{\mu}_p, \widehat{\phi}_p) \doteq (r^{(2)}_p - r^{(1)}_p, \mu^{(2)}_p - \mu^{(1)}_p, \phi^{(2)}_p - \phi^{(1)}_p)$ solve a version of the system (\ref{PSSESF:1})-(\ref{PSSESF:5}), with the following alterations.
\begin{enumerate}
    \item All terms $\overline{\Psi}$ are replaced by $\overline{\Psi}'$.
    \item All terms $\Psi_p$ are replaced by $\widehat{\Psi}_p$.
\end{enumerate}
The differences achieve the data 
\begin{equation}
    \label{eq:region3data1}
    \widehat{\phi}_p(u,0) = 0, \ \ \ \widehat{r}_p(u,0) = 0,
\end{equation}
\begin{equation}
  \label{eq:region3data2}
    \widehat{\phi}_p(u_2(1+\sigma),v) = -\phi^{(1)}_p(u_2,v),  \ \ \ \widehat{r}_p(u_2(1+\sigma),v) = -r^{(1)}_p(u_2,v).
\end{equation}
In contrast to the setting of the approximate interior spacetimes, the data for the problem in $\mathcal{Q}^{(in),u_2(1+\sigma)}$ is trivial on the ingoing surace $\{v= 0, \ u \leq u_2(1+\sigma)\}$, but is non-trivial on the outgoing surface $\{u = u_2(1+\sigma)\}$. 

The idea is to now apply the results of Section \ref{subsec:approxint} in the region $\mathcal{Q}^{(in),u_2(1+\sigma)}$ to conclude estimates on the $\widehat{\Psi}_p$. Let $\widehat{\mathcal{I}}^{(in)}_{\alpha}$ denote the initial data norm (\ref{eq:intinitialdatanorm}) computed on the data (\ref{eq:region3data1})-(\ref{eq:region3data2}). Similarly, define $\overline{\mathfrak{N}}_1'$ as in (\ref{eq:totalbackgroudnorm}), substituting the values of the background solution
$(\overline{r}', \overline{\mu}', \overline{\phi}')$. Finally one can define $\mathfrak{N}^{(in)}_{tot,\alpha}$ as in (\ref{eq:inttotalspacetimenorm}). Observe that $\overline{\mathfrak{N}}_1' < \infty$ follows from the estimates of Section (\ref{subsec:approxint}) applied to $(r^{(1)}_p, \mu^{(1)}_p, \phi^{(1)}_p)$, and it is precisely here that we require the auxiliary higher order estimates on the solution derived in Section \ref{subsubsec:thirdorderestimates}.

It is straightforward to check that 
$$\widehat{\mathcal{I}}^{(in)}_\alpha \lesssim \mathcal{E}_{1,\alpha}.$$
There now exists an $\alpha$ large (perhaps larger than the value chosen in Section \ref{subsec:approxint}) and $\epsilon$ small depending only on $\mathcal{E}_{1,\alpha}$ such that we have uniform control on the solution norm $\mathfrak{N}^{(in)}_{tot,\alpha}$ in $\mathcal{Q}^{(in),u_2(1+\sigma)}$. The solution norm in turn controls the $\mathcal{Y}$ norm of the differences. More precisely, we have the sequence of bounds
\begin{equation*}
    \|(r^{(1)}_p,\mu^{(1)}_p,\phi^{(1)}_p) - (r^{(2)}_p,\mu^{(2)}_p,\phi^{(2)}_p) \|_{\mathcal{Y}} \lesssim \mathfrak{N}^{(in)}_{tot,\alpha} \lesssim \widehat{\mathcal{I}}^{(in)}_\alpha \lesssim \mathcal{E}_{1,\alpha}.
\end{equation*}
Although the bound is uniform, it does not provide decay as $|u_2| \rightarrow 0$. To generate the decay we consider the bound with $\alpha' = \alpha - 1$ instead, and conclude 
\begin{equation*}
   \mathfrak{N}^{(in)}_{tot,\alpha'} \lesssim \widehat{\mathcal{I}}^{(in)}_{\alpha'} \lesssim \mathcal{E}_{1,\alpha}|u_2|.
\end{equation*}
As $|u_2| \rightarrow 0$ we thus conclude $\|(r^{(1)}_p,\mu^{(1)}_p,\phi^{(1)}_p) - (r^{(2)}_p,\mu^{(2)}_p,\phi^{(2)}_p) \|_{\mathcal{Y}} \rightarrow 0$, as desired.

\vspace{10pt}
We have thus shown that the sequence $(r_p^{u_\delta},\mu_p^{u_\delta}, \phi_p^{u_\delta})$ is Cauchy in the space $\mathcal{Y}$, and therefore there exists a limit $(r_{p,\infty}, \mu_{p,\infty}, \phi_{p,\infty}) \in C^1(\mathcal{Q}^{(in)})\times C^0(\mathcal{Q}^{(in)}) \times C^1(\mathcal{Q}^{(in)})$. Moreover, the sequence converges pointwise uniformly in $\mathcal{Q}^{(in)}$. Define the limiting spacetime
\begin{equation}
    (r_\infty,\mu_\infty,\phi_\infty) \doteq (r_{p,\infty}, \mu_{p,\infty}, \phi_{p,\infty}) + (\overline{r}, \overline{\mu}, \overline{\phi}).
\end{equation}
It remains to study the limit, and show the following:
\begin{enumerate}[(i)]
    \item $(r_\infty,\mu_\infty,\phi_\infty)$ is a pointwise solution to (\ref{SSESF:2:1})-(\ref{SSESF:2:5}),
    \item The solution achieves the boundary conditions (\ref{eq:axisbc1})-(\ref{eq:axisbc2}) on the axis and the data along $\{v=0 \}$,
    \item The solution has the stated regularity, namely it is $C^1$ away from $\mathcal{O} \cup \{v=0\}$ and is globally BV away from $\mathcal{O}$.
    \item The solution converges asymptotically to the background solution as $u \rightarrow 0$ with the required rates.
\end{enumerate}

To show that the equations are satisfied, let $\mathcal{U} \subset \mathcal{Q}^{(in)}$ be a subset supported away from the origin, i.e. $\mathcal{U} \subset \mathcal{Q}^{(in)} \cap \{u \leq u_{\delta_1}\}$ for some $u_{\delta_1} < 0$. By hypothesis, all elements of the sequence $(r_p^{u_\delta},\mu_p^{u_\delta}, \phi_p^{u_\delta})$ are pointwise solutions to (\ref{PSSESF:1})-(\ref{PSSESF:6}). We have also shown that $r_p^{u_\delta}, \nu_p^{u_\delta}, \lambda_p^{u_\delta}, \mu_p^{u_\delta},\phi_p^{u_\delta},\partial_u \phi_p^{u_\delta},\partial_v \phi_p^{u_\delta} $ all converge pointwise in $\mathcal{U}$. 

Of course, it also then follows that $r^{u_\delta}, \nu^{u_\delta},\lambda^{u_\delta}, \mu^{u_\delta}, \phi^{u_\delta}, \partial_u \phi^{u_\delta}, \partial_v \phi^{u_\delta}$ converge pointwise, as these are just the sum of the (fixed) background solution and the converging sequence of differences.

The equations (\ref{SSESF:2:1})-(\ref{SSESF:2:5}) are satisfied along the sequence, and the right hand sides converge pointwise. It follows that the derivatives $\partial_u \lambda^{u_\delta}, \partial_v \nu^{u_\delta}, \partial_u \mu^{u_\delta}, \partial_v \mu^{u_\delta}, \partial_u\partial_v(r^{u_\delta}\phi^{u_\delta})$ converge uniformly away from $\mathcal{O}$, and that the limit is also a solution to (\ref{SSESF:2:1})-(\ref{SSESF:2:5}).

A similar argument, considering the equations satisfied by the sequence of second order unknowns
$$\partial_u \nu^{u_\delta}, \ \partial_v \lambda^{u_\delta}, \ \partial_u^2(r^{u_\delta}\phi^{u_\delta}), \ \partial_v^2(r^{u_\delta}\phi^{u_\delta}),$$
shows that these quantities also converge pointwise uniformly, and the limit is a solution to the differentiated Einstein-scalar field system. 

Given the uniform bounds on the sequence computed via Proposition \ref{approxint:mainboot}, as well as the convergence of the sequence and its derivatives, we conclude the stated properties (i)-(iv) above.

%% file: Figures/fig_approxint_convergence.tex
\begin{tikzpicture}
\node[circle, label= above:{$\mathcal{O}$}, draw, inner sep =0, minimum size = .2cm] (n1) at (0,0) {};

\draw[thick] (n1) -- ++(-90:6) -- ++(45:4.24) -- cycle;
\draw[thick] (-90:2) -- ++ (45: 1.35); 
\draw[thick] (-90:4) -- ++ (45:2.81);

\node[label = {[rotate=45]right:$\{u=u_1\}$ }] at ($(-90:2) + (45: 1.35)$ ) {};
\node[label = {[rotate=45]right:$\{u=u_2(1+\sigma)\}$ }] at ($(-90:4) + (45: 2.81)$ ) {};
\node[] at ($(-90:1.4) + (45:.5)$) {$\mathcal{X}_{\text{I}} $};
\node[] at ($(-90:3.2) + (45:1.2)$) {$\mathcal{X}_{\text{II}} $};
\node[] at ($(-90:5) + (45:1.8)$) {$\mathcal{X}_{\text{III}} $};
\end{tikzpicture}

%% file: sec5.tex
\section{Exterior solution: Proof of Theorem \ref{thm2:exterior}}
\label{sec:proofthm2}
The starting point for the construction of the exterior solution is the characteristic initial value problem with data along $\underline{\Sigma}_0 \cup \Sigma^{(ex)}_{-1}$. For convenience we recall here the form of the initial data:
\begin{align}
    \label{intdata1.1}
    r(u,0) &= \overline{r}(u,0), \\[.5em]
    \label{intdata2.1}
    \partial_u(r\phi)(u,0) &= \partial_u(\overline{r\phi})(u,0) + \epsilon f_0(u), \\[.5em]
     \label{extdata1.1}
    \lambda(-1,v) &= \overline{\lambda}(-1,v) + \chi_1(v)\Lambda, \\[.5em]
    \label{extdata2.1}
    \partial_v(r\phi)(-1,v) &= \partial_v (\overline{r\phi})(-1,v)+ \chi_1(v) \Phi +\epsilon \Bigl \{ \chi_1(v) v + \Big(1-\chi_1(v)\Big) \Bigr \}  g_0(v).
\end{align}
The values of $r, \phi$ at the intersection point $(-1,0)$ of the null hypersurfaces are determined by the limit of these quantities along $\{v=0\}$, or equivalently, as limits of their respective values in the interior region of spacetime. As in the interior region, we require the asymptotic boundary conditions $(r\phi)(u,0), m(u,0) \rightarrow 0$ as $u \rightarrow 0$.

In the region $v \lesssim 1$, the analysis takes place away from both the axis $\Gamma$ and null infinity. Here it is preferable to work directly with the scalar field derivatives $\partial_u \phi, \partial_v \phi$, rather than quantities associated to $r\phi$ as was done in the interior. One can check from the prescribed outgoing data for $\lambda, \partial_v(r\phi)$, and the matching conditions at $(-1,0)$, that the following estimates hold for $\partial_v \phi(-1,v)$ in $v \lesssim 1$:
\begin{equation}
    |\partial_v \phi_p(-1,v) - \partial_v \phi^{(in)}(-1,0)|\lesssim \epsilon v, \ \ \  |\partial_v^2 \phi_p|(-1,v) \lesssim \epsilon,
\end{equation}
where the implied constants depend on $\mathcal{E}_2$. 

We now turn to the construction, beginning with a discussion of the local solution in a neighborhood of $\{v=0\}$.

\subsection{Local existence}
\label{subsec:extlocalexist}
For a fixed $(\wt{u},\wt{v}) \in \mathcal{Q}^{(ex)},$ we generate a local solution in the characteristic rectangle $\mathcal{Q}^{(ex)}_{\wt{u},\wt{v}}$. As $\mathcal{Q}^{(ex)}_{\wt{u},\wt{v}}$ is localized away from $\mathcal{O}$, we are able to apply standard characteristic local wellposedness. The initial data is regular, contains no trapped surfaces, and has $r$ bounded below away from zero. It follows that a BV solution exists in 
\begin{equation*}
    \mathcal{U}_{\wt{u},\wt{v}} \doteq \mathcal{Q}^{(ex)}_{\wt{u},\wt{v}} \cap \{0 \leq v \leq v_{\wt{u},\wt{v}} \},
\end{equation*}
for some $v_{\wt{u},\wt{v}} > 0$ depending on $(\wt{u},\wt{v})$. The solution is in fact in $C^1$ up to the possible blowup of $\partial_v \lambda, \partial_v^2 \phi$ as $v \rightarrow 0$.

Fix now a sequence of $\tilde{u}_i \rightarrow 0$, and consider the local solutions defined in $\mathcal{U}_i \doteq \mathcal{U}_{\tilde{u}_i, 1}$. Uniqueness implies these solutions may be glued together on their common domain of definition, and one arrives at a local solution defined on a region $\mathcal{U} = \bigcap_{i}\mathcal{U}_i$.

Summarizing the conclusion of the above argument, we arrive at the following:
\begin{prop}
There exists a BV solution to (\ref{SSESF:2:1})-(\ref{SSESF:2:5}) achieving the data in a region $\mathcal{U}$ of the form 
$$\mathcal{U} \doteq \{(u,v) | -1 \leq u < 0, \ 0 \leq v \leq T(u) \}, $$
where $T(u)$ is a continuous, decreasing function on $[-1,0]$ satisfying $T(0) = 0$ and $T(-1) = 1$.
\end{prop}
Note that we are not guaranteed local existence in a full self-similar neighborhood of $\{v=0\}$. However, standard continuation criteria imply that to extend the solution uniformly in a \textit{given} rectangle $\mathcal{Q}_{\wt{u},\wt{v}}$, it is enough to bound pointwise the double null unknowns $$\Psi \in \{r,\nu,\lambda,\mu,\partial_u \nu, |v|^a \partial_v \lambda, \partial_u \phi, \partial_v \phi, \partial_u^2 \phi, |v|^{b}\partial_v^2 \phi\}.$$ 

\subsection{Region I}
Fix a small parameter $v_1 \ll 1$, and consider the self-similar neighborhood of $\{v=0\}$ given by 
\begin{equation}
\mathcal{R}_{\text{I}} \doteq \{(u,v) \ | \ 0\leq \frac{v}{|u|^{q_1}}\leq v_1, -1 \leq u < 0 \}.
\end{equation}
We first define a renormalized set of first order unknowns, for which we will propagate the main estimates of this section. These renormalized quantities are extensions of the differences $\Psi_p$ considered in the interior region, and have the advantage of vanishing along $\{v=0\}$. As a preliminary step we define various \enquote{corrector} functions. For a given $\Psi \in \{r, \nu, \lambda,\mu,  \partial_u \phi, \partial_v \phi\},$ define 
\begin{equation}
    \label{eqn:correctordfn}
    \Psi_c(u,v) = \lim_{v\rightarrow 0^-}\Psi^{(in)}(u,v) - \overline{\Psi}(u,0),
\end{equation}
where $\Psi^{(in)}(u,v)$ denotes the value in the interior region $\mathcal{Q}^{(in)}$. With this notation introduce the renormalized variables
\begin{equation}
\wt{\Psi}(u,v) \doteq \Psi(u,v)  - \overline{\Psi}(u,v) - \Psi_c(u,v).
\end{equation}
Note that we will not work with renormalized values of the second order quantities $\partial_u \nu, \partial_v \lambda, \partial_u^2 \phi, \partial_v^2 \phi$. For these quantities it is sufficient to directly estimate the differences $\Psi_p$, after closing the estimates on the first order quantities.

Finally, for any function $F$ of double null unknowns, denote the restriction to $\{v=0\}$ by 
\begin{equation*}
    \mathring{F}(u,v) \doteq F(u,0).
\end{equation*}
Before turning to analysis of the $\wt{\Psi}$, we note the following simple consequence of the choice of ingoing data along $\{v=0\}$ and Theorem \ref{thm1:interior}.
\begin{lem}
\label{lem:correcestimates}
The corrector functions $r_c, \ \nu_c, \ \partial_u \nu_c$ vanish identically. The remaining functions satisfy the following estimates uniformly in $\mathcal{R}_{\text{I}}$:
\begin{align}
    |\lambda_c| &\lesssim \epsilon |u|^{\kappa_1+\delta}, \ \ \ | \mu_c| \lesssim \epsilon |u|^{\delta}, \ \ \  |\partial_u \phi_c| \lesssim \epsilon |u|^{-1+\delta}, \  \ \ |\partial_v \phi_c| \lesssim \epsilon |u|^{-1+\kappa_1+\delta}.
\end{align}
Here, $\delta < 1$ is an arbitrary constant.
\end{lem}
In the remainder of the proof of Theorem \ref{thm2:exterior}, the parameter $\delta$ will serve to track the increased regularity of the differences $\Psi_p$ near $\mathcal{O}$. It will be chosen small in the proof as a function of the background solution.

We conclude this section by deriving relevant equations for the $\wt{\Psi}$, which will be needed in the estimates below. The equations are analogous to (\ref{PSSESF:1})-(\ref{PSSESF:6}); however, the $\wt{\Psi}$ are defined by subtracting away corrector terms which are not exact solutions to the scalar field system. Therefore the equations for the renormalized quantities contain additional inhomogeneous terms. 

\begin{lem}
\label{lem:exregion1maineqs}
    The renormalized unknowns satisfy the equations 
    \begin{align}
        \partial_v \wt{r} &= \wt{\lambda} + \lambda_c, \label{eq:extPSSESF1}\\
        \partial_v \wt{\nu} &= \mathcal{G}_1 \wt{\nu} + (\mathcal{G}_1 - \overline{\mathcal{G}}_1)\overline{\nu}, \label{eq:extPSSESF2}\\
        \partial_v \wt{\mu} &= -\mathcal{G}_3 \wt{\mu} +  (\overline{\mathcal{G}}_3-\mathcal{G}_3)\overline{\mu} + (\mathcal{I}_3 - \overline{\mathcal{I}}_3) - \mathcal{G}_3\mu_c, \label{eq:extPSSESF3}\\
        \partial_v \wt{\partial_u \phi} &=  -\mathcal{T}_1 \wt{\partial_u \phi}  + \mathcal{F}_1, \label{eq:extPSSESF5}\\
        \partial_u \wt{\lambda} &= \mathcal{G}_2 \wt{\lambda} +  (\mathcal{G}_2-\mathring{\mathcal{G}}_2 )\lambda_c + (\mathcal{H}_1 - \mathring{\mathcal{H}}_1), \label{eq:extPSSESF7}\\
        \partial_u \wt{\partial_v \phi} &= -\mathcal{T}_1\wt{\partial_u \phi} - \mathcal{T}_2\wt{\partial_v \phi}  - (\mathcal{T}_1 - \mathring{\mathcal{T}}_1)(\partial_u \phi)_c - (\mathcal{T}_2 - \mathring{\mathcal{T}}_2)(\partial_v \phi)_c+(\mathcal{H}_2 - \mathring{\mathcal{H}}_2), \label{eq:extPSSESF9}
    \end{align}
    where
    \begin{align}
        \mathcal{F}_1 &\doteq- \mathcal{T}_2 \wt{\partial_v \phi} - (\mathcal{T}_1 - \overline{\mathcal{T}}_1)(\partial_u \overline{\phi}) - (\mathcal{T}_2 - \overline{\mathcal{T}}_2)(\partial_v \overline{\phi}) - \mathcal{T}_1 (\partial_u \phi)_c - \mathcal{T}_2 (\partial_v \phi)_c,\\
        \mathcal{H}_1 &\doteq (\mathcal{G}_2 - \overline{\mathcal{G}}_2)\overline{\lambda}, \\
        \mathcal{H}_2 &\doteq -(\mathcal{T}_1-\overline{\mathcal{T}}_1)\partial_u \overline{\phi}-(\mathcal{T}_2-\overline{\mathcal{T}}_2)\partial_v \overline{\phi}. 
    \end{align}
Moreover, each term appearing on the right hand side of equations (\ref{eq:extPSSESF7})-(\ref{eq:extPSSESF9}) vanishes along $\{v=0\}$.
\end{lem}

\begin{figure}
    \centering
  \includestandalone[]{Figures/fig_region1}
  \caption{$\mathcal{R}_{\text{I}}$}
  \label{fig:region1}
\end{figure}
\subsubsection{Norms and bootstrap assumptions}
Fix $(\wt{u},\wt{v}) \in \mathcal{R}_{\text{I}}$ and $0<\beta \ll 1$ sufficiently small, and define the solution norms

\begin{align}
    \mathfrak{M}_1 \doteq &\sup_{\mathcal{Q}^{(ex)}_{(\wt{u},\wt{v})}} \Big|\frac{1}{|u|^{1+\delta}}\wt{r}\Big| +   \sup_{\mathcal{Q}^{(ex)}_{(\wt{u},\wt{v})}} \Big|\frac{1}{|u|^{\delta}}\wt{\nu}\Big|  
    +  \sup_{\mathcal{Q}^{(ex)}_{(\wt{u},\wt{v})}} \Big|\frac{1}{|u|^{\delta}}\wt{\mu}\Big| + \sup_{\mathcal{Q}^{(ex)}_{(\wt{u},\wt{v})}} \Big||u|^{1-\delta}\wt{\partial_u \phi}\Big| ,\\[1.5em]
    \mathfrak{M}_2 \doteq &\sup_{\mathcal{Q}^{(ex)}_{(\wt{u},\wt{v})}}\Big| \frac{1}{|u|^{\kappa_1 + \delta}}\wt{\lambda}\Big| 
    +\sup_{\mathcal{Q}_{(\wt{u},\wt{v})}} \Big||u|^{q_1-\delta}\Big(\frac{|u|^{q_1}}{v} \Big)^{1-a-\beta}\wt{\partial_v \phi}\Big| 
\end{align}
as well as the total spacetime norm
\begin{equation}
\mathfrak{M} \doteq \mathfrak{M}_{1} + \mathfrak{M}_2.
\end{equation}
These norms are defined relative to a fixed $\mathcal{Q}^{(ex)}_{\wt{u},\wt{v}}$, but the estimates will be insensitive to the choice of $(\wt{u},\wt{v}) \in \mathcal{R}_{\text{I}}$. For this reason the notation does not include a reference to the specific characteristic rectangle.

A continuity argument implies that existence in Region I follows from the following proposition:
\begin{prop}
\label{prop:extregion1boot}
Assume a local solution is given in a rectangle $\mathcal{Q}_{\wt{u},\wt{v}}$, and that the bound
\begin{equation}
    \label{eq:extbootass1}
    \mathfrak{M} \leq 2 A \epsilon
\end{equation}
holds. Then for $v_1$ sufficiently small depending on the background solution, and $\epsilon$ sufficiently small depending on $A$ and initial data, we have the improved bound 
\begin{equation}
     \label{eq:extbootass1improved}
    \mathfrak{M} \leq A \epsilon.
\end{equation}
\end{prop}

\subsubsection{Recovering bootstrap assumptions}
An immediate consequence of the bootstrap assumption (\ref{eq:extbootass1}) is that for $\epsilon$ small enough, the \textit{total} solution in Region I is comparable with the background solution:
\begin{lem}
\label{lem:extregion1leadingorderbehavior}
Assume the bound (\ref{eq:extbootass1}). For $\epsilon$ sufficiently small the geometry satisfies
\begin{equation}
    r \sim |u|, \ \ \quad (-\nu) \sim 1, \ \ \quad \lambda \sim |u|^{\kappa_1}, \ \ \quad (1-\mu) \sim 1.
\end{equation}
Moreover, the scalar field satisfies the bounds
\begin{align}
   |\partial_u \phi| \lesssim |u|^{-1}, \ \ \  |\partial_v \phi| \lesssim |u|^{-q_1}.
\end{align}
\end{lem}
\begin{proof}
For a given double null unknown one can write $\Psi = \overline{\Psi} + \wt{\Psi} + \Psi_c$. The bound (\ref{eq:extbootass1}) controls $\wt{\Psi}$, and Lemma \ref{lem:correcestimates} controls $\Psi_c$. In both cases the bounds imply that the deviation from the background is $\epsilon$-small, and so for $\epsilon$ sufficiently small the leading order behavior of the background solution is unaffected.
\end{proof}

We now turn to improving the bootstrap assumption. Having established the norms, the estimates will all follow a similar strategy, with differences only between quantities satisfying $u$ and $v$ equations. After closing the bootstrap, bounds for $\partial_u \nu, \partial_v \lambda, \partial_u^2 \phi, \partial_v^2 \phi$ will follow by a simple preservation of regularity argument. 

With the singular weights appearing in the $u$ coordinate, it is easier to first close estimates on $v$ equations. We therefore start with the following:
\begin{lem}
Assume (\ref{eq:extbootass1}). Then for $v_1$, $\epsilon$ sufficiently small we have
\begin{equation}
    \mathfrak{M}_{1} \leq \frac{1}{10}A\epsilon. 
\end{equation}
Moreover, we have the following bounds consistent with vanishing along $\{v=0\}$:
\begin{align}
    |\wt{r}| &\lesssim A\epsilon|u|^{1+\delta}\Big(\frac{v}{|u|^{q_1}}\Big), \ \ \ |\wt{\mu}| \lesssim A\epsilon|u|^{\delta}\Big(\frac{v}{|u|^{q_1}} \Big), \label{eq:ext1temp001}\\[.4em]
    |\wt{\nu}| &\lesssim A\epsilon|u|^{\delta} \Big(\frac{v}{|u|^{q_1}} \Big), \ \ \ |\wt{\partial_u \phi}| \lesssim A\epsilon|u|^{-1+\delta}  \Big(\frac{v}{|u|^{q_1}} \Big). \label{eq:ext1temp002}
\end{align}
\end{lem}
\begin{proof}
We illustrate the method with $\wt{\nu}$, and estimate equation (\ref{eq:extPSSESF2}). Inserting bootstrap bounds into the right hand side and crudely estimating $\frac{v}{|u|^{q_1}} \lesssim 1$ gives $|\partial_v \wt{\nu}| \lesssim A\epsilon |u|^{-1+\delta+\kappa_1}.$ Fix $(u,v) \in \mathcal{R}_{\text{I}}$, and integrate this estimate in $v$ from data to conclude
\begin{equation*}
    |\wt{\nu}|(u,v) \lesssim A\epsilon |u|^{-1+\delta+\kappa_1}v \lesssim A\epsilon |u|^{\delta}\Big(\frac{v}{|u|^{q_1}}\Big) \lesssim A\epsilon |u|^{\delta}v_1.
\end{equation*}
The estimate (\ref{eq:ext1temp002}) for $\wt{\nu}$ follows. Note in the above we have used that $\wt{\nu}(u,0) = 0$. Finally, choosing $v_1$ sufficiently small in the last inequality improves the bootstrap assumption on $\wt{\nu}$.

An identical argument applies to $\wt{r}, \wt{\mu},$ and $ \wt{\partial_u \phi},$ by estimating the terms appearing in equations (\ref{eq:extPSSESF1}), (\ref{eq:extPSSESF3}), (\ref{eq:extPSSESF5}), and integrating from data. 
\end{proof}

\begin{lem}
Assume (\ref{eq:extbootass1}). Then for $v_1$, $\epsilon$ sufficiently small we have
\begin{equation}
    \mathfrak{M}_{2}\leq \frac{1}{10}A\epsilon.
\end{equation}
\end{lem}
\begin{proof}
Start with $\wt{\lambda}$, satisfying (\ref{eq:extPSSESF7}). The bootstrap assumptions imply the zeroth order term has a good sign, i.e. $\mathcal{G}_2 \leq 0$. We will not require estimates on its magnitude, so we turn to the inhomogeneous terms appearing in (\ref{eq:extPSSESF7}). It is key that all terms satisfy bounds consistent with vanishing along $\{v=0\}$, as the integration in $u$ will not be able to remove any singular $v$ behavior. The form of the $u$ equations in Lemma \ref{lem:exregion1maineqs}, as well as the estimates (\ref{eq:ext1temp001})-(\ref{eq:ext1temp002}), reflect this vanishing.

Expanding the right hand side of (\ref{eq:extPSSESF7}) and estimating gives
\begin{align*}
    |(\mathcal{G}_2 - \mathring{\mathcal{G}}_2)\lambda_c+\mathcal{E}_1 - \mathring{\mathcal{E}}_1| \lesssim A\epsilon |u|^{\kappa_1+\delta-1}\Big( \frac{v}{|u|^{q_1}} \Big)^{1-a}.
\end{align*}
This estimate is a \textit{worst} case bound, and uses that the Hölder control on $\partial_v \overline{\lambda}$ implies we in general only have the bound
\begin{equation*}
    |\overline{\lambda} - \mathring{\overline{\lambda}}| \lesssim |u|^{\kappa_1}\Big(\frac{v}{|u|^{q_1}} \Big)^{1-a}.
\end{equation*}
It follows that the equation for $\wt{\lambda}$ is schematically
\begin{equation}
    \label{eq:temptemptemp}
    \partial_u \wt{\lambda} = \mathcal{G}_2 \wt{\lambda} + O \Big(A\epsilon |u|^{\kappa_1+\delta-1}\Big( \frac{v}{|u|^{q_1}} \Big)^{1-a} \Big).
\end{equation}
The total $u$ weight appearing in the error term is $\kappa_1 + \delta - 1 - q_1(1-a) = -2 + 2\kappa_2 + aq_1 + \delta$. Recalling $a < 1-p_1\kappa_1$, it follows that $\delta$ can be chosen small enough such that this term is strictly less than $-1$. In particular, when integrating in the increasing $u$ direction from data, we can drop the contribution of the integral on $\{u=-1\}$.

Integrating (\ref{eq:temptemptemp}) in $u$ thus gives
\begin{align*}
    |\wt{\lambda}|(u,v) &\lesssim A\epsilon \int\limits \frac{1}{|u'|^{2-2\kappa_1-aq_1-\delta}}v^{1-a}du' \\
    & \lesssim A\epsilon |u|^{\kappa_1+\delta}\Big( \frac{v}{|u|^{q_1}} \Big)^{1-a}.
\end{align*}
In particular, for $v_1$ sufficiently small we improve the bootstrap assumption. Note we also have concluded a bound consistent with $\wt{\lambda}$ vanishing as $v \rightarrow 0$. This bound will be required in the following.

For $\wt{\partial_v \phi}$ we turn to (\ref{eq:extPSSESF9}), and estimate the zeroth order coefficient $-\mathcal{T}_1$ to be
\begin{equation}
    \label{eqn:T_1estimate}
    -\mathcal{T}_1 = \frac{1}{|u|} + \frac{O(v_1+\epsilon)}{|u|}.
\end{equation}
This estimate uses that the exact value of $\mathcal{T}_1$ along $\{v=0\}$ is known (via the gauge condition $r = |u|$), and that bounds on $\partial_v r, \partial_v \nu$ imply deviations from the value along $\{v=0\}$ are of order $O(v_1+\epsilon)$. In particular, for all $v_1, \epsilon$ small this coefficient has an unfavorable sign, and needs to be dealt with by conjugating through with positive $|u|$ weights.

The remaining terms in (\ref{eq:extPSSESF9}) can be estimated via the bootstrap assumptions. Denoting the terms collectively by $\mathcal{C}_1$, we have
\begin{equation}
    \label{eq:temptemptem1}
   |\mathcal{C}_1| \lesssim A\epsilon |u|^{-2+\kappa_1+\delta} \Big(\frac{v}{|u|^{q_1}} \Big)^{1-a}.
\end{equation}
A priori, terms of the form $\partial_v \overline{\phi} - (\partial_v \overline{\phi})^{\circ}$ that vanish as $v\rightarrow 0$ at the slower rate $v^{1-b}$ could appear in the right hand side of (\ref{eq:extPSSESF9}). However, such terms always appear multiplied by factors that themselves vanish along $\{v=0\}$, and it follows that the worst cast bound is in fact (\ref{eq:temptemptem1}). 

Define the weight $w_3 \doteq |u|^{q_1-\delta}\Big(\frac{|u|^{q_1}}{v} \Big)^{1-a-\beta}$. Conjugating (\ref{eq:extPSSESF9}) through by $w_3$ yields the schematic equation
\begin{equation}
    \partial_u (w_3 \wt{\partial_v \phi}) + \frac{q_1(2-a-\beta)-\delta-1+O(v_1+\epsilon)}{|u|} (w_3 \wt{\partial_v \phi})= O\Big(A\epsilon |u|^{-1}\Big(\frac{v}{|u|^{q_1}} \Big)^{\beta} \Big).
\end{equation}
The assumption (\ref{eqn:a,brequirements}) on $a$ implies one can choose $\delta, \beta, v_1, \epsilon$ small enough to render the zeroth order coefficient positive. Moreover, $\beta > 0$ and so the error term has singular $|u|$ weight strictly greater than $1$. Integrating in $u$ as for $\wt{\lambda}$ thus concludes the bound. The data term does not vanish, but is bounded by the initial data norm $\mathcal{E}_{2,\alpha}$.
\end{proof}

To close this section, we argue that the bound on $\mathfrak{M}$ is sufficient to propagate regularity on the second order unknowns.
\begin{lem}
For $v_1$, $\epsilon$ sufficiently small we have the following bounds in $\mathcal{R}_{\text{I}}$:
\begin{align}
    |\partial_u \nu_p| &\lesssim |u|^{-1+\delta}, \ \ \ |\partial_v \lambda_p| \lesssim |u|^{-1+2\kappa_1+\delta}\Big(\frac{|u|^{q_1}}{v} \Big)^{a}, \label{eqn:region1higherorder1}\\[.5em]
    |\partial_u^2 \phi_p| &\lesssim |u|^{-2+\delta}, \ \ \ |\partial_v^2 \phi_p| \lesssim |u|^{-2+2\kappa_1+\delta}\Big(\frac{|u|^{q_1}}{v} \Big)^{b}.\label{eqn:region1higherorder2}
\end{align}
\end{lem}
\begin{proof}
To estimate $\partial_u \nu_p$ and $\partial_u^2 \phi_p$, it suffices to insert bounds following from (\ref{eq:extbootass1improved}) into (\ref{PSSESF:8}), (\ref{PSSESF:10}) and estimate, yielding
\begin{equation*}
    |\partial_v \partial_u \nu_p| \lesssim \epsilon |u|^{-2+\kappa_1+\delta}, \ \ \  |\partial_v \partial_u^2 \phi_p| \lesssim \epsilon |u|^{-3+\kappa_1+\delta}.
\end{equation*}
Integrating in $v$ from data and estimating $\frac{v}{|u|^{q_1}} \lesssim 1$ gives the stated estimate.

We now turn to $\partial_v \lambda_p$. The zeroth order term has a favorable sign, and so it remains to estimate the inhomogeneous terms, denoted collectively by $\mathcal{C}_2$. The most singular are those containing $\partial_v \overline{\lambda}$, which potentially blowup as $v \rightarrow 0$. We estimate
\begin{equation*}
    |\mathcal{C}_2| \lesssim \epsilon |u|^{-2+2\kappa_1 + \delta}\Big(\frac{|u|^{q_1}}{v} \Big)^a.
\end{equation*}
The coefficient of the $|u|$ weight is again strictly less than $-1$, and thus integrating in $u$ from $\{u=-1\}$ gives the desired bound.

Finally we estimate $\partial_v^2 \phi_p$ using (\ref{PSSESF:9}). The zeroth has a bad sign, and is estimated as in (\ref{eqn:T_1estimate}) above. Denote the remaining inhomogenous terms collectively by $\mathcal{C}_3$, and estimate 
\begin{equation*}
    |\mathcal{C}_3| \lesssim  \epsilon |u|^{-3+2\kappa_1+\delta}\Big(\frac{|u|^{q_1}}{v}\Big)^a,
\end{equation*}
Conjugating through by $|u|^{1+\sigma}$ and choosing $\sigma, v_1, \epsilon$ small enough gives rise to a favorable zeroth order term. The total singular $u$ weight on the error term is $-2+2\kappa_1 + \delta + \sigma + aq_1,$ which after choosing $\delta, \sigma$ smaller still, is strictly less than $-1$. Integrating the conjugated equation from $\{u=-1\}$ yields the desired estimate.
\end{proof}

\subsection{Change of gauge}
The remainder of the construction takes place away from $\{v=0\}$, and it is natural to transition to a gauge adapted to $\{u=0\}$. The coordinate change was first introduced in Section \ref{sec:assumptionsbackground}, and we recall the essential definitions here.

Define a double null coordinate system $(U,V)$ in $\mathcal{Q}^{(ex)}$ via 
\begin{equation}
    U = -|u|^{q_2},  \ \ \ V = v^{p_1}.
\end{equation}
Coordinate derivatives transform as 
\begin{equation}
    \frac{\partial}{\partial u} = q_2 |u|^{-\kappa_2}\frac{\partial}{\partial U} , \ \ \ \frac{\partial}{\partial v} = p_1 v^{p_1 \kappa_1}\frac{\partial}{\partial V},
\end{equation}
and we denote the gauge-dependent coordinate derivatives $\partial_U r, \ \partial_V r \ $ as $^{(U)}\nu, \ ^{(V)}\lambda$ respectively. The unknowns $r, \mu, \phi$ are gauge-independent, and are therefore unaffected by the change of coordinates.

With respect to the $(U,V)$ gauge, $\mathcal{R}_{\text{I}}$ takes the coordinate form
\begin{equation}
    \mathcal{R}_{\text{I}} = \{(U,V) \ | \ \ 0 \leq \frac{V^{q_2}}{|U|} \leq V_1, \ -1 \leq U < 0\},
\end{equation}
where $V_1 = v_1^{p_1 q_2}$.

The estimates for the solution in $\mathcal{R}_{\text{I}}$ imply corresponding bounds for all double null unknowns on spacelike surfaces of constant $\frac{V^{q_2}}{|U|}$ contained in $\mathcal{R}_{\text{I}}$. The following lemma records the estimates we shall need below; it is a simple consequence of the bounds in Proposition \ref{eq:extbootass1improved} and (\ref{eqn:region1higherorder1})-(\ref{eqn:region1higherorder2}), the estimates for the corrector functions in Lemma \ref{lem:correcestimates}, and the definition of the $(U,V)$ gauge.
\begin{lem}
\label{lem:region2dataest}
Along the surface $\{\frac{V^{q_2}}{|U|} = V_1\}$, the following bounds hold for the $\Psi_p$:
\begin{alignat}{2}
 |r_p| &\lesssim \epsilon V^{1+\delta}, \quad \quad \quad  \quad \quad \quad \quad  |\mu_p| &&\lesssim \epsilon V^\delta, \label{eqn:lem5.7.1}\\[.5em]
 |^{(U)}\nu_p| &\lesssim \epsilon V^{\kappa_2 + \delta}, \quad  \quad \quad \quad \quad |^{(V)}\lambda_p| &&\lesssim \epsilon V^{\delta}, \\[.5em]
 |{\partial_U}{^{(U)}\nu_p}| &\lesssim \epsilon V^{-1+2\kappa_2+\delta}, \ \  \ \quad |{\partial_V}{^{(V)}\lambda_p}| &&\lesssim \epsilon V^{-1+\delta}\\[.5em]
 |\partial_U \phi_p| &\lesssim \epsilon V^{-1+\kappa_2+\delta}, \ \ \ \ \ \quad \quad |\partial_V \phi_p| &&\lesssim \epsilon V^{-1+\delta}, \\[.5em]
 |\partial_U^2 \phi_p| &\lesssim \epsilon V^{-2+2\kappa_2+\delta}, \ \ \ \quad \quad  |\partial_V^2\phi_p| &&\lesssim \epsilon V^{-2+\delta}. \label{eqn:lem5.7.2}
\end{alignat}
\end{lem}

\subsection{Region II}
Let $V_2 \gg V_1$ denote a large constant, to be chosen below. In this section we propagate the solution to Region II, given by 
\begin{equation}
\mathcal{R}_{\text{II}} \doteq \{(U,V) \ | \  V_1 \leq \frac{V^{q_2}}{|U|}\leq V_2, -1 \leq U < 0\}.  
\end{equation}
Note that
\begin{equation}
  V_1^{p_2} |U|^{p_2} \leq V \leq V_2^{p_2} |U|^{p_2}
\end{equation}
holds in Region II. In the following we treat $V_1$ as a fixed (small) constant; however we retain the dependence of inequalities on $V_2$, anticipating that $V_2$ will only be chosen in Region IV. In particular, we have the freedom to interchange $U$ and $V$ weights (up to powers) in this region, which we use to express singular weights with respect to the $V$ coordinate.

The values of the double null unknowns on the past boundary of Region II,
\begin{equation}
    \{\frac{V^{q_2}}{|U|} = V_1, \ -1 \leq U < 0\} \cup \{U =-1, \ V_1 \leq V \leq V_2\},
\end{equation}
are induced by the solution constructed in $\mathcal{R}_{\text{I}},$ as well as by the choice of outgoing data along $\{U=-1\}$. This past boundary comprises the \enquote{initial data surface} for the solution in $\mathcal{R}_{\text{II}}.$ Estimates for the unknowns along the spacelike portion of the past boundary are given in Lemma \ref{lem:region2dataest}.

The analysis of this section will again take place in characteristic rectangles $\mathcal{Q}^{(ex)}_{\wt{U},\wt{V}}$. By an abuse of notation, in this section we implicitly restrict such rectangles to $\mathcal{R}_{\text{II}}$ (and similarly for Regions III-IV below).

Define the weight
\begin{equation}
    w^* \doteq \exp{\big(-D\frac{V^{q_2}}{|U|}\big)},
\end{equation}
which will be used to generate favorable lower order terms in the transport estimates. Here $D$ denotes a large constant depending on $V_2$. The constants are ordered such that for fixed $V_2$ and $D$, $\epsilon$ is chosen small enough such that 
\begin{equation}
    \epsilon^{\frac{1}{100}} \leq e^{-D V_2}.
\end{equation}
This choice of scales for $\epsilon$ and $D$ implies that $w^*$ is bounded above and below as
\begin{equation}
    \label{eq:ext2wbound}
  \epsilon^{\frac{1}{100}}  \leq  w^* \leq 1.
\end{equation}

\begin{rmk}
The choice of $\frac{1}{100}$ is immaterial, and any small enough constant would do. The need for a loss of only a small $\epsilon$ power is due to the bounds of this section, which lead to estimates $|w^* \Psi_p| \lesssim \epsilon$. As $V_2$ (and $D$) increases, without a corresponding lower bound on $w^*$ we would not deduce uniform bounds on $\Psi_p$. With the small power loss of $\epsilon$, we will be able to estimate $|\Psi_p| \lesssim \epsilon^{1-\frac{1}{100}}$ uniformly in $V_2$. This uniformity will prove valuable in Regions III-IV, where $V_2^{-1} $ is used as a smallness parameter.
\end{rmk}

\begin{figure}
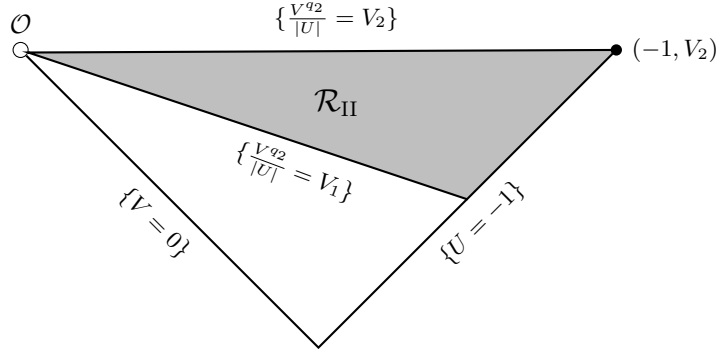

    \centering
  \includestandalone[]{Figures/fig_region2}
  \caption{$\mathcal{R}_{\text{II}}$}
  \label{fig:region2}
\end{figure}

\subsubsection{Norms and bootstrap assumptions}
For fixed $(\wt{U},\wt{V}) \in \mathcal{R}_{\text{II}}$, define the solution norms
\begin{align}
\mathfrak{T}_1 \doteq &\sum_{i=0}^1\sup_{\mathcal{Q}^{(ex)}_{(\wt{U},\wt{V})}} \Big|\frac{1}{V^{\delta-i}} w^* {\partial_V^i}  ^{(V)}\lambda_p \Big| + \sum_{i=1}^2 \sup_{\mathcal{Q}^{(ex)}_{(\wt{U},\wt{V})}}\Big|\frac{1}{V^{\delta-i}}  w^* \partial_V^i \phi_p \Big| ,\\[.4em]
    \mathfrak{T}_2 \doteq &\sup_{\mathcal{Q}^{(ex)}_{(\wt{U},\wt{V})}} \Big|\frac{1}{V^{1+\delta}}w^*r_p \Big| + \sum_{i=0}^1\sup_{\mathcal{Q}^{(ex)}_{(\wt{U},\wt{V})}} \Big|\frac{1}{V^{\kappa_2+\delta-q_2i}} w^* {\partial_U^i}  ^{(U)}\nu_p \Big| \\
    &+ \sup_{\mathcal{Q}^{(ex)}_{(\wt{U},\wt{V})}} \Big|\frac{1}{V^{\delta}} w^* \mu_p \Big| + \sum_{i=1}^2 \sup_{\mathcal{Q}^{(ex)}_{(\wt{U},\wt{V})}} \Big| \frac{1}{V^{\delta-q_2i}} w^* \partial_U^i \phi_p \Big|, 
\end{align}
as well as the total spacetime norm
\begin{equation}
    \mathfrak{T} \doteq \mathfrak{T}_1 + \mathfrak{T}_2.
\end{equation}
Note that we work directly with the $\Psi_p$ in this section, and a renormalization as in $\mathcal{R}_{\text{I}}$ is not necessary. The main result of this section is the following, ensuring uniform estimates on the norm $\mathfrak{T}$. Importantly, these estimates are independent of $V_2$.
\begin{prop}
Assume a local solution is given in a rectangle $\mathcal{Q}_{\wt{U},\wt{V}},$ and that the bound 
\begin{equation}
    \label{eq:extbootass2}
    \mathfrak{T} \leq 2A\epsilon
\end{equation}
holds. Then for $V_2$ arbitrarily large (but fixed), $D$ sufficiently large as a function of $V_2,$ and $\epsilon$ sufficiently small as a function of the background solution and $D$, we have the improved bound 
\begin{equation}
    \label{eq:extbootass2improved}
    \mathfrak{T} \leq A\epsilon.
\end{equation}
\end{prop}

\subsubsection{Recovering bootstrap assumptions}
Analogously to Lemma \ref{lem:extregion1leadingorderbehavior}, the following is a direct consequence of the bootstrap assumptions:
\begin{lem}
\label{lem:extregion2prelims}
Assume the bound (\ref{eq:extbootass2}). Then for $\epsilon$ sufficiently small the geometry satisfies
\begin{equation}
     r \sim V, \ \ \ (-^{(U)}\nu) \sim V^{\kappa_2}, \ \ \  ^{(V)}\lambda \sim 1, \ \ \ (1-\mu) \sim 1.
\end{equation}
Moreover, we have the bounds
\begin{alignat}{2}
    |\partial_U \phi| &\lesssim V^{-1+\kappa_2}, \ \ \ \ \quad \quad \quad  |\partial_V \phi| &&\lesssim V^{-1}, \\[.5em]
    |{\partial_U}{^{(U)}\nu}| &\lesssim V_2^{c}V^{-1+2\kappa_2}, \ \  \ \ \ |{\partial_V}{^{(V)}\lambda}| &&\lesssim V^{-1}, \\[.5em]
    |\partial_U^2 \phi| &\lesssim V_2^{d}V^{-2+2\kappa_2}, \ \ \ \quad \ \   |\partial_V^2 \phi| &&\lesssim V^{-2}.
\end{alignat}
\end{lem}

We turn to estimating the bootstrap norm $\mathfrak{T}$, and comment on key features of the analysis in this region: 
\vspace{-.5em}
\begin{enumerate}
    \item Given $r \sim V$, the singular weights associated to $\mathcal{O}$ will naturally appear as $V$ weights. It is a straightforward to integrate $U$ equations directly, and use an upper bound on $\frac{|U|}{V^{q_2}}$ which is, importantly, independent of $V_2$. When integrating $V$ equations it is natural to first replace the singular $V$ weights with $U$ weights (at no cost of factors of $V_2$), integrate in $U$, and then use bounds on $\frac{V^{q_2}}{|U|}$. One does incur unfavorable dependence on $V_2$ here, but choosing $D$ large as a function of $V_2$ compensates. Therefore we do not need to deal with any logarithmic divergences.
    \vspace{-.5em}
    \item Compute 
    \vspace{-.25em}
    $$\partial_U w^* = -\frac{D }{|U|}\frac{V^{q_2}}{|U|}w^*, \ \ \ \partial_V w^* = -\frac{Dq_2}{V}\frac{V^{q_2}}{|U|}w^*. $$ 
    Conjugation by $w^*$ thus produces good zeroth order terms with large ($\approx D$) coefficients. The smallness mechanism in Region II should be compared with the interior, where zeroth order terms with large ($\approx \alpha$) coefficients were utilized.
\end{enumerate}

\begin{lem}
Assume the bound (\ref{eq:extbootass2}). Then for $D$ large enough, $\epsilon$ sufficiently small, we have 
\begin{equation}
    \label{ext:lambbootimproved}
    \mathfrak{T}_1 \leq \frac{1}{10}A\epsilon.
\end{equation}
\end{lem}
\begin{proof}
The technique being the same for unknowns satisfying $U$ equations, we only show the estimate for $^{(V)}\lambda_p$ in detail. We start by estimating the terms appearing in (\ref{PSSESF:1}). Inserting the bootstrap assumptions shows 
\begin{equation*}
    |w^* \partial_U \, ^{(V)} \lambda_p| \lesssim A \epsilon V^{-1+\kappa_2+\delta}.
\end{equation*}
We have included a factor of $w^*$ in the above estimate, anticipating that we eventually conjugate by $w^*$. In estimating the right hand side of (\ref{PSSESF:1}), we write $\Psi_p = (w^*)^{-1}(w^*\Psi_p)$ for each unknown that appears, using the bootstrap assumptions to estimate $(w^*\Psi_p)$ and the bound $(w^*)^{-1} \lesssim \epsilon^{-\frac{1}{100}}$ to estimate factors of $w^*$. It is important already here that we have a lower bound on $w^*$, which loses only a small power of $\epsilon$.

Conjugating (\ref{PSSESF:1}) through by $V^{-\delta}w^*$ gives the equation 
\begin{equation}
    \label{eq:ext2lambdatemp}
    \partial_U (V^{-\delta} {w^*} \  ^{(V)}\lambda_p) = \frac{\big(-D\frac{V^{q_2}}{|U|}\big)}{|U|}(V^{-\delta} {w^*} \ ^{(V)}\lambda_p) + \mathcal{A},
\end{equation}
where $ |\mathcal{A}| \lesssim A\epsilon  \frac{1}{V^{q_2}}.$
Using the lower bound $\frac{V^{q_2}}{|U|} \geq V_1$, we estimate that the coefficient of the zeroth term is bounded above by $ -\frac{DV_1}{|U|} $. Note $V_1$ is fixed independently of $\epsilon$.

It remains to integrate (\ref{eq:ext2lambdatemp}) in $U$ from the boundary, the spacelike and null component of which we consider separately. First take $(U,V) \in \mathcal{R}_{\text{II}}$ with $V \leq V_1$. Contracting (\ref{eq:ext2lambdatemp}) with $V^{-\delta}{w^*} ^{(V)}\lambda_p$ and integrating gives
\begin{align*}
    |V^{-\delta}{w^*} ^{(V)}\lambda_p&|^2(U,V) +  DV_1 \int\limits_{(-V^{q_2}V_1^{-1},V)}^{(U,V)}\frac{1}{|U'|}(V^{-\delta}{w^*} ^{(V)}\lambda_p)^2(U',V)dU'  \nonumber\\
    &\leq |V^{-\delta}{w^*} ^{(V)}\lambda_p|^2(-V^{q_2}V_1^{-1},V) +\frac{DV_1}{2} \int\limits_{(-V^{q_2}V_1^{-1},V)}^{(U,V)}\frac{1}{|U'|}(V^{-\delta}{w^*} ^{(V)}\lambda_p)^2(U',V)dU' \\
    &+ \overline{C}D^{-1}V_1^{-1}A^2\epsilon^2  \frac{1}{V^{2q_2}}\int\limits_{(-V^{q_2}V_1^{-1},V)}^{(U,V)}|U'|dU' \nonumber\\
    &\leq\epsilon^2 +\frac{DV_1}{2} \int\limits_{(-V^{q_2}V_1^{-1},V)}^{(U,V)}\frac{1}{|U'|}(V^{-\delta}{w^*} ^{(V)}\lambda_p)^2(U',V)dU'+ \overline{C}D^{-1}V_1^{-1}A^2\epsilon^2. \label{eq:ext2lambdaimprovetemp}
\end{align*}
Absorbing the remaining integral expression, taking square roots, and choosing $D$ large enough as a function of the background solution (through the constant $\overline{C}$) and $V_1$ gives the improvement.

Next consider $(U,V)$ with $V > V_1$. The past directed constant $V$ line emanating from $(U,V)$ intersects the boundary at $(-1,V)$. The analogous estimate becomes
\begin{align*}
    |V^{-\delta}{w^*} ^{(V)}\lambda_p|^2(U,V) \leq |V^{-\delta}{w^*} ^{(V)}\lambda_p|^2(-1,V) + \overline{C}D^{-1}V_1^{-1}A^2\epsilon^2 \frac{1}{V^{2q_2}}\int\limits_{(-1,V)}^{(U,V)}|U'|dU'.
\end{align*}
The $V$ weights are uniformly bounded $V^{-1} \leq V_1^{-1}$, and the integral is bounded by a constant. The boundary term may simply be estimated $|^{(V)}\lambda_p|(-1,V) \lesssim \epsilon,$ and after choosing $D$ large we again improve the bootstrap assumption.

The case of ${\partial_V}{^{(V)}\lambda}$ will follow similarly once consistent decay bounds for the terms in (\ref{PSSESF:7}) are shown. Inserting bootstrap bounds gives
\begin{equation*}
    |{w^* \partial_U {\partial_V}}{^{(V)}\lambda_p}| \lesssim A\epsilon V^{-2+\kappa_2+\delta}.
\end{equation*}
Conjugating through by $V^{-1+\delta}w^*$, contracting, and integrating in $U$, one proceeds as for $^{(V)}\lambda_p$ to conclude the estimate.

Estimating the terms in (\ref{PSSESF:7}), (\ref{PSSESF:9}) implies 
\begin{equation*}
    |w^* \partial_U \partial_V \phi_p| \lesssim A\epsilon V^{-2+\kappa_2+\delta},
\end{equation*}
\begin{equation*}
    |w^* \partial_U \partial_V^2 \phi_p| \lesssim A\epsilon V^{-3+\kappa_2+\delta}.
\end{equation*}
Conjugating by the bootstrap weight and integrating in $U$ leads to the desired estimate for $\partial_V \phi$, $\partial_V^2 \phi$ after choosing $D$ sufficiently large. Boundary terms along the spacelike boundary are estimated using Lemma \ref{lem:region2dataest}. Boundary terms along $\{U=-1\}$ are automatically bounded by the choice of initial data.

\end{proof}

It remains to consider the unknowns satisfying $V$ equations. Some care will need to be taken in exchanging $V$ for $U$ weights such that all singular weights remain integrable. 

\begin{lem}
Assume the bound (\ref{eq:extbootass2}). Then for $D$ large enough, $\epsilon$ sufficiently small, we have 
\begin{equation}
    \mathfrak{T}_2 \leq \frac{1}{10}A\epsilon.
\end{equation}
\end{lem}
\begin{proof}
Inserting bootstrap bounds into (\ref{PSSESF:1}) implies 
\begin{equation*}
    |{w^*\partial_V}{^{(U)}\nu_p}| \lesssim A\epsilon V^{-1+\kappa_2 + \delta} \lesssim A\epsilon |U|^{-p_2}V^{\kappa_2+\delta}.
\end{equation*}
In the final inequality we have exchanged part of the singular $V$ weight for a singular $U$ weight\footnote{This step serves to avoid logarithmic divergences in the $V$ integrals, but is not strictly necessary. One can proceed as in the interior region (cf. the proof of Lemma \ref{lem:inttheonlylemmareferredto}), namely by conjugating through by an extra $V^\sigma$ weight in order to render all error terms integrable. }. The implied constant only depends on $V_1$, which is assumed to be fixed.

Conjugating (\ref{PSSESF:1}) by $V^{-\kappa_2-\delta}w^*$ yields
\begin{equation*}
    \partial_V (V^{-\kappa_2-\delta}{w^*} ^{(U)}\nu_p) = \frac{\Big(-Dq_2\frac{V^{q_2}}{|U|}\Big)}{V}(V^{-\kappa_2-\delta}{w^*} ^{(U)}\nu_p) + \mathcal{A},
\end{equation*}
where $|\mathcal{A}| \lesssim A\epsilon \frac{1}{|U|^{p_2}}$. Contracting with $V^{-\kappa_2-\delta}{w^*} ^{(U)}\nu_p $ and integrating in $V$ from the spacelike portion of the past boundary of $\mathcal{R}_{\text{II}}$ gives
\begin{align*}
   |V^{-\kappa_2-\delta}{w^*} ^{(U)}\nu_p&|^2(U,V)+ Dq_2V_1\int\limits_{(U,(V_1|U|)^{p_2})}^{(U,V)}\frac{1}{|V'|} (|V'|^{-\kappa_2-\delta}{w^*} ^{(U)}\nu_p)^2(U,V')dV' \\
   &\leq  |V^{-\kappa_2-\delta}{w^*} ^{(U)}\nu_p|^2(U,(V_1|U|)^{p_2}) +\frac{Dq_2V_1}{2}\int\limits_{(U,(V_1|U|)^{p_2})}^{(U,V)}\frac{1}{|V'|} (|V'|^{-\kappa_2-\delta}{w^*} ^{(U)}\nu_p)^2(U,V')dV' \\
   &+ \overline{C}D^{-1}V_1^{-1}A^2\epsilon^2 |U|^{-2p_2} \int\limits_{(U,(V_1|U|)^{p_2})}^{(U,V)}|V'| dV' \\
   &\leq \epsilon^2  +\frac{Dq_2V_1}{2}\int\limits_{(U,(V_1|U|)^{p_2})}^{(U,V)}\frac{1}{|V'|} (|V'|^{-\kappa_2-\delta}{w^*} ^{(U)}\nu_p)^2(U,V')dV' + \overline{C}D^{-1}V_1^{-1}V_2^{2p_2}A^2\epsilon^2.
\end{align*}
After absorbing the integral term, the remaining error term is of size $\approx \overline{C} D^{-1}V_1^{-1}V_2^{2p_2}A^2\epsilon^2$. Choosing $D$ large enough as a function of both implied constants and of $V_2$, we improve the bootstrap assumption.

The remaining estimates follow similarly. Given an unknown $\Psi_p$ satisfying a $V$ equation, it is sufficient to show an estimate $|w^* \partial_V \Psi_p| \lesssim A\epsilon |U|^{-p_2} V^{r(\Psi)},$ where $V^{r(\Psi)}$ is the $V$ weight contained in the bootstrap norm. The boundary terms along $\{\frac{V^{q_2}}{|U|} = V_1\}$ resulting from integration have already been bounded in a manner consistent with the bootstrap norms, so an estimate of the above form allows us to run an identical argument to that of $^{(U)}\nu_p$.

For (\ref{PSSESF:8}), inserting bootstrap norms and estimates on the background shows
\begin{equation*}
    |{w^*\partial_V \partial_U}{^{(U)}\nu_p}| \lesssim A\epsilon V_2^{c} V^{-2+2\kappa_2+\delta} \lesssim A\epsilon V_2^{c} |U|^{-p_2} V^{-1+2\kappa_2+\delta}.
\end{equation*}
This estimate has the desired structure, up to added factors of $V_2$. These added factors can always be absorbed by the smallness of $D^{-1}$, and thus do not affect the argument.

We turn now to $\mu_p$, and estimate using (\ref{PSSESF:4})
\begin{equation*}
    |{w^*} \partial_V \mu_p| \lesssim A\epsilon V^{-1+\delta} \lesssim A\epsilon |U|^{-p_2} V^{\delta},
\end{equation*}
which has the desired structure.

Next we consider $\partial_U \phi, \partial_U^2 \phi$. Estimate using (\ref{PSSESF:5}), (\ref{PSSESF:10})
\begin{align}
    |w^* \partial_V \partial_U \phi_p| &\lesssim A\epsilon V^{-2+\kappa_2+\delta} \lesssim A\epsilon |U|^{-p_2} V^{-1+\kappa_2+\delta}, \\[.5em]
    |w^* \partial_V \partial_U^2 \phi_p| &\lesssim A\epsilon V_2^{d}V^{-3+2\kappa_2+\delta} \lesssim A\epsilon V_2^{d}|U|^{-p_2}V^{-2+2\kappa_2+\delta}.
\end{align}
Both estimates have the proper structure, and upon integrating in $V$ we conclude the desired improvements.

Finally, $r_p$ is estimated by the fundamental theorem of calculus and the improved bound on $^{(V)}\lambda_p$ as in Region I.
\end{proof}

\subsection{Region III}
\label{subsec:regionIII}
The analysis thus far has extended the solution up to a spacelike hypersurface $\{\frac{V^{q_2}}{|U|} = V_2\}$, with restrictions on $\epsilon$ ensuring the estimates are independent of $V_2$. Unpacking the bootstrap norms of the previous section gives the following:
\begin{lem}
\label{lem:region3databounds}
Assume the bound (\ref{eq:extbootass2improved}). Then there exists $s \ll 1$ small such that for $\epsilon$ small as a function of $V_2$, along $\{\frac{V^{q_2}}{|U|} = V_2\}$ we have that the bounds (\ref{eqn:lem5.7.1})-(\ref{eqn:lem5.7.2}) continue to hold, after replacing $\epsilon \rightarrow \epsilon^{1-s}.$
\end{lem}

In solving up to $\{U=0\}$, the analysis will parallel that of Region I. There it was convenient to have explicit control on the gauge along $\{v=0\}$, with the condition $\nu(u,0) = -1$. In order to establish similar control along $\{U=0\}$, define the modified coordinate $\hat{V}$ by
\begin{equation*}
    \hat{V} \doteq {\int\limits_{0}^V} {^{(V)}}\overline{\lambda}(0,V')dV' \doteq t(V) V,
\end{equation*}
where
\begin{equation*}
    t(V) \doteq \frac{1}{V}{\int\limits_{0}^V} {^{(V)}}\overline{\lambda}(0,V')dV'.
\end{equation*}
It follows by the assumptions on the geometry of the background spacetime that $t(V) \sim 1$, and one can check that $^{(\hat{V})}\overline{\lambda}(0,V) \doteq \partial_{\hat{V}}\overline{r}(0,V) = 1.$ In effect we have straightened out the $V$ coordinate along $\{U=0\}$ so that $\overline{r}(0,\hat{V}) = \hat{V}$ holds.

The $V$ coordinate derivative is transformed to
\begin{equation*}
   \frac{\partial}{\partial \hat{V}} = \frac{1}{^{(V)}\overline{\lambda}(0,V)} \frac{\partial}{\partial V}.
\end{equation*}
In this section it is convenient to work both in the $(U,V)$ and $(U, \hat{V})$ coordinate systems. Note $V \sim \hat{V}$, and it thus makes no difference which outgoing coordinate is used for displaying singular/decaying weights. The ability to switch into a $\hat{V}$ coordinate will be useful when considering transport equations in the outgoing direction.

We now discuss the extension to Region III, given by
\begin{equation}
    \mathcal{R}_{\text{III}} \doteq \{(U,V) \ | \ V_2 \leq \frac{V^{q_2}}{|U|} < \infty, \ -1 \leq U < 0, \ 0 < V \leq 1  \},
\end{equation}
containing a portion of the future light cone $\{U = 0\}$ of the singular point. To close the estimates of this section, we will only have access to the smallness parameter $V_2^{-1}$. For $V_2$ large enough the past boundary of Region III consists of a spacelike piece
\begin{equation}
    \{\frac{V^{q_2}}{|U|}=V_2, \ -1 \leq U < 0, \ 0 < V \leq 1\},
\end{equation}
with the value of all double null unknowns provided by the solution constructed in Region II. 

\begin{figure}
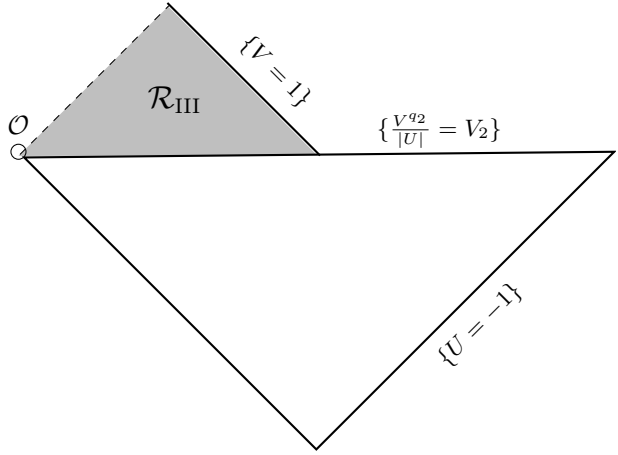

    \centering
  \includestandalone[]{Figures/fig_region3}
  \caption{$\mathcal{R}_{\text{III}}$ }
  \label{fig:region3}
\end{figure}

\subsubsection{Norms and bootstrap assumptions}
Fix $(\wt{U},\wt{V}) \in \mathcal{R}_{\text{III}}$ and define the norms
\begin{align}
    \mathfrak{S}_1 \doteq &\sup_{\mathcal{Q}^{(ex)}_{(\wt{U},\wt{V})}} \Big|\frac{1}{V^{1+\delta}}r_p \Big| + \sum_{i=0}^1 \sup_{\mathcal{Q}^{(ex)}_{(\wt{U},\wt{V})}}\Big|\frac{1}{V^{\delta-i}}{\partial_{V}^i} ^{(V)}\lambda_p\Big| + \sup_{\mathcal{Q}^{(ex)}_{(\wt{U},\wt{V})}} \Big|\frac{1}{V^{\delta}}\mu_p \Big| + \sum_{i=1}^2 \sup_{\mathcal{Q}^{(ex)}_{(\wt{U},\wt{V})}} \Big|V^{i-\delta}\partial_{V}^i \phi_p \Big|, \\[1em]
    \mathfrak{S}_2 \doteq &\sup_{\mathcal{Q}^{(ex)}_{(\wt{U},\wt{V})}} \Big|{\frac{1}{V^{\kappa_2+\delta}}} {^{(U)}\nu_p} \Big| +  \sup_{\mathcal{Q}^{(ex)}_{(\wt{U},\wt{V})}} \Big|V^{1-\kappa_2-\delta}   \partial_U \phi_p \Big|,
\end{align}
as well as the total solution norm
\begin{equation}
    \mathfrak{S} \doteq \mathfrak{S}_{1} + \mathfrak{S}_{2}.
\end{equation}
Note do not include explicit bounds on ingoing derivatives of $^{(U)}\nu_p, \ \partial_U \phi_p$ in the bootstrap norm; estimates for these quantities will follow after having closed the main bootstrap. Existence in Region III is a consequence of the following proposition, showing that bounds on $\mathfrak{S}$ can be improved independently of $(\wt{U},\wt{V}) \in \mathcal{R}_{\text{III}}$. 
\begin{prop}
Assume a local solution is given in a rectangle $\mathcal{Q}^{(ex)}_{(\wt{U},\wt{V})},$ and that the bound 
\begin{equation}
    \label{eq:extbootass3}
    \mathfrak{S} \leq 2A\epsilon^{1-s}
\end{equation}
holds, where $s$ is the parameter appearing in Lemma \ref{lem:region3databounds}. Then there exists $V_2$ sufficiently large depending on the background solution and $\epsilon$ sufficiently small as a function of data such that we have the improved bound 
\begin{equation}
    \label{eq:extbootass3improved}
    \mathfrak{S} \leq A\epsilon^{1-s}.
\end{equation}
\end{prop}

\subsubsection{Recovering bootstrap assumptions}
The following lemma is a direct consequence of the bootstrap assumptions and the definition of the $(U,\hat{V})$ coordinate system.
\begin{lem}
\label{lem:extbootconsequences3}
Assume the bound (\ref{eq:extbootass3}). Then for $\epsilon$ sufficiently small the geometry satisfies
\begin{equation}
    r \sim V, \ \ \  (-^{(U)}\nu) \sim V^{\kappa_2}, \ \ \ ^{(V)}\lambda \sim 1, \ \ \ (1-\mu) \sim 1.
\end{equation}
Moreover, we have the self-similar bounds
\begin{alignat}{2}
|\partial_U \phi| &\lesssim V^{-1+\kappa_2}  , \ \ \ |\partial_{V} \phi| &&\lesssim V^{-1}, \\[.5em]
|{\partial_V} ^{(V)}\lambda| &\lesssim V^{-1}, \ \ \ \quad \  |\partial_{V}^2 \phi| &&\lesssim V^{-2}.
\end{alignat}
Finally, $^{(\hat{V})} \lambda, {\partial_{\hat{V}}}{^{(\hat{V})} \lambda}, \partial_{\hat{V}} \phi, \partial_{\hat{V}}^2 \phi$ satisfy analogous bounds, with $\hat{V}$ in place of $V$.
\end{lem}

We turn now to estimating the $\Psi_p$, starting with unknowns satisfying $U$ equations.

\begin{lem}
Assume the bound (\ref{eq:extbootass3}). Then for $V_2$ sufficiently large and $\epsilon$ sufficiently small we have 
\begin{equation}
    \mathfrak{S}_{1} \leq \frac{1}{10}A\epsilon^{1-s}.
\end{equation}
\end{lem}
\begin{proof}
Applying (\ref{eq:extbootass3}) and estimating (\ref{PSSESF:1}) gives
\begin{equation*}
    |{\partial_U}{^{(V)}\lambda_p}| \lesssim A\epsilon^{1-s} V^{-1+\kappa_2+\delta}.
\end{equation*}
Conjugating by $V^{-\delta}$ and integrating in $U$ from a point $(U,V)$ to the past boundary point $(-V^{q_2}V_2^{-1},V)$ yields
\begin{align*}
    |{V^{-\delta}}  \ ^{(V)}\lambda_p|(U,V) \lesssim |{V^{-\delta}} \ ^{(V)}\lambda_p|(-V^{q_2}V_2^{-1},V) + A\epsilon^{1-s}V_2^{-1}.
\end{align*}
The boundary term is of size $\epsilon$ by the results of Lemma \ref{lem:region3databounds}, and thus choosing $V_2$ sufficiently large improves the bootstrap assumption.

It follows that for a given $\Psi_p$ satisfying a $U$ equation, it suffices to establish a bound $|\partial_u \Psi_p| \lesssim A\epsilon^{1-s} V^{-q_2 + r(\Psi)},$ where $r(\Psi)$ is the $V$ weight appearing in the bootstrap norm. 

We next estimate (\ref{PSSESF:7}), giving
\begin{equation*}
    |{\partial_U \partial_V} {^{(V)}\lambda_p}| \lesssim V^{-2+\kappa_2+\delta},
\end{equation*}
which has the desired structure.

Considering (\ref{PSSESF:3}), (\ref{PSSESF:4}), and (\ref{PSSESF:9}) shows 
\begin{equation*}
    |\partial_U \mu_p| \lesssim A\epsilon^{1-s}V^{-1+\kappa_2+\delta},
\end{equation*}
\begin{equation*}
    |\partial_U \partial_V \phi_p| \lesssim A\epsilon^{1-s}V^{-2+\kappa_2+\delta},
\end{equation*}
\begin{equation*}
    |\partial_U \partial_V^2 \phi_p| \lesssim A\epsilon^{1-s} V^{-3+\kappa_2+\delta},
\end{equation*}
all of which have the required structure. Finally, integrating $\partial_V {r_p =} {^{(V)}\lambda_p}$ and using bounds on $^{(V)}\lambda_p$ yields the desired improvement for $r_p$.
\end{proof}
It remains to consider $^{(U)}\nu_p, \partial_U \phi_p$, which satisfy $V$ equations. 

\begin{lem}
Assume the bound (\ref{eq:extbootass3}). Then for $V_2$ sufficiently large and $\epsilon$ sufficiently small we have 
\begin{equation}
    \mathfrak{S}_{2} \leq \frac{1}{10}A\epsilon^{1-s}.
\end{equation}
Moreover, there exists a constant $0 < \sigma \ll 1$ for which we have the following higher order estimates:
\begin{equation}
    \label{eqn:region3higherorder}
    |{\partial_U} {^{(U)}\nu_p}| \lesssim \epsilon V^{-1+2\kappa_2+\delta}\Big(\frac{V^{q_2}}{|U|}\Big)^{1-\sigma}, \ \ \
    |\partial_U^2\phi_p| \lesssim \epsilon V^{-2+2\kappa_2+\delta}\Big(\frac{V^{q_2}}{|U|}\Big)^{1-\sigma}.
\end{equation}
\end{lem}
\begin{proof}
In this lemma we work in the $(U, \hat{V})$ gauge. We therefore must translate the bounds obtained for $^{(V)}\lambda_p, \partial_V \phi_p$, and derivatives thereof, into corresponding bounds in the new coordinates.

The coordinate change $V \rightarrow \hat{V}$ is explicit in terms of $\overline{\lambda}(0,v)$, and we compute
\begin{equation*}
    |^{(\hat{V})}\lambda_p| \lesssim |\overline{\lambda}(0,v)|^{-1} |^{(V)}\lambda_p | \lesssim A\epsilon^{1-s}\hat{V}^{\delta}.
\end{equation*}
Similarly,
\begin{align*}
   |{\partial_{\hat{V}}}^{(\hat{V})}\lambda_p| &\lesssim A\epsilon^{1-s} \hat{V}^{-1+\delta}, \\[.5em]
   |{\partial_{\hat{V}}} \phi| &\lesssim A\epsilon^{1-s} \hat{V}^{-1+\delta},  \\[.5em]
   |{\partial^2_{\hat{V}}} \phi| &\lesssim A\epsilon^{1-s} \hat{V}^{-2+\delta}.
\end{align*}
In translating between coordinate derivatives we have used the self-similar bound ${\partial_V}{^{(V)}\overline{\lambda}}(0,V) \lesssim V^{-1}$. Unless otherwise specified, all computations are now done in $(U, \hat{V})$ coordinates. 

To estimate $^{(U)}\nu_p$ we view (\ref{PSSESF:2}) as an equation for $\frac{^{(U)}\nu_p}{^{(U)}\overline{\nu}}$, giving 
\begin{equation}
\label{eq:ext3tempest0}
    \partial_{\hat{V}} \log \Big(1 + \frac{^{(U)}\nu_p}{^{(U)}\overline{\nu}}\Big) = (\mathcal{G}_1)_p.
\end{equation}
This form emphasizes that $^{(U)}\nu_p$ may be estimated in terms of $\mathcal{G}_1$, which contains only quantities that have been estimated previously. Estimate
\begin{equation}
    \label{eq:ext3tempest1}
    |(\mathcal{G}_1)_p| \lesssim \epsilon^{1-s}\hat{V}^{-1+\delta}.
\end{equation}
Consider $(U,\hat{V}) \in \mathcal{R}_{\text{III}}$, and denote by $(U,\hat{V}_*)$ the intersection point of the constant $U$ line through $(U,\hat{V})$ with the past boundary of $\mathcal{R}_{\text{III}}$. Explicitly, $\hat{V}_* = t((|U|V_2)^{p_2})(|U|V_2)^{p_2}$. Integrating (\ref{PSSESF:2}) and inserting (\ref{eq:ext3tempest1}) gives
\begin{align*}
    \Big| \log \Big(1+ \big(\frac{^{(U)}\nu_p}{^{(U)}\overline{\nu}}\big)(U,\hat{V})\Big)\Big| &\lesssim \log \Big| \Big(1+ \big(\frac{^{(U)}\nu_p}{^{(U)}\overline{\nu}}\big)(U,\hat{V}_*)\Big)\Big| + \epsilon^{1-s}\hat{V}^{\delta} \\
    &\lesssim \epsilon^{1-s} \hat{V}^{\delta}.
\end{align*}
For $\epsilon$ sufficiently small we conclude
\begin{equation*}
    |^{(U)}\nu_p| \lesssim \epsilon^{1-s} \hat{V}^{\delta}|^{(U)}\overline{\nu}| \lesssim \epsilon^{1-s} \hat{V}^{\kappa_2+\delta}.
\end{equation*}
We next consider $\partial_U \phi_p$. From (\ref{PSSESF:5}) it follows that the zeroth order term has a good sign, and may be estimated to give
\begin{equation*}
    -\mathcal{T}_1 = -\frac{^{(\hat{V})}\lambda}{r} = -\frac{1 + O(V_2^{-1} + \epsilon)}{\hat{V}}.
\end{equation*}
Here we have used the gauge condition $^{(\hat{V})}\overline{\lambda}(0,\hat{V}) = 1$. It follows that $\hat{V}^{1-\kappa_2}\partial_U \phi_p$ satisfies an equation of the form
\begin{equation*}
    \partial_{\hat{V}}(\hat{V}^{1-\kappa_2}\partial_U \phi_p) = -\frac{\kappa_2 + O(V_2^{-1} + \epsilon)}{\hat{V}} (\hat{V}^{1-\kappa_2}\partial_U \phi_p) + \mathcal{A},
\end{equation*}
where $|\mathcal{A}| \lesssim \epsilon^{1-s}\hat{V}^{-1+\delta},$ and consists of terms that have already been estimated. Integrating in $\hat{V}$ and dividing by $\hat{V}^{\delta}$ gives the desired estimate.

Having now estimated all first order quantities (and thus closed the bootstrap argument), we are able to turn to ${\partial_U}{^{(U)}\nu_p}$. Commuting (\ref{eq:ext3tempest0}) with $\partial_U$ gives
\begin{equation*}
    \partial_{\hat{V}}\partial_U  \log \Big(1+ \frac{^{(U)}\nu_p}{^{(U)}\overline{\nu}}\Big) = \partial_U (\mathcal{G}_1)_p,
\end{equation*}
where $|\partial_U (\mathcal{G}_1)_p| \lesssim \epsilon^{1-s}\hat{V}^{-2+\kappa_2+\delta}$. This estimate is not integrable in $\hat{V}$, and will in fact lead to (badly) singular behavior in $U$. More precisely, we arrive at an estimate of the form
\begin{align*}
     |\partial_U  \log \Big(1+ \big(\frac{^{(U)}\nu_p}{^{(U)}\overline{\nu}}\big)(U,\hat{V})\Big)| &\lesssim  |\partial_U  \log \Big(1+ \big(\frac{^{(U)}\nu_p}{^{(U)}\overline{\nu}}\big)(U,\hat{V}_*)\Big)| + \epsilon^{1-s}\hat{V}_*^{-1+\kappa_2+\delta} \\[.5em]
     &\lesssim \epsilon^{1-s}|U|^{-1+p_2\delta}.
\end{align*}
Expanding the derivative and re-arranging gives 
$|{\partial_U}{ ^{(U)} \nu_p}| \lesssim \epsilon^{1-s}\hat{V}^{\kappa_2} |U|^{-1+p_2\delta},$ after choosing $\delta$ sufficiently small as a function of the parameter $c$ controlling the blowup rate of ${\partial_U} {^{(U)}\nu}$.

Finally, equation (\ref{PSSESF:10}) gives
\begin{equation*}
    \partial_{\hat{V}} \partial_U^2 \phi_p = -\frac{1+O(V_2^{-1}+\epsilon)}{\hat{V}}\partial_U^2 \phi_p  + O(\epsilon^{1-s}\hat{V}^{-2+\kappa_2}|U|^{-1+p_2\delta}).
\end{equation*}
Conjugate by $\hat{V}^{1-\frac{1}{2}\kappa_2},$ preserving the favorable sign of the zeroth order coefficient and rendering the error term integrable in $\hat{V}$. Integrating in $\hat{V}$ from data thus yields
\begin{align*}
    |\hat{V}^{1-\frac{1}{2}\kappa_2}\partial_U^2 \phi_p|(U,\hat{V}) &\lesssim \hat{V}_*^{1-\frac{1}{2}\kappa_2}|\partial_U^2 \phi_p|(U,\hat{V}_*) + \epsilon^{1-s}\hat{V}^{\frac{1}{2}\kappa_2}|U|^{-1+p_2\delta} \\[.5em]
    &\lesssim \epsilon^{1-s}\hat{V}^{\frac{1}{2}\kappa_2}|U|^{-1+p_2\delta}.
\end{align*}
In both of the above estimates we have used Lemma \ref{lem:region3databounds} to estimate the data terms. We now choose $\sigma = p_2\delta$ to arrive at (\ref{eqn:region3higherorder}).

\end{proof}

\subsection{Region IV}
\label{subsec:extregionIV}
We finally extend the solution to the asymptotically flat region
\begin{equation}
    \mathcal{R}_{\text{IV}} \doteq \{(U,V)\ | \ V_2 \leq \frac{V^{q_2}}{|U|} < \infty, \ -1\leq U < 0, \ V \geq 1 \}.
\end{equation}
The data for the double null unknowns along the past boundary of $\mathcal{R}_{\text{IV}}$ is induced by the solution in $\mathcal{R}_{\text{II}} \cup \mathcal{R}_{\text{III}},$ as well as outgoing data along $\{U=-1\}$. It will be helpful to decompose this past boundary accordingly, as the union of the three pieces $\mathcal{J}_i$, $i = 1,2,3$, where
\begin{align}
    \mathcal{J}_1 &\doteq \{V=1, -V_2^{-1} \leq U < 0\}, \\[.5em]
    \mathcal{J}_2 &\doteq \{\frac{V^{q_2}}{|U|} = V_2, \ -1 \leq U <  -V_2^{-1}\},\\[.5em]
    \mathcal{J}_3 &\doteq \{U = -1, \ V > V_2\}.
\end{align}
Before turning to the bootstrap assumptions, we discuss the bounds satisfied by initial data along $\mathcal{J}_1 \cup \mathcal{J}_2 \cup \mathcal{J}_3$. In Region IV we will eventually choose $V_2$ large in order to close the estimates. A consequence of the freedom to choose $V_2$ is that the surface $\mathcal{J}_2$ potentially extends into regions of arbitrarily large $V$ coordinate, where we expect asymptotically flat bounds to hold. It follows that the data along $\mathcal{J}_2$ must transition between the bounds thus far considered in $\mathcal{R}_{\text{II}} \cup \mathcal{R}_{\text{III}}$ adapted to the near-singularity region, and the asymptotically flat bounds consistent with the outgoing data along $\mathcal{J}_3$. The following lemma revisits the bounds induced on $\mathcal{J}_2$ by the solution in Region II (cf. Lemma \ref{lem:region3databounds}) to render them consistent with the outgoing data.

\begin{lem}
\label{lem:newinitialdataregionIV}
Assume the bounds (\ref{eq:extbootass2improved}), (\ref{eq:extbootass3improved}), (\ref{eqn:region3higherorder}) in $\mathcal{R}_{\text{II}} \cup \mathcal{R}_{\text{III}}$. There exist constants $\sigma \ll 1 $, $s' \ll 1$ small such that for $\epsilon$ small as a function of $V_2$, along $\mathcal{J}_1 \cup \mathcal{J}_2$ we have 
\begin{alignat}{2}
 |r_p| &\lesssim \epsilon^{1-s'}, \quad \quad  \ \ \quad \quad \quad \quad  |m_p| &&\lesssim \epsilon^{1-s'}, \\[.5em] 
 |^{(U)}\nu_p| &\lesssim \epsilon^{1-s'}, \ \ \ \quad \ \ \ \  \ \quad \quad  |^{(V)}\lambda_p| &&\lesssim \epsilon^{1-s'}V^{-2},\\[.5em] 
 |\partial_U (r\phi_p)| &\lesssim \epsilon^{1-s'}, \ \quad  \ \ \ \  \quad \quad |\partial_V (r\phi_p)| &&\lesssim \epsilon^{1-s'}V^{-2}, \label{eqn:letmebedonetemp3}\\[.5em]
    |{\partial_U} {^{(U)}\nu_p}| &\lesssim \epsilon^{1-s'}|U|^{-1+\sigma}, \ \ \   |{\partial_V} {^{(V)}\lambda_p}| &&\lesssim \epsilon^{1-s'}V^{-3}, \\[.5em]
    |\partial_U^2 (r\phi_p)| &\lesssim \epsilon^{1-s'}|U|^{-1+\sigma}, \ \ \ \ |\partial_V^2 (r\phi_p)| &&\lesssim \epsilon^{1-s'}V^{-3}.
\end{alignat}
The implied constants are independent of $V_2$, for all $V_2 < \infty$.
\end{lem}

\begin{proof}
We begin by estimating $r\phi_p$ and its derivatives. From (\ref{eq:extbootass2improved}), (\ref{eq:extbootass3improved}) it follows that $|\partial_U \phi_p| \lesssim \epsilon^{1-s}$ holds in $\mathcal{J}_1 \cup (\mathcal{R}_{\text{II}}\cap\{V \gtrsim 1\}).$ We have used that the coordinate $V$ is bounded below in the region of interest in order to drop decaying $V$ factors.

Integrating in $U$ from $\{U=-1\}$ gives $|\phi_p| \lesssim \epsilon^{1-s}.$ It then follows by expanding the derivatives and inserting bounds on $^{(U)}\nu_p, \ ^{(V)}\lambda_p, \ \phi_p, \ \partial_U \phi_p, \ \partial_V \phi_p $ along $\mathcal{J}_1 \cup \mathcal{J}_2$ that $$|\partial_U(r\phi_p)|, \ \ |\partial_V (r\phi_p)| \lesssim \epsilon^{1-s} V_2.$$ 
These estimates grow with $V_2$, and are not yet consistent with (\ref{eqn:letmebedonetemp3}). However, one can insert decaying $V$ weights at the cost of adding additional factors of $V_2$. Choosing $\epsilon$ small as a function of $V_2$ then gives the stated bound, after choosing $s$ slightly larger. Assume some $s' > s$ has been chosen.

We are essentially using that along $\mathcal{J}_1 \cup \mathcal{J}_2$, the equivalence $1 \lesssim V \lesssim V_2 $ holds. Weights that decay in $V$ can thus be added freely to any estimate, as long as $\epsilon$ is chosen small to absorb growing powers of $V_2$.

The remaining estimates follow similarly. Note the singular weights $|U|^{-1+\sigma}$ appearing in the estimates for ${\partial_U^2} {^{(U)}\nu_p}, \ \partial_U^2 \phi_p$ are a consequence of the corresponding bounds in (\ref{eqn:region3higherorder}).

Finally, to estimate $m_p$ we first translate bounds on $\mu_p$ to those on $m_p$ using (\ref{eq:mpvsmup}), and then apply the above procedure.
\end{proof}

\begin{figure}
    \centering
  \includestandalone[]{Figures/fig_region4}
  \caption{$\mathcal{R}_{\text{IV}}$ }
  \label{fig:region4}
\end{figure}

\subsubsection{Norms and bootstrap assumptions}
Fix $(\wt{U},\wt{V}) \in \mathcal{R}_{\text{IV}}$ and define the norms
\begin{align}
    \mathfrak{A}_1 &\doteq   \sum_{i=0}^1 \sup_{\mathcal{Q}^{(ex)}_{(\wt{U},\wt{V})}}\Big| V^{i}{\partial_V^i} ^{(V)}\lambda_p\Big| +  \sup_{\mathcal{Q}^{(ex)}_{(\wt{U},\wt{V})}}\Big|  \mu_p\Big| + \sum_{i=1}^2 \sup_{\mathcal{Q}^{(ex)}_{(\wt{U},\wt{V})}}\Big| V^{i} \partial_V^i(r\phi_p)\Big|, \\[1em]
    \mathfrak{A}_2 &\doteq  \sup_{\mathcal{Q}^{(ex)}_{(\wt{U},\wt{V})}}{\Big|}{^{(U)}\nu_p }\Big| + \sup_{\mathcal{Q}^{(ex)}_{(\wt{U},\wt{V})}}\Big| |U|^{1-\sigma} {\partial_U} ^{(U)}\nu_p\Big| +  \sup_{\mathcal{Q}^{(ex)}_{(\wt{U},\wt{V})}}\Big|\partial_U(r\phi_p) \Big|+\sup_{\mathcal{Q}^{(ex)}_{(\wt{U},\wt{V})}}\Big| |U|^{1-\sigma}  \partial_U^2(r\phi_p)\Big|,
\end{align}
as well as the total solution norm
\begin{equation}
    \mathfrak{A} \doteq \mathfrak{A}_1 +  \mathfrak{A}_2.
\end{equation}
Existence in Region IV is a consequence of the following proposition, showing that bootstrap bounds on $\mathfrak{A}$ can be improved independently of $(\wt{U},\wt{V}) \in \mathcal{R}_{\text{IV}}$.
\begin{prop}
Assume a local solution is given in a rectangle $\mathcal{Q}^{(ex)}_{\wt{U},\wt{V}},$ and that the bound 
\begin{equation}
    \label{eq:extbootass5}
    \mathfrak{A} \leq 2A\epsilon^{1-s'}
\end{equation}
holds, where $s'$ is the parameter in Lemma \ref{lem:newinitialdataregionIV}. There exists $V_2$ sufficiently large depending on the background solution, and $\epsilon$ sufficiently small depending on data, such that we have the improved bound 
\begin{equation}
    \label{eq:extbootass5improved}
    \mathfrak{A} \leq A\epsilon^{1-s'}.
\end{equation}
\end{prop}

\subsubsection*{Recovering bootstrap assumptions}
The following lemma is a straightforward consequence of the bootstrap assumptions.
\begin{lem}
\label{lem:region4backgroundmatch}
Assume the bound (\ref{eq:extbootass5}). Then for $\epsilon$ sufficiently small the geometry satisfies
\begin{align}
     (-^{(U)}\nu) \sim 1, \ \ \  ^{(V)}\lambda \sim 1, \ \ \ (1-\mu) \sim 1.
\end{align} 
Moreover, we have the bounds
\begin{alignat}{2}
       |\partial_U (r\phi)| &\lesssim 1, \quad \quad \quad \quad  \quad |\partial_V (r\phi)| &&\lesssim V^{-1}, \\[.5em]
    |{\partial_U} {^{(U)}\nu}| &\lesssim |U|^{-1+\delta}, \quad \ \ \   |{\partial_V} {^{(V)}\lambda}| &&\lesssim V^{-1}, \\[.5em]
    |\partial_U^2 (r\phi)| &\lesssim |U|^{-1+\delta},  \quad \quad |\partial_V^2 (r\phi)| &&\lesssim V^{-2}. 
\end{alignat}
\end{lem}

\vspace{5pt}
\noindent
The bounds captured in the norm $\mathfrak{A}$ are far from sharp in terms of $V$ decay.  The following lemma establishes improved estimates.
\begin{lem}
\label{lem:please}
Assume the bound (\ref{eq:extbootass5}). For $\epsilon$ sufficiently small we have
\begin{equation}
    \label{eq:please0}
    |r_p| \lesssim A\epsilon^{1-s'}, \ \ \ |m_p| \lesssim A\epsilon^{1-s'}, \ \ \ |r\phi_p| \lesssim A\epsilon^{1-s'}.
\end{equation}
\vspace{-1.5em}
\begin{align}
    \label{eq:please.5}
    |^{(V)}\lambda_p| &\lesssim A\epsilon^{1-s'}V^{-2}, \ \ \ \quad  |{\partial_V}{^{(V)}\lambda_p}| \lesssim A\epsilon^{1-s'}V^{-3}, \\[1em]
    \label{eq:please.6}
    |\partial_V(r\phi_p)| &\lesssim A\epsilon^{1-s'}V^{-2}, \ \ \ \quad |\partial_V^2(r\phi_p)| \lesssim A\epsilon^{1-s'}V^{-3},
\end{align}
Moreover, the bounds
\begin{equation}
    \label{eq:please1}
    r \sim V, \ \ \ \quad \mu \lesssim V^{-1},
\end{equation}
hold throughout $\mathcal{R}_{\text{IV}}$.
\end{lem}
\begin{proof}
We start with a preliminary logarithmic bound on $r_p$. Applying the fundamental theorem of calculus in the $V$ direction, and the bound $|^{(V)}\lambda_p| \lesssim A\epsilon^{1-s'}V^{-1}$, we conclude $|r_p| \lesssim A\epsilon^{1-s'}\log V$. This bound is already better as $V \rightarrow \infty$ than the corresponding bound on the background solution, $\overline{r} \sim V$, and so after choosing $\epsilon$ small as a function of $A$ we arrive at the first statement in 
(\ref{eq:please1}).

Similarly, integrating the bound $|\partial_U(r\phi_p)| \lesssim A\epsilon^{1-s'}$ in $U$ gives $|r\phi_p| \lesssim A\epsilon^{1-s'}.$ Here we have used bounds on $r\phi_p$ along $\mathcal{J}_2, \ \mathcal{J}_3$.

It follows from the identity $\partial_U \phi_p = r^{-1}(\partial_U(r\phi_p) - ^{(U)}\nu \phi_p)$, and the bounds just derived, that $|\partial_U \phi_p| \lesssim A\epsilon^{1-s'}V^{-1}$. A similar argument shows $\partial_U \overline{\phi} \lesssim V^{-1}.$

From the evolution equation for $m$ in the $U$ direction we derive
\begin{equation*}
    \partial_U m_p = \Big(\frac{r^2}{2 ^{(U)}\nu}(1-\mu)(\partial_U \phi)^2 \Big)_p.
\end{equation*}
For $\epsilon$ sufficiently small the right hand side can be bounded by $A\epsilon^{1-s'}.$ Integrating in the $U$ direction gives the desired bound on $m_p$, noting that $m_p(-1,V) \lesssim \epsilon$ follows from the assumptions on outgoing initial data. Having established boundedness of $m_p$ it follows that $m$ is bounded, and thus $\mu \lesssim V^{-1}$ follows.

To conclude better decay on $^{(V)}\lambda_p$, revisit (\ref{PSSESF:1}). Bootstrap assumptions imply the zeroth order term has a good sign, and the remaining inhomogeneous term is bounded by $A\epsilon^{1-s'}V^{-2}.$ Note we have used the improved bound on $m_p$ in this estimate. Integrating in $U$ from data along $\{U=-1\}$, we conclude $|^{(V)}\lambda_p| \lesssim A\epsilon^{1-s'}V^{-2}$.

An analogous argument applies to give the improved estimates on ${\partial_V} {^{(V)}\lambda},\  \partial_V (r\phi_p),$ and $\partial_V^2(r\phi_p)$. Finally, integrating the improved bound on $^{(V)}\lambda_p$ in the $V$ direction gives the stated bound on $r_p$.
\end{proof}

\vspace{5pt}
\noindent
We now turn to improving the bootstrap assumptions, beginning with unknowns satisfying $U$ equations. 

\begin{lem}
Assume the bound (\ref{eq:extbootass5}). Then for $\epsilon$ sufficiently small we have
\begin{equation}
    \mathfrak{A}_1  \leq \frac{1}{10}A\epsilon^{1-s'}.
\end{equation}
\end{lem}
\begin{proof}
Let 
\begin{equation*}
    \Psi_p \in \{V \, {^{(V)}}\lambda_p, \ V^2 \, {\partial_V} {^{(V)}\lambda_p}, \ \mu_p, \ V\partial_V(r\phi_p), \ V^2\partial_V^2(r\phi_p) \}.
\end{equation*}
Estimating the right hand sides of (\ref{PSSESF:1}), (\ref{PSSESF:3}), (\ref{PSSESF:6}), (\ref{PSSESF:7}), (\ref{PSSESF:9}) yields the uniform bound
\begin{equation*}
    |\partial_U \Psi_p| \lesssim A\epsilon^{1-s'}V^{-1}.
\end{equation*}
For $(U,V)$ with $V \geq V_2$, integration in $U$ picks up a boundary term along $\mathcal{J}_3$. We conclude
\vspace{-1em}
\begin{align*}
    |\Psi_p|(U,V) &\lesssim |\Psi_p|(-1,V) + A\epsilon^{1-s'}\int\limits_{(-1,V)}^{(U,V)}V^{-1}dU' \\
    &\lesssim \epsilon + A\epsilon^{1-s'}V_2^{-1}.
\end{align*}
In the above we have used the assumptions on the outgoing data along $\{U=-1\}$.  Now we choose $V_2$ large to render the error term appropriately small, improving the bootstrap assumption.

For $(U,V)$ with $V \leq V_2$, integration in $U$ picks up a boundary term at $(-V^{q_2}V_2^{-1}, V ) \in \mathcal{J}_2$. We conclude in a similar manner, using the bounds of Lemma \ref{lem:newinitialdataregionIV} and the estimate
\begin{align*}
    A\epsilon^{1-s'}V^{-1}\int\limits_{(-V^{q_2}V_2^{-1}, V )}^{(U,V)}dU' \lesssim  A\epsilon^{1-s'}V_2^{-1}.
\end{align*}
\end{proof}
\vspace{-2.5em}
\begin{lem}
Assume the bound (\ref{eq:extbootass5}). Then for $\epsilon$ sufficiently small we have
\begin{equation}
    \mathfrak{A}_2  \leq \frac{1}{10}A\epsilon^{1-s'}.
\end{equation}
\end{lem}
\begin{proof}
We will not be able to use the smallness of $V_2^{-1}$ when integrating in the outgoing direction. Instead, as for Region III we rely on an ordering of the estimates in which each step only relies on already improved bounds. Note that the quantities $r_p, r\phi_p$ can be estimated solely in terms of data and $\mathfrak{A}_1$, which we assume has been done.

To estimate $^{(U)}\nu_p$ we consider (\ref{eq:ext3tempest0}) as an equation for $\frac{{^{(U)}\nu_p}}{^{(U)}\overline{\nu}}$. Inserting the improved bootstrap bound on $\mathfrak{A}_1$ and the bounds of Lemma \ref{lem:please} gives for $\epsilon$ sufficiently small, 
\begin{equation*}
    |(\mathcal{G}_1)_p| \lesssim \epsilon^{1-s'}V^{-2}.
\end{equation*}
Integrating (\ref{eq:ext3tempest0}) in $V$ from $\mathcal{J}_1 \cup \mathcal{J}_2$ and solving for $^{(U)}\nu_p$, we conclude the stated bound.

A similar argument applies to $\partial_U(r\phi_p)$. Estimating the right hand side of (\ref{PSSESF:5}) yields
\begin{equation*}
    |\partial_V \partial_U (r\phi_p)| \lesssim \epsilon^{1-s'}V^{-2}.
\end{equation*}
Integrating in $V$ gives the desired bound. Note we required the improved bound on $^{(U)}\nu_p$ in this estimate.

Proceeding to ${\partial_U}{^{(U)}\nu_p}$ and $\partial_U^2(r\phi_p)$ respectively, we apply the same strategy to the $U$-commuted (\ref{eq:ext3tempest0}) and to (\ref{PSSESF:12}) respectively. The key difference is that the data terms along $\mathcal{J}_1$ contribute singular weights in $U$. As these are $U$ weights, they do not affect the integration in $V$.
\end{proof}

\subsection{Concluding the proof of Theorem \ref{thm2:exterior}}
We are now in a position to conclude the proof of Theorem \ref{thm2:exterior}. Let $(r,m,\phi)$ denote the exterior solution constructed on $\mathcal{Q}^{(ex)}$ in Sections \ref{subsec:extlocalexist}-\ref{subsec:extregionIV}.
\begin{lem}
The exterior solution $(r,m,\phi)$ in $\mathcal{Q}^{(ex)}$ satisfies the following properties:
\begin{enumerate}
    \item The spacetime is asymptotically flat, and contains an incomplete $\mathcal{I}^+$.
    \item The functions $\Psi \in \{{r,} {^{(U)}\nu,} {^{(V)}\lambda,} \  \mu, {\partial_V} {^{(V)}\lambda,} \ \partial_U \phi, \ \partial_V \phi, \ \partial_V^2 \phi\}$ extend continuously as functions on $\{U=0\}$ away from $(U,V) = (0,0)$, and there exists an $s', \delta>0$ such that the limits along $\{U=0\}$ satisfy bounds
    \begin{equation}
       |\lim_{U \rightarrow 0} \Psi_p|(V) \lesssim \epsilon^{1-s'} \min(V^{\delta},1).
    \end{equation}
\end{enumerate}
\end{lem}
\begin{proof}
The asymptotic flatness of outgoing null hypersurfaces follows from the estimates of Lemma \ref{lem:please}. Given the explicit bounds on the perturbations $\Psi_p$ as $V \rightarrow \infty$, one can use the same argument as in Section \ref{subsubsec:asympflattrunc} to show that $\mathcal{I}^+$ is incomplete.

The existence of the continuous extensions to $\{U=0\}$ follows from estimates (\ref{eq:extbootass3improved}), (\ref{eq:extbootass5improved}), and the system of equations for the various $\partial_U \Psi_p$. These rates are integrable in $U$ as $U \rightarrow 0$, giving the existence of the limits. Moreover, the added $V^{\delta}$ decay as $V\rightarrow 0$ is contained in the norm (\ref{eq:extbootass3improved}), implying the enhanced regularity of the perturbations as $V \rightarrow 0$.
\end{proof}

%% file: Figures/fig_region1.tex
\begin{tikzpicture}[scale=1.4]
 \node[circle, label= above:{$\mathcal{O}$}, draw, inner sep =0, minimum size = .2cm] (n1) at (0,0) {};
 
 \draw[thick, fill = lightgray] (n1) -- node[sloped,below]{\small$\{v=0\}$}(-45:4)-- (-45:4) -- node[sloped,below]{\small$\{u=-1\}$}+(45:2)--  node[circle, fill = black, inner sep = 0, minimum size = .15cm, label = right:{\small$(-1,v_1)$}]{} +(45:2) -- cycle ;
 
 \node[sloped] (n2) at ($(-45:2.5) + (45:.6) $ ) {\large$\mathcal{R}_{\text{I}}$ };
 
\end{tikzpicture}

%% file: Figures/fig_region2.tex
\begin{tikzpicture}[scale=1.4]
\node[circle, label= above:{$\mathcal{O}$}, draw, inner sep =0, minimum size = .2cm] (n1) at (0,0) {};
\draw[thick] (n1) -- ++(-45:4) -- +(45:2);
\draw[thick, fill = lightgray] (n1) -- (-18.5:4.47) -- ($(-45:4) + (45:4) $) -- cycle ;

\node[label = {[rotate = -45]below:\small$\{V=0\}$}] (n2) at (-45:2) {};
\node[label = {[rotate=45]below:\small$\{U=-1\}$}] (n3) at ($(-45:4) + (45:2)$) {};
\node[circle, fill, inner sep = 0, minimum size = .15cm, label = {right:\small$(-1,V_2)$}] (n4) at ($(-45:4) + (45:4)$) {};
\node[label = {[rotate = -20.5]below:\small$\{\frac{V^{q_2}}{|U|}=V_1\}$}] (n5) at ($(-45:2.5) + (45:1.3)$) {};
\node[label = {above:\small$\{\frac{V^{q_2}}{|U|}=V_2\}$}] (n6) at (-1:3) {};
\node[] (n7) at ($(0:3) + (-90:.5)$) {\large$\mathcal{R}_{\text{II}}$};
\end{tikzpicture}

%% file: Figures/fig_region3.tex
\begin{tikzpicture}[scale=1.4]

\fill[lightgray] (0,-.04) -- ++(45:2.01) -- + (-45:2.00) ;
\node[circle, label= above:{$\mathcal{O}$}, draw, inner sep =0, minimum size = .2cm] (n1) at (0,0) {};  

\draw[thick] (n1) -- ++(-45:4) -- +(45:4) -- cycle;
\draw[dashed] ($(0,-.04)+(45:.02)$) -- +(45:2);
\draw[thick] (45:2) -- + (-45:2.05);

\node[label = {[rotate = -45]above:\small$\{V=1\} $ }] (n2) at ($(45:2) + (-45:1.25)$) {};
\node[label = {[rotate=45]below:\small$\{U=-1\}$}] (n3) at ($(-45:4) + (45:2)$) {};
\node[label = {above:\small$\{\frac{V^{q_2}}{|U|}=V_2\}$}] (n4) at (-2:4) {};
\node[] (n5) at ($(45:1.4) + (-45:.7) $) {\large$\mathcal{R}_{\text{III}}$};
\end{tikzpicture}

%% file: Figures/fig_region4.tex
\begin{tikzpicture}[scale=1.0]
\node[circle, label= above:{$\mathcal{O}$}, draw, inner sep =0, minimum size = .2cm] (n1) at (0,0) {};
\fill[lightgray] (45:2) -- ++ (-45:2.05) -- ++ (0:2.81) --  ++ (45:2) --++ (135:4.05); 
\draw[thick] (n1) -- ++(-45:4) -- +(45:4) -- cycle;
\draw[dashed] (0,0) -- + (45:6);
\draw[thick] ($(-45:4)+(45:4)$) -- +(45:2);
\draw[thick] (45:2) -- + (-45:2.05);
\node[label = {[rotate = -45]below:\small$\{V=1\} $  }] (n2) at ($(45:2.1) + (-45:.9) $) {};

\draw[dashed] (45:4) -- +(-45:4);
\node[label = {[rotate = -45]below:\small$\{V=V_2\}$}] (n3) at ($(45:4.1) + (-45:1.7) $) {};

\node[] (n4) at ($(45:3) + (-45:1.5) $) {\large$\mathcal{R}_{\text{IV}}$};
\node[label = {[rotate=45]below:\small$\{U=-1\}$}] (n6) at ($(-45:4) + (45:3)$) {};
\node[label = {[rotate=45]above:\small$\{U=0\}$}] (n7) at (45:3) {};
\node[label = {below:\small$\{\frac{V^{q_2}}{|U|}=V_2\}$}] (n8) at (2:3) {};

\node[] (n9) at ($(45:4.6) + (-45:2) $ ) {};
\node[] (n10) at ($(45:4.6) + (-45:2) + (45:1)$) {};
\node[label = {[rotate = 45]above:\small$V \rightarrow \infty$} ] (n11) at ($(45:4.6) + (-45:2) + (45:.5)$) {};

\draw[->] (n9) -- (n10);

\end{tikzpicture}

%% file: appA.tex
\appendix 

\section{Construction of $(\overline{g}_k,\overline{\phi}_k)$}
\label{appA}
This appendix discusses the existence and global properties of $k$-self-similar naked singularities in double null gauge, cf. Theorem \ref{thm:christodoulou_solutions}. Recall that the original argument of \cite{chris2} takes place in Bondi coordinates $(u,r)$, with the area radius $r$ treated as a coordinate. For the purposes of proving the main results of the paper, Theorems \ref{thm1:interior}-\ref{thm2:exterior}, we require analytical tools adapted to the Einstein-scalar field system in double null coordinates. Thus the main role of the appendix is to effect a translation of the results in \cite{chris2} to a particular gauge, here termed \textit{renormalized double null coordinates}. The proof presented here is not self-contained, and relies crucially on the self-similar analysis of the spherically symmetric Einstein-scalar field system presented in \cite{chris2}.

Section \ref{subsec:appA1} motivates the class of $k$-self-similar solutions, and defines \textit{self-similar double null coordinates} adapted to the symmetry. Section \ref{subsec:A2} reduces the full Einstein-scalar field system to an autonomous ODE system, which is shown to be identical with the system obtained in \cite{chris2}. Applying the ODE analysis of the latter paper and transforming back to a double null gauge, we discuss in Sections \ref{appa:regularcoords}, \ref{subsec:appAext} the resulting interior and exterior regions of spacetime. The final step in the construction consists of an asymptotically flat truncation in the exterior region. Although technically straightforward, we discuss the truncation procedure explicitly and collect various bounds on the solution. 

A key insight from the construction is the breakdown of self-similar double null coordinates along the past and future similarity horizons, motivating the introduction of a renormalized double null gauge (see Section \ref{appa:regularcoords}). The properties of the self-similar solutions in the renormalized gauge, including blowup rates of various double null unknowns, directly motivate the assumptions for admissible spacetimes outlined in Section \ref{sec:assumptionsbackground}.

\subsection{Consequences of $k$-self-similarity}
\label{subsec:appA1}
The point of departure for Christodoulou's construction of the $(\overline{g}_k,\overline{\phi}_k)$ solutions is the notion of $k$-self-similarity. Assume that the spacetime $(\mathcal{M},\textbf{g}_{\mu \nu})$ admits a spherically symmetric, conformally Killing vector field $S$, which generates dilations about a central point $\mathcal{O}$. It follows that the metric and scalar field quantities satisfy 
\begin{equation}
    \label{eq:app1}
    \mathcal{L}_S g_{\mu \nu} = 2 g_{\mu \nu}, \ Sr = r, \ S\phi = -k,
\end{equation}
where $g_{\mu \nu}$ is the metric on the quotient spacetime $\mathcal{Q}$. Explicitly, $g = -\Omega^2(u,v)dudv$. Moreover, we will work exclusively with $k \in \mathbb{R}_+$ satisfying $k^2 \in (0, \frac{1}{3})$. While \cite{chris2} considers a wider range of values, it is shown that such an extension does not lead to new naked singularity solutions. In particular, the so called \textit{scale-invariant} solutions with $k=0$ can be written down explicitly, and do not model singularity formation. For details, see \cite{chris1}.

The solutions will be constructed in \textit{self-similar double null coordinates}, adapted to the scaling symmetry. Wherever this system is regular, it will follow that $S$ takes the simple form
\begin{equation}
    \label{eq:app2}
    S = \hat{u}\partial_{\hat{u}} + \hat{v}\partial_{\hat{v}}.
\end{equation}
Armed with this coordinate condition, the assumption (\ref{eq:app1}) places strong restrictions on the form of the metric functions and the scalar field. Compute 
$$(\mathcal{L}_S g)_{\mu \nu} = S(g(\partial_{\mu}, \partial_\nu)) - g([S,\partial_\mu],\partial_\nu) - g(\partial_\mu, [S,\partial_\nu]), $$ 
with $\mu,\nu \in \{\hat{u},\hat{v} \}.$ The nonvanishing commutators are given by 
$$[S,\partial_{\hat{u}}] = - \partial_{\hat{u}}, \ \ [S,\partial_{\hat{v}}] = - \partial_{\hat{v}},$$ implying 
\begin{align*}
    (\mathcal{L}_S g)_{\hat{u} \hat{u}} &= (\mathcal{L}_S g)_{\hat{v} \hat{v}} = 0, \ 
     (\mathcal{L}_S g)_{\hat{u} \hat{v}} = -S(\Omega^2) - 2\Omega^2.
\end{align*}
It follows from (\ref{eq:app1}) that 
\begin{equation}
    \label{eq:app3}
    S(\Omega^2) = 0, \ Sr = r, \ S(\phi) = -k,
\end{equation}
and it remains to analyze these reduced equations. Define the variable $z = -\frac{\hat{v}}{\hat{u}}$, which is well-defined on $\mathcal{Q}$, and introduce the coordinate system 
\begin{equation}
    (\wt{u},z) \doteq (\hat{u}, -\frac{\hat{v}}{\hat{u}}). 
\end{equation}
The coordinate $\tilde{u}$ is distinguished from $u$ as a reminder that the coordinate derivatives do not agree. In fact, we have the relations
\begin{equation}
    \label{eq:app4}
   \partial_{\hat{u}} =  \partial_{\wt{u}} - \frac{z}{\wt{u}}\partial_z, \ \partial_{\hat{v}} = -\frac{1}{\wt{u}}\partial_z.
\end{equation}
The conformal Killing field is given in the new coordinates by 
\begin{equation*}
    S = \wt{u}\partial_{\wt{u}}.
\end{equation*}
The equations (\ref{eq:app3}) are equivalent to 
\begin{equation*}
    S(\Omega^2) = S(\frac{r}{\wt{u}}) = S(\phi + k\log(-\wt{u})) = 0,
\end{equation*}
and thus there exist functions $\mathring{\Omega}(z), \mathring{r}(z), \mathring{\phi}(z)$ such that 
\begin{equation}
    \label{eq:app5}
    \Omega^2(\wt{u},z) = \mathring{\Omega}^2(z), \ \  r(\wt{u},z) = -\wt{u}\mathring{r}(z), \ \  \phi(\wt{u},z) = \mathring{\phi}(z) - k\log (-\wt{u}).
\end{equation}
We can in fact derive similar expressions for all the double null unknowns appearing in the system (\ref{SSESF:1:1})-(\ref{SSESF:1:5}). Applying (\ref{eq:app5}) along with the coordinate derivative expression (\ref{eq:app4}) gives
\begin{align}
    \nu &= \partial_{\hat{u}} r = (\partial_{\wt{u}} - \frac{z}{\wt{u}}\partial_z)(-\wt{u}\mathring{r}(z)) = -\mathring{r}(z) + z\partial_z \mathring{r}(z) \doteq \mathring{\nu}(z) \label{eq:app6}\\
    \lambda &= \partial_{\hat{v}}r =  -\frac{1}{\wt{u}}\partial_z (-u\mathring{r}(z)) = \partial_z \mathring{r}(z) \doteq \mathring{\lambda}(z),\label{eq:app7}\\
    \mu &= 1 + \frac{4\lambda \nu}{\Omega^2} = 1 + \frac{4\mathring{\lambda}(z) \mathring{\nu}(z)}{\mathring\Omega(z)^2} \doteq \mathring{\mu}(z), \label{eq:app8}\\
    m &= \frac{\mu r}{2} \doteq -\wt{u} \mathring{m}(z),\label{eq:app9}\\
    \partial_{\hat{u}} \phi &= (\partial_{\wt{u}} - \frac{z}{\wt{u}}\partial_z)(\mathring{\phi}(z) - k\log(-\wt{u})) = -\frac{z}{\wt{u}}\partial_z \mathring{\phi}(z) - \frac{k}{\wt{u}}, \label{eq:app10}\\
    \partial_{\hat{v}} \phi &= -\frac{1}{\wt{u}}\partial_z (\mathring{\phi}(z) - k\log(-\wt{u})) = -\frac{1}{\wt{u}}\partial_z \mathring{\phi}(z),\label{eq:app11}
\end{align}
where we have introduced functions $\mathring{\nu}(z), \mathring{\lambda}(z), \mathring{\mu}(z),$ and $\mathring{m}(z)$. Comparing (\ref{eq:app6}) and (\ref{eq:app7}) gives the useful algebraic relationship
\begin{equation}
\label{eq:appalgrel1}
\mathring{\nu} + \mathring{r} = z\mathring{\lambda}. \newline
\end{equation}
We conclude this section with an overview of the solution manifold. With respect to a self-similar double null coordinate system $(\hat{u},\hat{v})$, the $(\bar{g}_k,\bar{\phi}_k)$ solutions will initially be defined on
\begin{equation*}
    \mathcal{Q} = \{(\hat{u},\hat{v})\ |  -\infty \leq \hat{u} < 0, \ \hat{u} \leq \hat{v} < \infty\}.
\end{equation*}
In Section \ref{subsubsec:asympflattrunc} this spacetime will be truncated in the region $\hat{v} \gg 1$ to generate an asymptotically flat spacetime; however, the underlying self-similar solution is defined on the whole of $\mathcal{Q}$.

The timelike curve $\{\hat{u} = \hat{v}, \hat{u} < 0\}$ will be denoted $\Gamma$, and coincides with the set $\{r(\hat{u}, \hat{v}) = 0\}$. An interesting question concerns the possible extensions to $\hat{u} > 0$, and the existence of a regular center and/or singularities in the extended spacetime. For details on extensions with a regular center with $\hat{u} > 0$, see \cite{chris2}. We will not discuss extensions further here.

The axis is generated by the vector field $\partial_{\hat{u}} + \partial_{\hat{v}}$, which will be useful for translating regularity conditions along $\Gamma$ to conditions on the coordinate derivatives. In particular, for a suitably regular solution the relation $\mathring{\lambda} + \mathring{\nu} = 0$ holds along the axis. We will in fact have that the solution is smooth in a $z$-neighborhood of $\Gamma$, and thus the natural regularity conditions will be applicable.

We have not yet exhausted the gauge freedom inherent to self-similar double null coordinates. A $1$-parameter scaling freedom 
\begin{equation*}
    \hat{u} \rightarrow a\hat{u}, \ \hat{v} \rightarrow a\hat{v}
\end{equation*}
remains, for any $a > 0$. Under this scaling, $\Omega^2(\hat{u},\hat{v})$ transforms as $\Omega^2 \rightarrow a^{-2}\Omega^2$, and so $a$ can be fixed in order to set $\mathring{\Omega}^2(-1) = 1$. Make this choice, and thereby fix the choice of self-similar coordinate system.

\subsection{The interior solution}
\label{subsec:A2}
The previous section explored the consequences of the self-similar ansatz, and identified a natural self-similar coordinate system spanned by $(\tilde{u},z).$ Using the expressions (\ref{eq:app5})-(\ref{eq:app11}) to rewrite the scalar field system in terms of self-similar coordinates, we arrive at the following proposition:
\begin{prop}
Under the assumption of $k$-self-similarity, (\ref{SSESF:1:1})-(\ref{SSESF:1:5}) formally implies the following system for the variables $(\mathring{r},\mathring{\nu},\mathring{\lambda},\mathring{\Omega},\mathring{\phi}')$. We use the notation $\frac{d}{dz} \mathring{\phi} \doteq \mathring{\phi}'$.
\begin{align}
    \label{eq:app12}
    \mathring{r}z\frac{d}{dz} \mathring{\lambda} &= -\mathring{\nu} \mathring{\lambda} - \frac{1}{4}\mathring{\Omega}^2, \\[.5em]
    \label{eq:app13}
    \mathring{r}\frac{d}{dz} \mathring{\nu} &= -\mathring{\nu} \mathring{\lambda} - \frac{1}{4}\mathring{\Omega}^2,\\[.5em]
    \label{eq:app14}
    2 \mathring{\Omega}^{-1} \mathring{\nu} z \frac{d}{dz}\mathring{\Omega} &= z\frac{d}{dz}\mathring{\nu}+\mathring{r}(z\mathring{\phi}'+k)^2,\\[.5em]
    \label{eq:app15}
    2\mathring{\Omega}^{-1}\mathring{\lambda}\frac{d}{dz}\mathring{\Omega} &= \frac{d}{dz}\mathring{\lambda} + \mathring{r} (\mathring{\phi}')^2,\\[.5em]
    \label{eq:app16}
    \mathring{r}z\frac{d}{dz}\mathring{\phi}' &= -2\mathring{\lambda}z \mathring{\phi}' - k\mathring{\lambda}.
\end{align}
\end{prop}
The self-similar reduction of the wave equation for $\mathring{\Omega}$ is not included in the above set, and is not required for generating the autonomous system below. The strategy for constructing the interior solution will be to solve the ODE system (\ref{eq:app12})-(\ref{eq:app16}) with \enquote{initial data} along $\Gamma$, or equivalently, $\{z=-1\}$. It will follow from the reduction below that $\Gamma$ is a critical point for the ODE system, and therefore care will have to be taken when discussing local existence of solutions.

In addition to the condition $\mathring{r}(-1) = 0$, the gauge choice implies $\mathring{\Omega}(-1) = 1$. It therefore remains to identify $\mathring{\lambda}(-1), \mathring{\nu}(-1),$ and $\mathring{\phi}'(-1).$ The solution will contain a regular center, and thus the quantities $\mathring{\lambda}(-1), \ \mathring{\nu}(-1)$ are related along the axis by regularity requirements for $\mathring{r}, \mathring{\mu}$. In particular, we have
\begin{equation*}
    \mathring{\lambda}(-1) = - \mathring{\nu}(-1),
\end{equation*}
and
\begin{equation*}
   0  = \mathring{\mu}(-1) = 1 + \frac{4\mathring{\nu}(-1)\mathring{\lambda}(-1)}{\mathring{\Omega}^2(-1)}.
\end{equation*}
Taken together, these equations imply $\mathring{\nu}(-1) = -\frac{1}{2}, \ \mathring{\lambda}(-1) = \frac{1}{2}. $
Finally, regularity of $\frac{d}{dz}\mathring{\phi}'$ at $z=-1$ implies the right hand side of (\ref{eq:app16}) must vanish, i.e. $\mathring{\phi}'(-1) = \frac{k}{2}.$
In summary, the initial conditions are given by
\begin{equation}
    \mathring{r}(-1) = 0, \ \ \ \mathring{\Omega}(-1) = 1, \ \ \ \mathring{\nu}(-1) = -\frac{1}{2}, \ \ \ \mathring{\lambda}(-1) = \frac{1}{2}, \ \ \ \mathring{\phi}(-1) = \frac{k}{2}.
\end{equation}
\subsubsection*{Reduction to an autonomous system}
The full system (\ref{eq:app12})-(\ref{eq:app16}) is difficult to analyze in its present form. Thankfully, one may reduce the system to a pair of autonomous equations. The starting point for this reduction is an algebraic identity for $\mathring{\Omega}^2$, allowing its elimination from the system. 

\begin{lem}
The following algebraic identity holds:
\begin{equation}
    \label{eq:appalgrel2}
    \frac{1}{4}\mathring{\Omega}^2 = -\mathring{\nu}\mathring{\lambda} + \mathring{r}^2z(\mathring{\phi}')^2  + 2k z\lambda  \mathring{r} \mathring{\phi}'+ \mathring{r}\mathring{\lambda} k^2.
\end{equation}
\end{lem}
\begin{proof}
Comparing (\ref{eq:app14}) and (\ref{eq:app15}) and eliminating $2\mathring{\Omega}^{-1}\frac{d}{dz}\mathring{\Omega}$ gives the equality
\begin{equation*}
    \frac{1}{\mathring{\nu}}\frac{d}{dz}\mathring{\nu} + \frac{\mathring{r}}{z\mathring{\nu}}(z\mathring{\phi}'+k)^2 = \frac{1}{\mathring{\lambda}}\frac{d}{dz}\mathring{\lambda} + \frac{\mathring{r}}{\mathring{\lambda}}(\mathring{\phi}')^2.
\end{equation*}
Inserting (\ref{eq:app12}), (\ref{eq:app13}) and clearing denominators gives
\begin{equation*}
    (\mathring{\nu}\mathring{\lambda}+\frac{1}{4}\mathring{\Omega}^2)(z\mathring{\lambda}-\mathring{\nu}) = \mathring{r}^2 \mathring{\lambda}(z\mathring{\phi}'+k)^2 - \mathring{r}^2z\mathring{\nu}(\mathring{\phi}')^2.
\end{equation*}
Using (\ref{eq:appalgrel1}) and simplifying gives the result.
\end{proof}

Define the functions 
\begin{equation}
    \psi(z) \doteq \frac{\mathring{r}}{\mathring{r}+\mathring{\nu}},
    \ \ \ \theta(z) \doteq z \psi(z) \mathring{\phi}'(z).
\end{equation}
The goal is to reduce the above ODE system to one for the quantities $\psi(z), \ \theta(z)$. Differentiating these quantities and inserting (\ref{eq:app12})-(\ref{eq:app16}), (\ref{eq:appalgrel2}) gives
\begin{align*}
     \frac{d}{dz}\psi &= \frac{\lambda}{\mathring{r}+\mathring{\nu}} - \frac{\mathring{r}(\mathring{\lambda}+\frac{d}{dz}\mathring{\nu})}{(\mathring{r}+\mathring{\nu})^2} \\
    &= \frac{1}{z} - \frac{\psi \mathring{\lambda}}{\mathring{r}+\mathring{\nu}} + \frac{\psi(\mathring{r}z(\mathring{\phi}')^2 +2kz \lambda \mathring{\phi}' + \mathring{\lambda} k^2 )}{(\mathring{r}+\mathring{\nu})} \\
    &= \frac{1}{z}((\theta+k)^2+(1-k^2)(1-\psi)).
\end{align*}
Similarly,
\begin{align*}
    \frac{d}{dz}\theta &= \psi \mathring{\phi}' + z \mathring{\phi}'\frac{d}{dz}\psi  + z\psi \frac{d}{dz}\mathring{\phi}' \\
    &= \psi \mathring{\phi}'  + \mathring{\phi}'((\theta+k)^2+(1-k^2)(1-\psi)) - \psi \frac{1}{\mathring{r}}(2\mathring{\lambda}z\mathring{\phi}' + k\mathring{\lambda})\\
    &= \frac{1}{z\psi}(\psi\theta + \theta((\theta+k)^2+(1-k^2)(1-\psi)) -2 \theta  -k\psi ) \\
    &= \frac{1}{z\psi}(k\psi(k\theta-1) + \theta( (\theta+k)^2-(1+k^2))).
\end{align*}
To make contact with the ODE system considered by Christodoulou in \cite{chris2}, assume a solution for $\psi(z)$ is given on some interval $z \in [-1,z_f)$, with $z\psi(z) \geq 0$ holding for all $z \in (-1,z_f)$. Choose any $(z_0, s_0)$ with $z_0 \in (-1,z_f)$, \ $s_0 > -\infty$, and define the change of variables $s(z)$ by 
\begin{equation}
    \label{eq:app16.5}
    s(z) \doteq s_0- \int\limits_{z}^{z_0}\frac{1}{z'\psi(z')}dz'.
\end{equation}
We first claim that $\lim_{z\rightarrow -1^{+}} s(z) = -\infty$. The boundary condition for $\mathring{\lambda}$ implies $\lim_{z \rightarrow -1} \partial_z \mathring{r} = \frac{1}{2},$ and therefore in any suitably regular solution to the ODE system we must have
\begin{equation*}
    \mathring{r}(z) = \frac{1}{2}|z+1| + O(|z+1|^2).
\end{equation*}
In combination with the limit 
\begin{equation*}
    \lim_{z \rightarrow -1} \frac{\psi}{\mathring{r}} = \lim_{z \rightarrow -1} \frac{1}{\mathring{r}+\mathring{\nu}} = -2,
\end{equation*}
it follows that 
\begin{equation}
 \label{eq:app16.6}
    \psi(z) = -|z+1| + O(|z+1|^2).
\end{equation}
Inserting into (\ref{eq:app16.5}) shows that in a neighborhood of $z=-1$, 
\begin{equation}
    \label{eq:app16.7}
    s(z) = \log |z+1| + O_{z\rightarrow -1}(1),
\end{equation}
where $O_{z\rightarrow -1}(1)$ denotes terms that are bounded as $z \rightarrow -1$.

The system for $(\psi,\theta)$ takes an especially nice form in terms of the parameter $s$. The following autonomous system is identical to that in \cite{chris2}, with the functions $(\alpha,\theta)$ of that paper replacing $(\psi, \theta)$ here. 
\begin{align}
    \frac{d}{ds}\psi &= \psi( (\theta+k)^2 + (1-k^2)(1-\psi)), \label{eq:app17}\\
    \frac{d}{ds}\theta &= k\psi(k\theta -1) + \theta( (\theta+k)^2-(1+k^2)).\label{eq:app18}
\end{align}

In the remainder of this section we use results of \cite{chris2}, in which a careful study of the existence and long time behavior of solutions to (\ref{eq:app17})-(\ref{eq:app18}) is undertaken. It remains to translate these results on the solution $(\psi, \theta)$ into control on the complete set of double null unknowns. Two complications arise in directly using the analysis of \cite{chris2}. One issue is that the parameter $s$ may only be translated back to the original parameter $z$ through the coordinate transformation defined by (\ref{eq:app16.5}), which itself depends on the function $\psi$. Another difficulty is that the unknowns considered in \cite{chris2}, by virtue of the Bondi coordinate system used, are not directly comparable to the double null unknowns we need to estimate here. Still, the bulk of the hard analysis in the construction lies in analyzing (\ref{eq:app17})-(\ref{eq:app18}), and this analysis may be used as a black box.

Before turning to the results of the ODE analysis, we discuss the initial conditions for $\psi,\theta$ as $s \rightarrow -\infty$ (equivalently, as $z \rightarrow -1$). The data for $\mathring{r},\mathring{\nu}, \mathring{\phi}'$ implies 
\begin{equation*}
    \lim_{s\rightarrow -\infty} \psi(s) = 0, \ \ \  \lim_{s\rightarrow -\infty} \theta(z) = 0,
\end{equation*}
and therefore the axis data constitutes a critical point of the system (\ref{eq:app17})-(\ref{eq:app18}). The linearization analysis shows $(\psi,\theta) = (0,0)$ to be a saddle point with eigenvalues $\pm 1$. Moreover, local existence and uniqueness follows given a choice of the limit 
\begin{equation}
    \lim_{s\rightarrow -\infty} \psi e^{-s}.
\end{equation}
Comparing (\ref{eq:app16.6}), (\ref{eq:app16.7}) gives
\begin{equation}
    \label{eq:app19}
    \lim_{s\rightarrow -\infty} \psi e^{-s} = -1.
\end{equation}
The above computations are formal, and one may phrase the logic more rigorously as follows. Choose the boundary condition as in (\ref{eq:app19}), and once the local existence theory gives the existence of a function $\psi(s)$, defining the coordinate $z(s)$ and writing all functions in terms of $z$ will reproduce the desired asymptotics (\ref{eq:app16.6}).

We are now in a position to state results concerning the local and global ODE analysis of (\ref{eq:app17})-(\ref{eq:app18}).

\begin{prop}
\label{prop:appchrisintmain}
There exists an $s_* < \infty$, and a unique solution to the system (\ref{eq:app17})-(\ref{eq:app18}) with asymptotic initial condition (\ref{eq:app19}) on a parameter range $s \in (-\infty, s_*)$ such that the following statements hold:
\begin{enumerate}
    \item $\psi(s) < 0$, \ $\theta(s) > 0 $ on $(-\infty, s_*).$
    \item $\lim_{s\rightarrow s_*^-} \psi(s) = -\infty $, \ \ \ $\lim_{s\rightarrow s_*^-} \theta(s) = \frac{1}{k}$.
    \item $\psi(s), \ \theta(s)$ are bounded, smooth functions on $ [-\infty,s_0)$ for any fixed $s_0 < s_*$.
    \item In a neighborhood of $s = -\infty$, the solution admits the expansion
    \begin{equation}
        \label{eq:app19.5}
        \psi(s) = - e^s + O(e^{2s}),
    \end{equation}
    \begin{equation}
        \label{eq:app19.6}
        \theta(s) = \frac{k}{2}e^{s} + O(e^{2s}).
    \end{equation}
    \item There exists a nonzero constant $a_1$, and an explicit constant $c(k)$ such that as $s \rightarrow s_*^- $ the solution $(\psi, \theta)$ admits the expansion
    \begin{align}
        \label{eq:app20}
        \frac{1}{\psi(s)} &= (1-k^2)(s-s_*) + O((s-s_*)^2), \\ 
        \label{eq:app21}
        \theta(s) &= \frac{1}{k} + c(k)\frac{1}{\psi(s)} + a_1 (-(1-k^2)(s-s_*))^{\frac{k^2}{1-k^2}} + O((s-s_*)^{\frac{2k^2}{1-k^2}}).
    \end{align}
\end{enumerate}
\end{prop}

A first step in recovering $(\psi(z),\theta(z))$ is understanding the behavior of the change of coordinates $z(s)$. 
\begin{lem}
\label{lem:appchrisintmain}
The function $z(s): [-\infty, s_*) \rightarrow [-1, z(s_*))$ is smooth and increasing. Moreover $\lim_{s \rightarrow s_*}z(s) = 0$, and there exists a constant $c_1>0$ such that as $s \rightarrow s_*^-$,
\begin{equation}
    \label{eq:app22}
    |s-s_*| = (c_1+ o_{z \rightarrow 0}(1))|z|^{1-k^2}.
\end{equation}
\end{lem}
\begin{proof}
By definition (\ref{eq:app16.5}) of the coordinate transformation, it follows that $z(s)$ is increasing wherever $z\psi(z) \geq 0$. As $\psi > 0$ for all $s < s_*$, it suffices to check that $z(s) < 0$ holds. Suppose by way of contradiction that there exists an $-\infty < s_1 < s_*$ with $z(s_1) = 0$. Fix any $s_0 < s_1$ finite. Applying (\ref{eq:app16.5}) between $(s_0, z(s_0))$ and $(s_1,0)$ gives
\begin{equation*}
s_1 = s_0 + \int\limits_{z_0}^{0}\frac{1}{z'\psi(z')}dz',
\end{equation*}
Since we assume $s_1 < s_*$, it follows that $|\psi|(z) \lesssim 1$ for all $z \in [z_0, z(s_0)]$. The above integral then diverges, a contradiction.

It thus follows that $z(s) < 0$ on $(-\infty,s_*)$, and that the map $z(s)$ is strictly increasing. Smoothness follows by differentiating (\ref{eq:app16.5}) and using the smoothness of $\psi$. 

It remains to consider the behavior of $z(s)$ as $s \rightarrow s_*$ (or $z \rightarrow 0$). The definition (\ref{eq:app16.5}) implies 
\begin{equation*}
    \frac{dz}{ds} = z\psi(z),
\end{equation*}
or rewriting and inserting the asymptotics (\ref{eq:app20}) gives
\begin{align}
    \log |\frac{z}{z_0}| &= \int\limits_{s_0}^{s(z)}\psi(s')ds' \nonumber\\
    &= (1-k^2)^{-1}\log |\frac{s(z)-s_*}{s_0-s_*}| + O_{z \rightarrow 0}(1).
\end{align}
Re-arranging gives (\ref{eq:app22}), which moreover implies $\lim_{s \rightarrow s_*}z(s) = 0$.
\end{proof}

The breakdown of the solution at $z=0$ physically corresponds to the solution arriving at the past light cone of the singular point $(\hat{u},\hat{v}) = (0,0)$. The region $-1 \leq z < 0$ on which the solution is defined is termed the interior of the naked singularity.

In the following sections the global properties of this interior solution are explored, with special attention paid to the consequences of (\ref{eq:app20})-(\ref{eq:app21}).

\subsubsection{Regular coordinates and global bounds}
\label{appa:regularcoords}
The goal is now to derive quantitative bounds on all the double null unknowns. As mentioned above, the main difficulties lie in 1) understanding rates as a function of the more natural coordinate $z$, and 2) converting bounds on $(\psi, \theta)$ into bounds on the full set of double null unknowns. 

Away from $\{s = s_*\},$ or equivalently $\{z=0\}$, Proposition \ref{prop:appchrisintmain} asserts that $(\psi,\theta)$ are smooth functions of $s$ (and by Lemma \ref{lem:appchrisintmain}, as functions of $z$). On any compact $z$ subinterval of $[-1,0)$ bounds will directly follow. It is therefore natural to start by considering the behavior as $z \rightarrow 0$. 

The following lemma is a restatement of (\ref{eq:app20})-(\ref{eq:app21}) in terms of $z$, applying the change of variables (\ref{eq:app22}):
\begin{lem}
As $z \rightarrow 0^-$, the following expansions hold for nonzero constants $c_1, c_2$:
\begin{equation}
    \label{eq:app23}
    \frac{1}{\psi(z)} = (c_1+o_{z \rightarrow 0}(1))|z|^{1-k^2} + O(|z|^{2(1-k^2)}),
\end{equation}
\begin{equation}
    \label{eq:app24}
    \theta(z) = \frac{1}{k} - (c_2+o_{z \rightarrow 0}(1))|z|^{k^2} + O(|z|^{2k^2}).
\end{equation}
\end{lem}
An immediate consequence of these expansions is the following set of bounds for the double null unknowns as $z \rightarrow 0$.
\begin{lem}
\label{lem:appblowuprates}
In a neighborhood of $z=0$, the following asymptotics hold:
\begin{alignat}{2}
\label{eq:app25}
    \mathring{r} &\sim 1, \quad \mathring{\lambda} \sim |z|^{-k^2}, \quad  (-\mathring{\nu}) \sim 1,  \\[.6em]
    \label{eq:app27}
    |\mathring{\phi}| &\lesssim 1, \quad |\mathring{\phi}'| \sim |z|^{-k^2}, \quad  |\mathring{\phi}''| \sim |z|^{-1-k^2}, \\[.6em]
    \label{eq:app26}
    \mathring{\Omega}^2 &\sim |z|^{-k^2}, \quad  0 < \mathring{\mu} < 1.
\end{alignat}
Moreover, the following identities hold:
\begin{equation}
 \label{eq:app28}
    \mathring{\nu}(0) = -\mathring{r}(0),
\end{equation}
\begin{equation}
 \label{eq:app29}
    \mathring{\mu}(0) = \frac{k^2}{1+k^2}.
\end{equation}
\end{lem}
\begin{proof}
Begin with $\mathring{r}$, $\mathring{\lambda}$, and $\mathring{\nu}$. By definition of $\psi(z)$ and (\ref{eq:app23})-(\ref{eq:app24}), it follows that
\begin{equation*}
    \frac{1}{\psi(z)} = \frac{z \partial_z \mathring{r}}{\mathring{r}} \sim |z|^{1-k^2},
\end{equation*}
and therefore
\begin{equation}
    \label{eq:appatemp1}
    \partial_z \log \mathring{r} \sim |z|^{-k^2},
\end{equation}
Integrating this approximate equality from a reference point $z_0 < 0$ yields $\mathring{r} \sim 1$ for $z$ close to $0$. Inserting this result into (\ref{eq:appatemp1}) gives the stated bound on $\mathring{\lambda} = \partial_z \mathring{r}$. The algebraic relationship (\ref{eq:appalgrel1}) combined with the rates for $\mathring{\lambda}, \mathring{r}$ gives the rate for $\mathring{\nu}$. Equation (\ref{eq:app28}) also follows.

\vspace{5pt}
\noindent
Now we turn to $\mathring{\phi}'$. Unpacking definitions of $\theta, \psi$ yields 
\begin{equation*}
    \mathring{\phi}' = \frac{\theta}{z\psi} = \frac{1}{k}(c_1+o_{z\rightarrow 0 }(1))|z|^{-k^2}.
\end{equation*}
Integrating this bound forward in $z$ from some reference $z_0 < 0$ gives a bound on $\mathring{\phi}$. To estimate $\mathring{\phi}''$, use the equation
\begin{equation*}
    \mathring{\phi}'' = -\frac{\mathring{\lambda}}{\mathring{r}z}(z\mathring{\phi}' + k).
\end{equation*}
The term in parentheses can be estimated as $k + o_{z\rightarrow 0 }(1)$, and the preceding factor is $\approx |z|^{-1-k^2}$. Combining these statements gives the estimate for $\mathring{\phi}''$.

\vspace{5pt}
\noindent
To estimate $\mathring{\Omega}$, recall the algebraic identity 
$$\frac{1}{4}\mathring{\Omega}^2 = -\mathring{\nu}\mathring{\lambda} + \mathring{r}^2z(\mathring{\phi}')^2  + 2k z\lambda  \mathring{r} \mathring{\phi}'+ \mathring{r}\mathring{\lambda} k^2.$$
The individual terms are of order $|z|^{-k^2}, |z|^{1-2k^2}, |z|^{1-2k^2}, |z|^{-k^2}$ respectively. The $|z|^{-k^2}$ terms have the same sign, and therefore do not cancel. Moreover, $k^2 < 1$ implies that the leading order behavior is determined by the $|z|^{-k^2}$ terms, which gives the result.

\vspace{2pt}
\noindent
Noting $\mathring{r} = -\mathring{\nu} + O(|z|^{1-k^2})$, we in fact have the more precise statement
\begin{equation*}
    \frac{1}{4}\mathring{\Omega}^2 = -(1+k^2)\mathring{\nu}\mathring{\lambda} + O(|z|^{1-2k^2}).
\end{equation*}
We conclude
\begin{equation*}
    \mathring{\mu} = 1+ \frac{4\mathring{\nu}\mathring{\lambda}}{\mathring{\Omega}^2} = 1 - \frac{1}{1+k^2} + O(|z|^{1-3k^2}) = \frac{k^2}{1+k^2}+O(|z|^{1-3k^2}).
\end{equation*}
The range of $k$ considered here implies $1-3k^2 > 0$, giving the remaining statements of the lemma.
\end{proof}
The rates proved in the previous lemma for $\mathring{\lambda}, \mathring{\Omega}^2, \mathring{\phi}',$ and $\mathring{\phi}''$ imply that these quantities blow up as $z \rightarrow 0$. Recalling the definition of the coordinate $z$, the blowup implies that along these unknowns do not have regular limits as $\hat{v} \rightarrow 0$, even \textit{away} from the singular point. This blowup should be distinguished from blowup near $(\hat{u},\hat{v}) = (0,0)$, corresponding to the presence of the singularity.

This discussion suggests the self-similar double null coordinate system becomes irregular as $\hat{v} \rightarrow 0$. While it is possible to work with a restricted set of double null unknowns that do have regular limits (e.g. quantities tangential to $\{\hat{v}=0\}$, alongside quantities weighted by suitable powers of $\mathring{\Omega}^2$), it is technically easier here to change coordinates into a regular double null gauge. The price to pay will be the symmetry of $\hat{u}$ and $\hat{v}$ built into the definition of the self-similar coordinate $z$.

Define a renormalized, non-self-similar double null coordinate pair $(u,v)$ by
\begin{equation}
    (u,v) = (\hat{u}, -|\hat{v}|^{1-k^2}).
\end{equation}
In moving between coordinate systems the $u$ coordinate derivative is unchanged, whereas 
\begin{equation}
\label{eq:app30}
\frac{\partial}{\partial v} = \frac{|\hat{v}|^{k^2}}{1-k^2}  \frac{\partial}{\partial \hat{v}}.  
\end{equation}
The axis $\Gamma$ becomes the set $\{u = -|v|^{\frac{1}{1-k^2}} \},$ with generator $$T = \partial_u + (1-k^2)|v|^{\frac{-k^2}{1-k^2}}\partial_{v}.$$
The conformal Killing field $S = \hat{u}\partial_{\hat{u}} + \hat{v}\partial_{\hat{v}}$ becomes $u\partial_u + (1-k^2)v \partial_{v}.$

The effect of defining $v$ is to introduce additional factors of $\hat{v}^{k^2}$ into the definition of the $\hat{v}$ derivative, and thereby compensate for blowup as $\hat{v} \rightarrow 0$. As mentioned above, however, one disadvantage of the formalism is that the natural self-similar coordinate $z = -\frac{\hat{v}}{\hat{u}}$ is replaced by the asymmetric $z = \frac{|v|^p}{u},$ where $p = (1-k^2)^{-1}.$

The next lemma shows that in the renormalized coordinate system all quantities at the first derivative level of the solution have finite limits as $v \rightarrow 0$. Moreover, the results apply not just in a neighborhood of $\{z = 0\}$, but globally in the interior region. Recall the notation $\mathcal{Q}^{(in)}$ for the interior region of the spacetime.
\begin{lem}
\label{lem:appinttranslatedestimates}
The solution in $\mathcal{Q}^{(in)}$ has the following properties:
\begin{itemize}
    \item $r(u,v) \geq 0$, and $r(u,v) > 0$ in $\mathcal{Q}^{(in)} \setminus \Gamma$.
    \item The renormalized coordinate derivatives $\nu \doteq \partial_u r, \ \lambda \doteq \partial_v r$ satisfy 
    \begin{equation}
        (-\nu) \sim 1, \ \ \ \lambda \sim |u|^{k^2}.
    \end{equation}
    \vspace{-.5em}
    \item Define the set $\mathcal{S}_{near} \doteq \mathcal{Q}^{(in)}\cap \{\frac{|v|}{|u|^{1-k^2}} \leq \frac{1}{2}\}$. Then 
    \vspace{-.5em}
    \begin{equation}
    r \lesssim |u| \ \text{in} \ \mathcal{Q}^{(in)}, \ \ \ r \sim |u| \ \text{in} \ \mathcal{S}_{near}.
\end{equation}
\vspace{-2em}
    \item There exists a $c_\mu < 1$ such that $\mu$ satisfies $0 \leq \mu \leq c_\mu < 1$. Moreover, the axis is regular in the following sense:
    \begin{equation}
        |\frac{\mu}{r^2}| \lesssim \frac{1}{|u|^2}.
    \end{equation}
    \vspace{-2em}
    \item The scalar field satisfies modified self-similar bounds 
    \begin{equation}
        |\partial_u \phi| \lesssim \frac{1}{|u|}, \ \ \ |\partial_v \phi| \lesssim \frac{1}{|u|^{1-k^2}}.
    \end{equation}
    \vspace{-2em}
    \item Along $\{v=0\}$, we have the identity
\begin{equation}
    \label{appa:temp2}
    \partial_u \phi(u,0) = \frac{k}{|u|}.
\end{equation}
\end{itemize}
\end{lem}
\begin{proof}
Start with $r = -\tilde{u} \mathring{r}$, which is a coordinate independent quantity. Let $\mathring{\lambda} = \partial_{\mathring{v}} r$ be associated to the self-similar coordinate system. To show non-negativity of $r$ it suffices to show $\mathring{\lambda}$ is always non-negative. In fact, suppose $\mathring{\lambda}(z_0) = 0$ for some $z \in (-1,0)$. The endpoints may be ruled out by considering initial conditions along $z = -1$, and the blowup $\lim_{z \rightarrow 0}\mathring{\lambda}(z) = \infty$.

Without loss of generality, choose $z_0$ to be the first point in $(-1,0)$ at which $\mathring{\lambda}(z_0)$ vanishes. Therefore, for all $z < z_0$ we have $\mathring{\lambda} > 0$. But then it follows that $\mathring{r}$ is a strictly increasing function on $[-1,z_0),$ and $\mathring{r}(z_0) > 0$. If $\mathring{\lambda}(z_0)=0$, then $\psi(z_0) = \frac{\mathring{r}(z_0)}{z_0\mathring{\lambda}(0)} = -\infty$, contradicting the statement that $|\psi|(z) < \infty$ for all $z < 0$. 

This argument shows $\mathring{\lambda}(z) > 0$ for all $z \in [-1,0]$. Compactness of the interval implies $\mathring{\lambda}$ has a positive lower bound. This in turn implies $\mathring{r}(z) >0$ for $z > -1$, and $r = -u\mathring{r} > 0$ for $(u,v)$ not contained in the axis $\Gamma$.

To see the bound on the renormalized $\lambda$, note $\lambda = p|\hat{v}|^{k^2}\mathring{\lambda}.$ It is helpful to divide $\mathcal{Q}^{(in)}$ into the regions $\mathcal{S}_{near}$ and $\mathcal{S}_{far} \doteq \mathcal{Q}^{(in)} \setminus \mathcal{S}_{near}.$ In $\mathcal{S}_{far}$ we estimate $\mathring{\lambda} \sim 1$, and therefore $\lambda \sim |\hat{v}|^{k^2} \sim |u|^{k^2}$. Here we used the equivalence $u \sim \hat{v} \sim -|v|^p$, which holds in $\mathcal{S}_{far}$. In $\mathcal{S}_{near}$ one estimates $$\lambda \sim |\hat{v}|^{k^2}\mathring{\lambda} \sim  |\hat{v}|^{k^2} \Big(\frac{|u|}{|\hat{v}|}\Big)^{k^2} \sim |u|^{k^2}.$$

\vspace{5pt}
\noindent
To estimate $\Omega$, start with an interval $[-1,z_0] \subset [-1,0]$ with $z_0 < 0$. Note $\log \mathring{\Omega}$ satisfies the equation
\begin{equation*}
    \frac{d}{dz}\log \mathring{\Omega} = \frac{1}{2\mathring{\lambda}}\Big(-\frac{\mathring{\nu}\mathring{\lambda}+\frac{1}{4}\mathring{\Omega}^2}{\mathring{r}z}+\mathring{r}(\mathring{\phi}')^2\Big).
\end{equation*}
The algebraic identity (\ref{eq:appalgrel2}) implies $\mathring{\Omega}$ remains finite on $[-1,z_0]$, and the positive lower bound on $\mathring{\lambda}$ implies one may integrate the above and bound $\log \mathring{\Omega}$ on the interval. Therefore $\mathring{\Omega}$ obeys positive upper and lower bounds on $[-1,z_0]$. 

The lapse in renormalized coordinates satisfies $$\Omega^2(u,v) = g(\partial_u, \partial_v) = p|\hat{v}|^{k^2}g(\partial_u, \partial_{\hat{v}}) = p|v|^{pk^2}\hat{\Omega}^2(u,v).$$
For $z \in [-1,z_0]$, on which we have shown $\mathring{\Omega}^2 \sim 1$, it follows that 
\begin{equation*}
    \Omega^2 \sim v^{pk^2} \mathring{\Omega}^2 \sim |u|^{k^2}.
\end{equation*}

It remains to consider a neighborhood of $z=0$. But we have shown $\mathring{\Omega}^2 \sim |z|^{-k^2}$, and so $\Omega^2 \sim v^{pk^2}|z|^{-k^2} \sim |u|^{k^2}.$

The quantity $\nu$ is unaffected by the coordinate transformation, and therefore retains its good bound near $\{z=0\}$. In a neighborhood of the axis, the bound $(-\nu) \sim 1$ also follows. It remains to consider $\mathcal{S}_{far} \setminus \Gamma$. Positivity of $\mathring{r}$ and negativity of $\psi$ implies $\mathring{r} + \mathring{\nu} = \mathring{r}\psi^{-1}$ must be negative, and thus $\mathring{\nu} < - \mathring{r} \lesssim -1$.  If $\mathring{\nu}$ did not satisfy a negative lower bound on this region, it would follow from the boundedness of $\mathring{r}$ that $\psi(z) = 0$ for some $z \in (-1,0)$, a contradiction. Therefore $(-\nu) \sim 1$ must hold.

\vspace{5pt}
To conclude a bound on $\mu$, note the positive lower bounds on $\lambda, -\nu$, as well as the bound $\Omega^2 \sim 1,$ imply via (\ref{eq:defnofm}) that
\begin{equation*}
    \mu \leq c_\mu < 1,
\end{equation*}
for some constant $c_\mu$. To see $\mu \geq 0$, note that along the axis $m=0$ holds. The equation (\ref{SSESF:2:5}) implies $m(u,v) \geq 0$ in $\mathcal{Q}^{(in)}$, and using the non-negativity of $r$ the result follows.

In addition to boundedness of $\mu$, we wish to show regularity at the axis consistent with a $C^1$ solution. One manifestation of this regularity is the statement $\mu = O(r^2)$ as $r \rightarrow 0$, which we now prove. Write (\ref{eq:defnofm}) as 
\begin{equation*}
    \mu = 4\Omega^{-2}(\frac{1}{4}\Omega^2 + \lambda \nu) = 4\mathring{\Omega}^{-2}(\frac{1}{4}\mathring{\Omega}^2 + \mathring{\lambda} \mathring{\nu}),
\end{equation*}
where we have used that $\mu$ is a coordinate independent quantity to write it in terms of the self-similar coordinate system. Dividing by $r^2$ and using the identity (\ref{eq:appalgrel2}) yields
\begin{equation*}
    \frac{\mu}{r^2} = |u|^{-2}\frac{\mu}{\mathring{r}^2} = 4|u|^{-2}\mathring{\Omega}^{-2}(z(\mathring{\phi}')^2 + 2k\mathring{\lambda}\frac{z\mathring{\phi}'+\frac{k}{2}}{\mathring{r}}).
\end{equation*}
The result near the axis follows provided one can bound $\frac{z\mathring{\phi}'+\frac{k}{2}}{\mathring{r}}$ near $\Gamma$. But 
\begin{equation*}
    |\frac{z\mathring{\phi}'+\frac{k}{2}}{\mathring{r}}| \lesssim \frac{|\frac{\theta}{\psi}+\frac{k}{2}|}{|\psi|},
\end{equation*}
and (\ref{eq:app19.5})-(\ref{eq:app19.6}) imply that this ratio is bounded. The boundedness of $\frac{\mu}{r^2}$ has not yet been shown near $\{v=0\}$. However, once the relation $r \sim |u|$ is shown in $\mathcal{S}_{near},$ the bound on $\frac{\mu}{r^2}$ will directly follow from a bound on $\mu$. The quantity $\frac{\mu}{r^2}$ is only an additional measure of regularity near the axis.

To see the remaining bounds on $r$, recall the expression $r = |u|\mathring{r}$. Global bounds have already shown $\mathring{r} \lesssim 1$ on $\mathcal{Q}^{(in)},$ and therefore $r \lesssim |u|$ holds in the interior region. In $\mathcal{S}_{near}$ we additionally have $\mathring{r} \gtrsim 1$, from which the bound $r \gtrsim |u|$ follows.

\vspace{5pt}
\noindent
Finally we consider the scalar field. The key will be to first consider the behavior of $\partial_z \mathring{\phi}$, and then use the coordinate vector fields (\ref{eq:app4}) to extract information on $u, v$ derivatives. Recall $z\mathring{\phi}' = \frac{\theta}{\psi}.$ Applying the asymptotics (\ref{eq:app19.5})-(\ref{eq:app21}) shows that 
\begin{equation*}
    |z\mathring{\phi}'| \lesssim 1 \  \text{in} \ \mathcal{S}_{far},\ \text{and} \ |z\mathring{\phi}'| \sim |z|^{1-k^2}\  \text{in} \ \mathcal{S}_{near}.
\end{equation*}
Boundedness of $z\mathring{\phi}'$ near $\{z = -1\}$ is equivalent to boundedness of $\mathring{\phi}'$, and applying (\ref{eq:app4}), (\ref{eq:app30}) yields in $\mathcal{S}_{far}$,
\begin{equation*}
    |\partial_u \phi| \lesssim \frac{|z\mathring{\phi}'|}{|u|}+\frac{k}{|u|} \lesssim \frac{1}{|u|}, \ \ \ |\partial_v \phi| \lesssim |\hat{v}|^{k^2}\frac{|\mathring{\phi}'|}{|u|} \lesssim \frac{1}{|u|^{1-k^2}}.
\end{equation*}
In both inequalities the relations $u \sim \hat{v}$, or $|z| \sim 1$ were used. Also recall that $\phi = \mathring{\phi} - k\log(-u),$ which accounts for the terms in $\partial_u \phi$ not arising from $\mathring{\phi}$. In $\mathcal{S}_{near}$ one computes
\begin{equation*}
    |\partial_u \phi| \lesssim \frac{|z|}{|u|} |z|^{-k^2} + \frac{k}{|u|} \lesssim \frac{v}{|u|} + \frac{k}{|u|} \lesssim \frac{1}{|u|},
\end{equation*}
and 
\begin{equation*}
    |\partial_v \phi| \lesssim |\hat{v}|^{k^2}\frac{1}{|u|} |z|^{-k^2} \lesssim \frac{1}{|u|^{1-k^2}},
\end{equation*}
as desired. Note that along $z=0$, the only term surviving in the expansion of $\partial_u \phi$ is precisely $\frac{k}{|u|},$ giving (\ref{appa:temp2}).
\end{proof}

From the bounds of the previous lemma, one can immediately read off a breakdown in regularity at $(u,v) = (0,0)$. The equality (\ref{appa:temp2}) along $\{v=0\}$ shows $\partial_u \phi$ is not integrable in $u$, implying there can be no BV extension of the solution in a neighborhood of $(0,0)$, for which $(0,0)$ is a regular center. Note the strong scalar field growth is directly related to the parameter $k \neq 0$.

It remains to verify that the solution is BV to the past of $(0,0)$, completing the picture of a loss of regularity. We therefore turn to estimating higher order derivatives of the solution in renormalized coordinates. In fact, we will find that the solution is more regular than BV in a self-similar neighborhood of the axis, a fact that will be important for solutions considered in the body of the paper.

In the following, recall the definition $p = (1-k^2)^{-1}.$
\begin{lem}
The following higher order estimates hold in $\mathcal{Q}^{(in)}$, with all quantities defined with respect to renormalized double null coordinates.
\vspace{-.5em}
\begin{align}
   |\partial_u \nu| &\lesssim \frac{1}{|u|}, \ \ \   |\partial_v \lambda| \lesssim \frac{1}{|u|^{1-2k^2}}, \label{eq:intexttemp1}\\[.5em]
   |\partial_u^2 \phi| &\lesssim \frac{1}{|u|^2}, \ \ \  |\partial_v^2 \phi| \lesssim \frac{1}{|u||v|^{1-pk^2}}, \label{eq:intexttemp2}\\[.5em]
   |\partial_u ^2 \nu| &\lesssim \frac{1}{|u|^{2}}, \ \ \  |\partial_u^3 \phi| \lesssim \frac{1}{|u|^3}.
\end{align}
In $\mathcal{S}_{far}$, we have
\begin{equation}
    |\partial_v^2 \lambda| \lesssim \frac{1}{|u|^{2-3k^2}}, \ \ \ |\partial_v^3 \phi| \lesssim \frac{1}{|u|^{3-3k^2}}.
\end{equation}
In $\mathcal{S}_{near}$, the following lower bound holds:
\begin{equation}
\label{appa:temp3}
    |\partial_v^2 \phi| \gtrsim \frac{1}{|u||v|^{1-pk^2}}.
\end{equation}
\end{lem}
Note that despite our efforts to work with a regular double null coordinate system, the second order derivative $\partial_v^2 \phi$ still blows up as $v \rightarrow 0$. This blowup is more fundamental than that of $\partial_v \phi$ in the self-similar coordinates, and should be seen as a reflection of the inherently limited regularity of the solution. In particular, $\phi$ is only in the Hölder class $C^{1,pk^2}(\mathcal{Q}^{(in)})$.

The above lemma goes beyond bounds required for a statement of BV or $C^1$ regularity, additionally estimating various third derivatives of $r, \phi$. These bounds are needed for technical reasons in the proof of Theorem \ref{thm1:interior}. Roughly, they assert higher regularity for the quantities tangential to $\{v=0\}$ throughout $\mathcal{Q}^{(in)},$ and for remaining quantities, higher regularity is available near $\Gamma$. One could extract more precise rates for all higher order derivatives near $\{v=0\}$, although we do not pursue this question here.
\begin{proof}
Start with $\partial_u \nu$. Write $\partial_u \nu = \frac{z}{u}\partial_z \mathring{\nu}$ and insert (\ref{eq:app13}) to give
\begin{equation*}
    |\partial_u \nu| \lesssim \frac{|z|}{\mathring{r}|u|}|\mathring{\nu}\mathring{\lambda}+\frac{1}{4}\mathring{\Omega}^2| \lesssim |z|\mathring{\Omega}^2 \frac{\mu}{r}.
\end{equation*}
One can further estimate $|z|\mathring{\Omega}^2 \lesssim 1$ and $\frac{\mu}{r} \lesssim |u|^{-1},$ which implies the desired bound on $\partial_u \nu$. 

The bound on $\partial_u^2 \phi$ follows similarly, namely by writing the derivatives in terms of coordinate $\partial_z$ derivatives, and estimating (\ref{eq:app16}). See also the case of $\partial_v^2\phi$ below for a related calculation.

To estimate $\partial_v \lambda$, write 
\begin{align*}
   (1-k^2)^2 \partial_v \lambda &= |\hat{v}|^{k^2}\partial_{\hat{v}}( |\hat{v}|^{k^2} \mathring{\lambda})\\
    &= |\hat{v}|^{2k^2}\frac{1}{|u|}\partial_z \mathring{\lambda} - k^2|\hat{v}|^{2k^2-1}\mathring{\lambda} \\
    &= |\hat{v}|^{2k^2}\frac{1}{r z}(\mathring{\nu}\mathring{\lambda}+\frac{1}{4}\mathring{\Omega}^2) - k^2|\hat{v}|^{2k^2-1}\mathring{\lambda}\\
    &=  \frac{|\hat{v}|^{2k^2}}{|u|}(\mathring{r}(\mathring{\phi}')^2 + 2k\mathring{\lambda}\mathring{\phi}' + z^{-1}\mathring{\lambda}k^2)- k^2|\hat{v}|^{2k^2-1}\mathring{\lambda}\\
    &= \frac{|\hat{v}|^{2k^2}}{|u|}(\mathring{r}(\mathring{\phi}')^2 + 2k\mathring{\lambda}\mathring{\phi}').
\end{align*}
Note that the terms with the most singular behavior in $v$ dropped out. Estimating $\mathring{r} \lesssim 1,$ $\mathring{\phi}' \lesssim |z|^{-k^2},$ $\mathring{\lambda} \lesssim |z|^{-k^2}$ gives 
\begin{align*}
    |\partial_v \lambda| \lesssim \frac{|\hat{v}|^{2k^2}}{|u|}\Big(\frac{|\hat{v}|}{|u|} \Big)^{-2k^2} \lesssim \frac{1}{|u|^{1-2k^2}}.
\end{align*}
Next turn to $\partial_v^2 \phi$. It will again be necessary to track a cancellation in the most singular $v$ terms. Compute 
\begin{align*}
   (1-k^2)^2 \partial_v^2 \phi 
    &= \frac{|\hat{v}|^{2k^2}}{|u|^2}\partial_z \mathring{\phi}' - k^2\frac{|\hat{v}|^{2k^2-1}}{|u|}\mathring{\phi}' \\
    &=-\frac{|\hat{v}|^{2k^2}}{\mathring{r}z|u|^2}(2\mathring{\lambda}z\mathring{\phi}' + k\mathring{\lambda} )- k^2\frac{|\hat{v}|^{2k^2-1}}{|u|}\mathring{\phi}'. 
\end{align*}
We proceed differently self-similar neighborhoods of $\Gamma$ and of $\{v=0\}$.
Near $\Gamma$ the right hand side can be estimated in magnitude by $|u|^{2k^2-2}.$  Bounding the first term requires a bound on $\mathring{r}^{-1}(z\mathring{\phi}' + \frac{k}{2}),$ already discussed in the proof of Lemma \ref{lem:appinttranslatedestimates}. Near $\{v=0\}$, rewrite the final line in the chain of equalities as 
\begin{equation}
    \label{eq:apptemp3}
    -2\frac{|\hat{v}|^{2k^2}}{\mathring{r}|u|^2}\mathring{\lambda}\mathring{\phi}' + k\frac{|\hat{v}|^{2k^2-1}}{|u|}(\frac{\mathring{\lambda}}{\mathring{r}} - k\mathring{\phi}' ),
\end{equation}
and note that the first term is bounded by $|u|^{2k^2-2}$. Rewrite 
\begin{equation*}
    \frac{\mathring{\lambda}}{\mathring{r}} - k\mathring{\phi}' = \frac{k}{z\psi}(\frac{1}{k} -\theta) = c + o_{z\rightarrow 0}(1),
\end{equation*}
where we have inserted the asymptotics (\ref{eq:app23})-(\ref{eq:app24}). The constant $c$ is nonzero precisely because $a_1 \neq 0$, where $a_1$ is the coefficient in (\ref{eq:app21}). This implies the second term in (\ref{eq:apptemp3}) is bounded in magnitude near $\mathcal{S}_{near}$ by 
$$\frac{1}{|u||\hat{v}|^{1-2k^2}} \sim \frac{1}{|u||v|^{1-pk^2}}. $$
Taken together, these estimates yield the desired statement for $\partial_v^2 \phi$. In $\mathcal{S}_{near}$ the term exhibiting blowup in $v$ dominates, giving the lower bound (\ref{appa:temp3}).

\vspace{5pt}
\noindent
We finally turn to the third order bounds on $r, \phi$. Using (\ref{eq:app12})-(\ref{eq:app16}), we collect bounds on various $z$ derivatives here:
\begin{align*}
    |\partial_z \mathring{\nu}| &\lesssim \mathring{r}|z|^{-k^2}, \  \quad \quad |\partial_z \mathring{\lambda}| \lesssim \mathring{r}|z|^{-1-k^2}, \\[.3em]
    |\partial_z^2 \mathring{\nu}| &\lesssim |z|^{-1-k^2}, \ \ \ \quad |\partial_z^2 \mathring{\lambda}| \lesssim |z|^{-2-k^2}, \\[.3em]
    |\partial_z \mathring{\Omega}| &\lesssim \mathring{r} |z|^{-1-\frac{1}{2}k^2}, \ \ \ |\partial_z \mathring{\phi}'| \lesssim |z|^{-1-k^2}, \ \ \ 
    |\partial_z^2 \mathring{\phi}'| \lesssim |z|^{-2-k^2}.
\end{align*}
Some care is required to take higher $z$ derivatives of the solution near $\Gamma$, due to potential singular factors of $\mathring{r}^{-1}.$ General smoothness considerations for solutions of ODEs of this type near critical points, however, imply that these higher derivatives must remain bounded in a neighborhood of $\Gamma.$

It is now a simple matter to translate $z$ derivatives into the appropriate $u,v$ derivatives. Recall $\partial_u = \partial_{\hat{u}},$ and compute
\begin{align*}
    \partial_{\hat{u}}^2 \nu &= \frac{2z}{\tilde{u}^2}\partial_z \mathring{\nu} + \frac{z^2}{\tilde{u}^2}\partial_z^2 \nu,
\end{align*}
and thus 
\begin{equation*}
    |\partial_u^2 \nu| \lesssim \frac{1}{|u|^2}.
\end{equation*}
It was important that the net power of $z$ in each term is non-negative, and thus is bounded as $v \rightarrow 0$. We estimate $|z|^{a}\lesssim 1$ for any $a \geq 0$, and the bound then reduces to tracking powers of $|u|$. Similarly,
\begin{align*}
    \partial_u^3 \phi &= \partial_u^3 \mathring{\phi} + \frac{2k}{|u|^3},
\end{align*}
and expanding the derivative gives $|\partial_u^3 \mathring{\phi}| \lesssim |u|^{-3}$. The stated estimate follows. 

For $\partial_v^2 \lambda, \partial_v^3 \phi$ we only require an estimate in $\mathcal{S}_{far}$, where $z \sim -1$ and $u \sim \hat{v}$. We proceed as above, inserting the definition of the coordinate derivative $\partial_v$ and expanding. Here, however, we drop all powers of $z$, and only track the $|u|$ and $v$ weights appearing in the definition of the coordinate derivative. For example,
\begin{align*}
    (1-k^2)^3 \partial_v^2 \lambda &= |\hat{v}|^{k^2}\partial_{\hat{v}}(|\hat{v}|^{2k^2}\frac{1}{|u|}\partial_z \mathring{\lambda} -k^2 |\hat{v}|^{2k^2-1}\mathring{\lambda} ) \\
    &= |\hat{v}|^{3k^2}\frac{1}{|u|^2}\partial_z^2 \mathring{\lambda} -3k^2 |\hat{v}|^{3k^2-1}\frac{1}{|u|}\partial_z \mathring{\lambda} +k^2(2k^2-1)|\hat{v}|^{3k^2-2}\mathring{\lambda},
\end{align*}
where each term contributes a total weight $|u|^{-2+3k^2}$.

\end{proof}

\subsection{The exterior solution}
\label{subsec:appAext}
Proposition \ref{prop:appchrisintmain} implies that the breakdown of the interior solution as $s\rightarrow s_*$ is driven by the blowup $\psi(s) \rightarrow -\infty$. Equivalently, $\lim_{s\rightarrow s_*}\frac{1}{\psi}(s) \rightarrow 0^{-}$. This section concerns the existence of extensions into $z > 0$, i.e. solutions $(\psi(s),\theta(s))$ to (\ref{eq:app17})-(\ref{eq:app18}) with $\frac{1}{\psi}(s) > 0$. Such extensions are not unique, and for each $k$ a $1$-parameter family of extensions exists. While \cite{chris2} studies the whole set of possible extensions, we focus attention on those that are consistent with naked singularity formation. Recall that we limit to $k^2 \in (0, \frac{1}{3})$.

\begin{prop}
\label{prop:appchrisintmain2}
There exists a solution to the system (\ref{eq:app17})-(\ref{eq:app18}) on a parameter range $s \in (s_*, \infty)$ such that the following statements hold:
\begin{enumerate}
    \item $1 <\psi(s) < \infty,$ on $(s_*, \infty).$
    \item $\lim_{s\rightarrow \infty}\psi(s) = 1,$ \ \ \ $\lim_{s\rightarrow \infty}\theta(s) = -k.$
    \item For any compact subinterval $I \subset (s_*, \infty)$, the solution $(\psi, \theta)$ is bounded and smooth on $I.$
    \item There exists a nonzero $b_1$ and an explicit constant $c(k)$ such that as $s \rightarrow s_*^+$, the solution admits the expansion
    \begin{align}
        \label{eq:app31}
        \frac{1}{\psi}(s) &= (1-k^2)(s-s_*)+ O((s-s_*)^2), \\[.3em]
        \label{eq:app32}
        \theta(s) &= \frac{1}{k} + c(k)\frac{1}{\psi(s)} + b_1((1-k^2)(s-s_*))^{\frac{k^2}{1-k^2}} + O((s-s_*)^{\frac{2k^2}{1-k^2}}).
    \end{align}
    \item In a neighborhood of $s=\infty$, there are nonzero constants $b_2, b_3$ such that the solution admits the expansion
    \begin{align}
        \label{eq:app33}
        \psi(s) &= 1 + b_2e^{-(1-k^2)s} + O(e^{-2(1-k^2)s}), \\[.3em] 
        \label{eq:app34}
        \theta(s) &= -k + \Big(\frac{1+k^2}{k}\Big)b_2e^{-(1-k^2)s} +  b_3e^{-s} + O(e^{-2s}).
    \end{align}
\end{enumerate}
\end{prop}

As in the interior, the first step to understanding global behavior is studying the change of coordinates $z(s)$. The asymptotics near $\{z=0\}$ in the exterior match those in the interior (compare (\ref{eq:app31})-(\ref{eq:app32}) with (\ref{eq:app20})-(\ref{eq:app21})). In the exterior $z \psi \geq 0$ continues to hold, and therefore the argument of Lemma \ref{lem:appchrisintmain} implies $z(s)$ is increasing, with $z \sim (s-s_*)^{1-k^2}$ in a neighborhood of $\{z=0\}$. 

It remains to consider the asymptotic regime $s \rightarrow \infty$. Integrating $\frac{dz}{ds} = z\psi(s)$ from a reference point $(s_0, z_0)$ with $s_0 > s_*, \ $ $z_0 > 0$, and inserting the asymptotic (\ref{eq:app33}) implies
\begin{align}
    \log |\frac{z}{z_0}| &= \int\limits_{s_0}^{s(z)}\psi(s')ds' \nonumber\\
    &= s(z) -C(s_0) + o_{s \rightarrow \infty}(1),\nonumber
\end{align}
where $C(s_0)$ is a constant depending on the choice of reference point. Re-arranging this expression gives that in a neighborhood of $\{s=\infty\}$,
\begin{equation}
    \label{eq:app35}
    z(s) \sim e^{s}.
\end{equation}
In particular, the limit $s \rightarrow \infty$ corresponds to $z \rightarrow \infty$.
\subsubsection*{Global bounds}
A direct consequence of (\ref{eq:app35}) is the following restatement of (\ref{eq:app33})-(\ref{eq:app34}):
\begin{lem}
There exists a nonzero constant $d_1$ such that the following expansion holds as $z \rightarrow \infty$:
\begin{align}
    \label{eq:app35.5}
    \psi(z) &= 1 + d_1|z|^{-(1-k^2)} + O(|z|^{-2(2-k^2)}), \\[.3em]
    \label{eq:app36}
    \theta(z) &= -k + d_1\Big(\frac{1+k^2}{k}\Big)|z|^{-(1-k^2)} + O(|z|^{-1}).
\end{align}
\end{lem}
The first result giving quantitative bounds on the exterior focuses on a  neighborhood (in $z$) of $\{v=0\}$. The behavior of the solution is identical to that of the interior region in the corresponding neighborhood. We record the result below.
\begin{lem}
\label{app:extmainestimates1}
In the region $\mathcal{Q}^{(ex)} \cap \{z \leq 1\}$ the results of Lemma \ref{lem:appinttranslatedestimates} continue to hold, as well as the bounds (\ref{eq:intexttemp1})-(\ref{eq:intexttemp2}).
\end{lem}

For the remaining bounds we transform to a different coordinate system, one adapted to the $\{u=0\}$ hypersurface. The transformation is analogous to that from the self-similar $(\hat{u}, \hat{v})$ system to the renormalized $(u,v)$ system. Effectively we are performing the same renormalization procedure with the roles of $u$ and $v$ interchanged. Define
\begin{equation}
    (U,V) = (-|\hat{u}|^{1-k^2}, \hat{v}),
\end{equation}
with the transformed $U$ coordinate derivative given by
\begin{equation}
    \frac{\partial}{\partial U} = \frac{|\hat{u}|^{k^2}}{1-k^2}\frac{\partial}{\partial \hat{u}}.
\end{equation}
The coordinate change transforms the horizons $\{v=0\}$, $\{u=0\}$ to $\{V=0\},$ $\{U=0\}$ respectively. 

The following estimates are stated with respect to the $(U,V)$ gauge. To avoid confusion we denote by ${^{(V)}\lambda,} \ {^{(U)}\nu,} \ {^{(U,V)}\Omega^2}$ the values of the gauge-dependent functions in the new coordinate system. We now proceed to study the solution in a neighborhood of $\{U=0\}$.

\begin{lem}
\label{app:extmainestimates2}
The solution in the region $\mathcal{Q}^{(ex)}\cap \{1 \leq z < \infty\}$ has the following properties:
\begin{itemize}
    \item $r(U,V) \geq 0, \ \text{and} \ \ r(U,V)=0 \iff (U,V) = (0,0)$. Moreover, $r \sim V$.
    \item The metric quantities satisfy the bounds
    \begin{equation}
        (-^{(U)}\nu) \sim V^{k^2}, \ \ \ ^{(V)}\lambda \sim 1, \ \ \  ^{(U),(V)}\Omega^2 \sim V^{k^2}.
    \end{equation}
    \begin{equation}
         |{\partial_U} {^{(U)}\nu}| \lesssim V^{-1+2k^2}, \ \ \ |{\partial_V}{^{(V)}\lambda}| \lesssim V^{-1}.
    \end{equation}
    \item There exists a $c_\mu < 1$ such that $\mu$ satisfies $0 \leq \mu \leq c_\mu < 1$. Moreover, $\mu \sim 1$.
    \item The scalar field satisfies self-similar bounds 
    \begin{equation}
        |\partial_U \phi| \lesssim V^{-1+k^2}, \ \ \quad \quad \quad \quad \quad \quad |\partial_V \phi| \lesssim V^{-1},
    \end{equation}
    \begin{equation}
        |\partial_U^2 \phi| \lesssim V^{-1+2k^2}|U|^{-(1-pk^2)}, \ \ \ \ \ |\partial_V^2 \phi| \lesssim V^{-2}.
    \end{equation}
\end{itemize}
Moreover, ${r,}{^{(U)}\nu,}{^{(V)}\lambda,} \ \mu, \partial_U \phi, \partial_V \phi$ extend continuously to functions on $\{U=0\}$ for all $V > 0$, and there exists a positive constant $c_r$ such that the following relations hold:
\begin{align}
    r(0,V) &= c_r V, \\
    ^{(V)}\lambda(0,V) &= c_r, \\
    \mu(0,V) &= \frac{k^2}{1+k^2}, \\
    \partial_V \phi(0,V) &= -\frac{k}{V}.
\end{align}
\end{lem}
\begin{proof}
The proof is largely a computation using (\ref{eq:app35.5})-(\ref{eq:app36}), akin to the proof of Lemma \ref{lem:appinttranslatedestimates}. We will start by showing the stated bounds in a region $\{ C_{large} \leq z < \infty \}$, where $C_{large}$ is a large constant. In this region we employ the series expansions (\ref{eq:app35.5})-(\ref{eq:app36}).

We first estimate $r$. Using the relation $\partial_z \log \mathring{r} = (z\psi(z))^{-1}$ and inserting (\ref{eq:app35.5}) gives
\begin{align*}
    \partial_z \log \mathring{r} &= \frac{1}{z\psi(z)} = z^{-1} + O(z^{-(2-k^2)}).
\end{align*}
Integrating from a from a fixed $z_0 \geq C_{large}$ yields
\begin{align*}
    \log \frac{\mathring{r}(z)}{\mathring{r}_0} = \log \frac{z}{z_0} + C(z_0) + O(|z|^{-(1-k^2)}),
\end{align*}
where $r_0 = r(z_0)$. Rearranging gives $\mathring{r} = (c_r + o_{z \rightarrow \infty}(1))z,$ which yields the desired estimate after recalling $r = \mathring{r}|\hat{u}|$. With this estimate on $\mathring{r}$, we see 
$$\mathring{\lambda} = \partial_z \mathring{r} =  \frac{\mathring{r}}{z \psi(z)} \sim 1.$$ The stated estimate for $^{(V)}\lambda = \mathring{\lambda}$ then follows. 

The definition $\psi = \frac{\mathring{r}}{\mathring{r}+\mathring{\nu}}$ implies $\mathring{\nu} = \frac{\mathring{r}(1-\psi)}{\psi}$, and inserting (\ref{eq:app35.5}) gives $$\mathring{\nu} = (c + o_{z \rightarrow(\infty)}(1))|z|^{k^2},$$ 
for a nonzero constant $c$. We therefore have $^{(U)}\nu = p|\hat{u}|^{k^2}\mathring{\nu} \sim V^{k^2}.$

To see the behavior of $\mathring{\phi}'$, recall $\mathring{\phi}' = \frac{\theta(z)}{z\psi(z)}.$ Inserting the series expansions for $\theta(z), \psi(z)^{-1}$ gives
\begin{equation}
    \label{eqn:appAAseriesexp}
    z\mathring{\phi}' = -k - \frac{d_1}{k}z^{-(1-k^2)} + O(z^{-1}).
\end{equation}
Differentiating $\phi = \mathring{\phi} - k\log|\hat{u}|$ in $U$ and applying (\ref{eqn:appAAseriesexp}) gives
\begin{align*}
    (1-k^2)\partial_U \phi &= |\hat{u}|^{k^2}\partial_{\hat{u}} \phi \\
    &= \frac{k}{|\hat{u}|^{1-k^2}} + \frac{1}{|\hat{u}|^{1-k^2}}(z\mathring{\phi}') \\
    &= - \frac{d_1}{k} \frac{1}{V^{1-k^2}} + o_{z \rightarrow \infty}(1).
\end{align*}
Similarly,
\begin{align*}
    \partial_V \phi &= \partial_{\hat{v}}\phi = \frac{1}{|\hat{u}|}(-kz^{-1} + O(z^{-(2-k^2)})) \\
    &= -\frac{k}{V} + o_{z \rightarrow \infty}(1).
\end{align*}
Estimates on $\Omega$ naturally follow from (\ref{eq:appalgrel2}). Rewriting in terms of $\theta, \psi$ gives the pair of identities 
\begin{align}
    &\frac{1}{4}\mathring{\Omega}^2 + \mathring{\nu}\mathring{\lambda} = -k^2(1-\psi)z\mathring{\lambda}^2 +  z\mathring{\lambda}^2(\theta+k)^2 \label{eq:apptemp6},\\[.3em]
    &\frac{1}{4}\mathring{\Omega}^2 = -(1-\psi)z\mathring{\lambda}^2(1+k^2) + z\mathring{\lambda}^2(\theta+k)^2 \label{eq:apptemp5}.
\end{align}
Estimating (\ref{eq:apptemp5}) gives
\begin{equation*}
    \frac{1}{4}\mathring{\Omega}^2= (c + o_{z \rightarrow \infty}(1))|z|^{k^2},
\end{equation*}
for a nonzero constant $c$. This implies the required estimate after calculating $^{(U),(V)}\Omega \sim |\hat{u}|^{k^2}\mathring{\Omega} \sim V^{k^2}.$

We next estimate $\mu$. Recalling (\ref{eq:app8}), estimating (\ref{eq:apptemp6}), and applying the estimate on $\mathring{\Omega}^2$ just derived gives
\begin{equation*}
    \mathring{\mu} = (1+o_{z \rightarrow \infty}(1))\frac{k^2}{1+k^2}.
\end{equation*}

We finally turn to the derivatives ${\partial_U}{^{(U)}\nu}, {\partial_V} {^{(V)}\lambda}, \partial_U^2 \phi, \partial_V^2 \phi$, which are estimated using the system (\ref{eq:app12})-(\ref{eq:app16}). Start with $^{(V)}\lambda$. Estimate
\begin{align*}
    |\partial_z \mathring{\lambda}| & \lesssim \frac{1}{z^{2-k^2}},
\end{align*} 
and thus 
\begin{align*}
    |{\partial_{V}}{^{(V)}\lambda}| \lesssim \frac{1}{|\hat{u}|}|\partial_z \mathring{\lambda}|  \lesssim \frac{1}{V}.
\end{align*}
To see the estimate for ${\partial_U}{^{(U)}\nu}$ we need to track cancellations in leading order $U$ weights. Expand
\begin{align}
    (1-k^2)^2{\partial_U}{^{(U)}\nu} &= |\hat{u}|^{k^2}\partial_{\hat{u}}(|\hat{u}|^{k^2}\mathring{\nu}) \nonumber \\
    &= |\hat{u}|^{-1+2k^2}(z \partial_z \mathring{\nu} - k^2\mathring{\nu}), \label{eq:temptemptempRoger}
\end{align}
and using (\ref{eq:app14}) we estimate
\begin{align*}
    z\partial_z \mathring{\nu} - k^2 \mathring{\nu} &= -\frac{z}{\mathring{r}}(\frac{1}{4}\mathring{\Omega}^2 + \mathring{\nu}\mathring{\lambda}) - k^2\mathring{\nu} \\
    &= k^2\frac{z^2}{\mathring{r}}\mathring{\lambda}^2(1-\psi) - k^2\frac{\mathring{r}(1-\psi)}{\psi} + O(z^{-(1-k^2)}) \\
    &= O(\mathring{r} (1-\psi)^2) = O(z^{-1+2k^2}).
\end{align*}
Inserting into (\ref{eq:temptemptempRoger}) yields $|{\partial_U}{^{(U)} \nu}| \lesssim V^{-1+2k^2}$, as desired.

We will need an expansion for $\partial_z \mathring{\phi}'$. A computation using (\ref{eq:app16}) gives
\begin{align*}
    \partial_z \mathring{\phi}' = z^{-2}\big(k +d_1(\frac{2}{k}-k)z^{-(1-k^2)} + O(z^{-1})\big).
\end{align*}
We therefore compute
\begin{align*}
    \partial_V^2 \phi = -\frac{1}{|\hat{u}|^{2}}\partial_z \mathring{\phi}' = O(\frac{1}{V^2}).
\end{align*}
The sharp rate for $\partial_U^2 \phi$ will again require tracking cancellations. Compute
\begin{align*}
    (1-k^2)^2\partial_U^2 \phi &= |\hat{u}|^{k^2}\partial_{\hat{u}}\big(|\hat{u}|^{k^2}\partial_{\hat{u}} (\mathring{\phi} - k\log |\hat{u}|)\big) \\
    &= \frac{1}{|\hat{u}|^{2-2k^2}}\big((2-k^2)z\mathring{\phi}' + z^2\partial_z \mathring{\phi}' + k(1-k^2) \big) \\
    &= \frac{1}{|\hat{u}|^{2-2k^2}}\big((2-k^2)(-k-d_1\frac{1}{k}|z|^{-(1-k^2)}) + k + d_1(\frac{2}{k}-k)z^{-(1-k^2)} +k(1-k^2) + O(|z|^{-1})\big) \\
    &= O(\frac{1}{z|\hat{u}|^{2-2k^2}}).
\end{align*}
Rewriting in terms of $U,V$ gives the stated bound $\partial_U^2\phi = O(\frac{1}{V |U|^{1-pk^2}})$.

\vspace{10pt}
It remains to extend these bounds to the transition region $\mathcal{Q}^{(ex)} \cap \{1 \leq z \leq C_{large} \}$. This region covers a compact range in the self-similar coordinate, and it is easy to see that one has uniform boundedness of $\mathring{r}, \mathring{\lambda}, \mathring{\nu}, \mathring{\Omega}^2, \mathring{\phi}',$ and $z$ derivatives thereof. Translating boundedness in the self-similar coordinates to statements about the $(U,V)$ coordinates, we recover all the stated upper bounds. It remains then to show the lower bounds $\mathring{r}, -\mathring{\nu}, \mathring{\lambda}, \mathring{\Omega}, \mathring{\mu} \gtrsim 1. $

To estimate $\mathring{\lambda} \sim 1$, it suffices to show $\mathring{\lambda}(z)$ cannot vanish for $z \in [1, C_{large}]$. But this is immediate from $\mathring{\lambda}(z) = \frac{\mathring{r}(z)}{z\psi(z)}$. If $z_0$ is the first $z$ in the interval satisfying $\mathring{\lambda}(z_0) = 0$, then by $\partial_z \mathring{r} = \mathring{\lambda}$ one must have $\mathring{r}(z_0) \geq \mathring{r}(1) > 0$. But $\psi(z) < \infty$ on $[1,C_{large}]$, and so it cannot be the case that $\mathring{\lambda}(z_0) = \frac{\mathring{r}(z_0)}{z_0\psi(z_0)} = 0$. Therefore $\mathring{\lambda} \sim 1$.

From $\mathring{\lambda} \sim 1$ it follows that $\mathring{r} \sim 1$, and thus $r = \mathring{r}|\hat{u}| \sim |\hat{u}|$. Note in this domain $|\hat{u}| \sim \hat{v} \sim V$ with the constants only depending on $C_{large}$. Therefore $r \sim |\hat{u}| \sim V$.

Next turn to $\mathring{\nu}$. Note $\psi(z) > 1$ holds for $z \in [1,C_{large}]$, and hence by compactness there exists a positive constant $c_\psi$ such that $\psi-1 \geq c_\psi > 0$. Moreover, away from $\{z=0\}$ we have that $\psi$ is bounded, and hence $\psi \leq C_{\psi}$ holds for some large constant $C_{\psi}$. Using the relation $-\mathring{\nu} = \frac{\mathring{r}(\psi-1)}{\psi}$, it immediately follows that $(-\mathring{\nu})$ is bounded above and below, as desired.

An upper bound for $\Omega^2$ follows directly from the algebraic relation (\ref{eq:appalgrel2}). A lower bound follows by considering 
\begin{equation*}
    \frac{d}{dz}\log \mathring{\Omega} = \frac{1}{2\mathring{\lambda}}\Big(-\frac{\mathring{\nu}\mathring{\lambda}+\frac{1}{4}\mathring{\Omega}^2}{\mathring{r}z}+\mathring{r}(\mathring{\phi}')^2\Big).
\end{equation*}
Integrating gives a bound on $|\log \mathring{\Omega}|$. This gives a positive lower bound on $\mathring{\Omega}$, and therefore we conclude the desired bound on $^{(U),(V)}\Omega \sim |\hat{u}|^{k^2}\mathring{\Omega}.$

Finally we consider $\mu$. The bound $\mu \leq c_\mu < 1$ follows by (\ref{eq:app8}) and the bounds already shown on $\mathring{\lambda}, \mathring{\nu},$ and $\mathring{\Omega}$.

The bound $\mu \geq 0$ follows by integrating (\ref{SSESF:2:5}) in $v$ outwards from $\{v=0\}$, where $\mu \geq 0$ is known. One uses that (\ref{SSESF:2:5}), in combination with bounds on $1-\mu$ and $\lambda$, implies $\partial_v m \geq 0$. Therefore $m(u,v) \geq 0$, and $\mu = \frac{2m}{r} \geq 0$.
\end{proof}

\subsubsection{Asymptotically flat truncation}
\label{subsubsec:asympflattrunc}
We have concluded the construction of the exact $k$-self-similar solution, and in an appropriate renormalized gauge have shown modified self-similar bounds for all double null unknowns. Importantly, however, these spacetimes are not asymptotically flat. An easy way to see this is to consider a null hypersurface $\{U = c\}$ for $c \in [-1,0),$ and compute 
\begin{equation}
    \lim_{V \rightarrow \infty} m(c,V) = \lim_{V \rightarrow \infty} \frac{1}{2}\mathring{\mu}\left(|c|^{-1}V^{q}\right) r(c,V) = \infty,
\end{equation}
where we have used that along $\{U=c\}$ the bounds $\mathring{\mu} \sim 1 \ $, $r \sim V$ hold. Thus although the spacetime has the proper singularity structure near $\mathcal{O}$, it cannot serve as a model for naked singularities. The solution employed in \cite{chris2} is to consider an asymptotically flat truncation of the exactly $k$-self-similar solution, leading to an exterior that naturally decomposes into two regions: 1) for $0 \leq V \leq V_{trunc}$, where $V_{trunc}$ is a fixed positive constant, the solution is precisely the $k$-self-similar one constructed above, and 2) for $V_{trunc} \leq V < \infty$, the solution is asymptotically flat.

The choice of $V_{trunc}$, and the asymptotically flat data along $\{U = -1, \ V \geq V_{trunc}\}$, renders the construction of the asymptotically flat region non-unique. The precise structure of this region is unimportant for the analysis of the paper, and so we choose a simple truncation procedure. Without loss of generality, take $V_{trunc}=1$.

Introduce a smooth, non-negative cutoff function $\chi(V)$ with the property that $\chi(V) = 1$ for $V \in [0,1]$, $\chi(V) = 0$ for $V \in [2,\infty)$, $\chi \leq 1$, and $|\partial_V^i \chi| \lesssim 1$  for all $i \in \mathbb{N}$. In the following, unknowns with the subscript $\Psi_s(U,V)$  will denote the values in the exactly self-similar solution. Quantities without a subscript refer to the solution being constructed.

Let $U_{min} < 0 $ be a small constant, which will be chosen in the proof. The asymptotically flat region will be a solution to the scalar field system on the truncated domain \begin{equation}
    \mathcal{Q}^{(trunc)} \doteq \mathcal{Q} \cap \{1 \leq V < \infty, \  U_{min} \leq U < 0\},
\end{equation}
with data posed along the null hypersurfaces $\{U=U_{min}, \ V \geq 1\}\cup\{U \geq U_{min}, \ V = 1\}$. More precisely, pose data
\begin{align}
    \phi(U,1) &= \phi_s(U,1), \label{app:trunc1}\\[.3em]
    ^{(U)}\nu(U,1)\ &{=} \ {^{(U)}\nu_s(U,1)}, \label{app:trunc2}\\[.3em]
    \phi(U_{min},V) &= \chi(V)\phi_s(U_{min},V), \label{app:trunc3}\\[.3em]
    ^{(V)}\lambda(U_{min},V) &= \chi(V)\lambda_s(U_{min},V) + \frac{1}{2}(1-\chi(V)).\label{app:trunc4}
\end{align}
It follows from the self-similar bounds of the previous section that the norm 
\begin{equation}
    \|r\|_{C^2(\{U=U_{min}, \ V \geq 1\})} +  \|\phi\|_{C^2(\{U=U_{min}, \ V \geq 1\})}
\end{equation}
of the outgoing data is bounded, independently of $U_{min}.$ The aim is now to show that for $U_{min}$ sufficiently small, a solution to the above problem exists in $\mathcal{Q}^{(trunc)}$, which moreover satisfies asymptotically flat bounds.

\begin{prop}
There exists $U_{min}$ sufficiently small as a function of the underlying self-similar solution, and a bounded variation solution to the system (\ref{SSESF:1:1})-(\ref{SSESF:1:5}) on $\mathcal{Q}^{(trunc)}$ assuming the data (\ref{app:trunc1})-(\ref{app:trunc4}). 

The glued solution  \begin{equation*}
\mathcal{Q}^{(glued)} \doteq (\mathcal{Q}^{(in)} \cap \{U \geq U_{min} \}) \cup (\mathcal{Q}^{(ex)} \cap \{V \leq 1, \ U \geq U_{min}\}) \cup \mathcal{Q}^{(trunc)}
\end{equation*}
is of bounded variation away from the singular point $\mathcal{O}$. Moreover, the solution in $\mathcal{Q}^{(trunc)}$ is asymptotically flat, and the spacetime contains an incomplete $\mathcal{I}^+$.
\end{prop}

\begin{proof}
By the local existence theory for the scalar field system with data posed on bifurcate null hypersurfaces, global existence on domains of the form $\mathcal{Q}^{(trunc)} \cap \{U \leq -\delta\}$ will follow from pointwise bounds on all quantities at the $C^1$ level of the solution. The bounds we are able to show on ${\partial_U}{^{(U)}\nu}, \ \partial_U^2 \phi$ will degenerate as $\delta \rightarrow 0$, but it will nevertheless follow that the solution extends in the BV class to $\{U=0\}$.

Introduce the norm 
\begin{equation}
    \mathfrak{H} \doteq \sup_{\mathcal{Q}^{(trunc)}}|{\log}{^{(V)}\lambda}| +\sup_{\mathcal{Q}^{(trunc)}}|V^2 {\partial_V}{^{(V)}\lambda}|   +\sup_{\mathcal{Q}^{(trunc)}}|m|+ \sup_{\mathcal{Q}^{(trunc)}}|V^2 \partial_V (r\phi)|+ \sup_{\mathcal{Q}^{(trunc)}}|V^3 \partial_V^2 (r\phi)|,
\end{equation}
as well as the bootstrap assumption
\begin{equation}
    \mathfrak{H} \leq 2 A,
\end{equation}
for a large constant $A$ independent of $U_{min}$. We now turn to improving this assumption, and along the way estimate the remaining double null unknowns. By the local existence theory, it follows that in a characteristic neighborhood of $(U,V) = (U_{min},1)$ we have 
\begin{equation*}
    0 \leq 1-\mu < 1, \ \ \ ^{(U)}\nu < 0, \ \ \ ^{(V)}\lambda > 0.
\end{equation*}
These bounds are required to apply the monotonicity properties of the scalar field system. In particular, the sign conditions on $^{(U)}\nu,\  ^{(V)}\lambda$ imply 
\begin{equation*}
    \inf_{\mathcal{Q}^{(trunc)}} r(U,V) \geq r(0,1) \gtrsim 1,
\end{equation*}
and the bootstrap bound on $^{(V)}\lambda$ gives
\begin{equation*}
    |r(U,V) - r(U,1)| \lesssim C(A)V.
\end{equation*}
Monotonicity goes further, and in the region of local existence the equation
\begin{equation*}
    \partial_U \log \Big(\frac{^{(V)}\lambda}{1-\mu}\Big) = \frac{r}{^{(U)}\nu}(\partial_U \phi)^2 \leq 0
\end{equation*}
implies
\begin{equation*}
    \log \Big(\frac{^{(V)}\lambda}{1-\mu}\Big)(U,V) \leq \log \Big(\frac{^{(V)}\lambda}{1-\mu}\Big)(U_{min},V) \leq C(I),
\end{equation*}
where $C(I)$ is a constant depending on initial data. In particular, applying the assumed positive lower bound on $^{(V)}\lambda$ we see
\begin{equation*}
     1-\mu \gtrsim c(A).
\end{equation*}
With quantitative bounds on $r, ^{(V)}\lambda, m$, and $1-\mu$, appealing to (\ref{SSESF:2:2}) and integrating in $V$ from $\{V=1\}$ gives
\begin{align*}
   \Big| \log \frac{^{(U)}\nu(U,V)}{^{(U)}\nu(U,1)}\Big| &= \int\limits_{(U,1)}^{(U,V)}\frac{m}{1-\mu}\frac{^{(V)}\lambda}{r^2}dV' \\
   &\lesssim C(A) \int\limits_{r(U,1)}^{r(U,V)} \frac{1}{(r')^2}dr' \lesssim C(A).
\end{align*}
Therefore 
\begin{equation*}
  c(A) \lesssim  -^{(U)}\nu(U,V) \lesssim C(A).
\end{equation*}
Integrating ${\partial_U r =} \ {^{(U)}\nu}$ in from $\{U=U_{min}\}$ and choosing $|U_{min}|$ small gives $r \sim V$.

We similarly arrive at bounds on the scalar field quantities $\phi, \ \partial_U \phi$. Applying the fundamental theorem of calculus and the bootstrap bound on $\partial_V (r\phi)$ gives
\begin{align*}
    |(r\phi)(U,V) - (r\phi)(U,1)| \lesssim C(A)\int\limits_{(U,1)}^{(U,V)}\frac{1}{(V')^2}dV' \lesssim C(A),
\end{align*}
and integrating the wave equation (\ref{SSESF:2:6}) in the outgoing direction yields
\begin{align*}
    |\partial_U (r\phi)(U,V) - \partial_U (r\phi)(U,1)| &\lesssim \int\limits_{(U,1)}^{(U,V)} \frac{\mu \ ^{(V)}\lambda|^{(U)}\nu|}{(1-\mu)r^2}|r\phi|(U,V')dV' \lesssim C(A).
\end{align*}
Using bounds on the initial data along $\{V=1\}$, we conclude
\begin{equation*}
    |\partial_U (r\phi)|(U,V) \lesssim C(A).
\end{equation*}
We next discuss how the bootstrap assumptions on ${\log} {^{(V)}\lambda}$ and $V^2 \partial_V (r\phi)$ are improved. Writing (\ref{SSESF:2:1}) as an equation for ${\log} {^{(V)}\lambda}$, estimating the right hand side, and integrating in $U$ yields
\begin{equation}
    |{\log} {^{(V)}\lambda}|(U,V) \leq |{\log} {^{(V)}\lambda}|(U_{min},V) + C(A)|U_{min}|.
\end{equation}
One can pick $A$ sufficiently large such that $\sup_V|{\log} {^{(V)}\lambda}|(U_{min},V) \leq A$. Therefore, for $|U_{min}|$ sufficiently small this estimate improves the bootstrap assumption.

Similarly, estimating the right hand side of (\ref{SSESF:2:6}) and integrating yields
\begin{align*}
    |\partial_V(r\phi)(U,V)-\partial_V(r\phi)(U_{min},V)| \lesssim C(A) \int\limits_{(U_{min},V)}^{(U,V)} \frac{1}{V^2}dU' \lesssim C(A)|U_{min}|V^{-2}.
\end{align*}
Using $|\partial_V(r\phi)(U_{min},V)| \lesssim V^{-2},$ it is clear that choosing $|U_{min}|$ small enough improves the bootstrap assumption.

It remains to estimate ${\partial_U}{^{(U)}\nu}, \ \partial_U^2 \phi$, and improve the bootstrap assumptions on $m, \ {\partial_V}{^{(V)}\lambda}$, and $V^3 \partial_V^2 (r\phi)$. We start with preliminary estimates on $\partial_U \phi, \ \partial_V \phi$. Having already estimated $\partial_U(r\phi)$, $\partial_V(r\phi)$, and $r\phi$, it follows that 
\begin{equation*}
    |\partial_V \phi| \lesssim \frac{1}{r}|\partial_V (r\phi)|  + \frac{^{(V)}\lambda}{r^2}|r \phi| \lesssim C(A)V^{-2},
\end{equation*}
\begin{equation*}
    |\partial_U \phi| \lesssim \frac{1}{r}|\partial_U (r\phi)| + \frac{|^{(U)}\nu|}{r^2}|r\phi| \lesssim C(A) V^{-1}.
\end{equation*}
In conjunction with bounds on $r, ^{(U)}\nu,$ and $1-\mu$ we integrate (\ref{SSESF:2:4}) to give 
\begin{equation*}
    |m(U,V)-m(U_{min},V)| \lesssim C(A)|U_{min}|,
\end{equation*}
and thus 
\begin{equation*}
    |m(U,V)| \lesssim C(I) + C(A)|U_{min}|,
\end{equation*}
improving the bound on $m$ for $U_{min}$ small. Differentiating (\ref{SSESF:2:1}) in $V$ and estimating yields
\begin{align*}
    |{\partial_U \partial_V \log} {^{(V)}\lambda}| &= |\partial_V \Big(\frac{{\mu} {^{(U)}\nu}}{(1-\mu)r} \Big)| \lesssim C(A)V^{-3},
\end{align*}
where we have used (\ref{SSESF:2:2}), (\ref{SSESF:2:7}) to estimate terms involving $V$ derivatives of $^{(U)}\nu, \ \mu$ respectively. Integrating in the $U$ direction from $\{U=U_{min}\}$ and choosing $U_{min}$ sufficiently small improves the bootstrap assumption.

Similarly, differentiating (\ref{SSESF:2:6}) in $V$ and estimating yields
\begin{align*}
    |\partial_U \partial_V^2 (r\phi)| &= |\partial_V\Big(\frac{\mu \ ^{(V)}\lambda \ ^{(U)}\nu}{(1-\mu)r^2}(r\phi)\Big)| \lesssim C(A)V^{-3}.
\end{align*}
Integrating in $U$ improves the bootstrap assumption.

Differentiating (\ref{SSESF:2:2}) in $U$ and estimating gives
\begin{align*}
    |{\partial_V \partial_U \log}{^{(U)}\nu}| &= |\partial_U \Big(\frac{\mu ^{(V)}\lambda}{(1-\mu)r} \Big)| \lesssim C(A)V^{-2}.
\end{align*}
Integrating in $V$ from $\{V=1\}$ implies 
\begin{align*}
    |{\partial_U \log} {^{(U)}\nu}| &\lesssim |{\partial_V \log } {^{(U)}\nu}|(1,U) + C(A) \lesssim C(A).
\end{align*}
Finally, differentiating (\ref{SSESF:2:6}) in $U$ and estimating gives
\begin{align*}
    |\partial_V \partial_U^2 (r\phi)| &= |\partial_U \Big(\frac{\mu ^{(V)}\lambda ^{(U)}\nu}{(1-\mu)r^2}(r\phi)\Big)| \lesssim C(A)V^{-3}.
\end{align*}
Integrating in $V$ from data implies the estimate
\begin{equation*}
    |\partial_U^2(r\phi)| \lesssim |U|^{-1+pk^2},
\end{equation*}
where the singular $U$ dependence arises due to the behavior of the data $\partial_U^2(r\phi)(U,1)$. All quantities at the $C^1$ level have been estimated, implying that for $U_{min}$ sufficiently small the solution extends to $\mathcal{Q}^{(trunc)}$.

It remains to check the global features of the glued solution. It is clear that the outgoing data for $^{(V)}\lambda, \ \partial_V(r\phi)$ along $\{U=U_{min}\}$ is $C^1$, and thus that the glued solution is $C^1$ across $\{V=1\}$ to the past of $\{U=0\}$. Moreover, the solution extends to $\{U=0\}$ as a solution in the BV class.

To see asymptotic flatness, observe that the estimates proven thus far imply that along any constant $U$ surface with $U_{min} \leq U < 0$,
\begin{enumerate}
    \item $\lim_{V \rightarrow \infty} r(U,V) = \infty$.
    \item $\lim_{V\rightarrow \infty} (r\phi)(U,V) $ exists and is finite.
    \item $\lim_{V \rightarrow \infty} m(U,V)$ exists and is finite.
\end{enumerate}
It follows from $|{\partial_V}{^{(U)}\nu}| \lesssim V^{-2}$ that for all $U \in [U_{min},0)$, 
\begin{equation*}
   ^{(U)}\nu_{\infty}(U) \doteq \lim_{V \rightarrow \infty}\nu(U,V) 
\end{equation*}
exists, is finite, satisfies uniform bounds
\begin{equation*}
    0 < c(I) \leq -^{(U)}\nu_{\infty}(U) \leq C(I),
\end{equation*}
for constants $c(I), \ C(I)$ depending on initial data, and converges at a rate
\begin{equation*}
    |^{(U)}\nu(U,V) - ^{(U)}\nu_{\infty}(U)| \lesssim V^{-1}.
\end{equation*}
It similarly follows that 
\begin{equation*}
    |^{(V)}\lambda(U,V) - \frac{1}{2}| \lesssim V^{-2},
\end{equation*}
and 
\begin{equation}
    \label{eq:appasympflatest}
    |^{(U),(V)}\Omega^2(U,V) - 2\ ^{(U)}\nu_{\infty}(U)| \lesssim V^{-1}.
\end{equation}
As discussed in Section \ref{subsec:nakedsingdfn}, incompleteness of null infinity reduces to a statement about the proper time elapsed by ingoing null geodesics near null infinity. Define the vector field 
\begin{equation}
    {X(U,V) =} \ {^{(U),(V)}\Omega^2(U,V)}\frac{\partial}{\partial U}.
\end{equation}
A direct computation using the Christoffel symbols of the connection associated to a Lorentzian metric $-^{(U),(V)}\Omega^2(U,V)dUdV$ shows $\nabla_{X} X = 0$, i.e. $X$ is a parallel vector field. $X$ is the tangent vector to affinely parameterized ingoing null geodesics, and incompleteness of $\mathcal{I}^+$ is thus equivalent to the existence of a constant $C < \infty$ such that 
\begin{equation}
    \limsup_{V \rightarrow \infty}\int\limits_{(U_{min},V)}^{(0,V)}|^{(U),(V)}\Omega^2(U,V)|(U',V) dU' \leq C.
\end{equation}
This is a direct calculation using (\ref{eq:appasympflatest}).
\end{proof}.